\documentclass{theoretics}

\usepackage[font=normalsize]{subcaption}
\usepackage{comment}

\usepackage[T1]{fontenc}
\usepackage{mathtools}
\usepackage{hyphenat}
\usepackage[disable]{todonotes}
\usepackage{marginnote}

\usepackage{diagbox}
\usepackage{array} 
\usepackage{makecell}

\usepackage{pifont}
\usepackage{bbm}

\usepackage[super]{nth}

\usepackage{mathcommand}
\usepackage[notion, quotation, electronic]{knowledge}

\newenvironment{longdescription} 
{\begin{description}[style=unboxed]}
	{\end{description}}

\usepackage{tikz}
\usetikzlibrary{automata, snakes, positioning, arrows,arrows.meta,calc,decorations,decorations.markings,math,shapes.geometric}
\tikzset{
	>=stealth',
	-={stealth',ultra thick,scale=3} 
	node distance=1cm, 
	every state/.style={thick}, 
	initial text=$ $, 
}

\let\ab\allowbreak
\mathchardef\hyphen=45 

\newcommand\restr[2]{{
		\left.\kern-\nulldelimiterspace 
		#1 
		\littletaller 
		\right|_{#2} 
}}

\newcommand{\littletaller}{\mathchoice{\vphantom{\big|}}{}{}{}}

\newcommand{\ts}{\textsuperscript}

\definecolor{Green2}{HTML}{3EA514}
\definecolor{Red2}{HTML}{FF0400}
\definecolor{Orange2}{HTML}{E6670A}
\definecolor{Violet2}{HTML}{CE1ff9}
\definecolor{Green3}{HTML}{51A342}
\definecolor{Green4}{HTML}{45A229}
\definecolor{Navy}{HTML}{2943A2}

\DeclareMathAlphabet{\mathpzc}{OT1}{pzc}{m}{it}

\newrobustcmd{\NN}{\mathbb{N}}
\newrobustcmd{\ZZ}{\mathbb{Z}}
\newrobustcmd{\QQ}{\mathbb{Q}}
\newrobustcmd{\RR}{\mathbb{R}}
\newrobustcmd{\CC}{\mathbb{C}}
\newrobustcmd{\WW}{\mathbb{W}}

\newrobustcmd{\I}{\mathcal{I}}
\newrobustcmd{\F}{\mathcal{F}}
\newrobustcmd{\G}{\mathcal{G}}

\newrobustcmd{\M}{\mathcal{M}}
\newrobustcmd{\Q}{\mathcal{Q}}
\newrobustcmd{\C}{\mathcal{C}}
\newrobustcmd{\A}{\mathcal{A}}
\newrobustcmd{\B}{\mathcal{B}}
\newrobustcmd{\Z}{\mathcal{Z}}
\newrobustcmd{\R}{\mathcal{R}}
\newrobustcmd{\T}{\mathcal{T}}
\newrobustcmd{\W}{\mathcal{W}}

\renewcommand{\P}{\mathcal{P}}

\renewcommand{\O}{\mathcal{O}}

\renewcommand{\S}{\mathcal{S}}

\newrobustcmd{\kk}{\kappa}
\newrobustcmd{\uu}{\upsilon}
\newrobustcmd{\dd}{\delta}

\renewcommand{\ss}{\sigma}

\newrobustcmd{\rr}{\rho}

\renewcommand{\aa}{\alpha}

\newrobustcmd{\bb}{\beta}

\newrobustcmd{\oo}{\omega}
\newrobustcmd{\pp}{\varphi}

\renewcommand{\gg}{\gamma}
\newrobustcmd{\ee}{\varepsilon}

\renewcommand{\SS}{\Sigma}
\newrobustcmd{\GG}{\Gamma}
\newrobustcmd{\DD}{\Delta}

\knowledgerenewmathcommand\nu{\cmdkl{\LaTeXnu}}
\knowledgenewmathcommand\nuAcd{\cmdkl{\LaTeXnu}}
\knowledgerenewmathcommand\eta{\cmdkl{\LaTeXeta}}

\newrobustcmd\inv[1]{#1^{-1}}

\newcommand{\tand}{\text{ and }}
\newcommand{\tor}{\text{ or }}

\newcommand{\tif}{\text{ if }}
\newcommand{\tow}{\text{ otherwise}}
\newcommand{\tst}{\text{ such that }}

\newrobustcmd{\pow}[1]{2^{#1}}
\newrobustcmd{\powplus}[1]{\kl[\powplus]{2^{#1}_{+}}}
\knowledge{\powplus}{notion}
\newrobustcmd\SigmaInfty{\kl[\SigmaInfty]{\Sigma^\infty}}
\knowledge{\SigmaInfty}[E'^\infty | E^\infty | \GammaInfty]{notion}
\newrobustcmd\GammaInfty{\kl[\SigmaInfty]{\Gamma^\infty}}

\newrobustcmd{\restSubsets}[2]{\kl[\restSubsets]{\restr{#1}{#2}}}
\knowledge{\restSubsets}{notion}

\newrobustcmd{\partialF}{\kl[\partialF]{\rightharpoonup}}
\knowledge{\partialF}{notion}

\newrobustcmd{\complTS}[1]{\kl[\complTS]{\overline{#1}}}
\knowledge{\complTS}{notion}
\newrobustcmd{\complSet}[1]{\kl[\complSet]{\overline{#1}}}
\knowledge{\complSet}{notion}

\newrobustcmd{\prefix}{\mathrel{\kl[\prefix]{\sqsubseteq}}}
\knowledge{\prefix}{notion}
\newrobustcmd{\nprefix}{\mathrel{\kl[\nprefix]\sqsubset}}
\knowledge{\nprefix}{notion}

\newrobustcmd{\first}{\kl[\first]{\mathsf{First}}}
\knowledge\first{notion}
\newrobustcmd{\last}{\kl[\last]{\mathsf{Last}}}
\knowledge\last{notion}

\newrobustcmd{\minf}{\kl[\minf]{\mathsf{Inf}}}
\knowledge{\minf}{notion}
\newrobustcmd{\mocc}{\kl[\mocc]{\mathsf{Occ}}}
\knowledge{\mocc}{notion}

\newcommand{\re}[1]{\xrightarrow{#1}}
\newcommand{\rp}[1]{\overset{#1}{\rightsquigarrow}}
\usetikzlibrary{calc,decorations.pathmorphing,shapes}

\newcounter{sarrow}
\newcommand\lrp[1]{%
	\stepcounter{sarrow}%
	\mathrel{\begin{tikzpicture}[baseline= {( $ (current bounding box.south) + (0,-0.5ex) $ )}]
			\node[inner sep=.5ex] (\thesarrow) {$\scriptstyle #1$};
			\path[draw,<-,decorate,
			decoration={zigzag,amplitude=0.7pt,segment length=1.2mm,pre=lineto,pre length=4pt}] 
			(\thesarrow.south east) -- (\thesarrow.south west);
	\end{tikzpicture}}%
}

\newcommand\lrpE{\lrp{\phantom{w.}}}

\newrobustcmd{\TS}{\mathcal{T\hspace{-1.1mm}S}}

\newrobustcmd{\Graph}[1]{\kl[\Graph]G_{#1}}
\knowledge\Graph{notion}

\newrobustcmd{\underlyingGraph}[1]{\kl[\underlyingGraph]G_{#1}}
\knowledge\underlyingGraph{notion}

\newrobustcmd{\macc}[1]{\kl[\macc]{\mathrm{Acc}}_{#1}}
\knowledge\macc{notion}

\newrobustcmd{\msource}{\kl[\msource]{\mathsf{Source}}}
\knowledge{\msource}[source]{notion}
\newrobustcmd{\mtarget}{\kl[\mtarget]{\mathsf{Target}}}
\knowledge{\mtarget}[target]{notion}

\newrobustcmd{\msourcePath}{\kl[\msourcePath]{\mathsf{Source}}}
\knowledge{\msourcePath}[source@path]{notion}
\newrobustcmd{\mtargetPath}{\kl[\mtargetPath]{\mathsf{Target}}}
\knowledge{\mtargetPath}[target@path]{notion}

\newrobustcmd{\PathSet}[2]{\kl[\PathSet]{\mathpzc{Path}_{#2}(#1)}}
\knowledge\PathSet{notion}
\newrobustcmd{\PathSetFin}[2]{\kl[\PathSetFin]{\mathpzc{Path}^{\mathsf{fin}}_{#2}(#1)}}
\knowledge\PathSetFin{notion}
\newrobustcmd{\Runs}[1]{\kl[\Runs]{\mathpzc{Run}(#1)}}
\knowledge\Runs{notion}
\newrobustcmd{\RunsFin}[1]{\kl[\RunsFin]{\mathpzc{Run}^{\mathsf{fin}}(#1)}}
\knowledge\RunsFin{notion}
\newrobustcmd{\RunsInfty}[1]{\kl[\RunsInfty]{\mathpzc{Run}^{\infty}(#1)}}
\knowledge\RunsInfty{notion}
\newrobustcmd{\PathSetInfty}[2]{\kl[\PathSetInfty]{\mathpzc{Path}^{\infty}_{#2}(#1)}}
\knowledge\PathSetInfty{notion}

\newrobustcmd{\mout}{\kl[\mout]{\mathsf{Out}}}
\knowledge\mout{notion}
\newrobustcmd{\mIn}{\kl[\mIn]{\mathsf{In}}}
\knowledge\mIn{notion}

\newrobustcmd{\initialTS}[2]{\kl[\initialTS]{#1_{#2}}}
\knowledge\initialTS{notion}

\newrobustcmd{\size}[1]{\kl[\size]{|#1|}}
\knowledge\size{notion}

\newrobustcmd{\Vrec}{V_\mathrm{rec}}
\newrobustcmd{\Vtrans}{V_\mathrm{trans}}

\newrobustcmd{\transAut}[1]{\kl[\transAut]{\delta_{#1}}}
\knowledge\transAut{notion}

\newrobustcmd{\compositionAut}{\mathbin{\kl[\compositionAut]{\ltimes}}}
\knowledge\compositionAut{notion}

\newrobustcmd{\edgesProduct}{\kl[\edgesProduct]E^\ltimes}
\knowledge\edgesProduct{notion}

\newrobustcmd\Lang[1]{\kl[\Lang]{\mathcal{L}(#1)}}
\knowledge\Lang{notion}

\newrobustcmd{\prodMem}[1]{\mathbin{\kl[\prodMem]{\lhd_{#1}}}}
\knowledge\prodMem{notion}

\newrobustcmd\LangTS[1]{\kl[\LangTS]{\mathcal{L}_{\mathpzc{Runs}}(#1)}}
\knowledge\LangTS{notion}

\newrobustcmd{\piAut}{\kl[\piAut]\pi_{\A}}
\knowledge\piAut{notion}
\newrobustcmd{\lPlayers}{\kl[\lPlayers]l_{\mathsf{Players}}}
\knowledge\lPlayers{notion}
\newrobustcmd{\winRegion}[2]{\kl[\winRegion]{\W_{#1}(#2)}}
\knowledge\winRegion{notion}
\newrobustcmd{\strat}{\mathsf{strat}}

\newcommand{\Eve}{\mathrm{Eve}}
\newcommand{\Adam}{\mathrm{Adam}}

\newrobustcmd{\VEve}{\kl[\VEve]{V_{\mathrm{Eve}}}}
\knowledge\VEve[V_P|\widetilde {V}_\Eve|\tilde {V}_\Eve ]{notion}
\newrobustcmd{\VAdam}{\kl[\VAdam]{V_{\mathrm{Adam}}}}
\knowledge\VAdam[\widetilde {V}_\Adam ]{notion}

\newrobustcmd{\choicestrat}{\mathsf{choice}_{st}}

\newrobustcmd{\letterGame}[1]{\kl[\letterGame]{\G_{#1}}}

\knowledge\letterGame{notion}

\newrobustcmd{\attr}[2]{\kl[\attr]{\mathsf{Attr}_{#1}(#2)}}
\knowledge{\attr}{notion}

\newrobustcmd{\attrDec}[1]{\kl[\attrDec]{\mathcal{D}_{#1}}}
\knowledge{\attrDec}{notion}

\newrobustcmd{\orderAttr}{\mathrel{\kl[\orderAttr]{<_{\mathcal{D}}}}}
\knowledge{\orderAttr}{notion}

\newrobustcmd{\transMem}{\kl[\transMem]{\mu}}
\knowledge\transMem{notion}

\newrobustcmd{\nextmove}{\kl[\nextmove]{\sigma}}
\knowledge\nextmove{notion}

\newrobustcmd{\nextmoveResolver}{\kl[\nextmoveResolver]{\sigma}}
\knowledge\nextmoveResolver{notion}

\newrobustcmd{\projLetterGame}[1]{\kl[\projLetterGame]{\mathsf{aut}_{#1}}}
\knowledge\projLetterGame{notion}

\newrobustcmd{\piLetterGame}{\kl[\piLetterGame]{\pi_{\G}}}
\knowledge\piLetterGame{notion}

\newrobustcmd{\subgame}[1]{\kl[\subgame]{\G_{\A}(#1)}}
\knowledge\subgame{notion}

\newrobustcmd{\Muller}[1]{\kl[\Muller]{\textsf{Muller}(#1)}}
\newrobustcmd{\MullerC}[2]{\kl[\MullerC]{\textsf{Muller}_{#2}(#1)}}
\knowledge{\Muller}[\MullerC]{notion}
\newrobustcmd{\Rabin}[1]{\kl[\Rabin]{\textsf{Rabin}(#1)}}
\newrobustcmd{\RabinC}[2]{\kl[\RabinC]{\textsf{Rabin}_{#2}(#1)}}
\knowledge{\Rabin}[\RabinC]{notion}
\newrobustcmd{\Streett}[1]{\kl[\Streett]{\textsf{Streett}(#1)}}
\newrobustcmd{\StreettC}[2]{\kl[\StreettC]{\textsf{Streett}_{#2}(#1)}}
\knowledge{\Streett}[\StreettC]{notion}
\newrobustcmd{\parity}{\kl[\parity]{\textsf{parity}}}
\knowledge{\parity}{notion}
\newrobustcmd{\Buchi}[1]{\kl[\Buchi]{\textsf{B\"{u}chi}(#1)}}
\newrobustcmd{\BuchiC}[2]{\kl[\BuchiC]{\textsf{B\"uchi}_{#2}(#1)}}
\knowledge{\Buchi}[\BuchiC]{notion}

\newrobustcmd{\genBuchi}[1]{\kl[\genBuchi]{\textsf{genBüchi}(#1)}}
\newrobustcmd{\genBuchiC}[2]{\kl[\genBuchiC]{\textsf{genBüchi}_{#2}(#1)}}
\knowledge{\genBuchi}[\genBuchiC]{notion}

\newrobustcmd{\coBuchi}[1]{\kl[\coBuchi]{\textsf{coBüchi}(#1)}}
\newrobustcmd{\coBuchiC}[2]{\kl[\coBuchiC]{\textsf{coBüchi}_{#2}(#1)}}
\knowledge{\coBuchi}[\coBuchiC]{notion}
\newrobustcmd{\gencoBuchi}[1]{\kl[\gencoBuchi]{\textsf{genCoBüchi}(#1)}}
\newrobustcmd{\gencoBuchiC}[2]{\kl[\gencoBuchiC]{\textsf{genCoBüchi}_{#2}(#1)}}
\knowledge{\gencoBuchi}[\gencoBuchiC]{notion}

\newrobustcmd{\Weak}[1]{\kl[\Weak]{\textsf{Weak}_{#1}}}
\knowledge{\Weak}{notion}

\newrobustcmd{\WeakIndex}[1]{\kl[\WeakIndex]{\textsf{Weak}_{#1}}}
\knowledge{\WeakIndex}{notion}

\newrobustcmd{\MullerFamily}[1]{\kl[\MullerFamily]{\F_{#1}}}
\knowledge{\MullerFamily}{notion}

\newrobustcmd{\impliesMuller}[2]{\mathbin{\kl[\impliesMuller]{#1 \rightarrow #2}}}
\knowledge{\impliesMuller}{notion}

\newrobustcmd{\equivTrans}{\mathrel{\kl[\equivTrans]{\simeq}}}
\knowledge{\equivTrans}{notion}
\newrobustcmd{\equivCond}[1]{\mathrel{\kl[\equivCond]{\simeq}_{#1}}}
\knowledge{\equivCond}{notion}

\newrobustcmd{\cycles}[1]{\kl[\cycles]{\mathpzc{Cycles}(#1)}}
\knowledge\cycles{notion}
\newrobustcmd{\cyclesState}[2]{\kl[\cyclesState]{\mathpzc{Cycles}_{#2}(#1)}}
\knowledge\cyclesState{notion}
\newrobustcmd{\states}[1]{\kl[\states]{\mathsf{States}(#1)}}
\knowledge\states{notion}
\newrobustcmd{\localMuller}[2]{\kl[\localMuller]{\mathsf{LocalMuller}_{#2}(#1)}}
\knowledge\localMuller{notion}

\newrobustcmd{\rInit}{\kl[\rInit]{r_{\mathsf{Init}}}}
\knowledge\rInit[\trInit]{notion}
\newrobustcmd{\trInit}{\kl[\rInit]{\tilde{r}_{\mathsf{Init}}}}

\newrobustcmd{\rRuns}{\kl[\rRuns]r_{\mathpzc{Runs}}}
\knowledge\rRuns[\rRunsOption]{notion}
\newrobustcmd{\rRunsOption}[1]{\kl[\rRunsOption]r_{#1,\mathpzc{Runs}}}

\newrobustcmd{\ppRuns}{\kl[\ppRuns]\pp_{\mathpzc{Runs}}}
\newrobustcmd{\ppRunsP}[1]{\kl[\ppRuns]{#1}_{\mathpzc{Runs}}}
\knowledge\ppRuns[\ppRunsP]{notion}
\newrobustcmd{\Id}[1]{\kl[\Id]{\mathit{Id}_{#1}}}
\knowledge\Id{notion}

\newrobustcmd{\autMorphism}[1]{\kl[\autMorphism]{\mathcal{A}_{#1}}}
\knowledge\autMorphism{notion}

\newrobustcmd{\ancestor}{\mathrel{\kl[\ancestor]{\preceq}}}
\knowledge\ancestor{notion}
\newrobustcmd{\descendant}{\mathrel{\kl[\descendant]{\succeq}}}
\knowledge\descendant{notion}
\newrobustcmd{\roundnodes}{\kl[\roundnodes]{N_\bigcirc}}
\knowledge\roundnodes{notion}
\newrobustcmd{\squarenodes}{\kl[\squarenodes]{N_\Box}}
\knowledge\squarenodes{notion}
\newrobustcmd{\leaves}{\kl[\leaves]{\mathsf{Leaves}}}
\knowledge\leaves{notion}
\newrobustcmd{\children}{\kl[\children]{\mathsf{Children}}}
\knowledge\children{notion}
\newrobustcmd{\nextChild}{\kl[\nextChild]{\mathsf{Next}}}
\knowledge\nextChild{notion}
\newrobustcmd{\jump}{\kl[\jump]{\mathsf{Jump}}}
\knowledge\jump[intermediate node]{notion}
\newrobustcmd{\depth}{\kl[\depth]{\mathsf{Depth}}}
\knowledge\depth{notion}
\newrobustcmd{\orderTree}[1]{\kl[\orderTree]{\leq_{#1}}}
\knowledge\orderTree{notion}

\newrobustcmd\zielonkaTree[1]{\kl[\zielonkaTree]{\mathcal{Z}_{#1}}}
\knowledge\zielonkaTree{notion}
\newrobustcmd{\supp}{\kl[\supp]{\mathsf{Supp}}}
\knowledge\supp{notion}
\newrobustcmd{\parityNodes}{\kl[\parityNodes]{p_\Z}}
\knowledge\parityNodes{notion}
\newrobustcmd{\minparityZ}[1]{\kl[\minparityZ]{\min_{#1}}}
\knowledge\minparityZ{notion}
\newrobustcmd{\maxparityZ}[1]{\kl[\maxparityZ]{\max_{#1}}}
\knowledge\maxparityZ{notion}
\newrobustcmd{\memTree}[1]{\kl[\memTree]{\mathsf{rbw}(#1)}}
\knowledge\memTree{notion}

\newrobustcmd\zielonkaAutomaton[1]{\kl[\zielonkaAutomaton]{\mathcal{A}^{\mathsf{parity}}_{\mathcal{Z}_{#1}}}}
\knowledge\zielonkaAutomaton{notion}

\newrobustcmd\zielonkaHDAutomaton[1]{\kl[\zielonkaHDAutomaton]{\mathcal{A}^{\mathsf{Rabin}}_{\mathcal{Z}_{#1}}}}
\knowledge\zielonkaHDAutomaton{notion}
\newrobustcmd{\sizeHDRabin}{\memTree{\zielonkaTree{\F}{\SS}}}

\newrobustcmd\acd[1]{\kl[\acd]{\mathcal{ACD}_{#1}}}
\knowledge\acd{notion}
\newcommand\acdNoP{\mathcal{ACD}}
\newrobustcmd\altTree[1]{\kl[\altTree]{\mathsf{AltTree}(#1)}}
\knowledge\altTree{notion}
\newrobustcmd\nuStates{\kl[\nuStates]{\nu_{\mathsf{States}}}}
\knowledge\nuStates{notion}
\newrobustcmd\nodesAcdCycle[1]{\kl[\nodesAcdCycle]{N_{#1}}}
\knowledge\nodesAcdCycle{notion}
\newrobustcmd\treeVertex[1]{\kl[\treeVertex]{\mathpzc{T}_{#1}}}
\knowledge\treeVertex{notion}
\newrobustcmd\nodesTreeVertex[1]{\kl[\nodesTreeVertex]{N_{#1}}}
\knowledge\nodesTreeVertex{notion}

\newrobustcmd\acdVertex[2]{\kl[\acdVertex]{\mathcal{ACD}_{(#1,#2)}}}
\knowledge\acdVertex{notion}

\knowledgenewrobustcmd\roundnodesv{N_{v,\bigcirc}}
\knowledgenewrobustcmd\squarenodesv{N_{v,\Box}}

\newrobustcmd\autCyclePreimage[1]{\kl[\autCyclePreimage]{\A_{(\inv{\pp},#1)}}}
\knowledge\autCyclePreimage{notion}

\newrobustcmd\unfold{\kl[\unfold]{\mathsf{Unfold}}}
\knowledge\unfold{notion}

\newrobustcmd\nodesAcd[1]{\kl[\nodesAcd]{\mathsf{Nodes}(\mathcal{ACD}_{#1})}}
\knowledge\nodesAcd{notion}
\newrobustcmd\nodesAcdRound[1]{\kl[\nodesAcdRound]{\mathsf{Nodes}_\bigcirc(\mathcal{ACD}_{#1})}}
\knowledge\nodesAcdRound{notion}
\newrobustcmd\nodesAcdSquare[1]{\kl[\nodesAcdSquare]{\mathsf{Nodes}_\Box(\mathcal{ACD}_{#1})}}
\knowledge\nodesAcdSquare{notion}

\newrobustcmd{\leavesAcd}{\kl[\leavesAcd]{\mathsf{Leaves}(\acd{\TS}}}
\knowledge\leavesAcd{notion}

\newrobustcmd{\suppAcd}{\kl[\suppAcd]{\mathsf{Supp}}}
\knowledge\suppAcd{notion}

\newrobustcmd{\parityNodesAcd}{\kl[\parityNodesAcd]p_{\mathcal{ACD}}}
\knowledge\parityNodesAcd{notion}
\newrobustcmd{\minparityAcd}[1]{\kl[\minparityAcd]{{\min}_{#1}}}
\knowledge\minparityAcd{notion}
\newrobustcmd{\maxparityAcd}[1]{\kl[\maxparityAcd]{{\max}_{#1}}}
\knowledge\maxparityAcd{notion}

\newrobustcmd{\coloursNodesAcd}{\kl[\coloursNodesAcd]\gg_{\mathcal{ACD}}}
\knowledge\coloursNodesAcd{notion}

\newrobustcmd{\minparityAcdVertex}[2]{\kl[\minparityAcdVertex]{{\min}_{(#1,#2)}}}
\knowledge\minparityAcdVertex{notion}
\newrobustcmd{\maxparityAcdVertex}[2]{\kl[\maxparityAcdVertex]{{\max}_{(#1,#2)}}}
\knowledge\maxparityAcdVertex{notion}

\newrobustcmd\acdParityTransform[1]{\kl[\acdParityTransform]{\mathsf{ACD}_{\mathsf{parity}}(#1)}}
\knowledge\acdParityTransform{notion}

\newrobustcmd\acdRabinTransform[1]{\kl[\acdRabinTransform]{\mathsf{ACD}_{\mathsf{Rabin}}(#1)}}
\knowledge\acdRabinTransform{notion}
\newrobustcmd\acdRabinTransformGFG[1]{\kl[\acdRabinTransformGFG]{\mathsf{ACD}_{\mathsf{parity}}^{\mathsf{game}}(#1)}}
\knowledge\acdRabinTransformGFG{notion}
	
\newrobustcmd{\ppAcd}{\kl[\ppAcd]{\pp_{\acdNoP}}}
\knowledge\ppAcd{notion}

\newrobustcmd{\VAsucc}{\kl[\VAsucc]{V_{\mathrm{A}\hyphen\mathrm{succ}}}}
\knowledge\VAsucc{notion}
\newrobustcmd{\Vnormal}{\kl[\Vnormal]{V_{\mathrm{normal}}}}
\knowledge\Vnormal{notion}
\newrobustcmd{\Apred}{\kl[\Apred]{\mathsf{pred}}}
\knowledge\Apred{notion}

\newrobustcmd{\mletters}{\kl[\mletters]{\mathtt{Letters}}}
\knowledge\mletters{notion}
\newrobustcmd{\minCol}{\kl[\minCol]{\mathtt{MinColour}}}
\knowledge\minCol{notion}
\newrobustcmd{\AlternatingSets}{\kl[\AlternatingSets]{\mathtt{AlternatingSets}}}
\knowledge\AlternatingSets{notion}
\newrobustcmd{\SCCDec}{\kl[\SCCDec]{\mathtt{SCC\hyphen Decomposition}}}
\knowledge\SCCDec{notion}
\newrobustcmd{\leavingLetters}{\kl[\leavingLetters]{\mathtt{LeavingLetters}}}
\knowledge\leavingLetters{notion}
\newrobustcmd{\pop}{\kl[\pop]{\mathtt{pop}}}
\knowledge\pop{notion}
\newrobustcmd{\push}{\kl[\push]{\mathtt{push}}}
\knowledge\push{notion}
\newrobustcmd{\maxInclusion}{\kl[\maxInclusion]{\mathtt{MaxInclusion}}}
\knowledge\maxInclusion{notion}

\newcommand{\NP}{\ensuremath{\mathsf{NP}}}
\newcommand{\PTime}{\ensuremath{\mathsf{P}}}
\newcommand{\coNP}{\ensuremath{\mathsf{co}\hyphen\mathsf{NP}}}
\newcommand{\NPcomplete}{\ensuremath{\mathsf{NP}\hyphen\mathrm{complete}}}
\newcommand{\PSPACE}{\ensuremath{\mathsf{PSPACE}}}

\definecolor{Blue Sapphire}{HTML}{005f73} 
\definecolor{Blue Matt}{HTML}{094c7b} 
\definecolor{Blue Dark}{HTML}{1B4C6E} 

\definecolor{Gamboge}{HTML}{ee9b00}
\definecolor{Ruby Red}{HTML}{9b2226}

\definecolor{Dark Ruby Red}{HTML}{580507}
\definecolor{Dark Blue Sapphire}{HTML}{053641}
\definecolor{Dark Gamboge}{HTML}{be7c00}

\definecolor{Blue Marine}{HTML}{022687}

\IfKnowledgePaperModeTF{
}{
	\knowledgestyle{intro notion}{color={Dark Ruby Red}, emphasize}
	\knowledgestyle{notion}{color={Dark Blue Sapphire}}
	\hypersetup{
		colorlinks=true,
		breaklinks=true,
		linkcolor={}, 
		citecolor={}, 
		filecolor={Blue Marine}, 
		urlcolor={Blue Marine},
	}
}

\knowledge{notion}
| notion@example
|definition@example

\knowledge{notion}
| $\omega $-word
| $\omega $-words
| infinite words
| infinite word

\knowledge{notion}
| word
| words

\knowledge{notion}
 | empty word

\knowledge{notion}
| prefix

\knowledge{notion}
  | partial mapping

\knowledge{notion}
  | state complexity

\knowledge{notion}
| directed graph
| graph
| graphs

\knowledge{notion}
| path
| paths

\knowledge{notion}
| strongly connected

\knowledge{notion}
| subgraph

\knowledge{notion}
| strongly connected component
| SCC
| strongly connected components
| SCCs

\knowledge{notion}
| final@SCC
| final SCC

\knowledge{notion}
| recurrent

\knowledge{notion}
| transient
| transient vertex

\knowledge{notion}
  | labelled graph
  | labelled@graph
  | labelled graphs
  | labelled
  | (labelled)
  | labelling
  
\knowledge{notion}
| pointed graph
| pointed graphs
| (pointed) graphs

\knowledge{notion}
 | transition system
 | transition systems
 | transitions systems
 | Transition system
 
\knowledge{notion}
 | TS

\knowledge{notion}
| underlying graph
| underlying graphs
| game graphs

\knowledge{notion}
| initial vertex
| initial vertices
| initial state
| initial@state
| initial states
| initial

\knowledge{notion}
  | uncoloured edges
  | $\ee $-edges
  | uncoloured
  | coloured

\knowledge{notion}
  | colours
  | output colours
  | output colour
  | sets of colours
  | colour
  | colours@outAut

\knowledge{notion}
 | edge-colouring 
 | colouring function
 | colouring function@TS
 | colouring@TS
 | colouring

\knowledge{notion}
| output
| output@run

\knowledge{notion}
 | accessible@vertex
 | reachable
 | accessible
 
\knowledge{notion}
 | accessible from
 | accessible@fromVertex

\knowledge{notion}
 | accessible@set
 | accessible@sets

\knowledge{notion}
| accessible part

\knowledge{notion}
| accessible part of $\TS $ from $v$
| accessible part of $\P $ from $v_P$
| accessible part from a vertex
| accessible part of $\widetilde {\TS }$ from $\tilde {v}$

\knowledge{notion}
| run@transSys
| infinite run
| infinite runs
| finite runs
| finite run
| runs
| run
| run from $v_0$

\knowledge{notion}
| accepting@run
| accepting
| acceptance@run
| accepting run
| accepting runs
| acceptance@runs

\knowledge{notion}
| rejecting@run
| rejecting run

\knowledge{notion}
| labelled transition system
| labelled transition systems
| (labelled) transition system
| labelling with input letters
| (labelled) TS

\knowledge{notion}
| vertex-labelling
| vertex-labellings
| vertex-labelled

\knowledge{notion}
| edge-labelled
| edge-labelled@graph

\knowledge{notion}
| size@transSys

\knowledge{notion}
  | product construction
  | composition
  | product automaton
  | product@aut
  | composing
  | composition of automata
  | compositions
  
\knowledge{notion}
  | suitable for transformations
  | suitability for transformations

\knowledge{notion}
| acceptance set
| acceptance sets
| winning condition

\knowledge{notion}
| acceptance condition
| acceptance conditions
| conditions
| condition
| condition@accep
| condition@acc
| conditions@acc

\knowledge{notion}
  | prefix-independent
  | prefix-independence

\knowledge{notion}
| automaton
| (non-deterministic) 
| automata
| (non-deterministic) automaton
| non-deterministic automata
| non-deterministic
| non-deterministic automaton
| non-determinism
| determinism

\knowledge{notion}
| input letters
| input alphabet
| input letter

\knowledge{notion}
| deterministic
| deterministic automaton
| deterministic automata
| Deterministic

\knowledge{notion}
| unambiguous

\knowledge{notion}
 | strongly unambiguous
 | (strongly) unambiguity
 | strongly@unamb

\knowledge{notion}
 | equivalent@automaton
 | equivalent@automata
 | equivalent@aut
 
\knowledge{notion}
| transition function

\knowledge{notion}
  | $\SS $-complete
  | complete
  | completeness
  | completeness@aut

\knowledge{notion}
| run over $w$
| run@automaton
| run over
| run@aut
| over@run
| run $\rr $ over

\knowledge{notion}
  | subautomaton induced by
  | subautomaton@induced
  | subautomaton
  | -subautomaton
  | subautomata
  | automaton induced by
  | induces@aut
  | induced by@aut
  | subautomaton induced
  | subautomata induced by subgraphs

\knowledge{notion}
| accepted@word
| accepts@aut

\knowledge{notion}
| language accepted
| accepting@automaton
| recognising
| recognise
| recognises
| recognising@automaton
| recognising@aut
| recognised
| recognises@aut
| language recognised

\knowledge{notion}
  | $\oo $-regular language
  | $\oo $-regular languages
  | $\oo $-regular@language
  | $\oo $-regular automata

\knowledge{notion}
  | history-deterministic
  | HD-parity automaton
  | history-determinism
  | history-deterministic@aut
  | HD@aut
  | HD
  | history-determinism@aut
  | HD automaton
  | HD automata
  | (history-)deterministic
  | history-deterministic automaton
  | History-deterministic automata
  | history-deterministic automata
  | History-
  | deterministic@hist

\knowledge{notion}
  | Determinisable by pruning (DBP)
  | determinisable by pruning

\knowledge{notion}
| resolver
| resolvers
| resolver@aut
| resolvers@aut

\knowledge{notion}
| sound@aut
| sound resolver@aut
| soundness@resolverAut
| sound@resolverAut

\knowledge{notion}
| run induced by@aut
| induced by $r$@aut

\knowledge{notion}
  | reachable using the resolver $(r_0, r)$
  | reachable using some sound resolver
  | reachable using@resolver
  | reachable using a sound resolver
  | reachable using $(r_0, r)$

\knowledge{notion}
  | game
  | games
  | Games

\knowledge{notion}
  | play
  | plays
  | finite play

\knowledge{notion}
  | strategy
  | strategies

\knowledge{notion}
  | wins
  | wins $\G $ from $v$
  | wins $\G \compositionAut \A $ from $(v, q_0)$
  | win
  | from@winGame
  | winner
  | wins $\G $ from

\knowledge{notion}
  | consistent with the strategy
  | consistent with
  | consistent with@strat
  | consistency with@strat
  
\knowledge{notion}
  | winning region

\knowledge{notion}
| full winning region

\knowledge{notion}
  | players
  | player
  | controls
  | controlled by

\knowledge{notion}
  | Eve
  | Eve's
  | Eve-vertex
  | controlled by Eve

\knowledge{notion}
  | Adam
  | Adam-vertex
  | controlled by Adam
  | Adam's vertex
  | Adam's

\knowledge{notion}
  | winning strategy for Eve
  | winning strategy
  | winning@strat
  | winning for Eve
  | winning strategy from
  | winning

\knowledge{notion}
  | memory skeleton

\knowledge{notion}
  | update function@mem

\knowledge{notion}
  | memory structure
  | finite memory structure
  | memory structures
   | memory@game

\knowledge{notion}
  | next-move function

\knowledge{notion}
  | implements@memStrat
  | implemented by@mem
  | implementing@mem
  | implemented by
  | implementing
  | implements
  | implements@strat
  | implemented by@stratMem

\knowledge{notion}
  | finite memory strategy
  | finite memory@strat
  | memory requirements
  | memory@games

\knowledge{notion}
| memory structure@aut
| memory structure@resolver
| implemented by@aut
| finite-memory resolvers

\knowledge{notion}
  | composition@memory
  | product game@mem

\knowledge{notion}
  | B\"uchi language
  | Büchi languages
  | B\"uchi@language
  | B\"uchi
  | Büchi

\knowledge{notion}
 | Büchi language associated to
 
\knowledge{notion}
  | coB\"uchi language
  | coB\"uchi languages
  | coB\"uchi
  | (resp. coB\"uchi) language
  | coBüchi

\knowledge{notion}
  | coB\"uchi language associated to
  
\knowledge{notion}
  | Streett language
  | Streett@language
  
\knowledge{notion}
| Streett language associated to
| Streett condition
| Streett

\knowledge{notion}
  | Rabin pairs
  | Rabin pair

\knowledge{notion}
  | Rabin language associated to
  | Rabin
  | Rabin condition

\knowledge{notion}
  | accepted by the Rabin pair
  | accepted by@Rabinpair
  | accepted@Rabinpair
  | accepts@RabinPair
  | accepted by@RabinPair

\knowledge{notion}
  | Rabin language
  | Rabin@language
  | Rabin languages

\knowledge{notion}
  | $[c,d]$-parity language
  | $[0,1]$-parity languages
  | $[1,2]$-parity languages
  | $[d_{\min }, d_{\max }]$-parity language
  | $[\minparityZ {\F },\maxparityZ {\F }]$-parity language
  | $[1,d]$-parity
  | $[0,1]$-parity
  | $[1,2]$-parity
  
\knowledge{notion}
  | parity@language
  | parity language
  | Parity conditions
  | parity condition
  | parity
  | parity languages
  | parity conditions
  | Parity languages
  | Parity

\knowledge{notion}
  | Muller language associated to
  | Muller
  | language associated@Muller

\knowledge{notion}
  | Muller language
  | Muller@condition
  | Muller languages
  
\knowledge{notion}
  | DPA
  | DPAs
  
\knowledge{notion}
    | DMA
    | DMAs

\knowledge{notion}
  | $X$ condition
  | $X$ acceptance condition
  | $X$ transition system
  | $X$ automaton
  | $X$ language
  | $Y$ language
  | $X$ languages
  | $Y$ languages

\knowledge{notion}
  | cycles
  | cycle

\knowledge{notion}
| states of the cycle
| containing@cycle
| state in common
| cycle over $(q,m_i)$
| cycles over $q$
| set of states@cycle
| over $v$@cycle

\knowledge{notion}
| accepting@cycle
| accepting cycle
| accepting cycles
| accepting@cycles

\knowledge{notion}
| rejecting@cycle
| rejecting cycle
| rejecting cycles
| rejecting@cycles

\knowledge{notion}
  | accepting SCC
  | accepting@SCC

\knowledge{notion}
  | rejecting SCC
  | rejecting@SCC

\knowledge{notion}
| $d$-flower
| $d$-flowers
| flowers
  | $k$-flower
  | $k+1$-flower
  | $(d+1)$-flower
  | $(d-1)$-flower

\knowledge{notion}
| positive flower
| positive@flower

\knowledge{notion}
  | negative@flower
  | negative flowers

\knowledge{notion}
  | language of accepting runs of

\knowledge{notion}
  | deterministic parity hierarchy

\knowledge{notion}
  | parity index
  | parity index is not less
  | parity index@atMost  

\knowledge{notion}
  | parity index at least

\knowledge{notion}
  | parity-index-tight

\knowledge{notion}
  | tree
  | trees

\knowledge{notion}
  | depth
  | deepest
  | deepest@tree

\knowledge{notion}
  | subtree

\knowledge{notion}
  | ordered tree
  | ordered
  | ordered@tree
  | order@tree
 | ordered trees
  
\knowledge{notion}
  | nodes
  | node

\knowledge{notion}
  | on the left of
  | leftmost@tree
  | leftmost
  | from left to right

\knowledge{notion}
  | ancestor relation

\knowledge{notion}
  | ancestor
  | below@tree
  | ancestors
  | below
  | above
  | lowest@tree

\knowledge{notion}
  | descendant
  | descendants
  | descendent

\knowledge{notion}
  | $A$-labelled (ordered) tree
  | $\powplus {\SS }$-labelled tree
  | labelled tree
  
\knowledge{notion}
  | the root
  | root

\knowledge{notion}
  | leaves
  | leaf
  
\knowledge{notion}
  | children
  | child
  
\knowledge{notion}
  | height

\knowledge{notion}
  | subtree of~$T$ rooted at~$n$
  | subtree of~$\zielonkaTree {\F }$ rooted at~$n$
  | subtrees of $T$ rooted at
  | subtree of $\zielonkaTree {\F }$ rooted at
  | subtree rooted
  | subtree rooted at

\knowledge{notion}
  | forest

\knowledge{notion}
  | branch
  | branches

\knowledge{notion}
| morphism@graphs
| morphism of graphs
| morphism of (labelled) graphs

\knowledge{notion}
| morphism of pointed graphs
| morphism of (labelled) pointed graphs

\knowledge{notion}
| morphism of labelled graphs

\knowledge{notion}
  | weak morphism of transition systems
  | (weak) morphism@transSys
  | weak morphism
  | weak morphism of labelled transition systems
  | weak morphism of (labelled) transition systems
  | labelled transition systems@mapping
  | (weak) morphism
  | weak morphisms
  | weak morphism of (labelled) TS
  | weak morphism of TS

\knowledge{notion}
   | morphism of automata
   | of automata@morphism
   | automata@mapping
   | morphism@aut
   | of automata@morph
  
\knowledge{notion}
   | morphism of games
   | of games@morphism  
  
\knowledge{notion}
  | isomorphism@TS
  | isomorphism of transition systems
  | isomorphism

\knowledge{notion}
  | isomorphic@TS
  
\knowledge{notion}  
  | morphism of labelled transition system
  | morphism of (labelled) transition systems
  | morphism of labelled transition systems
  | morphism of labelled TS
  | morphism@transSys
  | morphisms@transSys
  | morphism of transition systems
  | morphism
  | morphisms
  | morphisms of transition systems
  | morphism@TS
  | morphism of TS
  | morphisms@TS

\knowledge{notion}
  | surjective
  | surjectivity
  | surjective morphism

\knowledge{notion}
  | preserves the acceptance
  | preserve the acceptance of runs

\knowledge{notion}
  | preserves accepting runs
  | preserve rejecting runs
  | preserve accepting runs

\knowledge{notion}
  | Locally surjective
  | locally surjective
  | local surjectivity
  | locally surjective morphism
  | locally surjective morphisms
  | Locally surjective morphisms

\knowledge{notion}
  | Locally injective
  | locally injective
  | local injectivity

\knowledge{notion}
  | Locally bijective
  | locally bijective
  | locally bijective morphism
  | Locally bijective morphisms
  | local bijectivity
  | local bijective
  | locally bijective morphisms

\knowledge{notion}
| automaton of the morphism
 | automaton@morphism
 | automaton of morphism $\pp $
 |automaton@autMorphism
 | of the morphism@aut

\knowledge{notion}
  | resolver simulating $\pp $
  | resolver simulating~$\pp $
  | resolver simulating
  | resolver@HDmorphism
  | simulating@resolver
  | resolver@morphism
  | resolver@mapping
  | simulating
  | resolver@HDmapping
  | resolver@morph
  | simulate@morph
  
\knowledge{notion}
  | sound@resolverMorphism
  | sound resolver@morph
  | sound@morph
  | soundness@morph

\knowledge{notion}
| run induced by@morphism 
| run induced by@morph
| induced by $r$@morph

\knowledge{notion}
  | history-deterministic mapping
  | history-deterministic mappings
  | history-deterministic@mapping
  | HD mappings
  | HD mapping
  | HD@mapping
  | history-determinism@mapping
  | history-determinism@morp
  | history-determinism@morph

\knowledge{notion}
  | HD mapping of labelled transition systems

\knowledge{notion}
| history-deterministic-for-games
  | HD@games
  | HD-for-games mapping
  | HD-for-games
  | HD-for-games mappings
  
\knowledge{notion}
  | resolver sound for $\G $
  | sound for $\G $
  | sound for $\TS $
  | sound for $\acdRabinTransformGFG {\G }$

\knowledge{notion}
  | consistent with@resolver
  | consistent with@mapping

\knowledge{notion}
  | Zielonka tree automaton

\knowledge{notion}
  | Zielonka tree
  | Zielonka trees

\knowledge{notion}
  | round nodes
  | round node
  | round

\knowledge{notion}
  | square nodes
  | square node
  | square

\knowledge{notion}
  | round-branching width

\knowledge{notion}
  | ZT-HD-Rabin-automaton
  | ZT-HD-Rabin- automaton

\knowledge{notion}
  | ZT-parity-automaton
  | Zielonka-tree-parity-automaton

\knowledge{notion}
| round-branching-width

\knowledge{notion}
  | $X$-closed@automaton
  | $X$-closed
  | $\SS _i$-closed

\knowledge{notion}
  | $X$-final strongly connected component
  | $X$-FSCC
  | $\nu (n_1)$-FSCC
  | $\nu (n_2)$-FSCC
  | $\nu (n_i)$-FSCC

\knowledge{notion}
  | implements@resolver
  | resolver implemented by
  | implemented by@resolver

\knowledge{notion}
  | letter game
  | Letter game

\knowledge{notion}
  | attractor to $X$
  | attractor to $0$
  | attractor to $x$
  
\knowledge{notion}
  | uniformly
  | uniformly@strat

\knowledge{notion}
| $x$-attractor decomposition
| attractor decompositions
| attractor decomposition
| $0$-attractor decomposition
| $2$-attractor decomposition

\knowledge{notion}
| $x$-recursive attractor decomposition
| $0$-recursive attractor decomposition
| $x+2$-recursive attractor decomposition
| decomposition@attrRecursive

\knowledge{notion}
| full attractor decomposition

\knowledge{notion}
  | $\SS _i$-region
  | $\SS _{i'}$-region
  | $\SS _i$-region@simple
  | $\SS _i$-region of the attractor decomposition
  | $\SS _i$-region@simpleAttr
  | $\SS _i$-regions

\knowledge{notion}
  | $1$-avoiding region
  | $x+1$-avoiding region
  | $y$-avoiding region@simpleAttr
  | $y+2$-avoiding region@simple
  | $3$-avoiding region
  
\knowledge{notion}
 | $\SS _i$-region of $\attrDec {\G '}$
 | $\SS _i$-region of $\attrDec {\letterGame {\A }}$
 | $\SS _i$-regions@recursive
 | $\SS _i$-region@rec
 | $\SS _i$-regions@rec

\knowledge{notion}
| $y$-avoiding region of $\attrDec {\G '}$
| $x-1$-avoiding region@recursive
| $y$-avoiding region@rec

\knowledge{notion}
| $X$-Adam-closed@letterGame
| $\SS _i$-Adam-closed@letterGame
| $X$-Adam-closed subgraph@letterGame
  
\knowledge{notion}
  | $X$-closed subgame
  | $\SS _i$-closed subgame
  | $\SS _i$-closed subgames

\knowledge{notion}
  | $X$-FSCC@letterGame
  | $\SS _i$-FSCC@letterGame

\knowledge{notion}
  | tree of alternating subcycles
  | trees of alternating subcycles

\knowledge{notion}
  | round nodes@acd
  | round node@acd
  | round@acd

\knowledge{notion}
  | square nodes@acd
  | square node@acd
  | square@acd
  
\knowledge{notion}
  | alternating cycle decomposition
  | ACD
  | Alternating cycle decomposition
  | ACDs

\knowledge{notion}
  | set of nodes of $\acd {\TS }$

\knowledge{notion}
  | local subtree at $v$
  | local subtree at $v_4$
  | local subtree at
  | local subtree

\knowledge{notion}
| local Muller condition of $\T $ at $q$
| local Muller condition of $\TS $ at $v$
| local Muller condition
| local Muller condition of $\TS '$ at $v'$
| local Muller conditions
| local Muller language of $\TS $ at $v$

\knowledge{notion}
  | positive@tree
  | positive tree

\knowledge{notion}
  | negative@tree

\knowledge{notion}
  | positive@acd

\knowledge{notion}
  | negative@acd

\knowledge{notion}
  | equidistant@acd

\knowledge{notion}
  | ACD-parity-transform
  | ACD-parity-transformation
  | ACD-parity@ACDtrans
  | ACD-transforms
  | ACD-transform

\knowledge{notion}
  | ACD-HD-Rabin-transform
  | ACD-HD-Rabin-transformation
  | ACD-HD-Rabin@ACDtrans
  
\knowledge{notion}
  | ACD-HD-Rabin-transform-for-games
  
\knowledge{notion}
| A-successor
| A-successors

\knowledge{notion}
  | cycle-preimage-automaton at $v'$

\knowledge{notion}
  | unfolding

\knowledge{notion}
| equivalence acceptance condition
| equivalent over $\S $
| equivalent parity condition over
| equivalent@condition
| equivalent to@condition
| equivalent over $G$
| equivalence of acceptance conditions
| equivalent to@accCond
| equivalent@accCond
| equivalent@acceptCond
| equivalent acceptance condition
| equivalent@acc
| equivalent over a same graph $G$
| equivalent over a same graph~$G$

\knowledge{notion}
  | $X$ type
  | Rabin type
  | Streett type
  | parity type
  | $[1,d+1]$-parity type
  | $[0,d-1]$ (resp. $[1,d]$)-parity type
  | Rabin@type
  | B\"uchi type
  | coB\"uchi type
  | B\"uchi@type
  | (resp. coB\"uchi) type
  | generalised B\"uchi@type
  | (resp. generalised coB\"uchi) type
  | generalised B\"uchi type
  | generalised coB\"uchi type
  | $[0,d-1]$ (resp. $[1,d]$ / $\WeakIndex {d}$)-parity type
  | typeness
  
\knowledge{notion}
| $\Weak {d}$ type

\knowledge{notion}
  | Rabin shape

\knowledge{notion}
  | Streett shape

\knowledge{notion}
  | Parity shape
  | parity shape
  
\knowledge{notion}
  | $[0,1]$-parity shape
  | $[1,2]$-parity shape

\knowledge{notion}
  | Rabin ACD
  | Rabin@ACD
  
\knowledge{notion}
  | Streett ACD

\knowledge{notion}
  | parity ACD
  | Parity ACD 

\knowledge{notion}
  | $[0,d-1]$-parity ACD
  | $[1,d]$-parity ACD
  | $[0,1]$-parity ACD
  | $[1,2]$-parity ACD
  | $[0,d]$-parity ACD
  | $[1,d+1]$-parity ACD
  | $[0,d-1]$(resp. $[1,d]$)-parity ACD
  | $[0,d-1]$ (resp. $[1,d]$)-parity ACD

\knowledge{notion}
  | normalised
  | normal form
  | normalisation
  | normality
  
\knowledge{notion}
| normalised@neg
| normal form@neg
  
\knowledge{notion}
  | negative@TS
  | not negative@TS
  | non-negative@TS

\knowledge{notion}
  | duplicated edges
  | duplicated transitions

\knowledge{notion}
| positive@SCC

\knowledge{notion}
| negative@SCC

\knowledge{notion}
  | generalised Büchi language associated to

\knowledge{notion}
  | generalised B\"uchi language
  | generalised B\"uchi@language
  | generalised B\"uchi
  | generalised B\"uchi languages
  | generalised Büchi
  | generalised B\"uchi-
  
\knowledge{notion}
  | generalised coB\"uchi language associated to

\knowledge{notion}
  | generalised coB\"uchi language
  | generalised coB\"uchi
  | coB\"uchi@generalised

\knowledge{notion}
  | B\"uchi shape
 | B\"uchi@shapeZT
 
\knowledge{notion}
  | coB\"uchi shape
 | (resp. coB\"uchi) shape
 
\knowledge{notion}
  | Generalised B\"uchi shape
  | generalised B\"uchi shape
 | generalised B\"uchi@shapeZT
 
\knowledge{notion}
  | Generalised coB\"uchi shape
  | generalised coB\"uchi shape
  | (resp. generalised coB\"uchi) shape

\knowledge{notion}
  | B\"uchi ACD

\knowledge{notion}
  | coB\"uchi ACD

\knowledge{notion}
  | Generalised B\"uchi ACD
  | generalised B\"uchi ACD

\knowledge{notion}
  | Generalised coB\"uchi ACD
  | generalised coB\"uchi ACD

\knowledge{notion}
| $\Weak {d}$ ACD

\knowledge{notion}
 | generalised weak
 | Generalised weak
 | weak automata
 | generalised weak automata
 | conditions that depend on the structure

\title{From Muller to Parity and Rabin Automata: Optimal Transformations Preserving (History) Determinism}
\ThCSshorttitle{From Muller to Parity: Optimal Transformations Preserving (History) Determinism} 

\ThCSauthor[1]{Antonio Casares}{antoniocasaressantos@gmail.com}[0000-0002-6539-2020]
\ThCSauthor[2,3]{Thomas Colcombet}{thomas.colcombet@irif.fr}[0000-0001-6529-6963]
\ThCSauthor[1,2,4]{Nathanaël Fijalkow}{nathanael.fijalkow@labri.fr}[0000-0002-6576-4680]
\ThCSauthor[2,5]{Karoliina Lehtinen}{lehtinen@lis-lab.fr}[0000-0003-1171-8790]

\ThCSaffil[1]{LaBRI, Université de Bordeaux, France}
\ThCSaffil[2]{CNRS, France}
\ThCSaffil[3]{IRIF, Université Paris Cité, France}
\ThCSaffil[4]{MIMUW, University of Warsaw, Poland}
\ThCSaffil[5]{Aix-Marseille Université, LIS, France}

\ThCSshortnames{A. Casares, T. Colcombet, N. Fijalkow and K. Lehtinen} 
\ThCSkeywords{Emerson-Lei automata, 
Good-for-games automata, 
Paritizing methods, 
Omega-regular languages.}

\ThCSthanks{We want to thank Klara J. Meyer and Salomon Sickert for their comments and for spotting a mistake in a previous version. We also thank Alexandre Duret-Lutz and Florian Renkin for stimulating discussions around the alternating cycle decomposition, and Corto Mascle for his precious suggestions about the presentation of this paper.
  We thank the anonymous reviewers for their valuable feedback.\\
  Thomas Colcombet: Supported by the European Research Council (ERC) under the European Union’s Horizon 2020 research and innovation programme (grant agreement No.670624) -- DuaLL -- and the DeLTA ANR project (ANR-16-CE40-0007).\\
This article is based on a paper that appeared at ICALP 2021~\cite{CCF21Optimal}, and incorporates material from another paper that appeared at ICALP 2022~\cite{CCL22SizeGFG}.}

\ThCSyear{2024}
\ThCSarticlenum{12}
\ThCSdoicreatedtrue
\ThCSreceived{May 19, 2023}
\ThCSaccepted{Jan 21, 2024}
\ThCSpublished{Apr 23, 2024}

\begin{document}
	\maketitle
	\begin{abstract}
	We study transformations of automata and games using Muller conditions into equivalent ones using parity or Rabin conditions.
	We present two transformations, one that turns a deterministic Muller automaton into an equivalent deterministic parity automaton, and another that provides an equivalent history-deterministic Rabin automaton.
	We show a strong optimality result: the obtained automata are minimal amongst those that can be derived from the original automaton by duplication of states.
	We introduce the notions of locally bijective morphisms and history-deterministic mappings to formalise the correctness and optimality of these transformations.
	
	The proposed transformations are based on a novel structure, called the alternating cycle decomposition, inspired by and extending Zielonka trees. In addition to providing optimal transformations of automata, the alternating cycle decomposition offers fundamental information on their structure. We use this information to give crisp characterisations on the possibility of relabelling automata with different acceptance conditions and to perform a systematic study of a normal form for parity automata.
	\end{abstract}

\paragraph*{}
This document contains hyperlinks.
\AP Each occurrence of a "notion@@example" is linked to its ""definition@@example"".
On an electronic device, the reader can click on words or symbols (or just hover over them on some PDF readers) to see their definition.

	\section{Introduction}\label{section:introduction}
	\subsection*{Context}
\subparagraph{Games and automata for LTL synthesis.} "Games" and "automata" over infinite words form the theoretical basis for the verification and synthesis of reactive systems; we refer to chapters 2,~4, and~27 of the recent Handbook of Model Checking~\cite{PitermanPnueli2018Handbook,Kupferman2018Handbook,BloemCJ2018Handbook} for a broad exposition of this research area.
A milestone objective is the synthesis of reactive systems with specifications given in \emph{Linear Temporal Logic} (LTL).
The original approach of Pnueli and Rosner~\cite{PR89Synthesis} using automata and games devised more than four decades ago is still at the heart of the state-of-the-art synthesis tools~\cite{EKRS17FromLTLtoParity,LMS20SynthesisLTL,MichaudColange18Synt,MullerSickert17LTLtoDeterministic}.
The limiting factor in this method is the transformation of the LTL formula to a "deterministic" "parity" "automaton". This "automaton" is then used to build a "game", and a controller for the reactive system can be obtained from a "winning" "strategy" for this game.
Most solutions to this problem (including the top-ranked tools in the SyntComp competitions~\cite{SyntCompReport22}, Strix~\cite{LMS20SynthesisLTL, MS21ModernisingStrix} and \texttt{ltlsynt}~\cite{MichaudColange18Synt})  first construct a  "Muller" (or Emerson-Lei) "automaton", and then transform it into an equivalent "parity" "automaton" (we remark that, nevertheless, synthesis procedures avoiding the construction of "deterministic" automata have been proposed, for example, via the use of universal "coB\"uchi" automata~\cite{KV05Safraless}).
The use of an intermediate "Muller" "automaton" is also present (although sometimes implicitly) in the most recent improvements in the determinisation  of "B\"uchi" "automata" towards "deterministic" "parity" "automata"~\cite{LodingP19,Piterman2006fromNDBuchi,Schewe2009tighter}. For this reason, understanding transformations of "Muller" "automata" and finding efficient procedures for them is of great importance.

\subparagraph{Which are the simplest acceptance conditions?}
There exist multiple kinds of "acceptance conditions" that are commonly employed by $\oo$-automata ("B\"uchi", "Rabin", "Muller"...). The use of "parity conditions" for LTL synthesis is justified by both practical and theoretical reasons. Firstly, there exist several high-performing algorithms solving "parity" "games"~\cite{DiStasio2021ParityGames,FriedmannM09ParityGames,SyntCompReport22,LvD20SymbolicParity,Dijk18Oink}, so the last step in the LTL synthesis method described above can be carried out smoothly once the "parity" "game" is obtained. 
From a theoretical point of view, "parity conditions" can be considered as the simplest family of conditions that can be used to "recognise" all "$\oo$-regular languages" with "deterministic" "automata"; it could even be argued that there is a canonical aspect to them:
\begin{itemize}
	\item The optimal number of "colours@@outAut" needed by a "parity" "automaton"  to "recognise" a language~$L$ reveals a fundamental piece of information about it, called its "parity index". 
	The "parity index" (sometimes called Mostowski index) yields a strict hierarchy both for "deterministic" automata over words and for non-deterministic automata over trees~\cite{Niwinski86Clones,Bradfield98MuHierarchy} (and these hierarchies are closely related~\cite{KSV96Relating,NiwinskiWalukievicz1998Relating}).
	In both cases, this index is a measure of the structural complexity of automata "recognising"~$L$~\cite{Wagner1979omega,NiwinskiWalukievicz1998Relating} and of its topological complexity~\cite{ADMN08Topological, Skrzypczak13Topological}.
	Whether we can decide the "parity index" of a language of infinite trees represented as a non-deterministic parity tree-automaton is a long-standing open problem~\cite{NW05DecidingHierarchy,CL08MostwHierarchy}, which is tightly related with the alternation depth of fixpoint operators in $\mu$-calculus formulas~\cite{Niwinski86Clones}
	
	\item "Parity languages" are exactly "Muller languages" corresponding to families $\F\subseteq \powplus{\GG}$ of subsets of "colours" such that both $\F$ and its complement are closed under union (Proposition~\ref{prop-typ:parityZielonkaShape}).
	\item "Parity languages" are bipositional~\cite{EmersonJutla91Determinacy} (in a "parity" "game", both players can play optimally using positional "strategies", that is, "strategies" that use no "memory@@game"). Moreover, over infinite "game graphs", these are the only bipositional languages~\cite{CN06}, and over finite "game graphs", these are the unique bipositional "Muller languages"~\cite{Zielonka1998infinite}.
	\item Solving "parity" "games" is both in $\NP$ and $\coNP$~\cite{EJS93ModelChecking} (more precisely, the problem is in $\mathsf{UP}\cap\mathsf{co\hyphen UP}$~\cite{Jurdzinski98UP}). They can be solved in quasi-polynomial time~\cite{CJKLS17}, and whether they can be solved in polynomial time is a major open question. This contrasts with the complexity of solving "Rabin" and "Muller" games, which is, respectively, $\NP$-complete~\cite{EmersonJutla99Complexity} and $\PSPACE$-complete~\cite{Dawar2005ComplexityBounds}.
\end{itemize}

However, these are not the only kind of conditions that deserve our attention. In this work, we further investigate transformations producing "automata" using a "Rabin" "acceptance condition". Although in practice solvers for "Rabin" "games" are not as developed, "Rabin languages" are a natural choice and interesting from a theoretical point of view: they are exactly the half-positional "Muller languages"~\cite{Zielonka1998infinite}, there exists a correspondence between "Rabin" "automata" and "memory structures" for "Muller" "games"~\cite{Casares2021Chromatic, CCL22SizeGFG}, and the determinisation of "B\"uchi" "automata" naturally produces "Rabin" "automata"~\cite{FKVW15Profile,Safra1988onthecomplexity, Schewe2009tighter}.

\subparagraph{Transformations of games and automata.} There are various existing techniques to transform "Muller" "automata" or "games" into "parity" ones. The majority of these methods involve "composing" the input automaton $\A$ with a "deterministic" "parity" "automaton" "recognising" the "acceptance condition" used by $\A$. The first such "parity" "automaton" was introduced by Gurevich and Harrington in the 1980s~\cite{Gurevich1982trees} and is known as the Latest Appearance Record (LAR).
L\"oding proved that the LAR is optimal in the worst case~\cite{Loding1999Optimal}: \textit{there exists} a family of "Muller languages" $L_i$
for which the LAR is minimal amongst "deterministic" "parity" "automata" "recognising" $L_i$.
However, the LAR is far from being minimal in every case, as it only uses the information about the size of the alphabet.
Since its introduction, many refinements of the LAR have been proposed for subclasses of "Muller languages"~\cite{Kretinski2017IAR,Loding1999Optimal}. 
The approach using "composition of automata" has one significant drawback: it disregards the structure of the original automaton, and only its "acceptance condition" is taken into account. Some works have explored heuristics to improve this aspect~\cite{KMWW21IARPreorders, MeyerSickert21OptimalPractical, RDP20PracticalParitizing}. These refined transformations do still have the following property: each original state $q$ is turned into multiple states of the form $(q, x)$ -- although this is done in a non-uniform way, with each state possibly being copied a different number of times.
In this work, we introduce "morphisms" of "transition systems" to formalise the idea of transformations of "automata" and "games"; if a "parity" "automaton" $\B$ has been obtained as a transformation of a "Muller" "automaton" $\A$, there will be a "morphism" $\pp\colon \B \to \A$ that sends states of the form $(q, x)$ to $q$. A theory of "morphisms" of "transition systems" is developed in Section~\ref{section:morphisms}.

\subparagraph{History-deterministic automata.} For the purposes of LTL synthesis and "game" transformations, it is imperative to eliminate "non-determinism" from "automata", since "non-deterministic" "automata" do not yield correct "games". Unfortunately, "deterministic" automata can be exponentially larger than "non-deterministic" ones.  
Recently, an intermediate model of "automata", named "history-deterministic" (also called good-for-games), has received considerable attention. The reason is that "history-determinism" exactly captures the features of "deterministic" "automata" that make them suitable for synthesis purposes, while being a less restrictive model.
A natural question that arises is whether "history-deterministic" "automata" can be more succinct than "deterministic" ones, and, in that case, which languages and automata types can benefit from this succinctness. It was not until several years after the introduction of "history-determinism"~\cite{HP06, Colcombet2009CostFunctions} that an example of an "$\oo$-regular language" for which "history-deterministic" "automata" are smaller than "deterministic" ones was exhibited~\cite{KS15DeterminisationGFG} (and it was even conjectured that such an "automaton" could not exist~\cite{Colcombet2012FormsND}).
"History-deterministic automata" are the focus of several lines of research (we refer to the survey~\cite{BL23SurveyHD} for a detailed exposition). Despite this, a complete understanding of "history-deterministic" automata remains elusive, and their scope of applicability is still uncertain. One key aspect that has not yet been addressed is how to design techniques as general as possible for building "history-deterministic automata". To the best of our knowledge, the only existing result in this direction is a polynomial-time algorithm to minimise "coBüchi" "history-deterministic" automata~\cite{AK22MinimizingGFG}.

\subparagraph{The Zielonka tree and the alternating cycle decomposition.} The starting point of our work is the notion of "Zielonka tree", introduced by Zielonka~\cite{Zielonka1998infinite} as an informative representation of "Muller languages" -- languages that can be described by a boolean
combination of atomic propositions of the form ``the letter `$a$' appears infinitely often''. The "Zielonka tree" captures many important properties of "Muller languages", such as being "Rabin" or "parity"~\cite{Zielonka1998infinite}, and, most importantly, it characterises their exact "memory requirements", both in two-player games~\cite{DJW1997memory} and stochastic games~\cite{Horn09RandomFruits}. 

The contribution at the core of this work is a generalisation of "Zielonka trees" to general "Muller" "automata" "recognising" any "$\oo$-regular language", which we call the "alternating cycle decomposition" (ACD). The "ACD", greatly inspired from Wagner's work on $\oo$-automata~\cite{Wagner1979omega}, is a data structure that provides an abridged representation of the "accepting@@cycle" and "rejecting@@cycle" "cycles" of the "automaton", encapsulating the interplay between the structure of the "underlying graph"  and the "acceptance condition" of a "Muller" "automaton".

\subsection*{Contributions}
In this work, we carry out an extensive study of transformations of "Muller" "automata" and "games". We outline next our main contributions.

\begin{longdescription}
	\item[1. Minimal automata for Muller languages.] The basis on which we build up our work is a study of minimal automata "recognising" "Muller languages". Using the "Zielonka tree", we propose a construction of a "deterministic" "parity" automaton "recognising" a "Muller language" (Section~\ref{subsec-zt: parity automaton}). This construction implicitly appears in the long version of~\cite{DJW1997memory}. We show a strong optimality result: \emph{for all} "Muller language" $L$, the "parity" "automaton" obtained from the "Zielonka tree" is minimal both amongst "deterministic" and "history-deterministic" "parity" "automata" "recognising"~$L$ (Theorem~\ref{thm-zt:strong_optimality_ZTparity}).\footnotemark{} Moreover, it uses the optimal number of "output colours" to "recognise"~$L$ (Theorem~\ref{thm-zt:optimality_ZTparity-priorities}).  The optimality result we obtain is much stronger than the worst case optimality result of the LAR transformation~\cite{Loding1999Optimal}, since it applies to every "Muller language". In particular, our characterisation yields an algorithm to minimise "deterministic" "parity" "automata" "recognising" "Muller languages" in polynomial time (Theorem~\ref{thm-min:minimisation_parity_automata}).
	In light of our result, we conclude that the use of "history-determinism" does not yield any gain in the "state complexity" of "parity" "automata" "recognising" "Muller languages". 
	
	\footnotetext{The optimality of the "Zielonka-tree-parity-automaton" amongst "deterministic" automata has also been obtained in the independent unpublished work \cite{MeyerSickert21OptimalPractical}.}
	
	We further propose a construction of a "history-deterministic" "Rabin" automaton "recognising" a "Muller language" (Section~\ref{subsec-zt: GFG-Rabin}), and prove that this automaton is minimal amongst "history-deterministic" "Rabin" automata (Theorem~\ref{thm-zt:optimality_ZT-HD-Rabin}). This construction is also based on the "Zielonka tree". 
	
	In essence, our results reinforce the idea that the "Zielonka tree" precisely captures the fundamental properties of  "Muller languages".

	\item[2. Introducing morphisms as witnesses of transformations.] In order to formalise transformations of "games" and "automata", we develop a theory of "morphisms of transition systems" (Section~\ref{section:morphisms}). Intuitively, a morphism $\pp\colon \B\to \A$ witnesses the fact that $\B$ has been obtained from $\A$ by blowing up each state $q\in \A$ to the states in $\inv{\pp}(q)$. However, this property on its own does not suffice to guarantee the semantic equivalence of $\A$ and $\B$. It is for this reason that we introduce different variants of "morphisms", offering a range of definitions with varying degrees of restrictiveness. Two kinds of "morphisms" will be of central importance: (1) "locally bijective" "morphisms", which generalise "composition" with "deterministic" "automata" and preserve "determinism", and (2) "history-deterministic mappings" ("HD mappings"), which generalise "composition" by "history-deterministic" "automata" and are defined using a minimal set of hypothesis guaranteeing the semantic equivalence of $\A$ and $\B$.

	\item[3. The alternating cycle decomposition and optimal transformations of Muller transition systems.] In order to generalise the fruitful applications of the "Zielonka tree" to "Muller" "automata" and "games", we introduce the "alternating cycle decomposition" (ACD), a data structure that captures the interplay of the "underlying graph" of these "transition systems" and their "acceptance condition"  (Section~\ref{section:acd}). Using the "ACD", we describe a construction that transforms a "Muller" "automaton" $\A$ into an equivalent "parity" "automaton" $\B$ while preserving the "determinism" of $\A$ (formally, there is a "locally bijective morphism" $\pp\colon \B\to \A$). This transformation comes with a strong optimality guarantee: for any other "parity" "automaton"~$\B'$ admitting a "locally bijective morphism" (or even "HD mapping") $\pp'\colon \B'\to \A$, the automaton~$\B$ is smaller than $\B'$ and it uses less "output colours" (Theorems~\ref{thm-acd:optimality-priorities_ACD-parity_transform} and~\ref{thm-acd:optimality-size_ACD-parity_transform}).
	An interesting corollary of our result is the following: if $\B$ is an "HD" "parity" "automaton" that is strictly smaller than any "deterministic" "parity" "automaton" "recognising" $\Lang{\B}$, then $\B$ cannot be derived from a "deterministic" "Muller" "automaton" (Corollary~\ref{cor-acd:HD-transformations-are-big}). This result sheds light on the difficulty to obtain succinct "HD automata" and their potential applicability.
	
	We also provide a transformation that translates a "Muller" "automaton" $\A$ into a "history-deterministic" "Rabin" "automaton" $\B$ in an optimal way: for any other "Rabin" "automaton" $\B'$ admitting an "HD mapping" $\pp'\colon \B'\to \A$, the "automaton" $\B$ is smaller than $\B'$.
	
	\item[4. Structural results for Muller transition systems.] The "ACD" does not only provide optimal transformations of "games" and "automata", it also features some of their fundamental structural properties. 
	As an application, we give a set of crisp characterisations for relabelling "automata" with different classes of "acceptance conditions" (Section~\ref{subsec-corollaries:typeness}). For instance,  we show that given a "Muller" "automaton" $\A$, we can define a "Rabin" "condition@@accep" over the "underlying graph" of $\A$ obtaining an "equivalent@@aut" "automaton" if and only if the union of "rejecting@@cycle" "cycles" of $\A$ is again a "rejecting@@cycle" "cycle". Our results unify and extend those from~\cite{BSW01Weakautomata,BKS10Paritizing,KPB94DetOmega,Zielonka1998infinite}.
	
	In Section~\ref{subsec-corollaries: normal-form}, we conduct a comprehensive examination of a "normal form" for "parity" "automata". This "normal form" implicitly appears in \cite{CartonMaceiras99RabinIndex}, and has since proven instrumental in proofs about "history-deterministic" "automata"~\cite{AK22MinimizingGFG,EhlersSchewe22NaturalColors,KS15DeterminisationGFG}, positionality of "$\oo$-regular languages"~\cite{BCRV22HalfPosBuchi} and learning of $\oo$-automata~\cite{BohnLoding23DetParityFromExamples}. Similar normalisation procedures are commonly applied to "parity" "games" to speed up algorithms solving them~\cite{FriedmannM09ParityGames}. We use the "ACD" to provide straightforward proofs of the fundamental properties which make automata in "normal form" practical in both theoretical proofs and applications.	
\end{longdescription}

\subparagraph*{Our model: transition systems and acceptance over edges.} We want to point out a few technical details about the model used in this paper. First, we work with general "transition systems" for two reasons: (1) to seamlessly encompass both "automata" and "games" models, and (2) to emphasise that the "ACD" and the transformations we propose do only depend on the "underlying graph" and the "acceptance condition"; we can view the "input letters" of an "automaton", or the partition of the vertices in a "game", as add-ons that do not affect the core of our approach.

Also, we define "acceptance conditions" over the edges of "transitions systems" -- instead of over the vertices. This choice has been shown to yield more canonical results in theory, for instance, in the study of "strategy" complexity for games~\cite{BCRV22HalfPosBuchi,Casares2021Chromatic,CN06,Zielonka1998infinite}, the determinisation of "B\"uchi" "automata"~\cite{Colcombetz2009tight,Varghese14Determinising}, or the minimisation of "history-deterministic automata"~\cite{AK22MinimizingGFG,EhlersSchewe22NaturalColors}. It has also proven to be more applicable in practical scenarios~\cite{Spot2.10CAV22,GiannakopoulouLerda2002Transitions}.  We believe that the present work provides further evidence to this claim, as the minimal "automaton" obtained from the "Zielonka tree", as well as the transformations based on the "ACD", substantially rely on the use of edge-based acceptance. 

\AP Finally, we remark that in this work we are concerned with ""state complexity"", that is, the efficiency of a construction is measured based on the number of states of the resulting "transition system". We do not focus on the representation of the "acceptance conditions"; for instance, we will not differentiate between "Muller" or Emerson-Lei conditions, as they have the same expressive power (see also Remark~\ref{rmk-prelim:representation_Muller_Languages}).

\paragraph*{Follow-up work.}

Despite its recent introduction~\cite{CCF21Optimal}, the "alternating cycle decomposition" has already found applications in both practical and theoretical scenarios.
The "ACD-parity-transform" has been implemented in two open-source tools: Spot 2.10~\cite{Spot2.10CAV22} and Owl 21.0~\cite{KMS18Owl}, and it is used in the LTL-synthesis tools \texttt{ltlsynt}~\cite{Spot2.10CAV22} and Strix~\cite{MS21ModernisingStrix}. These implementations were presented in the conference paper~\cite{CDMRS22Tacas}, where transformations based on the "ACD" are compared to the state-of-the-art existing paritizing methods.

The "typeness" results stemming from the "ACD" (Section~\ref{subsec-corollaries:typeness}) have also been proven instrumental in theoretical applications. They have been used to show a correspondence between "Rabin" "automata" and "memory structures" for "games"~\cite{Casares2021Chromatic}, and to provide lower bounds in the size of "deterministic" "Rabin" "automata"~\cite{CCL22SizeGFG}.

	\section{Preliminaries}\label{section:preliminaries}
	In this section, we introduce definitions that will be used throughout the paper.

\paragraph*{Basic definitions.}
\AP For a set $A$, we let $|A|$ denote its cardinality, $\pow{A}$ its power set and $\intro*\powplus{A} = \pow{A} \setminus \{\emptyset\}$.
For natural numbers $i\leq j$, $[i,j]$ stands for $\{i,i+1, \dots, j-1,j \}$.

\AP For a set $\Sigma$, a ""word"" over $\Sigma$ is a sequence of elements from $\Sigma$. 
An ""$\omega$-word"" (or simply an \emph{infinite word}) is a word of length $\oo$.
\AP The sets of finite and infinite words over $\Sigma$ will be written $\Sigma^*$ and $\Sigma^{\oo}$, respectively, 
\AP and we let $\intro*\SigmaInfty = \Sigma^* \cup \Sigma^\oo$. Subsets of $\SS^*$ and $\SS^\oo$ will be called \emph{languages}.  For a word $w\in \SigmaInfty$ and $i\geq 0$ we write $w_i$ to represent the $i$-th letter of $w$.
We let $\ee$ denote the ""empty word"", and let $\SS^+ = \SS^*\setminus\{\ee\}$. The concatenation of two words $u\in \Sigma^*$ and $v\in \SigmaInfty$ is written $u\cdot v$, or simply $uv$. \AP If $u=v\cdot w$, for $v\in \Sigma^*$ and $u,w\in \SigmaInfty$, we say that $v$ is a  ""prefix"" of $u$, and we write $v \intro*\prefix u$. 
\AP For a word $w\in \Sigma^\oo$, we let $\intro*\minf(w)=\{ a\in \Sigma \mid w_i=a \text{ for infinitely many } i\in \NN\}$.

\AP We say that a language $L\subseteq \SS^\oo$ is ""prefix-independent"" if for all $w\in \SS^\oo$ and $u\in \SS^*$ we have that  $uw\in L$ if and only if $w\in L$.

Given a map $\alpha:A \rightarrow B$, we will extend $\alpha$ to "words" component-wise, i.e., $\aa: A^\infty \rightarrow B^\infty$ will be defined as $\aa(w_0w_1w_2\dots)=\aa(w_0)\aa(w_1)\aa(w_2)\dots$. We will use this convention throughout the paper without explicitly mentioning it. 
If $A'\subseteq A$, we denote $\restr{\aa}{A'}$ the restriction of $\aa$ to $A'$.
\AP We let $\intro*\Id{A}$ be the identity function on $A$.
\AP We write $\alpha:A \intro*\partialF B$ if $\aa$ is a ""partial mapping"" (it is defined only over some subset of $A$).

\AP In this work, we will use the term \emph{graph} to denote what is sometimes called a \emph{directed multigraph}: A ""graph"" is a tuple $G=(V,E,\msource,\mtarget)$, where $V$ is a set of vertices, $E$ a set of edges and $\intro*\msource\colon E\to  V$ and $\intro*\mtarget\colon E\to  V$ are maps indicating the source and target for each edge. 
\AP A ""path"" is a (finite or infinite) sequence $\rr = e_0e_1...\in \kl{E^\infty}$ such that  $\msource(e_{i})=\mtarget(e_{i-1})$, for all $i>0$. For notational convenience, we  write $v_0 \re {e_0} v_1 \cdots  \re{e_{n-1}} v_n$ to denote a finite path from $v_0 = \msource(e_0)$ to $v_n = \mtarget(e_{n-1})$, and we let $\intro*\msourcePath(\rr) = v_0$ and $\intro*\mtargetPath(\rr) = v_n$.
\AP For $A\subseteq V$, we let $\intro*\PathSetFin{G}{A}$ and $\intro*\PathSet{G}{A}$ denote, respectively, the set of finite and infinite "paths"  on $G$ starting from some $v\in A$ (we omit brackets if $A=\{v\}$ is a singleton).  We let $\intro*\PathSetInfty{G}{A} = \PathSetFin{G}{A}\cup \PathSet{G}{A}$.
\AP For a subset of vertices $A\subseteq V$ we write:
\begin{itemize}
	\item $\intro*\mIn(A)=\{ e\in E \mid \mtarget(e)\in A \}$,
	\item $\intro*\mout(A)=\{ e\in E \mid \msource(e)\in A \}$.
\end{itemize}

All "graphs" considered in this paper will be finite.

\AP A graph is ""strongly connected"" if there is a "path" connecting each pair of vertices. \AP A ""subgraph"" of $(V,E,\msource,\mtarget)$ is a graph $(V',E',\msource',\mtarget')$ such that $V'\subseteq V$, $E'\subseteq E$ and $\msource'$ and $\mtarget'$ are the restrictions  of $\msource$ and $\mtarget$ to $E'$, respectively. A ""strongly connected component"" (SCC) is a maximal "strongly connected" "subgraph". We say that a SCC is ""final@@SCC"" if there is no edge leaving it.
\AP We say that a vertex $v$ is ""recurrent"" if it belongs to some "SCC", and that it is ""transient"" on the contrary.

\subsection{Transition systems, automata and games}\label{subsect-prelim:transition-systems}

\subparagraph*{Transition systems.}
\AP A ""pointed graph"" $G = (V,E, \msource, \mtarget, I)$ is a "graph" 
 together with a non-empty subset of ""initial vertices"" $I\subseteq V$.
\AP An ""acceptance condition"" over $G$ is a tuple $\mathrm{Acc} = (\gg, \GG, \WW)$ where $\GG$ is a finite set of ""colours"", $\gg:E\to \GG\cup\{\ee\}$ is an ""edge-colouring"" of $G$ and $\WW\subseteq \GG^\oo$ is a language of infinite words called the ""acceptance set"". We allow ""uncoloured edges"" ($\ee$-edges), but we impose the condition that no infinite path of $G$ is eventually composed exclusively of $\ee$-edges (that is, every "cycle" contains some edge $e$ with $\gg(e)\neq \ee$). 

\AP A ""transition system"" (abbreviated ""TS"") is a tuple $\TS = (\intro*\underlyingGraph{\TS}, \intro*\macc{\TS})$, consisting in a "pointed graph" $\underlyingGraph{\TS} = (V,E, \msource, \mtarget, I)$, called the ""underlying graph"" of $\TS$, and an "acceptance condition" $\macc{\TS} = (\gg, \GG, \WW)$  over $\underlyingGraph{\TS}$. We will also refer to vertices and edges as states and transitions, respectively. We write $v \re c v'$ if there is $e\in E$ such that $\msource(e)=v$, $\mtarget(e)=v'$ and $\gg(e)=c$.
We will assume for technical convenience that "transition systems" contain no sink, that is, every vertex has at least one outgoing edge.
\AP For any non-empty subset of vertices $\tilde{I} \subseteq V$, we let $\intro*\initialTS{\TS}{\tilde{I}}$ be the "transition system" obtained from $\TS$ by setting $\tilde{I}$ to be its set of "initial vertices".
\AP The ""size@@transSys"" of a transition system $\TS$ is the cardinality of its set of vertices, written $\intro*\size{\TS}$.

\AP A ""run@@transSys"" on a "transition system" $\TS$ (or on a "pointed graph") is a (finite or infinite) 
"path" $\rr=e_0e_1\dots \in \kl{E^\infty}$ starting from an "initial vertex", that is, $\msource(e_0)\in I$.
We let $\intro*\RunsFin{\TS}$ and $\intro*\Runs{\TS}$ be the set of finite and infinite "runs" on $\TS$, respectively, and we let $\intro*\RunsInfty{\TS} = \RunsFin{\TS}\cup \Runs{\TS}$. (We note that $\Runs{\TS} = \PathSet{\underlyingGraph{\TS}}{I}$.)

\AP The ""output"" of a "run@@transSys" $\rr \in \RunsInfty{\TS}$ is the sequence of colours in $\GG^\infty$ obtained by removing the occurrences of $\ee$ from $\gamma(\rr)$; which we will also denote $\gamma(\rr)$ by a small abuse of notation. 
\AP A "run@@transSys" $\rr$ is ""accepting@@run"" if $\gamma(\rr)\in \WW$, and ""rejecting@@run"" otherwise (in particular, finite runs will be rejecting).
We write $\rr = v\lrp{w} v'$ to denote a "run" with $\msourcePath(\rr)=v$, $\mtargetPath(\rr)=v'$ and $\gg(\rr)=w$.

\AP We say that a vertex $v\in V$ is ""accessible@@fromVertex"" (or \emph{reachable}) from a vertex $v_0$ if there exists a finite "path" from $v_0$ to $v$. We say that $v$ is ""accessible@@vertex"" if it is accessible from some "initial vertex". A set of states $B\subseteq V$ is ""accessible@@set"" if every state $v\in B$ is "accessible@@vertex". 
\AP The ""accessible part"" of a "transition system" is the set of accessible states. We define analogously the ""accessible part from a vertex"" $v_0$.

\AP A ""labelled graph"" $\left(G, (l_V, L_V), (l_E, L_E)\right)$ is a "graph"  together with labelling functions $l_V: V \rightarrow L_V$, $l_E: E \rightarrow L_E$, where $L_V$ and $L_E$ are sets of labels for vertices and edges, respectively. If only the first (resp. the second) of these labelling functions appears, we will use the terms ""vertex-labelled"" (resp. ""edge-labelled"") "graphs". 
A ""labelled transition system"" is a "transition system" with "labelled@@graph" "underlying graph".

\begin{remark}\label{rmk:colours=edges}
	We remark that, whenever necessary, we can assume without loss of generality that in the "acceptance condition" $\mathrm{Acc} = (\gg, \GG, \WW)$ of a "transition system",  $\GG=E$ is the whole set of edges and $\gamma$ is the identity function.
	Indeed, an "equivalent acceptance condition" can always be defined by using the "acceptance set" $\WW' = \{w\in E^\oo \mid \gg(w)\in \WW\} \subseteq E^\oo$.
\end{remark}

\subparagraph*{Automata.}
A ""(non-deterministic) automaton"" over $\SS$ is an "edge-labelled" "transition system" $\A = (\underlyingGraph{\A}, \macc{\A}, (l_\SS, \SS))$, where $\SS$ is a finite set of ""input letters"".
Let $\A$ be an "automaton" with $\underlyingGraph{\A} = (Q,\DD,\msource, \mtarget, I)$ as "underlying graph" and $\macc{\A} = (\gg, \GG, \WW)$ as "acceptance condition".
We write $e = q\re{a:c}q'$ to denote that $e\in \DD$ satisfies $l_\SS(e) = a$ and $\gg(e) = c$.
We can assume that $\DD\subseteq Q\times \SS\times \GG \times Q$.
\AP We define:
 \[\intro*\transAut{}(q, a) = \{ (q', c)\in Q\times \GG \mid \text{ there is } e = q\re{a:c} q' \in \DD\}.\]

We note that we can now recover the classical representation of an automaton as a tuple $\A = (Q,\Sigma, I, \Gamma, \transAut{}, \WW)$, (or $(Q,\Sigma, I, \Gamma, \DD, \WW)$) which we might use when working exclusively with "automata".

\AP We say that an "automaton" $\A$ is ""deterministic"" if $I$ is a singleton and for every $q\in Q$ and $a\in \Sigma$, $|\transAut{}(q,a)| \leq 1$. 
\AP We say that $\A$ is  ""complete""
if for every $q\in Q$ and $a\in \Sigma$, $|\transAut{}(q,a)| \geq 1$. 
We remark that we can assume that automata are "complete" without loss of generality by adding a sink state.

\AP Given an "automaton" $\A$ and a word $w\in \SigmaInfty$, a  ""run over $w$"" in $\A$ is a "run@@transSys" $\rr = e_0e_1\dots  \in \RunsInfty{\A}$ such that $l_\SS(e_i) = w_i$ for all $i \geq 0$.
\AP A word $w\in \SS^\oo$ is  ""accepted@@word"" by $\A$ if there is a "run over $w$" that is "accepting" (that is, a "run@@automaton" $\rr$ such that $\gg(\rr)\in \WW$).	
\AP The ""language accepted"" (or recognised) by an automaton $\A$ is the set 
\[ \intro*\Lang{\A}:= \{ w\in \Sigma^\oo \mid w \text{ is "accepted@@word" by } \A \}.\]
\AP Two "automata" "recognising" the same language are said to be ""equivalent@@aut"".

We remark that if $\A$ is "deterministic" (resp. "complete"), there is at most one (resp. at least one) "run over" $w$ for each $w\in \SigmaInfty$. 

\AP Given a "subgraph" $G'$ of the "underlying graph" of an "automaton" $\A$ and a subset of states $I'$ in $G'$, the ""subautomaton induced by"" $G'$ with "initial states" $I'$ is the "automaton" having $G'$ as "underlying graph", $I'$ as set of initial states, and whose "acceptance condition" and "labelling with input letters" are the restrictions of those of $\A$.

\subparagraph*{History-deterministic automata.} 
Let $\A$ be a (non-deterministic) "automaton" over $\SS$ with $\DD$ as set of transitions and $I$ as set of "initial states".
\AP A ""resolver@@aut"" for $\A$ is a pair $(r_0,r)$, consisting of a choice of an "initial state",\footnotemark{} $r_0\in I$, and a function $r\colon \DD^* \times \SS \to \DD$ such that for all words $w = w_0w_1\dots \in \SS^\oo$, the sequence $e_0e_1\dots \in \DD^\oo$, called the ""run induced by@@aut"" $r$ over $w$ and defined by $e_i = r(e_0\dots e_{i-1}, w_i)$ is actually a "run over" $w$ in $\A$ starting from $r_0$.
We say that the resolver is ""sound@@aut"" if it satisfies that for every $w\in \Lang{\A}$, the "run induced by@@aut" $r$ over $w$ is an "accepting run". 
In other words, $r$ should be able to construct an "accepting run" in $\A$ letter-by-letter with only the knowledge of the word so far, for all words in $\Lang \A$.
\AP An "automaton" $\A$ is called ""history-deterministic"" (shortened HD, also called \emph{good-for-games} in the literature) if there is a "sound resolver@@aut" for it.

\footnotetext{Sometimes in the literature~\cite{Boker2013NondetUnknownFuture,BL23SurveyHD,HP06} the initial state $r_0$ is not required to be specified. This would permit to choose it after the first letter $w_0$ is given. We consider that a resolver constructing a run without guessing the future should pick the initial state before the first letter is revealed, hence the introduction of $r_0$ in the definition of a resolver. The suitability of this choice will be further supported by the generalisation of "HD automata" to "HD mappings" (Section~\ref{subsect-morp: hist-det-morphisms}).}

\begin{remark}\label{rmk-prel:deterministic-implies-HD}
	"Deterministic" "automata" are "history-deterministic", and they admit a unique "resolver".
\end{remark}

\begin{example}\label{ex-prelim:HD-aut}
	In Figure~\ref{fig-prelim:HD-aut-example}, we show an "automaton" $\A$ over $\SS = \{a,b,c\}$ that is not "deterministic" (as it has two $b$-transitions from $q_1$) but is "history-deterministic". Its set of "output colours" is $\GG=\{1,2\}$ and its "acceptance set" is $\WW = \{u\in \{1,2\}^\oo \mid u \text{ contains finitely many } 1s\}$ (this is a "coB\"uchi" condition, as introduced in the next section).
	It is easy to check that $\A$ "recognises" the language:
	\[ \Lang{\A} = \{w\in \SS^\oo \mid \minf(w)\subseteq \{a,b\} \tor \minf(w)\subseteq \{b,c\} \}. \]
	
	A "resolver@@aut" for $\A$ only has to take a decision when the automaton is in the state $q_1$ and letter $b$ is provided. In this case, a "sound@@aut" "resolver@@aut" is obtained by using the following strategy: if the last letter seen was $a$, we take the transition leading to state $q_0$; if it was $c$, we take the transition leading to $q_2$. This strategy ensures that, if eventually only letters in $\{a,b\}$ (resp. $\{b,c\}$) are seen, the "run" will end up in state $q_0$ (resp. $q_2$) and remain there indefinitely, without producing any colour $1$.
	
	\begin{figure}[ht]
		\centering
		\includegraphics[width=0.5\textwidth]{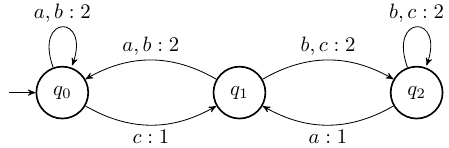}
		\caption{An example of a "history-deterministic" automaton that is not "deterministic". 		
		The "acceptance set" is $\WW = \{u\in \{1,2\}^\oo \mid u \text{ contains finitely many } 1s\}$.		
		An arrow of the form $q\re{a,b:2} q'$ represents two different transitions with input letters $a$ and $b$, respectively. The "initial state" $q_0$ is marked with one incoming arrow. }
		\label{fig-prelim:HD-aut-example}
	\end{figure}
	
\end{example}

\AP We say that a state $q$ is ""reachable using the resolver $(r_0, r)$"" if there is a finite run $\rr = r_0 \lrpE q$ such that $\rr$ is the "run induced by@@aut" $r$ over some word $w\in\SS^*$.

The next remark indicates that we can assume without loss of generality that all states in an "HD automaton" are "reachable using@@resolver" some "sound@@aut" "resolver@@aut".
\begin{remark}\label{rmk-prelim:restriction-to-accessible}
	Let $\A$ be an "HD automaton", let $(r_0,r)$ be a "sound@@aut" "resolver@@aut" for it and let $\tilde{\A}$ be the "subautomaton" induced by the set of states "reachable using $(r_0, r)$", with initial state $r_0$. Then, $\Lang{\A} = \Lang{\tilde{\A}}$.
\end{remark}

The next lemma provides a simplification for "automata" "recognising" "prefix-independent" languages. Its proof can be found in Appendix~\ref{sec:appendix-prefix-independent}. Together with Remark~\ref{rmk-prelim:restriction-to-accessible}, it indicates that when dealing with "HD automata" for this kind of languages, we can assume that any state of the automaton is the "initial" one.
In particular there will be no need to specify the "initial states" of "subautomata induced by subgraphs" of "HD automata" "recognising" "prefix-independent" languages.

\begin{restatable}{lemma}{prefixIndepAutomaton}\label{lemma-prelims:prefix-indep-automaton}
	Let $\A$ be a "history-deterministic" "automaton" "recognising" a "prefix-independent" language and using as "acceptance set" a "prefix-independent" language. For any state $q$ of $\A$ that is "reachable using some sound resolver", it holds that $\A$ "recognises" the same language if we fix $q$ as "initial state", that is, $\Lang{\A} = \Lang{\initialTS{\A}{q}}$. Moreover, $\initialTS{\A}{q}$ is also "history-deterministic".  In particular, if~$\A$ is "deterministic", this is the case for any "reachable" state $q$.
\end{restatable}

\subparagraph*{Games.} \AP A ""game"" is a "vertex-labelled" "transition system" $\G = \left(\underlyingGraph{\G}, \macc{\G}, (\intro*\lPlayers, \{\Eve, \Adam\})\right)$, with $\underlyingGraph{\G} = (V,E,\msource, \mtarget, I)$ a "pointed graph", 
and $l_{\mathrm{Players}}\colon V \to \{\Eve, \Adam\}$ a "vertex-labelling" function inducing a partition of $V$ into vertices controlled by two ""players"" that we refer to as ""Eve"" and ""Adam"". We let $\intro*\VEve = \inv{\lPlayers}(\Eve)$ and $\intro*\VAdam = \inv{\lPlayers}(\Adam)$.

During a play, players move a token from one vertex to another for an infinite amount of time. 
\AP The player who owns the vertex $v$ where the token is placed chooses an edge in $\mout(v)$ and the token travels through this edge to its "target". In this way, they produce an infinite "run" $\rr$ on $\G$ (that we also call a ""play""). 
The objective of Eve is to produce an "accepting run" (a sequence of "colours" in $\WW$), and Adam tries to prevent it.

\AP A ""strategy"" from $v\in V$ for "Eve" is a (partial) function $\strat_{v}: \PathSetFin{\G}{v}  \partialF E$, defined for finite paths from $v$ ending in a vertex in $\VEve$, that tells "Eve" which move to choose after any possible "finite play".
\AP We say that a "play" $\rr\in \PathSetInfty{\G}{v}$ is ""consistent with the strategy"" $\strat_v$ if after each finite prefix $\rr' \prefix \rr$ ending in a vertex controlled by "Eve", the next edge in $\rr$ is $\strat_v(\rr')$.
\AP  We say that $\strat_v$ is a ""winning strategy for Eve"" if all infinite "plays" from $v$ "consistent with" $\strat_v$ are "accepting@@run". We say that Eve ""wins"" the game $\G$ from $v$ if there is a "winning strategy" from $v$ for her. Strategies for Adam are defined dually.

\AP Given a game $\G$, the ""winning region"" of $\G$ for "Eve", written $\intro*\winRegion{\Eve}{\G}$, is the set of "initial vertices" $v\in I$ such that she "wins" the game $\G$ from $v$. 
\AP The ""full winning region"" of $\G$ for "Eve" is her "winning region" in the "game" $\initialTS{\G}{V}$ where all vertices are initial, that is, the set of vertices $v\in V$ such that "Eve" "wins" the game $\G$ from $v$. 

In some proofs, we will need to take a close look into the "strategies" used in games, for which we need to introduce \emph{finite memory strategies}.
\AP For a set $X$ (usually the set of edges of a "game"), we define a ""memory skeleton"" over $X$ as an "edge-labelled@@graph" "pointed graph" $\M = (M, E_M, \msource, \mtarget, m_0)$ with a single "initial state" $m_0$ and labels $l_M\colon E_M\to X$ inducing a deterministic structure, that is, satisfying that for each $m\in M$ and $x\in X$ there is at most one transition $e\in \mout(m)$ labelled~$x$.
\AP We denote $\intro*\transMem\colon M \times X \partialF M$ the ""update function@@mem"" given by $\transMem(m, x) = m'$ if $m\re{x}m'$ is the (only) transition from $m$ labelled~$x$. We extend $\transMem$ to $\transMem\colon M\times X^* \partialF M$ by induction ($\transMem(m,\ee) = m$ and $\transMem(m,x_1\dots x_n) = \transMem(\transMem(m,x_1\dots x_{n-1}),x_n)$).
\AP A ""memory structure"" (for "Eve") for  a "game" $\G$ is a "memory skeleton" over the set $E$ of edges of $\G$ together with a ""next-move function"" $\intro*\nextmove\colon \VEve\times M \to E$.
\AP We say that $(\M, \nextmove)$ ""implements@@memStrat"" a "strategy" $\strat_{v}: \PathSetFin{\G}{v}  \partialF E$  if for any "finite play" $\rr\in \PathSetFin{\G}{v}$ ending in $\VEve$, $\strat_{v}(\rr) = \nextmove(\mtargetPath(\rr), \transMem(m_0, \rr))$. We remark that a  "memory structure" for $\G$ implements at most one "strategy" from a given vertex.
We say that $\strat_{v}$ is a ""finite memory strategy"" if it can be "implemented by@@mem" a finite "memory structure".

\subparagraph*{Composition of a transition system and an automaton.}
We now present the construction of the composition (or product) of a "transition system" with an "automaton", which constitutes the standard method for transforming a "transition system" that uses an "acceptance set" $\WW_1$ to another one using a different "acceptance set" $\WW_2$. To guarantee the correctness of the resulting "transition system" (that is, that it has the same semantic properties as the original one), the automaton must be "deterministic" or "history-deterministic" (see Propositions~\ref{prop-prelim:composition_automata},~\ref{prop-prelim:composition_games_HD}, and~\ref{prop:HD-iff-GFG}). 

Let $\TS =(\underlyingGraph{\TS}, \macc{\TS})$ be a "transition system", with $\underlyingGraph{\TS} = (V, E, \msource_\TS, \mtarget_\TS, I_{\TS})$ and $\macc{\TS} = (\gg_{\TS}, \SS, \WW_{\TS})$, 
and let $\A=\left(\underlyingGraph{\A}, \macc{\A}, (l_\SS, \SS)\right)$ be a "complete" "automaton" over the alphabet $\Sigma$, where $\underlyingGraph{\A} = (Q,  \DD, \msource_\A, \mtarget_\A, I_{\A})$ and $\macc{\A} = (\gg_{\A}, \GG, \WW_{\A})$. 
\AP The ""composition"" of $\TS$ and $\A$ (also called their \emph{product}) is the "transition system" $\TS \intro*\compositionAut \A$ defined as follows:
\begin{itemize}
	\item The set of vertices is the cartesian product $V\times Q$.
	\item The set of "initial vertices" is $I_{\TS} \times I_{\A}$.
	\item The set of edges $\intro*\edgesProduct$ contains a transition $(v,q)\re{c} (v', q')$ if there is $a\in \SS$ and transitions $e_1 = v\re{a} v'\in E$  and $e_2 = q\re{a:c} q'\in \DD$. It also contains "$\ee$-edges" $(v,q)\re{\ee} (v', q)$ if $v\re{\ee}v'\in E$. Formally, 
	\[ \edgesProduct = \{(e_1,e_2) \in E\times \DD \mid \gg_{\TS}(e_1) = l_\SS(e_2)\} \; \cup \; \{e_1 \in E \mid \gg_{\TS}(e_1) = \ee \} \; \subseteq \; (E \times \DD) \cup E. \]
	\item The "acceptance condition" is inherited from that of $\A$: the "colouring function" $\gamma'\colon \edgesProduct \to \Gamma$ is defined as $\gamma'(e_1,e_2)= \gamma_{\A}(e_2)$, and the "acceptance set" is $\WW_\A \subseteq \Gamma^\oo$.
\end{itemize}

We remark that if $\TS$ does not contain an "uncoloured" "cycle", neither does $\TS \compositionAut \A$. Also, $\TS \compositionAut \A$ does not contain sinks by "completeness" of $\A$.

If $\TS$ is a "labelled transition system", labelled by the functions $l_V$ and $l_E$, we consider $\TS \compositionAut \A$ as a "labelled transition system" with the functions $l_V^\ltimes(v,q) = l_V(v)$ and $l^\ltimes_E(e_1, e_2) = l_E(e_1)$ (resp. $l^\ltimes_E(e_1) = l_E(e_1)$ if $e_1$ is an "uncoloured" edge).

Intuitively, a computation in $\TS \compositionAut \A$ happens as follows: we start from a vertex $v_0\in I_{\TS}$ in $\TS$ and from $q_0 \in I_{\A}$. When we are in a position $(v, q)\in V \times Q$, a transition $e$ between $v$ and $v'$ takes place in $\TS$, producing a letter $a\in \Sigma$ as "output". Then, the automaton $\A$ proceeds using a transition corresponding to $a$, producing an output in $\Gamma$. In this way, a word in $\Gamma^\oo$ is generated, and we can use the "acceptance set" $\WW_\A \subseteq \Gamma^\oo$ of the "automaton" as the "acceptance set" for $\TS \compositionAut \A$.

In particular, we can perform this operation if $\TS$ is an "automaton". We obtain in this way a new "automaton" that uses the "acceptance condition" of $\A$. 

We could, of course, apply this construction to a "game" $\G$, obtaining a new "game" $\G\compositionAut \A$ in which the "player" who makes a move in $\G$ also chooses a transition in $\A$ corresponding to the letter produced by the selected move. However, in most applications, we intend to obtain an asymmetric form of product game in which one player has full control of the transitions of the automaton (we take the point of view of Eve and want her to choose these transitions). For this reason, we restrain the class of games to which we can apply the "product construction" by a non-deterministic automaton. 

\AP We say that a "game" is ""suitable for transformations"" if it satisfies that for every edge $e = v\re{}v'$ such that $v\in \VAdam$, the edge $e$ is "uncoloured" ($\gg(e)=\ee$), $v'\in \VEve$, and $e$ is the only incoming edge to $v'$ ($\mIn(v') = \{e\}$). 
We remark that any game $\G$ can be made "suitable for transformations" with at most a linear blow up on the size by inserting an intermediate "Eve-vertex" in each edge outgoing from an "Adam-vertex". A formal construction, as well as further motivation for this definition, can be found in Appendix~\ref{sec:appendix-games-transformations}.

The following results are well known and constitute the main application of automata composition. They can be seen as corollaries of our results from Section~\ref{subsec-morph:semantic-properties}, which generalise them.

\begin{proposition}[Folklore]\label{prop-prelim:composition_automata}
	Let $\B$ be an automaton with "acceptance set" $\WW_{\B}$  and let $\A$ be an "automaton" "recognising" $\Lang \A=\WW_{\B}$. Then, $\Lang{ \B \compositionAut \A}=\Lang \B$.
	Moreover, if $\A$ and $\B$ are "deterministic" (resp. "history-deterministic"), so is $\B \compositionAut \A$.
\end{proposition}

\begin{proposition}[\cite{HP06}]\label{prop-prelim:composition_games_HD}
	Let $\G$ be a "game" that is "suitable for transformations" with "acceptance set" $\WW_\G$, and let $\A$ be a "history-deterministic" "automaton" "recognising" $\Lang \A=\WW_{\G}$. Then, the "winning region" of Eve in $\G$ is the projection of her "winning region" in $\G\compositionAut \A$, that is, Eve "wins $\G$ from" an "initial vertex" $v$ if and only if she "wins $\G\compositionAut \A$ from $(v, q_0)$", for $q_0$  some initial vertex of $\A$. 
\end{proposition}

Proposition~\ref{prop-prelim:composition_games_HD} fails if the "automaton" is not "HD". In fact, this property characterises "history-determinism", which is the reason why "HD" automata are also called \emph{good-for-games} in the~literature. However, it should be noted that "history-determinism" and good-for-gameness have been generalised to other contexts in which they do not necessarily yield equivalent notions~\cite{BL21HDvsGFG,Colcombet2009CostFunctions}.

\begin{proposition}[\cite{HP06}]\label{prop:HD-iff-GFG}
	Let $\A$ be an "automaton" "recognising" $\WW_{\G}\subseteq \Sigma^\oo$ satisfying that for every "game" $\G$ "suitable for transformations"  with "acceptance set" $\WW_\G$, "Eve" "wins" the "game" $\G$ from an "initial vertex" $v$ if and only if she "wins" $\G \compositionAut \A$ from $(v,q_0)$, for $q_0$  some initial vertex of $\A$. Then, $\A$ is "history-deterministic".
\end{proposition}

\subsection{Muller languages, cycles and the parity hierarchy}\label{subsect-prelim:acceptance-conditions}

\subparagraph{Languages commonly used as acceptance sets.}
We now define the main classes of languages used by "$\oo$-regular automata" as "acceptance sets".
We let $\Gamma$ stand for a finite set of "colours".

\begin{description}	
	\item[Büchi.] \AP Given a subset $B\subseteq \GG$, we define the ""Büchi language associated to"" $B$ as:
	\[ \intro*\BuchiC{B}{\GG}  = \{w\in \GG^\oo \mid \minf(w) \cap B \neq \emptyset\}. \]
	\AP We say that a language $L\subseteq \GG^\oo$ is a ""B\"uchi language"" if there is a set $B\subseteq \GG$ such that $L = \BuchiC{B}{\GG}$.

	
	\item[coB\"uchi.] \AP Given a subset $B\subseteq \GG$, we define the ""coB\"uchi language associated to"" $B$ as:
	\[ \intro*\coBuchiC{B}{\GG}  = \{w\in \GG^\oo \mid \minf(w) \cap B = \emptyset\}. \]
	\AP We say that a language $L\subseteq \GG^\oo$ is a ""coB\"uchi language"" if there is a set $B\subseteq \GG$ such that $L = \coBuchiC{B}{\GG}$.
	
	\item[Rabin.] \AP A Rabin language is represented by a family $R=\{(G_1,R_1),\dots,(G_r,R_r)\}$  of ""Rabin pairs"", where $G_j,R_j\subseteq \Gamma$.
	The ""Rabin language associated to"" $R$ is defined as:
	\[ \intro*\RabinC{R}{\GG} = \{w\in \GG^\oo \mid [\minf(w)\cap G_j \neq \emptyset \text{ and } \minf(w)\cap R_j = \emptyset] \text{ for some index } j \}. \]
	\AP If $[\minf(w)\cap G_j \neq \emptyset \text{ and } \minf(w)\cap R_j = \emptyset]$, we say that $w$ is ""accepted by the Rabin pair"" $(G_j, R_j)$.
	\AP We say that a language $L\subseteq \GG^\oo$ is a ""Rabin language"" if there is a family of "Rabin pairs" $R$ such that $L = \RabinC{R}{\GG}$.
	
	
	\item[Streett.] \AP The ""Streett language associated to"" a family  $S=\{(G_1,R_1),\dots,(G_r,R_r)\}$ of "Rabin pairs" is defined as:
	\[ \intro*\StreettC{S}{\GG} = \{w\in \GG^\oo \mid [\minf(w)\cap G_j \neq \emptyset \text{ implies } \minf(w)\cap R_j \neq \emptyset] \text{ for all indices } j \}. \]
	
	\AP We say that a language $L\subseteq \GG^\oo$ is a ""Streett language"" if there is a family of "Rabin pairs" $S$ such that $L = \StreettC{S}{\GG}$.
	
	\item[Parity.] \AP We define the \emph{parity language} over the alphabet $[d_{\min}, d_{\max}]\subseteq \NN$ 
	 as:
	\[ \intro*\parity_{[d_{\min}, d_{\max}]} = \{ w\in [d_{\min}, d_{\max}]^\oo \mid \min \minf(w) \text{ is even}\}.\]
	\AP We say that a language $L\subseteq \GG^\oo$ is a ""$[d_{\min}, d_{\max}]$-parity language"" if there is a mapping $\phi\colon \GG \to [d_{\min}, d_{\max}]$ such that for all $w\in \GG^\oo$, $w\in L$ if and only if $\phi(w)\in \parity_{[d_{\min}, d_{\max}]}$. We say that $L$ is a ""parity language"" if there are $d_{\min}, d_{\max}\in \NN$ such that $L$ is a "$[d_{\min}, d_{\max}]$-parity language".

	\item[Muller.] \AP We define the ""Muller language associated to"" a family $\F\subseteq \powplus{\Gamma}$ of non-empty subsets of $\Gamma$ as:
	\[ \intro*\MullerC{\F}{\GG} = \{w\in \GG^\oo \mid \minf(w) \in \F\}. \]
	\AP We say that a language $L\subseteq \GG^\oo$ is a ""Muller language"" if there is a family $\F\subseteq \powplus{\Gamma}$ such that $L = \MullerC{\F}{\GG}$. 
\end{description}

We drop the subscript $\GG$ (resp. $[d_{\min}, d_{\max}]$) whenever the set of colours is clear from the context. 
We remark that all languages of the classes above are "prefix-independent" (for all $w\in \GG^\oo$ and $u\in \GG^*$,  $uw\in L$ if and only if $w\in L$).

\AP We say that an "acceptance condition" (resp. "transition system", "automaton") is an ""$X$ condition"" (resp. $X$ transition system, $X$ automaton), for $X$ one of the classes of languages above, if its "acceptance set" is an $X$ language. 
In the case of "parity" "transition systems", we will always assume that the set of "colours" is a subset of $\NN$ and $\phi$ is the identity function.

\AP We let ""DPA"" stand for "deterministic" "parity" "automaton" and ""DMA"" for "deterministic" "Muller" "automaton".

We discuss further classes of languages in Appendix~\ref{sec:appendix-weak-conditions} ("generalised B\"uchi" and "coB\"uchi@@generalised" languages, as well as "generalised weak" acceptance). We refer to the survey~\cite{Boker18WhyTypes} for a more detailed account on different types of acceptance conditions.

\begin{remark}[Inclusions between classes]
	We observe that there are many inclusions between the classes of languages that we have introduced. For example, "Büchi languages" are exactly "$[0,1]$-parity languages", and "parity languages" are "Rabin languages"~\cite{Mostowski1984RegularEF}. 
	In particular, all  classes above are special cases of "Muller languages".
	The relations between these classes of languages are outlined in Figure~\ref{fig-prelim:classes-acceptance}. 
\end{remark}

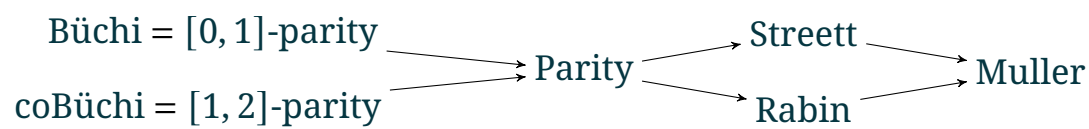
\begin{figure}[ht]
	\centering
	\begin{tikzpicture}[square/.style={regular polygon,regular polygon sides=4}, align=center,node distance=2cm,inner sep=3pt]
		\node at (-1.5,0.5)  (Buchi) {"B\"uchi" $=$ "$[0,1]$-parity"};
		\node at (-1.7,-0.5)  (co-Buchi) {"coB\"uchi" $=$ "$[1,2]$-parity"};
		\node at (3.4,0)  (Parity) {"Parity"};
		\node at (6.3,0.5)  (Streett) {"Streett"};
		\node at (6.3,-0.5)  (Rabin) {"Rabin"};
		\node at (9.3,0)  (Muller) {"Muller"};
	
		\draw
		(Buchi) edge[->] (Parity)
		(co-Buchi) edge[->] (Parity)
		(Parity) edge[->] (Streett)
		(Parity) edge[->] (Rabin)
		(Streett) edge[->] (Muller)
		(Rabin) edge[->] (Muller);	
	\end{tikzpicture}
	\caption{Relations between subclasses of "Muller languages". An arrow from a class $X$ towards a class $Y$ means that if a language $L\subseteq \GG^\oo$ is an "$X$ language", then it is also a "$Y$ language". Arrows obtained by transitivity have been omitted. Inclusions are strict: if an arrow from $X$ to $Y$ cannot be obtained by transitivity, then there are "$X$ languages" that are not "$Y$ languages"~\cite{Zielonka1998infinite}.}
	\label{fig-prelim:classes-acceptance}
\end{figure}

\begin{remark}\label{rmk-prelim:characterisation_Muller_Languages}
	A language $L\subseteq \GG^\oo$ is a "Muller language" if and only if it satisfies:
	\[ \text{For all } w,w'\in \GG^\oo, \tif \minf(w) = \minf(w'), \text{ then } w\in L \iff w' \in L. \]
\end{remark}

\begin{remark}[Representation of acceptance conditions]\label{rmk-prelim:representation_Muller_Languages}
	\AP In practice, there exists a variety of ways to represent "Muller languages" and "acceptance conditions" of "automata": using boolean formulas (Emerson-Lei conditions), as a list of accepting subsets of edges, etc.
	The complexity and practicality of algorithms manipulating "automata" and "games" may greatly differ depending on the representation of their "acceptance conditions"~\cite{Horn2008Explicit, Dawar2005ComplexityBounds}.	
	However, in this work, we are mostly interested in the expressive power of "acceptance conditions", and the results we present will not depend on how they are represented. 
\end{remark}

\begin{example}\label{ex-prelim:multiple-automata-examples}
	In Figure~\ref{fig-prelim:multiple-aut-ex} we show three different types of "automata" over the alphabet $\SS=\{a,b\}$ "recognising" the language
	\[ L = \{ w \in \SS^\oo \mid w=ub^\oo \text{ or } \left(w=ua^\oo \text{ and } u\text{ has an even number of `b's }\right) \}. \]
\end{example}
	\begin{figure}[ht]
	\centering 
	\includegraphics[width=0.8\textwidth]{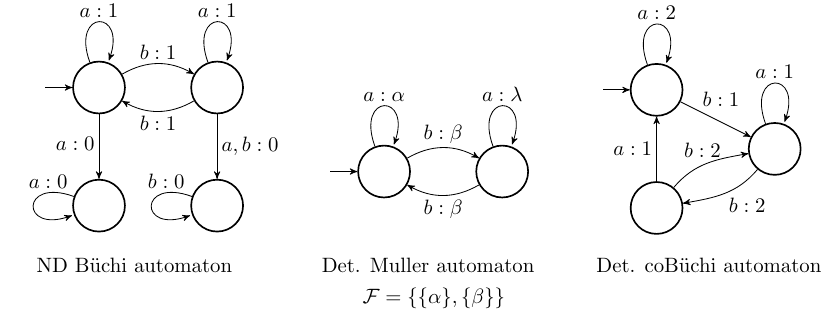}
	\caption{Different types of "automata" "recognising" the language $L = \{ w \in \SS^\oo \mid w=ub^\oo \text{ or } \left(w=ua^\oo \text{ and } u\text{ has an even number of `b's }\right) \}$.}
	\label{fig-prelim:multiple-aut-ex}
\end{figure}

\subparagraph{$\omega$-regular languages.} The class of "$\oo$-regular languages" plays a central role in the theory of formal languages and verification. The significance of "$\oo$-regular languages" is (partly) due to the robustness of its definition, as they admit multiple equivalent characterisations relating different areas of study.

\begin{proposition}[\cite{Mostowski1984RegularEF,Muller1963infSequences}]
	Let $L\subseteq \SS^\oo$ be a language of infinite words. The following properties are equivalent:
	\begin{itemize}
		\item $L$ can be "recognised" by a non-deterministic "B\"uchi" "automaton".
		\item $L$ can be "recognised" by a "deterministic" "parity" "automaton".
		\item $L$ can be "recognised" by a "non-deterministic" "Muller" "automaton".
	\end{itemize}
\end{proposition}
\AP A language satisfying the previous conditions is called ""$\oo$-regular@@language"". 
Many other equivalent definitions exist. Notably, "$\oo$-regular languages" are exactly the languages that can be defined using \emph{monadic second-order logic}~\cite{Buchi1960decision}, those that can be described by using \emph{$\oo$-regular expressions}~\cite{McNaughton1966Testing}, and those that can be "recognised" by an \emph{$\oo$-semigroup}~\cite[Chapter~2]{PP04InfiniteWords}.

\subparagraph*{Cycles.} Let $\TS$ be a "transition system" with $V$ and $E$ as set of vertices and edges, respectively. A ""cycle"" of $\TS$ is a subset $\ell \subseteq E$ such that there is a finite "path" $v_0 \re {e_0}v_1 \re {e_1}v_2 \re{} \dots v_r \re {e_r} v_0$ with $\ell = \{e_0, e_1, \dots, e_r\}$. 
We remark that we do not require this "path" to be simple, that is, edges and vertices may appear multiple times. 
\AP The set of ""states of the cycle"" $\ell$ is $\intro*\states{\ell} = \{v_0, v_1, \dots v_r\}$. 
\AP The set of "cycles" of a "transition system" $\TS$ is written $\intro*\cycles{\TS}$. 
We will consider the set of "cycles" ordered by inclusion.
\AP For a state $v\in V$, we denote $\intro*\cyclesState{\TS}{v}$ the subset of "cycles" of $\TS$ "containing@@cycle" $v$.
We remark that a vertex $v$ is "recurrent" if and only if $\cyclesState{\TS}{v}\neq \emptyset$.
We note that $\cyclesState{\TS}{v}$ is closed under union; moreover, the union of two "cycles" $\ell_1, \ell_2\in \cycles{\TS}$ is again a "cycle" if and only if there is some state $v$ such that both $\ell_1$ and $\ell_2$ contain $v$.

Let $\TS$ be a "Muller" "transition system" with "acceptance condition" $(\gg, \GG, \MullerC{\F}{\GG})$.
\AP Given a "cycle" $\ell\in \cycles{\TS}$, we say that $\ell$ is ""accepting@@cycle"" (resp. ""rejecting@@cycle"") if $\gg(\ell) \in \F$ (resp. $\gg(\ell) \notin \F$). 
We remark that the maximal "cycles" of a "transition system" are exactly the sets of edges of its "strongly connected components". In particular, we can apply the adjectives ""accepting@@SCC"" and ""rejecting@@SCC"" similarly to the "SCCs" of a "Muller" "transition system".

We note that, by definition, the "acceptance@@run" of a "run" in a "Muller" "transition system" only depends on the set of transitions taken infinitely often. For any infinite "run" $\rr \in \Runs{\TS}$, the set of transitions taken infinitely often forms a "cycle", $\minf(\rr) = \ell_\rr\in \cycles{\TS}$, and $\rr$ is an "accepting run" if and only if $\ell_\rr$ is an "accepting cycle".

\subparagraph{The deterministic and history-deterministic parity hierarchy.}
As we have mentioned, every "$\oo$-regular language" can be recognised by a "deterministic" "parity" "automaton", but the number of colours required to do so might be arbitrarily large. We can assign to each "$\oo$-regular language" the optimal number of colours needed to recognise it using a "deterministic automaton". We obtain in this way the  \AP""deterministic parity hierarchy"", having its origins in the works of Wagner~\cite{Wagner1979omega}, Kaminski~\cite{Kaminski85classification}, and Mostowski~\cite{Mostowski1984RegularEF}. We represented this hierarchy in Figure~\ref{fig-prelim:ParityHierarchy}. This hierarchy is strict, that is, for each level of the hierarchy there are languages that do not appear in lower levels~\cite{Wagner1979omega}. It is known that we can decide in polynomial time the "parity index" of an "$\oo$-regular language" represented by a "deterministic" "parity" "automaton"~\cite{CartonMaceiras99RabinIndex}, but this problem is $\NPcomplete$ if the language is given by a "deterministic" "Rabin" or "Streett" "automaton"~\cite{KPB95Structural}.

\begin{figure}[ht]
	\centering 
	\begin{tikzpicture}[square/.style={regular polygon,regular polygon sides=4}, align=center,node distance=2cm,inner sep=3pt]
		
		\node at (0,0)  (00) {$[0,0]$};
		\node at (4,0)  (11) {$[1,1]$};
		\node at (2,0.75)  (Weak1) {$\WeakIndex{1}$};
		
		\node at (0,1.5)  (01) {$[0,1]$};
		\node at (4,1.5)  (12) {$[1,2]$};
		\node at (2,2.25)  (Weak2) {$\WeakIndex{2}$};
		
		\node at (0,3)  (02) {$[0,2]$};
		\node at (4,3)  (13) {$[1,3]$};
		\node at (2,3.75)  (Weak3) {$\WeakIndex{3}$};			
		
		\node at (0,4.25)  (dotsLeft) {$\vdots$};
		\node at (4,4.25)  (dotsRight) {$\vdots$};
		\node at (2,4.5)  (dotsRight) {$\vdots$};

		\draw   
		(00) edge (Weak1)
		(11) edge (Weak1)
		(Weak1) edge (01)
		(Weak1) edge (12)
		
		(01) edge (Weak2)
		(12) edge (Weak2)
		(Weak2) edge (02)
		(Weak2) edge (13)
		
		(02) edge (Weak3)
		(13) edge (Weak3);
		
	\end{tikzpicture}
	\caption{The (history-)deterministic parity hierarchy.}
	\label{fig-prelim:ParityHierarchy}
\end{figure}
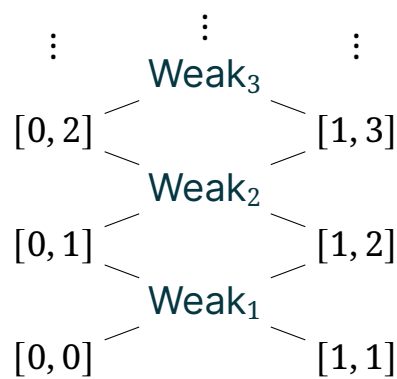

\begin{definition}[Parity index of a language]\label{def:parity_index_language}
	Let $L\subseteq \SS^\oo$ be an "$\oo$-regular language". 
	\AP We say that $L$ has ""parity index at least"" $[0, d-1]$ (resp. $[1, d]$) if any "DPA" "recognising" $L$ with a "parity" "acceptance condition" over the set of colours $[d_{\min}, d_{\max}]$ satisfies that $d_{\max} - d_{\min} \geq d-1$, and in case of equality $d_{\min}$ is even (resp. odd).
	We say that the ""parity index"" of $L$ is $[0, d-1]$ (resp. $[1, d]$) if, moreover, there is a "DPA" "recognising" $L$ with a "parity" "acceptance condition" over the set of colours $[0,d-1]$ (resp. $[1,d]$). 
	
	\AP We say that $L$ has \emph{parity index at least} $\intro*\WeakIndex{d}$ if any "DPA" "recognising" $L$ with a "parity" "acceptance condition" over the set of colours $[d_{\min}, d_{\max}]$ satisfies that $d_{\max} - d_{\min} \geq d$. We say that the parity index of $L$ is $\WeakIndex{d}$ if, moreover, there are "DPAs" $\A_1$ and $\A_2$ recognising $L$ with "parity" "acceptance conditions" over the sets of colours $[0,d]$ and $[1,d+1]$, respectively.  
	
\end{definition}

If follows from the definition that for each "$\oo$-regular language" $L$, there is a unique $d$ such that either $L$ has "parity index" $[0,d-1]$, $[1,d]$ or $\WeakIndex{d}$, and these options are mutually exclusive. See also Appendix~\ref{sec:appendix-weak-conditions} for more details about languages of "parity index" $\WeakIndex{d}$.

One of our contributions is to show that the "parity index" also applies to "Muller" "automata": any "deterministic" or "HD" "Muller" "automaton" "recognising" an "$\oo$-regular language" of "parity index" $[0, d-1]$ uses at least $d$ different colours (Proposition~\ref{prop-typ:min-colours-Muller-automata}).

The following proposition states that the notion of "parity index" of a language does not change by using "HD" automata instead of "deterministic" ones in the definition.
However, for "non-deterministic automata", the hierarchy collapses at level $[0,1]$ ("Büchi" "automata")~\cite{McNaughton1966Testing}.

\begin{proposition}[{\cite[Theorem~19]{BL19GFGFromND}}]\label{prop-prelim:parity_index-HDAutomata}
	Let $\A$ be an "HD" "parity" "automaton" recognising a language~$L$, and assume that the "parity index" of $L$ is $[0,d-1]$ (resp. $[1,d]$). Then, the "acceptance condition" of $\A$ uses at least $d$ output colours, and if it uses exactly $d$ colours, the least of them is even (resp. odd). If the "parity index" of $L$ is $\WeakIndex{d}$, then $\A$ uses at least $d+1$ output colours.
\end{proposition}

We show next that the "parity index" of an "$\oo$-regular language" can be read directly from a "deterministic" "Muller" "automaton". 
	
Let $\TS$ be a "transition system" using the "Muller@@condition" "acceptance condition" $(\gg, \GG, \MullerC{\F}{\GG})$. 
\AP A ""$d$-flower"" over a state $v$ of $\TS$ is a set of $d$ "cycles" $\ell_1, \ell_2, \dots, \ell_{d}\in \cyclesState{\TS}{v}$ such that $\ell_i \supsetneq \ell_{i+1}$ and $\gg(\ell_i)\in \F \iff \gg(\ell_{i+1})\notin \F$. We say that it is a ""positive flower""  if $\gg(\ell_{1})\in \F$ and that it is ""negative@@flower"" otherwise.

\begin{lemma}[Flower Lemma, \cite{NiwinskiWalukievicz1998Relating, Wagner1979omega}]\label{lemma:flower-lemma}
	Let $\A$ be a "DMA". If $\A$ admits an "accessible" "positive@@flower" (resp. "negative@@flower") "$d$-flower", then $\Lang{\A}$ has "parity index at least" $[0,d-1]$ (resp. $[1,d]$). If $\A$ admits both "accessible" "positive@@flower" and "negative@@flower" "$d$-flowers", then $\Lang{\A}$ has "parity index at least" $\WeakIndex{d}$.
	
	Conversely, if an "$\oo$-regular language" $L$ has "parity index at least" $[0,d-1]$ (resp. $[1,d]$), then any "DMA" "recognising" $L$ admits a "positive@@flower" (resp. "negative@@flower") "$d$-flower".
\end{lemma}

\subsection{Trees}\label{subsect-prelim:trees}
We introduce some technical notations that will be used to define "automata" based on the "Zielonka tree" (Sections~\ref{subsec-zt: parity automaton} and~\ref{subsec-zt: GFG-Rabin}) and the transformations based on the "ACD" (Sections~\ref{subsec-acd: parity-transformation} and~\ref{subsec-acd: HD-Rabin-transformation}).

\AP A ""tree""~$T=(N,\ancestor)$ is a non-empty finite set of ""nodes""~$N$ equipped with an order relation~$\intro*\ancestor$ called the ""ancestor relation"" (we say that $x$ is an ""ancestor"" of $y$, or that $y$ is below $x$ if $x\ancestor y$), such that
(1) there is a minimal "node" for~$\ancestor$, called ""the root"", and
(2) the "ancestors" of an element are totally ordered by~$\ancestor$.
\AP The converse relation $\intro*\descendant$ is the ""descendant"" relation.
\AP Maximal nodes are called ""leaves"", and the set of leaves of $T$ is denoted by $\intro*\leaves(T)$. 
\AP The minimal strict "descendants" of a node are called its ""children"". The set of children of $n$ in $T$ is written $\intro*\children_T(n)$.
\AP The ""depth"" of a node $n$ is the number of strict "ancestors" of it. We note it $\intro\depth(n)$.
\AP The ""height"" of a tree~$T$ is the maximal length of a chain for the "ancestor" relation. 
\AP A ""subtree""  of $T=(N,\ancestor)$ is a "tree" $T'=(N', \ancestor')$ such that $N'\subseteq N$, $\ancestor'$ is the restriction of $\ancestor$ to $N'$ and $\children_{T'}(n') \subseteq \children_T(n')$ for all $n'\in N'$.
\AP Given a "node"~$n$ of a tree~$T$, the ""subtree of~$T$ rooted at~$n$"" is the "subtree" of $T$ whose nodes are the nodes of $T$ that have~$n$ as ancestor. 
\AP A ""branch"" is a maximal chain of the order $\ancestor$.

\AP An ""ordered tree"" is a "tree" $T = (N, \ancestor)$ together with a total order $\leq_n$ over $\children_T(n)$, for each node $n\in N$ that is not a "leaf". We remark that a "subtree" of an "ordered tree" can be seen as an "ordered tree" with the restrictions of these total orders to the existing children. 
\AP These orders induce a total order $\intro*\orderTree{T}$ on $T$ (the depth-first order): let $n, n'\in N$. If $n \ancestor n'$, we let $n \leq_T n'$. If $n$ and $n'$ are incomparable for the "ancestor" relation, let $n_m$ be the "deepest" common ancestor, and let $n_1, n_2\in \children_{T}(n_m)$ such that $n_1\ancestor n$ and $n_2 \ancestor n'$. We let $n \leq_T n'$ if and only if $n_1 \leq_{n_m} n_2$.
\AP In the latter case, we say that $n$ is ""on the left of"" $n'$.

We will make use of these orders through some auxiliary functions.
The function $\nextChild(n)$ gives the next sibling of $n$ in the tree, in a cyclic order. Two examples are shown on the left of Figure~\ref{fig-prelim:trees}.
The function $\jump(n, n_m)$ (for $n_m$ an "ancestor" of $n$) outputs the node given by the following  procedure: we go up the tree from $n$ to $n_m$; then, we change to the next branch below $n_m$ (in a cyclic way) and go down again taking the leftmost "leaf" below it. 
Examples are given on the right of Figure~\ref{fig-prelim:trees}.

\begin{figure}[ht]
	\hspace{3mm}
	\begin{minipage}[b]{0.45\textwidth} 
		\includegraphics[scale=1.1]{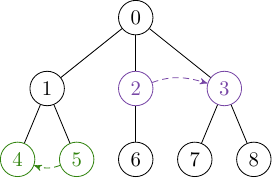}
  \newline
		$\nextChild(5)=4$ and $\nextChild(2)=3$.
	\end{minipage} 
	\hspace{2mm}
	\begin{minipage}[b]{0.45\textwidth} 
		\includegraphics[scale=1.1]{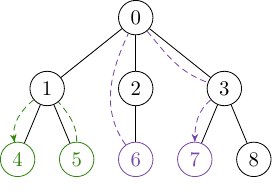}
  \newline
		$\jump(5,1)=4$ and $\jump(6,0)=7$.
	\end{minipage} 
	\hspace{0mm}
	\caption{Illustration of the functions $\nextChild$ and $\jump$.}
	\label{fig-prelim:trees}
\end{figure}

We give the formal definition now. We also need to define these notions taking into account some "subtree" $T'$ of $T$: the input can be any node in $T$, but the final output is restricted to be a node in $T'$. 
Examples~\ref{example-ZT:zielonka-tree} and~\ref{ex-acd:example-acd} further illustrate these notations.

\AP  Let $T'$ be a "subtree" of $T$ and $n'$ a node of $T'$ that is not a "leaf" in $T'$. For $n\in \children_T(n')$, we let 
\[ \intro*\nextChild_{T'}(n)=\begin{cases}
	\min_{\leq_{n'}}\{n''\in \children_{T'}(n') \mid n <_{n'} n''\} \; \text{ if this set is not empty},\\[2mm]
	\min_{\leq_{n'}}\{ n''\in \children_{T'}(n') \} \tow.
\end{cases} \]
That is, the function $\nextChild_{T'}$ maps each child of~$n'$ to a sibling that
is its successor in $T'$ for the $\leq_n$-order, in a cyclic way.

Let $T'=(N',\ancestor')$ be a "subtree" of $T = (N,\ancestor)$. Let $n'\in N'$
and $n\in N$ such that $n'$ is a (non-strict) "ancestor" of $n$ ($n' \ancestor n$).
\AP If $n'$ is a "leaf" of $T'$, we define  $\intro*\jump_{T'}(n, n') = n'$.
For $n'=n$, we define $\jump_{T'}(n, n')$ to be the "leftmost" "leaf" of $T'$ "below@@tree" $n'$.
In any other case, we define  $\jump_{T'}(n, n') = l_{\mathrm{dest}}\in \leaves(T')$ to be the only node satisfying that there are two children of $n'$ in $T$, $n_1, n_2\in \children_T(n')$ such that:
\begin{itemize}
	\item $n_1 \ancestor n$,
	\item $n_2 = \nextChild_{T'}(n_1)$ (in particular, $n_2\in N'$),
	\item $l_{\mathrm{dest}} \descendant n_2$ is the "leftmost@@tree"\footnotemark{} "leaf" in $T'$ (minimal for $\orderTree{T'}$) below $n_2$.
\end{itemize}
\footnotetext{The choice of the leftmost "leaf" is arbitrary. In all our uses of the function $\jump$, it could be replaced by any "leaf" below~$n_2$.}

We remark that $n_1=n_2$ if $n_1$ is the only child of $n'$ in $T'$.

\AP An ""$A$-labelled (ordered) tree"" is an ("ordered@@tree") "tree" $T$ together with a labelling function $\nu \colon N \to A$.
\AP A set of "trees" is called a ""forest"".

 	
	\section{Morphisms as witnesses of transformations}\label{section:morphisms}
	As mentioned in the introduction, all existing methods transforming a "Muller" into a "parity" "automaton" follow a common approach: they turn each state $q$ into multiple states of the form $(q,x)$, where $x$ stores some information about the "acceptance condition". 
It is reasonable to put forward this characteristic as the defining trait establishing that an "automaton" has been obtained as a transformation of another.
In this section, we introduce "morphisms of transition systems", which formalise this idea: a "morphism" $\pp \colon \B\to \A$ witnesses that each state $q\in \A$ has been augmented to $\inv{\pp}(q)$. To ensure that $\B$ is semantically equivalent to $\A$, the morphism has to grant a further guarantee, namely, we need to be able to simulate "runs" of $\A$ in $\B$. We will examine two properties of "morphisms" that allow to do this: "local bijectivity" and "history-determinism@@morp" for mappings.

We note that almost identical notions of morphisms were considered by Sakarovitch~\cite[Section~2]{Sakarovitch98Hidden} and Sakarovitch and de Souza~\cite[Section~2.5]{SS10LexDecomposition} in the context of transducers over finite words.\footnote{We thank Géraud Sénizergues for pointing us to the works of Sakarovitch and De Souza.}
Similar ideas to the ones presented here were used by Colcombet to characterise "history-deterministic automata": an "automaton" is "history-deterministic@@aut" if it is the homomorphic image of a (possibly infinite) "deterministic" automaton for the same language~\cite[Definition~13]{Colcombet2012FormsND}.\\

In all of this section, $\TS = (\underlyingGraph{}, \macc{})$ and $\TS' = (\underlyingGraph{}', \macc{}')$ will stand for "transition systems" with "underlying graphs" $\underlyingGraph{} =(V ,E , \msource , \mtarget, I)$ and $\underlyingGraph{}' =(V' ,E' , \msource' , \mtarget', I')$, and "acceptance conditions" $\macc{} = (\gg, \GG, \WW)$ and  $\macc{}' = (\gg', \GG', \WW')$.

\subsection{Morphisms of transition systems}\label{subsec:morphisms-TS}

\begin{definition}\label{def:morphism_graphs}
	\AP A ""morphism of graphs"" from $G$ to $G'$ is a pair of mappings $\pp = (\pp_V: V \rightarrow V', \pp_E: E \rightarrow E')$ preserving edges, that is:
	\begin{itemize}
		\item $\msource'(\pp_E(e))=\pp_V(\msource(e))$ for every $e\in E$,
			
		\item $\mtarget'(\pp_E(e))=\pp_V(\mtarget(e))$ for every $e\in E$.
	\end{itemize}			
	\AP We say that $\pp$ is a ""morphism of pointed graphs"" if, moreover, it preserves "initial vertices":
	\begin{itemize}
		\item $\pp_V(v_0)\in I'$ for every $v_0\in I$. 
	\end{itemize}
	\AP If $(G,(l_V,L_V),(l_E,L_E))$ and $(G,(l'_V,L'_V),(l'_E,L'_E))$ are "labelled graphs", we say that $\pp$ is a ""morphism of labelled graphs"" if, in addition, $L_V \subseteq L_V'$, $L_E \subseteq L_E'$ and $\pp$ preserves labels:
\begin{itemize}
	\item $l_V'(\pp_V(v))=l_V(v)$ for every $v\in V$,
	
	\item $l_E'(\pp_E(e))=l_E(e)$ for every $e\in V$. 
\end{itemize}
\end{definition}

We will write $\pp\colon G \to G'$ to denote a "morphism@@graphs" $\pp$. We will drop the subscript in $\pp_V$ and $\pp_E$ whenever it can be deduced from its use.
\AP We say that $\pp$ is ""surjective"" (resp. injective) if $\pp_V$ is.

Note that the mapping $\pp_V$ does not completely determine a "morphism@@graphs" $\pp$, as multiple edges might exist between two given vertices. However, if $G$ has no isolated vertices, the mapping $\pp_E$ does determine it. It will be convenient nonetheless to also keep the notation for $\pp_V$.

\AP We remark that the image of a "run" in $G$ by a "morphism of pointed graphs" is a "run" in~$G'$. Therefore, a "morphism of pointed graphs" $\pp\colon G \to G'$ induces a mapping \[\intro*\ppRuns\colon \RunsInfty{G}\to \RunsInfty{G'}.\] 

\begin{definition}\label{def:morphism_transition_systems}
	\AP Let $\TS$ and $\TS'$ be two ("labelled") "transition systems". A ""weak morphism of (labelled) transition systems"" $\pp\colon \TS \to \TS'$ is a "morphism of (labelled) pointed graphs" between their "underlying graphs", $\pp: G \to G'$.
	\AP We say that it is a ""morphism of (labelled) transition systems"" if it ""preserves the acceptance"" of "runs", that is:
	\begin{itemize}
		\item for every "infinite run" $\rr \in \Runs{\TS}$,  $\gg(\rr) \in \WW \; \iff \; \gg'(\ppRuns(\rr)) \in \WW'$. 
	\end{itemize}
\end{definition}

\AP A "morphism of labelled TS" between "automata" (resp. between "games") will be called a ""morphism of automata"" (resp. ""morphism of games"").

\AP We say that a "morphism of TS" $\pp\colon \TS \to \TS'$ is an ""isomorphism@@TS"" if $\pp_V$ and $\pp_E$ are bijective and $\inv{\pp} = (\inv{\pp_V}, \inv{\pp_E})$ is a "morphism" from $\TS'$ to $\TS$. In that case, we say that $\TS$ and $\TS'$ are ""isomorphic@@TS"".

\subsection{Local properties of morphisms}\label{subsect-morp: loc-bijective-morphisms}

\begin{definition}\label{def:loc_bijective_morphisms}
	A "morphism of pointed graphs" $\pp: G \rightarrow G'$ is called:	
	\begin{itemize}
		\item  \AP""Locally surjective"" if it verifies: 
		\begin{enumerate}
			\item For every $v_0'\in I'$ there exists $v_0\in I$ such that $\pp(v_0)=v_0'$.
			\item 
			For every $v\in V$ and every $ e'\in \mout(\pp(v))$
			there exists $e\in \mout(v) $ such that $ \pp(e)=e' $. 
		\end{enumerate}
		
		\item \AP""Locally injective"" if it verifies:
		\begin{enumerate}
			\item For every $v_0'\in I'$, there is at most one $v_0\in I$ such that $\pp(v_0)=v_0'$. 
			
			\item For every $ v\in V$ and every couple $ e_1,e_2\in \mout(v) $,  $\pp(e_1)=\pp(e_2)$ implies $e_1=e_2$.
		\end{enumerate}
		\item \AP""Locally bijective"" if it is both "locally surjective" and "locally injective".
	\end{itemize}
\end{definition}

Equivalently, a "morphism of pointed graphs" $\pp$ is "locally surjective" (resp. "locally injective") if for every $v\in V$ the restriction of $\pp_E$ to $\mout(v)$ is a surjection onto $\mout(\pp(v))$ (resp. an injection into $\mout(\pp(v))$), and the restriction of $\pp_V$ to $I$ is a surjection onto $I'$ (resp. an injection into $I'$).

Let $\pp: \TS \rightarrow \TS'$ be a  "(weak) morphism", and let $\rr'= v_0'\re{e_0'}v_1'\re{e_1'}\dots$ be a "run" in $\TS'$. If $\pp$ is "locally surjective", we can pick an initial vertex $v_0$ in $\inv{\pp}(v_0')$ and build step-by-step a run $\rr$ in $\TS$ from $v_0$ that is sent to $\rr'$ under $\pp$. If $\pp$ is moreover "locally bijective", the choices of the "initial vertex" and the edges at each step are unique, so "runs" in $\TS'$ can be simulated in $\TS$ via $\pp$ in a unique way.
Said differently, if $\pp: \TS \rightarrow \TS'$ is a  "locally bijective morphism", we can see $\TS$ as an "automaton" that processes "runs" of $\TS'$ in a deterministic fashion (this idea is formalised in Section~\ref{subsubsec-acd:optimality-parity-condition}).
This property will allow us to show that a "locally bijective" "morphism@@TS" witnesses the semantic equivalence of $\TS$ and $\TS'$ (see Section~\ref{subsec-morph:semantic-properties}).

We note that the notion of "locally bijective" "morphisms of transition systems" almost coincide with the usual concept of bisimulation. The main difference is that "locally bijective morphisms" treat the "acceptance@@run" of a run as a whole; we do not impose the "output colour" of an edge $\gg(e)$ to coincide with the "colour" $\gg'(\pp(e))$. This allows us to compare "transition systems" using different types of "acceptance conditions".

\begin{remark}\label{rmk-morph:loc-morph-bijection-runs}
	Let $\pp$ be a "morphism of pointed graphs".
	\begin{enumerate}
		\item If $\pp$ is "locally surjective", then $\ppRuns$ is surjective.
		\item If $\pp$ is "locally injective", then $\ppRuns$ is injective.
		\item If $\pp$ is "locally bijective", then $\ppRuns$ is bijective.
	\end{enumerate}
\end{remark}


In the following, the "weak morphisms" under consideration will be "locally surjective". The next lemma ensures that we can assume that they are "surjective" without loss of generality.

\begin{lemma}\label{lemma-morph:locSurjMorph_OntoAccessible}
	If $\pp\colon \TS \to \TS'$ is a "locally surjective" "weak morphism", it is onto the "accessible part" of $\TS'$. That is, for every "accessible" state $v'\in \TS'$, there exists some state $v\in \TS$ such that $\pp_V(v)=v'$. In particular, if every state of $\TS'$ is "accessible", $\pp$ is "surjective".
\end{lemma}
\begin{proof}
	 Let $v'$ be an accessible state of $\TS'$. By definition, there exists a "finite run" $\rr'$ from an initial vertex of $\TS'$ to $v'$. By surjectivity of $\ppRuns$, there is a finite "run" $\rr\in \RunsFin{\TS}$ such that $\ppRuns(\rr) = \rr'$. As $\pp$ is a "morphism of graphs", we have that $\pp(\mtargetPath(\rr)) = v'$.
\end{proof}

\begin{example}
	In Figure~\ref{fig-morp:loc-bij-morphism} we provide an example of a "locally bijective morphism" between the two rightmost "transition systems" from Figure~\ref{fig-prelim:multiple-aut-ex} (we have removed "input letters" for simplicity).
	We recall that the "acceptance set" of the rightmost "transition system" is the "Muller language associated to" $\F = \{\{\alpha\},\{\beta\}\}$.
	The morphism is given by $\pp_V(v_1) = \pp_V(v_2) = v'$ and $\pp_V(v_2) = v_2'$. In this case, the mapping $\pp_V$ determines a unique "morphism"; the (uniquely determined) mapping $\pp_E$ is represented by the colours of the edges in the figure.
	It is easy to check that this mapping "preserves the acceptance" of "runs" and that it is "locally bijective".
	
	\begin{figure}[ht]
		\centering
		\includegraphics[scale=1.15]{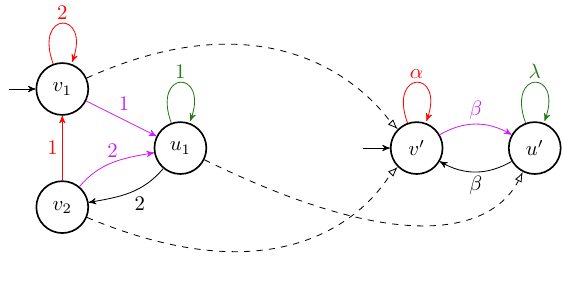}
		\caption{A "locally bijective morphism" from a "parity" "TS" to a "Muller" "TS" with "acceptance set" given by $\F = \{\{\alpha\},\{\beta\}\}$. We use dashed arrows to represent the images of vertices, and colours to represent the image of edges (that can be inferred from $\pp_V$).}
		\label{fig-morp:loc-bij-morphism}
	\end{figure}
\end{example}

\subsection{History-deterministic mappings}\label{subsect-morp: hist-det-morphisms}

"Locally bijective morphisms" are a natural generalisation of the "composition" of a "transition system" with a "deterministic" "automaton". They guarantee the semantic equivalence of the two involved "transition systems", but at the cost of the use of some strong hypothesis, as the outgoing edges of a vertex $v$ must exactly correspond to the outgoing edges of its image $\pp(v)$.
We can imagine correct transformations that do not satisfy this requirement.
Notably, "history-deterministic" "automata" have been introduced as a method to bypass this restriction, with the hope of outperforming transformations that are witnessed by "locally bijective morphisms".
In general, if $\A$ is an "HD@@aut" "automaton" "recognising" the "acceptance set" ot $\TS$, the "composition" $\TS \compositionAut \A$ does not admit a "locally bijective morphism" to $\TS$, although it shares most semantic properties with it (Proposition~\ref{prop-prelim:composition_games_HD}).

We introduce next "HD mappings", which are "weak morphisms" with the minimal set of hypothesis ensuring that, if $\pp\colon \TS \to \TS'$ is an "HD mapping", we can "simulate@@morph" "runs" of $\TS'$ in $\TS$ via $\pp$ while preserving their "acceptance@@run". This will allow us to show that $\pp$ witnesses the semantic equivalence of $\TS$ and $\TS'$ (Section~\ref{subsec-morph:semantic-properties}).

\begin{figure}[ht]
	\centering
	\includegraphics[scale=0.3]{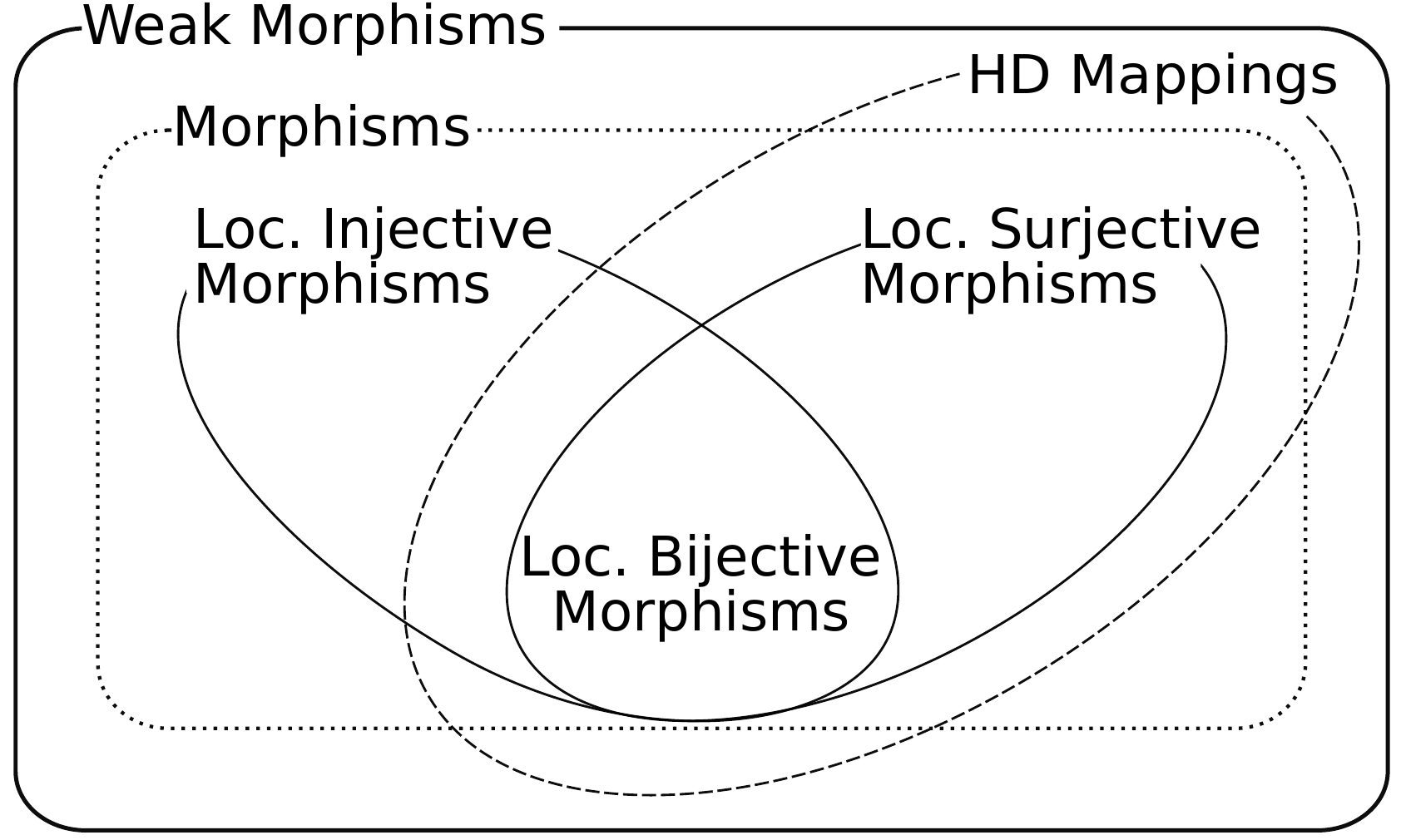}
	\caption{Different types of morphisms and the relations between them. The fact that "locally surjective morphisms" are "HD mappings" is given by Lemma~\ref{lemma-morph:loc_surjective_implies_HD}. Note that "HD mappings" are also "locally surjective" "weak morphisms" (Remark~\ref{rmk-morph: HD-is-loc.surjective}).}
	\label{fig-morph:types-morphisms}
\end{figure}

\paragraph*{History-deterministic mappings.}

Let $\TS$ and $\TS'$ be "transition systems" and $\pp:\TS \to \TS'$ a "weak morphism" between them.
\AP A ""resolver simulating $\pp$"" consists in a pair of functions $\intro* \rInit\colon I' \to I$ and $r\colon E^* \times E' \to E$ such that:
\begin{enumerate}
	\item$\pp(\rInit(v_0')) = v_0'$ for all $v_0'\in I'$, \label{item-HD-map:map-initial-vertex}
	\item $\pp(r(\rr,e')) = e'$, for all $\rr\in E^*$ and $e'\in E'$, \label{item-HD-map:map-edges}
	\item if $e_0'\in \mout(I')$, $\msource(r(\ee,e_0')) = \rInit(\msource(e_0'))$, \label{item-HD-map:respect-initial-choice} and
	\item if $\rr$ is a finite run in $\TS$ ending in $v$ and $e'\in \mout(\pp(v))$, then $r(\rr, e')\in \mout(v)$. \label{item-HD-map:building-run}
\end{enumerate} 
\AP Given a "run" $\rr' = e_0'e_1'\dots \in \RunsInfty{\TS'}$ starting in some $v_0'\in I'$, the ""run induced by@@morphism"" $r$ is the sequence $\intro*\rRuns(\rr') = e_0e_1e_2\dots\in \RunsInfty{\TS}$ defined by $e_i = r(e_0\dots e_{i-1}, e_i')$, which is indeed a "run" in $\TS$.
We say that the "resolver@@HDmorphism" is ""sound@@resolverMorphism"" if for every "accepting run" $\rr'\in \Runs{\TS'}$, the "run" $\rRuns(\rr')$ is "accepting@@run" in $\TS$.
Note that we do not impose $\rRuns(\rr')$ to be "rejecting@@run" if $\rr'$ is.

\begin{remark}\label{rmk-morph: HD-is-loc.surjective}
	Provided that all states of $\TS'$ are "accessible", a "resolver simulating $\pp$" can only exist if $\pp$ is a "locally surjective" "weak morphism".
\end{remark}


Said differently, a "sound@@morph" "resolver simulating" $\pp$ is a winning strategy for the player Duplicator in the following game:
\begin{itemize}
	\item In round $0$, Spoiler picks an "initial vertex" $v_0'$ in $\TS'$. Duplicator responds by picking an "initial vertex" $v_0$ in $\TS$ such that $\pp(v_0) = v_0'$.
	\item In round $n>0$, Spoiler picks an edge $e_n'$ in $\TS'$, and Duplicator responds by picking an edge $e_n$ in $\TS$ such that $\pp(e_n) = e_n'$.
	\item Duplicator wins if either $e_1e_2\dots$ is an "accepting run" in $\TS$ from $v_0$ or $e_1'e_2'\dots$ is not an "accepting run" in $\TS'$ from $v_0'$ (it is either not a "run from $v_0$" or not "accepting@@run"). Spoiler wins otherwise.
\end{itemize}

\begin{definition}\label{def:HD_mapping}
	Let $\TS$ and $\TS'$ be ("labelled") "transition systems".
	\AP A ""history-deterministic mapping"" (HD mapping) of transition systems from $\TS$ to $\TS'$ is a pair of mappings $\pp = (\pp_V: V \rightarrow V', \pp_E: E \rightarrow E')$ such that:
	\begin{itemize}
		\item $\pp$ is a "weak morphism",
		\item $\pp$ ""preserves accepting runs"":  $\rr \in \Runs{\TS}$ and  $\gg(\rr) \in \WW \; \implies \; \gg'(\ppRuns(\rr)) \in \WW'$, and
		\item there exists a "sound@@morph" "resolver simulating" $\pp$. 
	\end{itemize}
\end{definition}

Even if a "history-deterministic mapping" is not necessarily "locally bijective" (and not even a "morphism of transition systems"), the existence of a "sound@@morph" "resolver@@HDmorphism" allows us to define a right inverse to $\ppRuns$ preserving the "acceptance@@runs" of "runs". 

\begin{lemma}\label{lemma-morph:facts_HD_morphisms}
	Let $\pp\colon \TS \to \TS'$ be an "HD mapping" and let $(\rInit, r)$ be a "sound@@morph" "resolver simulating" it. The following holds:
	\begin{itemize}
		\item $\ppRuns \circ \rRuns = \Id{\RunsInfty{\TS'}}$.
		\item $\rRuns$ preserves the "acceptance@@runs" of "runs" in $\TS'$, that is,  for every "run" $\rr' \in \Runs{\TS'}$,  $\rr'$ is "accepting@@run" if and only if $\rRuns(\rr')$ is "accepting@@run" in $\TS$.
	\end{itemize}
\end{lemma}
\begin{proof}
	The first item follows from the fact that $\pp(r(\rr,e')) = e'$ for every $\rr\in E^*$ and $e'\in E'$. 
	
	For the second item, the definition of a "sound resolver@@morph" imposes that if $\rr'$ is "accepting@@run", so is $\rRuns(\rr')$. For the other direction, if $\rRuns(\rr')$ is "accepting@@run", then $\ppRuns(\rRuns(\rr')) = \rr'$ has to be "accepting@@run", as an "HD mapping" "preserves accepting runs".
\end{proof}

\begin{example}
	In Figure~\ref{fig-morp:HD mapping} we give an example of a "weak morphism" $\pp\colon \TS \to \TS'$ that is a "history-deterministic mapping", but which is neither a "morphism", nor "locally bijective".
	"Transition system" $\TS$, on the left of the figure, is a "parity" "TS" (more precisely, a "coB\"uchi" "TS"). 
	"Transition system" $\TS'$, depicted on the right of the figure, is a "Muller" "TS" using as "acceptance set" the "Muller language associated to" $\F = \{\{\aa\},\{\aa,\bb\}, \{\aa,\lambda\}\}$; that is, a "run" in $\TS'$ is "accepting@@run" if and only if it eventually avoids either transition $e'$ or transition $f'$.  
	The "weak morphism" we propose is given by: $\pp(v_0)= \pp(v_1) = \pp(v_2) = v'$, and $\pp(u_1) = \pp(u_2) = u'$. The image of most edges is uniquely determined, and we use colours to represent them. We have named the only edges whose image is not uniquely determined, and we define $\pp(e_1)=\pp(e_2) = e'$ and $\pp(f_1)=\pp(f_2) = f'$.
	
	\begin{figure}[ht]
		\centering
		\includegraphics[scale=1.15]{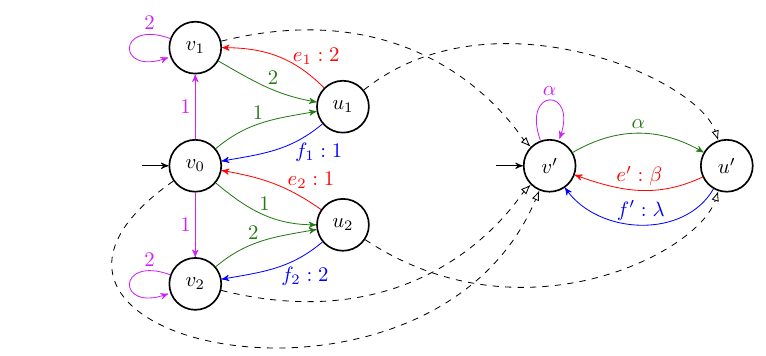}
		\caption{A "history-deterministic mapping" from a "parity" "TS" to a "Muller" "TS" with "acceptance set" given by $\F = \{\{\aa\},\{\aa,\bb\}, \{\aa,\lambda\}\}$. We use dashed arrows to represent the images of vertices, and colours to represent the image of edges.}
		\label{fig-morp:HD mapping}
	\end{figure}
	
	We remark that $\pp$ does not "preserve rejecting runs". Indeed, a run in $\TS$ alternating between $v_0$ and $u_1$, taking transition $f_1$ infinitely often, is "rejecting@@run", but its image is "accepting@@run" in $\TS'$. 
	However,  $\pp$ "preserves accepting runs": a "run" is "accepting@@run" in $\TS$ if and only if it eventually stays in $\{v_1, u_1\}$ or in $\{v_2, u_2\}$. In the first case, the image under $\pp$ avoids transition $f'$ in $\TS'$, and in the second case, its image avoids transition $e'$.

	Finally, we describe a "sound@@morph" "resolver simulating $\pp$". When simulating a run from $\TS'$ in $\TS$, we have a choice to make only when we are in state $v_0$. If the previous transition in $\TS'$ was~$e'$, we will go up, that is, $v'\re{\aa}u'$ is simulated by $v_0\re{1}u_1$ and $v'\re{\aa}v'$ is simulated by $v_0\re{1}v_1$.
	If the previous transition in $\TS'$ was $f'$, we will go down symmetrically.
	In this way, if  transition $f'$ is eventually not visited by the run in $\TS'$, we ensure to stay in $\{v_1, u_1\}$ in $\TS$ (and symmetrically, we ensure to stay in $\{v_2, u_2\}$ if $e'$ is avoided in $\TS'$).
\end{example}

\paragraph*{History-deterministic-for-games mappings.}

In the case of "games", we need to slightly strengthen the definition of "HD mappings" to guarantee that, if there is a suitable mapping $\pp\colon\G \to \G'$, then $\G$ and $\G'$ have the same "winner".
In order to show that if "Eve" wins $\G'$ then she "wins" $\G$, we need a method to transfer "strategies" in $\G'$ to $\G$. A regular "resolver simulating~$\pp$" does not suffice to do this, as it does not take into account the partition into "Eve" and "Adam" vertices. 
We need to be able to simulate a "play" of $\G'$ in $\G$ in a two-players-game fashion, "Adam's" moves will be simulated by "Adam", and "Eve's" moves by "Eve". This idea leads to the notion of "HD-for-games mapping".


Let $\G$ and $\G'$ be two "games", and $\pp:\G \to \G'$ be a "weak morphism" between them admitting a "resolver@@morphism" $(\rInit, r)$ "simulating@@resolver" $\pp$. 
\AP Given "runs" $\rr' = e_0'e_1'\dots \in \Runs{\G'}$ and $\rr = e_0e_1\dots \in \Runs{\G}$, we say that $\rr$ is ""consistent with@@resolver"" $(\rInit, r)$ over $\rr'$ if:
\begin{enumerate}
	\item $\msource(e_0) = \rInit(\msource(e_0')$,
	\item $\pp(e_i) = e_i'$, and 
	\item for every finite prefix $e_0e_1\dots e_{n-1} \prefix \rr$ ending in a vertex "controlled by Eve", the next edge in $\rr$ is $e_n = r(e_0\dots e_{n-1}, e_n')$.
\end{enumerate}  
We remark that there exists at least one run "consistent with@@mapping" $(\rInit, r)$ over $\rr'$, namely $\rRuns(\rr')$.
\AP We say that $(\rInit, r)$ is ""sound for $\G$"" if it verifies that for any "accepting run" $\rr'\in \Runs{\G'}$, all runs "consistent with@@resolver" $(\rInit, r)$ over $\rr'$ are "accepting" in $\G$. 

Said differently, a "resolver sound for $\G$" is a winning strategy for Duplicator in the following game:
\begin{itemize}
	\item In round $0$, Spoiler picks an "initial vertex" $v_0'$ in $\G'$. Duplicator responds by picking an "initial vertex" $v_0$ in $\G$ such that $\pp(v_0) = v_0'$.
	\item In round $n>0$, Spoiler picks an edge $e_n'$ in $\G'$. If $v_{n-1}$ is controlled by "Adam", Spoiler chooses an edge $e_n = v_{n-1}\re{} v_{n}\in \mout(v_{n-1})$ such that $\pp(e_n) = e_n'$. If $v_n$ is controlled by "Eve", it is Duplicator who chooses one such $e_n$. 
	\item Duplicator wins if either $e_1e_2\dots$ is an "accepting run" in $\G $ from $v_0$ or $e_1'e_2'\dots$ is not an "accepting run" in $\G'$ from $v_0'$. Spoiler wins otherwise.
\end{itemize}

\begin{definition}\label{def:HD mapping-for-games}
	An "HD mapping" "of games@@morphism" $\pp:\G \to \G'$ is called ""history-deterministic-for-games"" if it admits a "resolver" "sound for $\G$".
\end{definition}

Whenever we apply the term "HD-for-games" to a map $\pp:\TS \to \TS'$, it will implicitly imply that $\TS$ and $\TS'$ are "games" (that is, they have a fixed "vertex-labelling" $\lPlayers\colon V\to \{\Eve, \Adam\}$), and that $\pp$ preserves those "vertex-labellings").

In the next lemma, we prove that "HD@@mapping" and "HD-for-games" mappings are a strict generalisation of "locally surjective morphisms" (and therefore, also of "locally bijective" ones). On the other hand, we remark that "HD mappings" must be "locally surjective", but they are not necessarily "morphisms" (they might not "preserve rejecting runs").

\begin{lemma}\label{lemma-morph:loc_surjective_implies_HD}
	If $\pp: \TS \to \TS'$ is a "locally surjective morphism", it is also an "HD mapping". If $\TS$ and $\TS'$ are "games", $\pp$ is moreover "HD-for-games". 
\end{lemma}
\begin{proof}
	We need to define a "sound@@morph" "resolver simulating" $\pp$. Let $\rInit\colon I' \to I$ be any function choosing initial vertices satisfying that $\pp\circ\rInit = \Id{I'}$ (which exists by "local surjectivity" of $\pp$). For each $v\in V$ and edge $e'\in \mout(\pp(v))$ we choose one edge $f(v,e')\in \mout(v)$ such that $\pp(f(v,e')) = e'$ (which exists by "local surjectivity"), and we let $r$ be the "resolver" induced by these choices. Formally, we define $r \colon E^*\times E' \to E$ recursively. For the base case, if $e_0'\in \mout(v')$, with $v'\in I'$, we define $r(\ee, e_0') = f(\rInit(v'), e_0')$. Assume that $r$ has been defined for "runs" of length $\leq n$, and let $\rr\in E^*$ be of length $n+1$ and $e'\in E'$. If $\rr$ is not a "run" or $e'\notin \mout(\mtarget(\pp(\rr)))$, we let $r(\rr,e')$ be any edge in $\inv{\pp}(e')$. If not, let $v = \mtarget(\rr)$ and define $r(\rr,e')$ to be the edge $f(v,e')$.
	
	It is straightforward to check that $(\rInit, r)$ is indeed a "resolver@@mapping"  (for every "run" $\rr'\in \Runs{\TS'}$, the sequence $\rRuns(\rr')$ is a "run" in $\TS$ and $\rr'$ is its image under $\pp$). Finally, since $\pp$ is a "morphism", for every $\rr'\in \Runs{\TS'}$ and every $\rr\in \Runs{\TS}$ "consistent with@@resolver" $(\rInit, r)$ over $\rr'$, $\rr$ is "accepting@@run" in $\TS$ if and only if $\rr' = \ppRuns(\rr)$ is  "accepting@@run" in $\TS'$. We conclude that $(\rInit ,r)$ is a "sound resolver@@morph" (resp. "sound for $\TS$") and therefore $\pp$ is an "HD mapping" (resp. "HD-for-games mapping").
\end{proof}

\paragraph*{Restrictions and extensions of initial sets.}

The following simple lemma states that reducing the number of "initial vertices" preserves the "history-determinism@@mapping" of mappings. 
\begin{lemma}\label{lemma-morph:HD mappings-reducing-initial-vertices}
	Let $\TS$ and $\TS'$ be two "TS" such that there is an "HD@@mapping" (resp. "HD-for-games") mapping $\pp\colon \TS \to \TS'$. 
	For any non-empty subset $\tilde{I} \subseteq I'$, $\pp$ is also an "HD@@mapping" (resp. "HD-for-games") mapping between the "transition systems" $\initialTS{\TS}{\rInit(\tilde{I})}$ and $\initialTS{\TS'}{\tilde{I}}$; that is, the "transitions systems" obtained by setting $\rInit(\tilde{I})$ and $\tilde{I}$ as "initial vertices", respectively.
\end{lemma}

For arbitrary "acceptance conditions", enlarging the set of "initial vertices" does not preserve "history-determinism@@mapping". However, for "transition systems" using the "acceptance conditions" considered in this work, we can enlarge the set of "initial vertices" without loss of generality. The proof can be found in Appendix~\ref{sec:appendix-prefix-independent}.

\begin{restatable}{lemma}{HDmappingsPrefIndep}\label{lemma-morph:HD mappings-with-prefix-independent-conditions}
	Let $\TS$ and $\TS'$ be two "TS" such that all their states are "accessible", and let $\pp\colon \TS \to \TS'$ be an "HD@@mapping" (resp. "HD-for-games") mapping between them. If $\WW$ and $\WW'$ are "prefix-independent", the mapping $\pp$ is also "HD@@mapping" (resp. "HD-for-games") when considered between the "transition systems" $\initialTS{\TS}{V}$ and $\initialTS{\TS'}{V'}$, consisting of the "transition systems" $\TS$ and $\TS'$ where all the states are set to be "initial@@state".
\end{restatable}

\subsection{Preservation of semantic properties of automata and games}\label{subsec-morph:semantic-properties}

We start this section by showing that "locally bijective morphisms" and "HD mappings" are a strict generalisation of "compositions" by "deterministic" and "history-deterministic" "automata", respectively (Proposition~\ref{prop-morph:product-by-aut-induces-mapping}). 
Then, we prove that these mappings witness the semantic equivalence of the "transition systems" under consideration. That is, (1) if $\pp\colon \A\to \A'$ is an "HD mapping" "of automata@@morph", then $\Lang{\A} = \Lang{\A'}$, and if $\pp$ is "locally bijective", $\A$ is "deterministic" (or "unambiguous") if and only if $\A'$ is (Proposition~\ref{prop-morph:HD mappings-preserve-languages});\footnotemark{} and (2) if $\pp\colon \G\to \G'$ is an "HD-for-games mapping", $\G$ and $\G'$ have the same "winner" (Proposition~\ref{prop-morph:HD mappings-preserve-games} and Corollary~\ref{cor-morph:HD-map-preserve-full-WR}).

\footnotetext{The results in this section do not directly imply that if $\A$ is an "automaton" recognising the "acceptance set" of another automaton $\B$, then $\B\compositionAut \A$ "recognises" the same language as $\B$, if $\A$ is not "history-deterministic" (Proposition~\ref{prop-prelim:composition_automata}). In that case, the equality $\Lang{\B\compositionAut \A} = \Lang{\B}$ follows from the idea that "runs" in $\B$ can be simulated in $\B\compositionAut \A$ ``guided by the non-deterministic choices of $\A$'', which we do not formalise in this work.}

\paragraph*{Morphisms generalise composition by an automaton.}

\begin{proposition}\label{prop-morph:product-by-aut-induces-mapping}
	Let $\A$ be a "complete" "automaton" "accepting@@automaton" the language $\Lang{\A} = \WW\subseteq \SS^\oo$, and let $\TS$ be a "(labelled) TS" with "acceptance set" $\WW$. 
	Then, there exists a "locally surjective" "weak morphism of (labelled) TS" $\pp\colon \TS \compositionAut \A \to \TS$ that "preserves accepting runs". Moreover:
	\begin{enumerate}
		\item If $\A$ is "deterministic", $\pp$ can be chosen to be a "locally bijective morphism".
		\item If $\A$ is "HD@@aut", then $\pp$ can be chosen to be an "HD mapping".
		\item If $\A$ is "HD@@aut" and $\TS$ is a "game" "suitable for transformations", then $\pp$ can be chosen to be an "HD-for-games mapping".
	\end{enumerate}	
\end{proposition}

\begin{proof}
	We recall that the set of states of $\TS \compositionAut \A$ is $V\times Q$ and its set of transitions $\edgesProduct$ is a subset of $(E\times \DD) \sqcup E$, where $V$ and $Q$ (resp. $E$ and $\DD$) are the states (resp. transitions) of $\TS$ and $\A$, respectively.
	We let $\WW_\A\subseteq \GG^\oo$ be the "acceptance set" of $\A$.
	We define $\pp_V(v,q) = v$ and $\pp_E(e_1,e_2) = e_1$ for $(e_1,e_2)\in E\times \DD$ and $\pp_E(e_1) = e_1$ for $e_1\in E$.
	It is immediate to check that $\pp$ is a "weak morphism". 
	
	Given a "run" $\rr = (v_0,q_0) \re{c_0}  (v_1,q_1) \re{c_1} \dots$ in $\TS \compositionAut \A$, we can consider its projection over $\TS$, $\ppRuns(\rr) = v_0 \re{a_0}  v_1 \re{a_1} \dots$. We note that there must exist a unique "run" in $\A$ of the form 
	\[ \pp_\A(\rr) = q_0 \re{a_0:c_0}  q_1 \re{a_1:c_1} \dots .\]
	(Formally, some letters $a_i$ might equal $\ee$, and in this case $q_{i} \re{a_i:c_i}  q_{i+1}$ does not appear in the "run" $\pp_\A(\rr)$). 
	
	We show that $\pp$ "preserves accepting runs". Let $\rr$ be an "accepting run" in $\TS\compositionAut \A$. In that case, $c_0c_1c_2\dots \in \WW_\A$, and therefore $\pp_\A(\rr)$ is an "accepting run" in $\A$ over $a_0a_1a_2\dots$, so we conclude that $a_0a_1a_2\dots\in \WW$ and $\ppRuns(\rr)$ is an "accepting run" in $\TS$.
	
	We prove next the "local surjectivity" of $\pp$. Clearly, $\pp$ induces a surjection between the "initial vertices" of $\TS \compositionAut \A$ (which are $I_\TS \times I_\A$) and those of $\TS$. 
	Let $(v,q)\in V\times Q$ and $e_1 = v\re{a}v'\in E$. If $a=\ee$, the edge $e_1$ belongs to $\edgesProduct$ and $\pp(e_1)=e_1$. If $a\neq \ee$, since $\A$ is "complete" there is a transition $e_2 = q\re{a}q'\in \DD$ and $\pp(e_1,e_2) = e_1$, so $\pp$ is "locally surjective".
	
	\begin{enumerate}
		
		\item Since $\A$ has a single "initial state" $q_0$, $\pp$ induces a bijection between the "initial vertices" of $\TS \compositionAut \A$ (which are $I_\TS \times \{q_0\}$) and those of $\TS$.
		Let $\edgesProduct_\mathrm{col}\subseteq E\times \DD$ and $\edgesProduct_\ee\subseteq E$ such that $\edgesProduct = \edgesProduct_\mathrm{col} \cup \edgesProduct_\ee$. We remark that $\restr{\pp}{\edgesProduct_\ee}$ is the identity function (so injective) and that $\pp(\edgesProduct_\mathrm{col}) \cap \pp(\edgesProduct_\ee) = \emptyset$ because $\pp(\edgesProduct_\mathrm{col})$ are exactly "coloured" transitions of $\TS$. Finally, let $(e_1, e_2)$ and $(e_1', e_2')$ in $\mout(v,q) \cap \edgesProduct_\mathrm{col}$. Their $\pp(e_1, e_2) = \pp(e_1', e_2')$ if and only if $e_1 = e_1'$. Let $a\in \SS$ be the colour of $e_1$. Since $\A$ is "deterministic", there is at most one transition from $q$ labelled by $a$, that must be $e_2=e_2'$. We conclude that $(e_1, e_2)=(e_1', e_2')$ and that $\pp$ is "locally injective".
		
		Let $\rr$ be a "rejecting run" in $\TS\compositionAut \A$ (we use the notations introduced above). In that case, $c_0c_1c_2\dots \notin \WW_\A$, and therefore $\pp_\A(\rr)$ is a "rejecting run" over $a_0a_1a_2\dots$. Since $\A$ is "deterministic", this is the only "run over" $a_0a_1a_2\dots$, so we conclude that it does not belong to $\WW.$ We conclude that $\ppRuns(\rr)$ is a "rejecting run" in $\TS$.
		
		\item Let $(r_0, r_\A)$ be a "resolver" for $\A$. We define a "resolver@@HDmapping" $(\rInit, r_\pp)$ "simulating" $\pp$. First, we let $\rInit(v_0) = (v_0,r_0)$ for all $v_0\in I_\TS$. We define $r_\pp\colon {\edgesProduct}^*\times E \to \edgesProduct$ by induction on the length of the "runs". 
		Let $e_0 = v_0\re{a}v_1\in \mout(v_0)$ be an edge in $\TS$. If $e_0$ is "uncoloured" ($a=\ee$), we let $r_\pp(\ee, e_0) = e_0 = (v_0,r_0)\re{\ee}(v_1,r_0)$. If not, we let $r_\pp(\ee, e_0) = (e_0, e_a)$, where $e_a = r(\ee,a)$. 
		Assume that $r_\pp$ has been defined for sequences of edges of $\TS \compositionAut \A$ of length $< n$ and let  $\rr =e_0e_1\dots e_{n-1}\in {\edgesProduct}^*$ be a sequence length $n+1$ and $e_\TS = v_n\re{a_n} v_{n+1}$ be an edge in $\TS$.
		If $\rr$ is not a "run" or if it does not end in $\inv{\pp}(v_n)$, we let $r_\pp(\rr, e_\TS)$ be any edge in $\inv{\pp}(e_\TS)$.
		Assume that $\rr$ is a "run" ending in $\inv{\pp}(v_n)$. If $a_n = \ee$, we define $r_\pp(\rr, e_\TS) = e_\TS$. As noted before, $\rr$ induces a "run" $\pp_\A(\rr) = q_0 \re{a_0:c_0}  q_1 \re{a_1:c_1} \dots \re{}q_n$ in $\A$. We let $e_\A = r_\A(\pp_\A(\rr), a_n)$ be the transition chosen by the "resolver" of $\A$ after this "run",  and we define $r_\pp(\rr, e_\TS) = (e_\TS, e_\A)$.
		
		It directly follows from this definition that $(\rInit, r_\pp)$ is indeed a "resolver". The proof that if $\rr\in \Runs{\TS}$ is an "accepting run" then $\rRunsOption{\pp}(\rr)$ is "accepting" follows the same lines as the previous item.
		
		\item We prove that, if $\TS$ is a "game" "suitable for transformations", the "resolver@@HDmapping" $(\rInit, r_\pp)$ defined in the previous item is "sound for $\TS$". We claim that if $\rr$ is a "run" in $\TS$, the only "run" "consistent with@@resolver" $(\rInit, r_\pp)$ over $\rr$ is $\rRunsOption{\pp}(\rr)$. This follows from the fact that if $(v,q)$ is a vertex in $\TS \compositionAut\A$ "controlled by Adam" and $e\in \mout(v)$, then there is a unique $e'\in \mout(v,q)$ such that $\pp(e') = e$. This is indeed the case: as $\TS$ is "suitable for transformations", if $v$ is an "Adam's vertex", every $e\in \mout(v)$ is "uncoloured", so by definition of $\pp$ we have that $\pp(e') = e \implies e'=e$. (This can be seen as that $\pp$ is locally bijective in Adam's vertices).
		We conclude that if $\rr$ is an "accepting run" in $\TS$ and $\rr^\ltimes$ is a "run" "consistent with@@resolver" $(\rInit, r_\pp)$ over $\rr$, then $\rr^\ltimes = \rRunsOption{\pp}(\rr)$, which is "accepting" by "soundness@@morph" of the "resolver@@mapping" $(\rInit, r_\pp)$.\qedhere
	\end{enumerate}	
\end{proof}

\paragraph*{Morphisms witness equivalence of automata.}
\begin{proposition}\label{prop-morph:HD mappings-preserve-languages}
	Let $\A$, $\A'$ be two "automata" over the same "input alphabet" such that there is an "HD mapping" of "automata@@mapping" $\pp\colon \A \to \A'$. Then, $\Lang{\A}=\Lang{\A'}$, and $\A$ is "HD@@aut" if and only if $\A'$ is "HD@@aut".
	If $\pp$ is moreover "locally bijective" and "surjective", $\A$ is "deterministic" (resp. "complete") if and only if $\A'$ is.
\end{proposition}
\begin{proof}
	Since $\pp$ "preserves accepting runs", it is clear that $\Lang{\A}\subseteq \Lang{\A'}$. Since $\pp$ admits a "sound resolver@@morph" $(\rInit, r)$, if $\rr$ is an "accepting@@run" "run over" $w\in \SS^\oo$ in $\A'$, then $\rRuns(\rr)$ is an "accepting@@run" "run over" $w$ in $\A$, so $\Lang{\A'}\subseteq \Lang{\A}$.
	
	Let $(\rInit, r_\pp)$ be a "sound@@morph" "resolver simulating" $\pp$.
	Assume that $\A$ is "HD@@aut", admitting a "resolver" $(r_0,r)$. A "resolver" $(r_0',r')$ for $\A'$ can be obtained just by composing $r_\pp$ and $\pp$, that is: $r_0' = \pp(r_0)$ and for $\rr'\in \RunsFin{\A'}$ and $a\in \SS$, $r'(\rr', a) = \pp(r(\rRunsOption{\pp}(\rr'),a)))$. That is, given a "run" $\rr'$ in $\A'$, we simulate it in $\A$ using $r_\pp$, then, we look at what is the continuation proposed by the resolver~$r$ when we give the letter $a$, and we transfer back this choice to $\A'$ using $\pp$. 
	Assume now that $\A'$ is "HD@@aut" and that $(r_0',r')$ is a "resolver" for it. We define a "resolver" $(r_0,r)$ for $\A$. We let $r_0 = \rInit(r_0')$, and for $\rr\in \RunsFin{\A}$ and $a\in \SS$, $r(\rr, a) = r_\pp((\ppRuns(\rr), r(\ppRuns(\rr),a))$. That is, given a "run" $\rr$ in $\A$, we simulate it in $\A'$ using $\pp$, then, we look at what is the continuation proposed by the resolver $r'$ when we give the letter $a$, and we transfer back this choice to $\A$ using $r_\pp$. 
	It is a direct check that the "resolvers" defined this way witness that $\A'$ and $\A$, respectively, are "HD@@aut".
	
	 The proof that $\A$ is "deterministic" (resp. "complete") if and only if $\A'$ is "deterministic" (resp. "complete"), assuming "surjectivity" and "local bijectivity" of $\pp$, follows the same lines.
\end{proof}

A subclass of "automata" with a restrictive amount of "non-determinism" that is widely study is that of \emph{unambiguous} automata (we refer to~\cite{Colcombet15Unambiguity, Carton2003UnambiguousBuchi} for a detailed exposition). An automaton is \AP""unambiguous"" if for every input word $w\in \SS^\oo$ there is at most one accepting "run over" $w$, and it is ""strongly unambiguous"" if there is at most one "run over" $w$. By Remark~\ref{rmk-morph:loc-morph-bijection-runs}, "locally bijective morphisms" also preserve "(strongly) unambiguity": if $\pp\colon \A \to \A'$ is a "locally bijective" "morphism@@aut" then $\A$ is ("strongly@@unamb") "unambiguous" if and only if  $\A'$ is.

\paragraph*{Morphisms preserve winning regions of games.}
\begin{lemma}\label{lemma-morph:surj-mappings-preserve-games}
	Let $\G, \G'$ be two "games", such that there is a "weak morphism" "of games@@morphism" $\pp\colon \G \to \G'$ that is "locally surjective" and "preserves accepting runs". If "Eve" "wins" the "game" $\G$ from an initial vertex $v$, then she "wins" $\G'$ from $\pp(v)$.
\end{lemma}
\begin{proof}
	Let $v' = \pp(v)$, and let $\strat_v\colon \PathSetFin{\G}{v}\to E$ be a "strategy" from $v$ for Eve in $G$. 
	Intuitively, we will define a "strategy" in $\G'$ as follows: for each finite run $\rr'$ from $v'$ in $\G'$, we pick a preimage $\rr\in \inv{\pp}(\rr')$ in $\G$, look at the decision made by $\strat_v$ at the end of $\rr$ and transfer it back to $\G'$ via $\pp$. In order to define a correct "strategy", the choices of the preimages have to be made in a coherent manner. We formalise this idea next.
	
	We will make use of a function $\choicestrat \colon \PathSetFin{\G'}{v'} \to \PathSetFin{\G}{v}$ satisfying that for any $\rr' = e_0'e_1'\dots e_{n-1}'e_n'\in \PathSetFin{\G'}{v'}$:
	\begin{itemize}
		\item The run $\choicestrat(\rr')$ has length $n+1$.
		\item  $\ppRuns(\choicestrat(\rr')) = \rr'$.
		\item Monotonicity: if $\tilde{\rr}' \prefix \rr'$ then $\choicestrat(\tilde{\rr}') \prefix \choicestrat(\rr')$.
		\item If there exists $e_n\in \inv{\pp}(e_n')$ such that $\choicestrat(e_0'e_1'\dots e_{n-1}')e_n$ is "consistent with@@strat" $\strat_v$, then $\choicestrat(\rr')$ is "consistent with@@strat" $\strat_v$.
	\end{itemize}

Assume for now that such a function exists, and define a "strategy" in $\G'$ as \[\strat'_{v'}(\rr') = \pp(\strat_v(\choicestrat(\rr'))), \;\text{ for } \rr'\in \PathSetFin{\G'}{v'}.\] 
We prove that $\strat'_{v'}$ is "winning@@strat".
 Let $\rr' =e_0'e_1'\dots\in \Runs{\G'}$ be an infinite "play" "consistent with@@strat" $\strat'_{v'}$. For each finite prefix $\tilde{\rr}'\prefix \rr'$, $\choicestrat(\tilde{\rr}')$ is a finite play in $\G$, and by the monotonicity assumption, we can define the limit of these runs as:
\[ \vec{\rr} = e_0e_1e_2\dots \in \Runs{\G}, \;\text{ where } e_0e_1\dots e_n = \choicestrat(e_0'e_1'\dots e_n'),  \]
which is indeed a "run" in $\G$. We show that $\vec{\rr}$ is "consistent with" $\strat_v$ by induction. Let $\rr_n = e_0e_1\dots e_{n-1}$ be the "prefix" of size $n$ of $\vec{\rr}$, and suppose that it ends in a vertex $v_n$ controlled by "Eve". We want to show that $e_n = \strat_v(\rr_n)$. By definition of $\strat'_{v'}$, $e_n' = \pp(\strat_v(\rr_n)) = \pp(e_n)$, and as $v_n$ is controlled by "Eve", $\strat_v(\rr_n)$ is the only continuation of $\rr_n$ "consistent with" $\strat_v$, so by the last property of $\choicestrat$, $e_n$ has to coincide with $\strat_v(\rr_n)$, as we wanted. As $\vec{\rr}$ is "consistent with" the "winning strategy" $\strat_v$, it is an "accepting run" in $\G$, and since $\pp$ "preserves accepting runs", $\rr' = \ppRuns(\vec{\rr})$ is also an "accepting run". 

Finally, we show how to build a function $\choicestrat \colon \PathSetFin{\G'}{v'} \to \PathSetFin{\G}{v}$ by induction on the length of the runs. Assume that $\choicestrat$ has been defined for "runs" of length $\leq n$, and let $e_0'e_1'\dots e_n'$ be a run of length $n+1$, with $\choicestrat(e_0'e_1'\dots e_{n-1}') = e_0e_1\dots e_{n-1}$. If $e_0e_1\dots e_{n-1}$ is not "consistent with" $\strat_v$, it ends in a vertex $v_n$ controlled by "Adam", or $\strat_v(e_0e_1\dots e_{n-1})\notin \inv{\pp}(e_n)$,  we let $e_n\in \inv{\pp}(e_n)\cap \mout(v_n)$ be any edge (one such edge exists by "local surjectivity"). 
On the contrary, we let $e_n = \strat_v(e_0e_1\dots e_{n-1})$.
We define $\choicestrat(e_0'e_1'\dots e_{n-1}'e_n') = e_0e_1\dots e_{n-1}e_n$. 
By construction, the obtained function fulfils the 4 requirements.
\end{proof}

\begin{proposition}\label{prop-morph:HD mappings-preserve-games}
	Let $\G, \G'$ be two "games" such that there is an "HD-for-games mapping" $\pp\colon \G \to \G'$. "Eve's" "winning region" in $\G'$ is the projection of her "winning region" in $\G$: $\winRegion{\Eve}{\G'} = \pp(\winRegion{\Eve}{\G})$.	
\end{proposition}
\begin{proof}
	If "Eve" "wins" $\G$ from an "initial vertex" $v$, Lemma~\ref{lemma-morph:surj-mappings-preserve-games} guarantees that she "wins" $\G'$ from $\pp(v)$.
	
	Assume now that "Eve" "wins" $\G'$ from an "initial vertex" $v'$ with a "strategy" $\strat'_{v'}\colon \PathSetFin{\G'}{v'}\to E'$. We need to show that she "wins" $\G$ from some "initial vertex" in $\inv{\pp}(v')$. Let $(\rInit,r)$ be a "resolver simulating" $\pp$ "sound for $\G$" and let $v=\rInit(v')$.  We define
	\[\strat_{v}(\rr) = r(\rr, \strat'_{v'}(\ppRuns(\rr)), \; \text{ for } \rr\in \PathSetFin{\G}{v}.\]
	 That  is, $\strat_{v}$ is a "strategy" in $\G$ from $v$ that, given a finite run $\rr$, simulates $\rr$ in $\G'$, looks at the move done by the "strategy" $\strat'_{v'}$ in there, and transfers this choice back to $\G'$ by using the "resolver@@morphism" $r$. We prove that $\strat_{v}$ is "winning for Eve" in $\G$. Let $\rr = e_0e_1\dots\in \PathSet{\G}{v}$ be a "play" consistent with $\strat_{v}$. We claim that $\pp(\rr)$  is "consistent with@@strat" $\strat'_{v'}$ and that $\rr$ is "consistent with@@resolver" $(\rInit,r)$ over $\pp(\rr)$. This implies the desired result; "consistency with@@strat" $\strat'_{v'}$ implies that $\pp(\rr)$ is "accepting", and since $(\rInit,r)$ is "sound for $\G$", $\rr$ would be "accepting@@run" in $\G$.
	
	We prove that $\pp(\rr)$ is "consistent with@@strat" $\strat'_{v'}$. Let $e_0'e_1'\dots e_{n-1}'$ be a subplay of $\pp(\rr)$ ending in a vertex $v_n$ "controlled by Eve". By definition of the strategy $\strat_{v}$, we have that $e_n = r(e_0\dots e_{n-1}, \strat'_{v'}(e_0'\dots e_{n-1}'))$, and by definition of a "resolver@@mapping" (item~\ref{item-HD-map:map-edges}), we obtain that $e_n' = \pp(e_n) = \strat'_{v'}(e_0'\dots e_{n-1}'))$, as we wanted.
	
	The fact that $\rr$ is "consistent with@@resolver" $(\rInit,r)$ over $\pp(\rr)$ follows directly from the definition of~$\strat_{v}$.
\end{proof}

The next corollary follows from the previous proposition and Lemma~\ref{lemma-morph:HD mappings-with-prefix-independent-conditions}.
\begin{corollary}\label{cor-morph:HD-map-preserve-full-WR}
	Let $\G, \G'$ be two "games" whose states are "accessible" and such that their "acceptance sets" $\WW_\G$ and $\WW_{\G'}$ are "prefix-independent". If there is an "HD-for-games mapping" $\pp\colon \G \to \G'$, then "Eve's" "full winning region" in $\G'$ is the projection of her "full winning region" in $\G$: $\winRegion{\Eve}{\initialTS{\G'}{V'}} = \pp(\winRegion{\Eve}{\initialTS{\G}{V}})$.	
\end{corollary} 	
	
	\section{The Zielonka tree: An optimal approach to Muller languages}\label{section:zielonka-tree}
	In this section, we take a close look into the "Zielonka tree", a structure introduced (under the name of \emph{split trees}) to study "Muller languages"~\cite{Zielonka1998infinite}.
We show how to use the "Zielonka tree" to construct minimal "deterministic" "parity" "automata" and minimal "history-deterministic" "Rabin" "automata" "recognising" "Muller languages".
In Section~\ref{subsec-zt: parity automaton}, we describe the construction of a minimal "deterministic" "parity" "automaton" $\zielonkaAutomaton{\F}$ for a given "Muller language" $\Muller{\F}$. Theorem~\ref{thm-zt:strong_optimality_ZTparity}, the main contribution of this section, states the minimality of $\zielonkaAutomaton{\F}$ both amongst "deterministic" and "HD@@aut" "parity" "automata". Theorem~\ref{thm-zt:optimality_ZTparity-priorities} states the optimality on the number of "colours" of the "acceptance condition" of $\zielonkaAutomaton{\F}$, and implies that we can determine the "parity index" of a "Muller language" from its "Zielonka tree".
We will use the optimality of "automaton" $\zielonkaAutomaton{\F}$ to provide a polynomial-time algorithm minimising "DPAs" "recognising" "Muller languages" in Section~\ref{subsec-corollaries: minimisation-parity}.

In Section~\ref{subsec-zt: GFG-Rabin}, we describe the construction of a minimal "history-deterministic" "Rabin" "automaton" $\zielonkaHDAutomaton{\F}$ for a "Muller language" $\Muller{\F}$. Its minimality amongst "HD@@aut" "automata" is shown in Theorem~\ref{thm-zt:optimality_ZT-HD-Rabin}, by using the characterisation of the "memory requirements" of a "Muller language" in terms of its "Zielonka tree"~\cite{DJW1997memory}.

On the other hand, it has been shown that finding a minimal "deterministic" "Rabin" "automaton" "recognising" a given "Muller language" is $\NPcomplete$, if the language is represented by a "parity" or "Rabin" "automaton", or even by its "Zielonka tree"~\cite{Casares2021Chromatic}.
Therefore, unless $\PTime = \NP$, there are "Muller languages" for which minimal "deterministic" "Rabin" automata are strictly larger than minimal "HD@@aut" "Rabin" "automata". Some explicit such languages were shown in~\cite[Section~4]{CCL22SizeGFG}.
A summary of the minimal automata "recognising" "Muller languages" appears in Table~\ref{table:minimal-automata}.

\newcolumntype{M}[1]{>{\centering\arraybackslash}m{#1}} 
\newcolumntype{N}{@{}m{0pt}@{}}
\newcommand{\mycolwidth}{1.3cm}
\begin{table}[ht]
	\centering
	\begin{tabular}{|M{3cm}||M{3cm}| M{3cm}|N} 
		\hline
		\textbf{Type of automata} &  "Deterministic"  & \begin{tabular}{@{}c@{}}"History-" \\ "deterministic@@hist"\end{tabular}  \\[4mm]
		\hline
		\hline

		"Parity"	&   $\zielonkaAutomaton{\F}$ &   $\zielonkaAutomaton{\F}$ &\\[4mm]
		
		\hline
		                  
		"Rabin"	& No characterisation 
    &   $\zielonkaHDAutomaton{\F}$ &\\[4mm]
\hline
	\end{tabular}
	\caption{Minimal automata "recognising" a "Muller language" $\Muller{\F}$, according to the type of "acceptance condition" ("parity" or "Rabin") and the form of "determinism".}
	\label{table:minimal-automata}
\end{table}

\subsection{The Zielonka tree}\label{subsec-zt: definition}

\begin{definition}[\cite{Zielonka1998infinite}]\label{def-ZT:zielonkaTree}
	\AP Let~$\F\subseteq\powplus{\SS}$ be a family of non-empty subsets over a finite set~$\SS$. A ""Zielonka tree"" for~$\F$ (over $\SS$),\footnotemark{} denoted $\intro*\zielonkaTree{\F} = (N, \ancestor, \intro*\nu :N \to \powplus{\SS})$ is a "$\powplus{\SS}$-labelled tree" with nodes partitioned into ""round nodes"" and ""square nodes"", $N= \intro*\roundnodes \sqcup \intro*\squarenodes$, such that:
	\begin{itemize}
		\item The "root" is labelled $\SS$.
		\item If a "node" is labelled~$X\subseteq \SS$, with~$X\in\F$, then it is a "round node", and it has a child for each maximal non-empty subset~$Y\subseteq X$ such that~$Y\not\in\F$, which is labelled $Y$.
		\item If a "node" is labelled~$X\subseteq \SS$, with~$X\not\in\F$, then it is a "square node", and it has a child for each maximal non-empty subset~$Y\subseteq X$ such that~$Y\in\F$, which is labelled $Y$. 
	\end{itemize}
\end{definition}

\footnotetext{The definition of $\zielonkaTree{\F}$, as well as most subsequent definitions, do not only depend on $\F$ but also on the alphabet $\SS$. Although this dependence is important,  we do not explicitly include it in the notations in order to lighten them, as most of the time the alphabet will be clear from the context.}

\begin{remark}
	We note that for each family of subsets $\F\subseteq\powplus{\SS}$, there is only one "Zielonka tree" up to renaming of its nodes, so we will talk of \emph{the} "Zielonka tree" of $\F$.
\end{remark}

\AP For a family of subsets $\F\subseteq \pow{\SS}$ and $\SS'\subseteq \SS$, we write $\intro*\restSubsets{\F}{\SS'} = \F \cap \pow{\SS'}$.

\begin{remark}
	We remark that if $n$ is a node of $\zielonkaTree{\F}$, then the "subtree of~$\zielonkaTree{\F}$ rooted at~$n$" is the "Zielonka tree" for the family $\restSubsets{\F}{\nu(n)}$ over the alphabet $\nu(n)$, that is, for the restriction of $\F$ to the subsets included in the label of $n$.
\end{remark}

\begin{remark}\label{rmk-ZT:union-changes-acceptance}
	Let $n$ be a node of $\zielonkaTree{\F}$ and let $n_1$ be a "child" of it. If $\nu(n_1) \subsetneq X \subseteq \nu(n)$, then $\nu(n_1)\in \F \iff X\notin \F \iff \nu(n)\notin \F$. In particular, if $n_1, n_2$ are two different "children" of $n$, then  
	$\nu(n_1)\in \F \iff \nu(n_2)\in \F \iff \nu(n_1)\cup \nu(n_2)\notin \F$.
\end{remark}


We equip "Zielonka trees" with an "order@@tree" to navigate in them. That is, we equip each set $\children_{\zielonkaTree{\F}}(n)$ with a total order, making $\zielonkaTree{\F}$ an "ordered tree". The precise order considered will be irrelevant for our purposes.
From now on, we will assume that all "Zielonka trees" are "ordered", without explicitly mentioning it.

\AP For a "leaf" $l \in \leaves(\zielonkaTree{\F})$ and a letter $a\in \SS$ we define $\intro*\supp(l,a) = n$ to be the "deepest" "ancestor" of $l$ (maximal for $\ancestor$) such that $a\in \nu(n)$.

\begin{example}\label{example-ZT:zielonka-tree}
	We will use the "Muller language associated to" the following family of subsets as a running example throughout the paper.
	Let $\SS= \{a,b,c\}$ and let $\F$ be:
	\[ \F = \{\{a,b\}, \{a,c\}, \{b\}\}. \]
	In Figure~\ref{fig-ZT:zielonkaTree-example} we show the "Zielonka tree" of $\F$. 
	We use Greek letters (in pink) to name the nodes of the tree.
	Integers appearing on the right of the tree will be used in the next section.
	
	We have that $\supp(\xi, c) = \lambda$ and $\supp(\xi, b) = \aa$.
	Also, $\jump(\xi, \lambda)=\zeta$  is the "leaf" reached by going from $\xi$ to $\lambda$, then changing to the next branch (in a cyclic way) and re-descend by taking the leftmost path. Similarly, $\jump(\xi,\aa)=\theta$.
	
	The "subtree rooted at" $\lambda$ contains the nodes $\{\lambda, \xi, \zeta\}$. We note that this is the "Zielonka tree" of $\restSubsets{\F}{\{a,c\}} = \{\{a,c\}\}$ (over the alphabet $\{a,c\}$).
\end{example}

	\begin{figure}[ht]
		\centering 
		\includegraphics[scale=1.2]{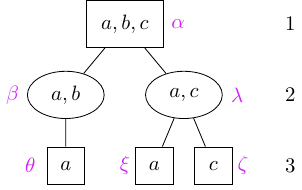}
		\caption{Zielonka tree $\zielonkaTree{\F}$ for 
			$\F=\{ \{a,b\}, \{a,c\}, \{b\}\}$.}
		\label{fig-ZT:zielonkaTree-example}
	\end{figure}

\subsection{A minimal deterministic parity automaton}\label{subsec-zt: parity automaton}
We present next the "Zielonka-tree-parity-automaton", a minimal "deterministic" "parity" "automaton" for a "Muller language" $\Muller{\F}$ built from the "Zielonka tree" $\zielonkaTree{\F}$. Our construction will furthermore let us determine the "parity index" of the language $\Muller{\F}$ from its "Zielonka tree".
\subsubsection{The Zielonka-tree-parity-automaton}\label{subsubsec-zt-parity: definition}

We associate a non-negative integer to each level of a "Zielonka tree" $\zielonkaTree{\F}= (N, \ancestor, \nu)$. 
\AP We let $\intro*\parityNodes: N \to \NN$ be the function defined as:
\begin{itemize}
	\item if $\SS\in \F$, $\parityNodes(n)=\depth(n)$, 
	\item if $\SS\notin \F$, $\parityNodes(n)=\depth(n) + 1$.	
\end{itemize}
\AP We let $\intro*\minparityZ{\F}$ (resp. $\intro*\maxparityZ{\F}$) be the minimum (resp. maximum) value taken by the function~$\parityNodes$. 

\begin{remark}\label{rmk:parity_levels_ZT}
	A node $n$ in the "Zielonka tree" $\zielonkaTree{\F}$ verifies that $\parityNodes(n)$ is even if and only if $\nu(n)\in \F$.
	If $\SS\in \F$, $\minparityZ{\F} = 0$ and $\maxparityZ{\F}$ equals the "height" of the "Zielonka tree" minus one. If $\SS\notin \F$, $\minparityZ{\F} = 1$ and $\maxparityZ{\F}$ equals the "height" of the "Zielonka tree". 
\end{remark}

\begin{example}
	The "Muller language" from Example~\ref{example-ZT:zielonka-tree} satisfies $\SS\notin \F$. The values taken by the function $\parityNodes$ are represented at the right of the "Zielonka tree" in Figure~\ref{fig-ZT:zielonkaTree-example}. We have $\parityNodes(\aa)=1$, $\parityNodes(\bb)=\parityNodes(\lambda) = 2$ and $\parityNodes(\theta) = \parityNodes(\xi) = \parityNodes(\zeta) = 3$, so $\minparityZ{\F}=1$ and $\maxparityZ{\F}=3$.
\end{example}

\begin{definition}[Zielonka-tree-parity-automaton]\label{def-ZT:parityAutomatonZT}
	\AP Given a family of non-empty subsets $\F\subseteq \powplus{\SS}$,
	we define the  ""ZT-parity-automaton"" $\intro*\zielonkaAutomaton{\F}=(Q, \SS, q_0, [\minparityZ{\F}, \maxparityZ{\F}], \transAut{}, \parity)$ as the "deterministic" "parity" "automaton" given by:
	\begin{itemize}
		\item $Q=\leaves(\zielonkaTree{\F})$, 
		\item $q_0$ is the "leftmost" "leaf" of $\zielonkaTree{\F}$,\footnotemark{}
		\item The transition reading $a\in\SS$ from $q\in Q$ goes to $\jump(q,\supp(q,a))$  and produces $\parityNodes(\supp(q,a))$ as "output", that is, \[\transAut{}(q,a) = \left(\jump(q,\supp(q,a)), \parityNodes(\supp(q,a))\right).\]
	\end{itemize}
\end{definition}

\footnotetext{Any state can be chosen as "initial state" (see Lemma~\ref{lemma-prelims:prefix-indep-automaton}).}

Intuitively, the transitions of the automaton are determined as follows: if we are in a "leaf"~$l$ and we read a colour~$a$, then we move up in the branch of $l$ until we reach a node~$n$ that contains the letter $a$ in its label. 
Then we pick the "child" of $n$ just on the right of the branch that we took before (in a cyclic way), and we move to the "leftmost" "leaf" below it. The colour produced as output is $\parityNodes(n)$, determined by the "depth" of $n$.

\begin{example}
	In Figure~\ref{fig-ZT:zielonka-parityAutomaton} we show the "ZT-parity-automaton" $\zielonkaAutomaton{\F}$ of the family of subsets from Example~\ref{example-ZT:zielonka-tree}.
\end{example}

\begin{figure}[ht]
	\centering 
	\includegraphics[width=0.5\textwidth]{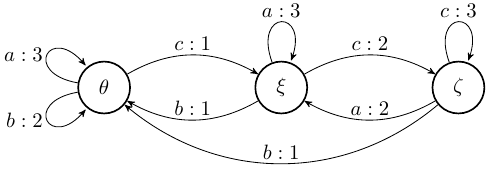}
	\caption{"ZT-parity-automaton" "recognising" the "Muller language associated to" $\F=\{ \{a,b\}, \{a,c\}, \{b\}\}$.}
	\label{fig-ZT:zielonka-parityAutomaton}
\end{figure}

\paragraph*{Correctness of the Zielonka-tree-parity-automaton.}

\begin{proposition}[Correctness]\label{prop-ZT:correctness_ZT_parity}
	Let $\F\subseteq \powplus{\SS}$ be a family of non-empty subsets. Then, 
	\[ \Lang{\zielonkaAutomaton{\F}} = \MullerC{\F}{\SS}. \]
	That is, a word $w\in \SS^\oo$ is "accepted@@word" by $\zielonkaAutomaton{\F}$ if and only if $\minf(w)\in \F$.
\end{proposition}

The following useful lemma follows directly from the definition of $\supp$ and $\jump$.
\begin{lemma}\label{lemma-ZT:transition_ZT-parity-automaton}
	Let $q$ be a "leaf" of $\zielonkaTree{\F}$ and let $n$ be a node "above" $q$. Then, $\supp(q, a)$ is a "descendant" of $n$ if and only if $a\in \nu(n)$, and in this case, $\jump(q,\supp(q,a))$ is a "descendant" of  $n$ too.
\end{lemma}

\begin{proof}[Proof of Proposition~\ref{prop-ZT:correctness_ZT_parity}]
	Let $w = w_0w_1w_2\dots \in \SS^\oo$ be an infinite word. For $i>0$, let 
	 $q_i$ be the "leaf" of $\zielonkaTree{\F}$ reached after the (only) "run over" $w_0w_1\dots w_{i-1}$ in $\zielonkaAutomaton{\F}$. 
	For $i\geq 0$ let $n_i = \supp(q_i,w_i)$ be the ``intermediate node'' used to determine the next state and the output colour of each transition, and let $c_i = \parityNodes(n_i) = \gg(q_i, w_i) \in [\minparityZ{\F}, \maxparityZ{\F}]$ be that output colour (the "output" of the "run over" $w$ being therefore $c_0c_1c_2\dots \in \NN^\oo$). 
	Let $q_\infty$ be a node appearing infinitely often in the sequence $q_0q_1q_2\dots$, and let $n_w$ be the "deepest" "ancestor" of $q_\infty$ such that $\minf(w)\subseteq \nu(n_w)$.
	\begin{claim}\label{claim:supp-below-n-CORRECT-ZT}
		 There is $K\in \NN$ such that for all $i \geq K$, $q_i \descendant n_w$ and $\supp(q_i, w_i) \descendant n_w$. In particular, $c_i \geq \parityNodes(n_w)$ for $i\geq K$.
	\end{claim}
	\begin{subproof}
		Let $K\in \NN$ be a position such that $w_i \in \minf(w)$ for all $i\geq K$ and $q_K = q_\infty$.
		The claim follows from Lemma~\ref{lemma-ZT:transition_ZT-parity-automaton} and induction.
	\end{subproof}
	\begin{claim}\label{claim:supp-n_inf_often-CORRECT-ZT}
		Let $n_{w,1},\dots, n_{w,s}$ be an enumeration of $\children_{\zielonkaTree{\F}}(n_w)$ "from left to right". It holds that:
		\begin{enumerate}
			\item $\supp(q_i,w_i) = n_w$ infinitely often. In particular, $c_i = \parityNodes(n_w)$ for infinitely many~$i$'s.
			\item There is no $n_{w,k}\in \children(n_w)$ such that $\minf(w) \subseteq \nu(n_{w,k})$.
		\end{enumerate}	
	\end{claim}
	\begin{subproof}
		We first remark that for all $n_{w,k}$ there are arbitrarily large positions $i$ such that $q_i$ is not "below" $n_{w,k}$. Suppose by contradiction that this is not the case. Then, for all $i$ sufficiently large we have that $\supp(q_i, w_i) \descendant n_{w,k}$, and by Lemma~\ref{lemma-ZT:transition_ZT-parity-automaton}, $\minf(w)\subseteq \nu(n_{w,k})$. In particular, $q_\infty$ is "below" $n_{w,k}$, contradicting the fact that $n_w$ is the "deepest" "ancestor" of $q_\infty$ containing $\minf(w)$.
		
		Let $K$ be like in the Claim~\ref{claim:supp-below-n-CORRECT-ZT}.
		We show that if $i\geq K$ and $q_i\descendant n_{w,k}$, then there is $j> i$ such that $w_j\notin \nu(n_{w,k})$, $\supp(q_{j},w_{j}) = n_{w}$ and $q_{j+1}\descendant n_{w,k+1}$ (by an abuse of notation we let $s+1 = 1$). 
	    It suffices to consider the least $j \geq i$ such that $\supp(q_{j},w_{j})\nsucceq n_{w,k}$ (which exists by the previous remark). Since $\minf(w)\subseteq \nu(n_w)$ we have that $\supp(q_{j},w_{j}) = n_{w}$, so $w_j\notin \nu(n_{w,j})$ (by Lemma~\ref{lemma-ZT:transition_ZT-parity-automaton}) and by definition of the transitions of $\zielonkaAutomaton{\F}$, $q_{j+1}$ will be a "leaf" "below" $n_{w,k+1}$.
	    
	    The fact that $q_{j+1}\descendant n_{w,k+1}$ implies that for any child $n_{w,k'}$, infinitely many states $q_i$ will be below $n_{w,k'}$ (we go around the children in a round-robin fashion). Therefore, for any $k$, there are arbitrarily large $j$ such that $w_j\notin \nu(n_{w,j})$ and $\supp(q_{j},w_{j}) = n_{w}$, implying both items in the claim.
	\end{subproof}
	Combining both claims, we obtain that the minimum of the colours that are produced as output infinitely often is $\parityNodes(n_w)$. By Remark~\ref{rmk:parity_levels_ZT}, $\parityNodes(n_w)$ is even if and only if $n_w$ is a "round node" (if $\nu(n_w)\in \F$). It remains to show that $\minf(w) \in \F$ if and only if $\nu(n_w)\in \F$, which holds by the second item in Claim~\ref{claim:supp-n_inf_often-CORRECT-ZT} and Remark~\ref{rmk-ZT:union-changes-acceptance}.	
\end{proof}

\subsubsection{Optimality of the Zielonka-tree-parity-automaton}\label{subsubsec-zt-parity: optimality}
We now state and prove the main results of this section: the optimality of the "ZT-parity-automaton" in both number of states (Theorem~\ref{thm-zt:strong_optimality_ZTparity}) and number of "colours" of the "acceptance condition" (Theorem~\ref{thm-zt:optimality_ZTparity-priorities}).
The minimality of the "ZT-parity-automaton" comes in two versions. A weaker one states its minimality only amongst "deterministic" "automata" (Theorem~\ref{thm-zt:weak_optimality_ZTparity}), and a stronger one states its minimality amongst all "history-deterministic"  "automata" (Theorem~\ref{thm-zt:strong_optimality_ZTparity}). Although the weaker version is implied by the stronger one, we find it instructive to provide a proof for this easier case. The proof of the stronger statement is one of the most technical parts of the paper, but the argument used in its proof is just a careful refinement of the ideas appearing in the weaker version.

\paragraph*{Statement of the results.}
\begin{theorem}[Optimality of the parity condition]\label{thm-zt:optimality_ZTparity-priorities}
	The "parity index" of a "Muller language" $\MullerC{\F}{\SS}$ is $[\minparityZ{\F}, \maxparityZ{\F}]$. That is, the "ZT-parity-automaton" of $\MullerC{\F}{\SS}$ uses the optimal number of colours to recognise this language.
\end{theorem}

\begin{theorem}[Minimality of the ZT-parity-automaton]\label{thm-zt:strong_optimality_ZTparity}
	Let $\A$ be a "history-deterministic" "parity" "automaton" "recognising@@automaton" a "Muller language" $\MullerC{\F}{\SS}$. Then, $|\zielonkaAutomaton{\F}|\leq \size{\A}$.
\end{theorem}

\begin{corollary}\label{cor-zt:no-small-HD-for-Muller}
	For every "Muller language" $L$, a minimal "deterministic" "parity" "automaton" "recognising" $L$ has the same size as a minimal "HD@@aut" "parity" "automaton" "recognising" $L$.
\end{corollary}

\AP We remark that, nonetheless, there are non-trivial "HD@@aut" "parity" "automata" "recognising" "Muller languages". The automaton provided in Example~\ref{ex-prelim:HD-aut} is an "HD@@aut" "coB\"uchi" "automaton" "recognising" a "Muller language" that cannot be made "deterministic" just by removing transitions. We note that the ("deterministic") "ZT-parity-automaton" for this "Muller language" has only $2$ states. 

\AP We say that an "automaton" $\A$ is ""determinisable by pruning"" if there is a subset $\DD'\subseteq \DD$ of its transitions and an "initial state" $q_0$ such that the "subautomaton induced" by $\DD'$ with initial state $q_0$ is "deterministic" and "recognises" $\Lang{\A}$.

\begin{proposition}
	There exists an  "HD@@aut" "parity" "automaton" "recognising" a "Muller language" that is not "determinisable by pruning".
\end{proposition}

\paragraph*{Optimality of the parity condition.}
\begin{proof}[Proof of Theorem~\ref{thm-zt:optimality_ZTparity-priorities}]
	Let $L= \MullerC{\F}{\SS}$.
	The "ZT-parity-automaton" of $L$ is a "parity" "automaton" recognising $L$ using colours in $[\minparityZ{\F}, \maxparityZ{\F}]$, therefore, the "parity index@@atMost" of $L$ is at most $[\minparityZ{\F}, \maxparityZ{\F}]$.
	
	To prove that the "parity index is not less" than $[\minparityZ{\F}, \maxparityZ{\F}]$, we use the Flower Lemma~\ref{lemma:flower-lemma}. The language $L$ is trivially recognised by a "deterministic" "Muller" "automaton" $\A_L$ with just one state $q$, transitions $q\re{a:a}q$ for each $a\in \SS$, and "acceptance condition" given by $L$ itself. Let $n_1\ancestor n_2\ancestor \dots \ancestor n_d$ be a branch of maximal length of $\zielonkaTree{\F}$ (that must verify $d= \maxparityZ{\F} - \minparityZ{\F}$, and that the "root" $n_1$ is a "round node" if and only if $\minparityZ{\F}$ is even). If we let $\ell_i$ be the "cycle" in $\A_L$ containing exactly the transitions corresponding to letters in $\nu(n_i)$, we obtain that $\ell_1\supsetneq \ell_2\supsetneq \dots \supsetneq \ell_d$ is a "$d$-flower" over $q$, which is "positive@@flower" if and only if  $n_1$ is a "round node". Lemma~\ref{lemma:flower-lemma} allows us to conclude.
\end{proof}

\paragraph*{Minimality of the ZT-parity-automaton with respect to deterministic automata.}

Before presenting the proof of Theorem~\ref{thm-zt:strong_optimality_ZTparity}, we prove a weaker result, namely, that the "ZT-parity-automaton" is minimal amongst \emph{deterministic} parity automata recognising a "Muller language".

\begin{theorem}[Minimality of the Zielonka Tree automaton with respect to deterministic automata]\label{thm-zt:weak_optimality_ZTparity}
	Let $\A$ be a "DPA"  "recognising@@automaton" a "Muller language" $\MullerC{\F}{\SS}$. Then, $|\zielonkaAutomaton{\F}| \leq \size{\A}$.
\end{theorem}

We recall that, by Remark~\ref{rmk-prelim:restriction-to-accessible} and Lemma~\ref{lemma-prelims:prefix-indep-automaton}, we can assume that all the states of "automata" "recognising" "Muller languages" are "accessible", and that any of them can be chosen to be "initial". When considering "subautomata" of these "automata", we will sometimes not mention their "initial state".

Let $\A = (Q, \SS, I, \GG, \Delta, \WW)$ be an "automaton", and let $X\subseteq \SS$ be a subset of the "input alphabet". 
\AP We say that a "subgraph" $\S$ of the "underlying graph" of $\A$ is ""$X$-closed@@automaton"" if for every state $q$ in $\S$ and every letter $a\in X$ there is some transition $q\re{a:c}q'$ in $\S$.
\AP An ""$X$-final strongly connected component"" ($X$-FSCC) of $\A$ is an "$X$-closed" "final SCC" in the graph obtained by taking the restriction of the "underlying graph" of $\A$ to the edges labelled by letters in $X$. We remark that a subset $S\subseteq Q$ is the set of states of an "$X$-FSCC" if and only if:
\begin{itemize}
	\item for any two states $q, q'\in S$ there is a finite word $w\in X^*$ labelling a finite "path" from $q$ to $q'$, and
	\item if $q\in S$ and there is a finite "path" from $q$ to $q'$ labelled with a word $w\in X^*$, then $q'\in S$.
\end{itemize}

\begin{lemma}\label{lemma-ZT:existance-X-FSCC}
	Let $\A$ be a "complete" "automaton". For every subset $X\subseteq \SS$, $\A$ contains an "accessible@@set" "$X$-FSCC".
\end{lemma}
\begin{proof}
	As any "graph" without sinks contains some "final SCC", the "accessible part" of the restriction of $\A$ to edges labelled by letters in $X$ contains one. By "completeness@@aut" of $\A$, one such "final SCC" has to be an "$X$-closed" "subgraph", so it is an "$X$-FSCC".
\end{proof}

\begin{lemma}\label{lemma-ZT:X-FSCC-induce-automata}
	Let $\A$ be a "DMA" "recognising" a "Muller language" $\MullerC{\F}{\SS}$, let $X\subseteq \SS$ and let $\S_X$ be an "accessible@@set" "$X$-FSCC" of $\A$. Then, the "automaton induced by" $\S_X$ is a "deterministic" "automaton" "recognising" $\MullerC{\restSubsets{\F}{X}}{X} = \{w\in X^\oo \mid \minf(w)\in \F\}$.
\end{lemma}
\begin{proof}
	Let $\A = (Q,\Sigma, q_0, \Gamma, \transAut{}, \WW)$ (where $\WW$ is a "Muller language"). Let $q_S$ be the state in $\S_X$ chosen to be "initial", and let $u_0\in \SS^*$ be a finite word such that the "run over" $u_0$ from $q_0$ ends in $q_S$.
	By "prefix-independence" of "Muller languages", a word $w\in X^\oo$ belongs to $\MullerC{\F}{\SS}$ if and only if $u_0w\in \MullerC{\F}{\SS}$, and therefore, $\A$ "accepts@@aut" $w$ if and only if it "accepts@@aut" $u_0w$. 
	Since the "run@@aut" in~$\A$ over $u_0w$ and the "run@@aut" in $\S_X$ over $w$ have a suffix in common, and by "prefix-independence" of~$\WW$, we have that $w\in \Lang{\S_X}$ if and only if $u_0w\in \Lang{\A}$ if and only if $\minf(w)\in \F$.
\end{proof}

The next lemma states that, in a "parity" "automaton", the union of two "accepting@@cycle" "cycles" must be "accepting@@cycle", and similarly for "rejecting cycles". 
In Section~\ref{subsec-corollaries:typeness}, we will see that this property is actually a characterisation of "parity" "transition systems" (Proposition~\ref{prop-typ:parity ACD type}).
\begin{lemma}\label{lemma-ZT:union_cycles_parity-TS}
	Let $\A$ be a "parity" "automaton". Let $\ell_1, \ell_2\in \cycles{\A}$ be two "cycles" with some "state in common". If $\ell_1$ and $\ell_2$ are both "accepting@@cycles" (resp. "rejecting@@cycles"), then $\ell_1 \cup \ell_2$ is also "accepting@@cycles" (resp. "rejecting@@cycles").
\end{lemma}
\begin{proof}
	Let $\gg\colon \DD\to \NN$ be the "colouring function" of $\A$. The "cycles" $\ell_1$ and $\ell_2$ are "accepting@@cycle" if and only if $d_i = \min \gg(\ell_i)$ is even, for $i=1,2$. In this case, $\min \gg(\ell_1 \cup \ell_2) = \min\{ d_1, d_2\}$ is even. The proof is symmetric if $\ell_1$ and $\ell_2$ are "rejecting@@cycle".
\end{proof}

By a small abuse of notation, we will say that two "SCC" $\S_1$ and $\S_2$ are disjoint, and write $\S_1\cap \S_2 = \emptyset$, if their sets of states are disjoint.

\begin{lemma}\label{lemma-ZT:disjoint-children-FSCC}
	Let $\F\subseteq \powplus{\SS}$ be a family of subsets with "Zielonka tree" $\zielonkaTree{\F} = (N, \ancestor, \nu)$, and let~$\A$ 
	be a "DPA" "recognising" $\MullerC{\F}{\SS}$. Let $n\in N$ be a node of the "Zielonka tree" of $\F$, and let $n_1, n_2\in \children_{\zielonkaTree{\F}}(n)$ be two different "children" of $n$. If $\S_1$ and $\S_2$ are two "accessible@@sets" "$\nu(n_1)$-FSCC" and "$\nu(n_2)$-FSCC" in $\A$, respectively, then $\S_1\cap \S_2 = \emptyset$.
\end{lemma}
\begin{proof}
	Without loss of generality, we can assume that all states in $\A$ are "accessible", and since the language that $\A$ recognises is "prefix-independent", we can also suppose that $\A$ is "complete". 
	Let $l_\SS\colon \DD \to \SS$ be the labelling of the transitions of $\A$ with "input letters".
	Let $\S_i$ be a "$\nu(n_i)$-FSCC" in~$\A$, for $i=1,2$, and let $\ell_i$ be its set of edges, which form a "cycle" satisfying $l_\SS(\ell_i) = \nu(n_i)$. 
	Suppose by contradiction that $\S_1\cap \S_2 \neq \emptyset$. Then $\ell_1$ and $\ell_2$ have some state in common, and their union is also a "cycle" satisfying $l_\SS(\ell_1\cup\ell_2) = \nu(n_1)\cup \nu(n_2)$. By Lemma~\ref{lemma-ZT:union_cycles_parity-TS}, we must have
	\[ \ell_1 \text{ "accepting" } \iff \ell_1\cup \ell_2 \text{ "accepting"},\]
	contradicting the fact that $\nu(n_1)\in \F$ if and only if $\nu(n_1)\cup \nu(n_2)\notin \F$ (Remark~\ref{rmk-ZT:union-changes-acceptance}).	
\end{proof}

\begin{proof}[Proof of Theorem~\ref{thm-zt:weak_optimality_ZTparity}]
	We proceed by induction in the "height" of $\zielonkaTree{\F}$. For "height" $1$, the result is trivial, since $|\zielonkaAutomaton{\F}| = 1$.
	Let $\A$ be a "DPA" "recognising" $\MullerC{\F}{\SS}$. Let $n_0$ be the "root" of $\zielonkaTree{\F}$ and $n_1, n_2, \dots, n_k$ be an enumeration of the "children" of $n_0$ in $\zielonkaTree{\F}$.
	By Lemma~\ref{lemma-ZT:existance-X-FSCC}, for each $i\in \{1, \dots, k\}$, $\A$ contains some "accessible" "$\nu(n_i)$-FSCC" $\S_i$, and by Lemma~\ref{lemma-ZT:disjoint-children-FSCC} these must be pairwise disjoint. By Lemma~\ref{lemma-ZT:X-FSCC-induce-automata}, each $\S_i$ induces a "deterministic" "subautomaton@@induced" "recognising" $\MullerC{\restSubsets{\F}{\nu(n_i)}}{\nu(n_i)}$. Let $\Z_i$ by the "subtree of $\zielonkaTree{\F}$ rooted at" $n_i$, which we recall that is the "Zielonka tree" for $\restSubsets{\F}{\nu(n_i)}$. By induction hypothesis, it must be the case that $|\leaves(\Z_i)| \leq |\S_i|$, so we can conclude:	
	\[ |\zielonkaAutomaton{\F}| = |\leaves(\zielonkaTree{\F})| = \sum\limits_{i=1}^k |\leaves(\Z_i)| \leq  \sum\limits_{i=1}^k |\S_i| \leq |\A|.\qedhere\]
\end{proof}

\paragraph*{Minimality of the ZT-parity-automaton with respect to HD automata.}
We intend to prove Theorem~\ref{thm-zt:strong_optimality_ZTparity}, that is, that for any $\F\subseteq \powplus{\SS}$, the automaton $\zielonkaAutomaton{\F}$ is minimal amongst "HD@@aut" "parity" "automata" "recognising" $\Muller{\F}$. 
We will follow the same proof scheme than in the "deterministic" case, performing an induction over the "height" of the "Zielonka tree". Assume that~$\A$ is an "HD@@aut" "parity" "automaton" for $\Muller{\F}$ and that $n_0$ is the "root" of $\zielonkaTree{\F}$ having $n_1, \dots, n_k$ as "children". For each "child" $n_i$ we want to find an "HD@@aut" "subautomaton" $\A_i$ "recognising" the "language associated@@Muller" to $\restSubsets{\F}{\nu(n_i)}$ in such a way that the automata $\A_i$ are pairwise disjoint, which would allow us to carry out the induction and obtain that $|\A| \geq |\leaves(\zielonkaTree{\F})| = |\zielonkaAutomaton{\F}|$. Our objective will be therefore to prove:

\begin{proposition}\label{prop-ZT:existance-disjoint-subautomata}
	Let $n_0$ be the "root" of the "Zielonka tree" of $\F$, and let $n_1, n_2, \dots, n_k$ be an enumeration of the "children" of $n_0$. If $\A$ is an "HD automaton" "recognising" $\MullerC{\F}{\SS}$, then, $\A$ contains $k$ pairwise disjoint "subautomata" $\A_1, \dots, \A_k$ that are "history-deterministic" and such that $\Lang{\A_i} = \MullerC{\restSubsets{\F}{\nu(n_i)}}{\nu(n_i)}$.
\end{proposition}

The "non-determinism" of $\A$ will make this task considerably more laborious than in the previous paragraph, and we will have to thoroughly examine the strategies used by the "resolvers@@aut" for $\A$.
By the inherently asymmetric semantics of "non-deterministic" "automata", there are two well-differentiated cases to consider, depending on whether the "root" of the "Zielonka tree" is "round" ($\SS\in \F$) or "square" ($\SS\notin \F$).

In order to simplify the proof, we will assume that all states are "reachable using a sound resolver" and that all automata have a single "initial state", which can be done without loss of generality since a "resolver@@aut" for an "HD automaton" fixes such initial state in advance.

\subparagraph*{Case 1: The root of the Zielonka tree is a "square node": $\SS\notin \F$.}

\AP Let $\A = (Q,\Sigma, q_0, \GG, \DD, \WW)$ be a "non-deterministic automaton". A ""memory structure@@aut"" for $\A$ is a "memory skeleton" $\M$ over~$\DD$ together with a function $\intro*\nextmoveResolver\colon Q\times M \times \SS \to \DD$, where $M$ is the set of states of $\M$.
\AP We say that $(\M, \nextmoveResolver)$ ""implements@@resolver"" a "resolver@@aut" $(q_0,r)$ if  for all $a\in \SS$, $r(\ee,a) = \nextmoveResolver(q_0,m_0,a)$ and for all $\rr\in \DD^+$, $r(\rr,a) = \nextmoveResolver(\mtargetPath(\rr), \transMem(m_0, \rr), a)$, where $m_0$ is the initial state of $\M$ and $\transMem\colon M\times \DD^*\to M$ is its "update function@@mem".

\begin{lemma}[\cite{Boker2013NondetUnknownFuture}]\label{lemma-ZT:finite-memory-resolvers}
	Every "HD@@aut" "parity" "automaton" admits a "sound resolver@@aut" "implemented by@@resolver" a "finite memory structure".
\end{lemma}

As $\M$ is a "pointed graph" labelled with the transitions of $\A$, we could consider the "product automaton" $\A\compositionAut \M$. We want to furthermore restrict the transitions of this automaton to those that are indicated by the next-move function $\nextmoveResolver$.
\AP Given an "automaton" $\A$ and a "memory structure@@resolver" $(\M, \nextmoveResolver)$, we define their ""composition@@memory"", which we write $\A \intro*\prodMem{\ss} \M = (Q\times M, \SS, (q_0,m_0), \GG, \DD', \WW)$ as the "automaton" having transitions $(q,m) \re{a:c} (q',m')$ if $\nextmoveResolver(q,m,a) = e = q\re{a:c} q'$ and $\mu(m,e) = m'$ (formally, $\DD'$ is a subset of $\DD\times E_\M$, where $E_\M$ are the edges of the "memory skeleton"). We note that $\A \prodMem{\nextmoveResolver} \M$ is "deterministic", and it is "complete" if $\A$ is.

The following lemma follows directly from the definition of "soundness@@resolverAut" of a "resolver" and the definition of "composition@@memory" of an "automaton" and a "memory structure".
\begin{lemma}\label{lemma:validity-product-by-memory-HD-aut}
	Let $\A$ be an "automaton" and $(\M, \nextmoveResolver)$ a "memory structure" for $\A$. The "resolver implemented by" $(\M, \nextmoveResolver)$ is "sound@@aut" if and only if $\A$ and $\A\prodMem{\nextmoveResolver} \M$ "recognise" the same language. 
\end{lemma}

For the rest of the paragraph, we let $\A = (Q,\Sigma, q_0, \NN, \DD, \parity)$ be a "complete" "history-deterministic@@aut" "automaton" "recognising" the "Muller language" $\MullerC{\SS}{\F}$ admitting a "sound resolver@@aut" $(q_0,r)$ implemented by a "memory structure"  $(\M, \nextmoveResolver)$.
\AP We let $\intro*\piAut\colon \A \prodMem{\nextmoveResolver} \M \to \A$ be the "morphism of automata" given by the projection into the first component: $\pi_{\A,V}(q,m) = q$ and $\pi_{\A,E}(e_1,e_2) = e_1$.

\begin{remark}\label{rmk-ZT:projection-same-labels}
	If $\rr$ is a path in $\A \prodMem{\nextmoveResolver} \M$ that is labelled by "input letters" $a_0a_1\dots \in \SigmaInfty$ and producing "output" $c_0c_1\dots \in \NN^\infty$, then the $\piAut$-projection of $\rr$ is a path in $\A$ labelled  by $a_0a_1\dots \in \SigmaInfty$ and producing $c_0c_1\dots \in \NN^\infty$ as "output". 
\end{remark}

\begin{lemma}\label{lemma-ZT:X-FSCC-in-product-induce-subautomaton}
	Let $X\subseteq \SS$ and let $\S_X$ be an "accessible@@set" "$X$-FSCC" of $\A \prodMem{\nextmoveResolver} \M$. Then, $\piAut(\S_X)$ "induces@@aut" an "HD@@aut" "subautomaton@@induced" of $\A$ recognising $\MullerC{\restSubsets{\F}{X}}{X} = \{w\in X^\oo \mid \minf(w)\in \F\}$.
\end{lemma}
\begin{proof}
	Let $q_S$ be a state in $\piAut(\S_X)$ chosen to be initial. Let $m_S$ be a state in $\M$ such that $(q_S, m_S)\in \S_X$. By Lemma~\ref{lemma-ZT:X-FSCC-induce-automata}, $\S_X$ induces a "deterministic" "subautomaton" with initial state $(q_S, m_S)$ "recognising@@aut" $\MullerC{\restSubsets{\F}{X}}{X}$. 
	On the one hand, since $\piAut(\S_X)$ is an "accessible" "subautomaton" of $\A$ having only transitions labelled by $X$ and by "prefix-independence" of $\Lang{\A}$, we have that 
	\[ \Lang{\piAut(\S_X)} \subseteq \Lang{\A}\cap X^\oo = \MullerC{\restSubsets{\F}{X}}{X}. \]
	On the other hand, the projection of any "accepting run" in $\S_X$ provides an "accepting run" in $\piAut(\S_X)$ (by Remark~\ref{rmk-ZT:projection-same-labels}), so 
	\[\Lang{\S_X} = \MullerC{\restSubsets{\F}{X}}{X}\subseteq \Lang{\piAut(\S_X)}.\] 
	Moreover, a "sound resolver@@aut" for $\piAut(\S_X)$ is "implemented by@@resolver" $(\initialTS{\M}{m_S}, \nextmoveResolver)$ (the "memory structure" with initial state set to $m_S$).
\end{proof}

\begin{lemma}\label{lemma-ZT:disjoint-projections-HD-parity}
	Let $n\in N$ be a "square node" of the "Zielonka tree" of $\F$ ($\nu(n)\notin \F$), and let $n_1, n_2\in \children_{\zielonkaTree{\F}}(n)$ be two different "children" of $n$. If $\S_1$ and $\S_2$ are two "accessible@@sets" "$\nu(n_1)$-FSCC" and "$\nu(n_2)$-FSCC" in $\A \prodMem{\nextmoveResolver} \M$, respectively, then $\piAut(\S_1)\cap \piAut(\S_2) = \emptyset$.
\end{lemma}
\begin{proof}
	Suppose by contradiction that there is some state $q$ in $\piAut(\S_1)\cap \piAut(\S_2)$, and let $m_1, m_2\in M$ be such that $(q,m_1)$ and $(q, m_2)$ are states in $\S_1$ and $\S_2$, respectively. For $i=1,2$, let $\ell_i\in \cyclesState{\A\prodMem{\nextmoveResolver} \M}{(q,m_i)}$ be the "cycle over $(q,m_i)$" containing all edges in $\S_i$. We note that $l_\SS(\ell_i) = \nu(n_i)$ and therefore $\min \gg(\ell_i)$ has to be even (as $\A\prodMem{\nextmoveResolver} \M$ is "deterministic"), where $l_\SS$ and $\gg$ are the labellings of $\A\prodMem{\nextmoveResolver} \M$ with "input letters" and "output colours", respectively. By Remark~\ref{rmk-ZT:projection-same-labels}, the $\piAut$-projections of $\ell_1$ and $\ell_2$ are "cycles over $q$" in $\A$ labelled with $\nu(n_1)$ and $\nu(n_2)$ and in which the minimal colour appearing is even.
	By alternating these two "cycles", we can build an "accepting run" in $\A$ "over@@run" a word $w\in \SS^\oo$ with $\minf(w) = \nu(n_1)\cup \nu(n_2)$, contradicting the fact that  $\nu(n_1)\cup \nu(n_2)\notin \F$ (Remark~\ref{rmk-ZT:union-changes-acceptance}).
\end{proof}

Lemmas~\ref{lemma-ZT:existance-X-FSCC},~\ref{lemma-ZT:X-FSCC-in-product-induce-subautomaton} and~\ref{lemma-ZT:disjoint-projections-HD-parity} imply Proposition~\ref{prop-ZT:existance-disjoint-subautomata} in the case in which the "root" of the "Zielonka tree" is a "square node".

\subparagraph*{Case 2: The root of the Zielonka tree is a "round node": $\SS\in \F$.}
Before presenting the formal proof, let us discuss why considering these two cases separately is necessary. A first idea to obtain the desired result would be to follow the same steps as in Case~1. However, this approach encounters a major difficulty: the argument used in the proof of Lemma~\ref{lemma-ZT:disjoint-projections-HD-parity} is not valid if $\SS\in \F$. Indeed, even if we can find two "rejecting@@cycles" "cycles" $\ell_1$, $\ell_2$ such that $l_\SS(\ell_i) = \nu(n_i)$, their $\piAut$-projections could a priori have a "state in common"; this would imply the \emph{existence} of a "rejecting run" over the set of letters $\nu(n_1)\cup \nu(n_2)\in \F$, which is not enough to conclude, as the non-determinism of~$\A$ leaves room for the existence of other "accepting runs" over this set of letters. 
To circumvent this difficulty, we need to take a closer look at the strategies used by the "resolver". Rather than considering any "finite memory strategy" resolving the "non-determinism" of $\A$, we will show that we can choose a specific "resolver@@aut" for which we will be able to obtain a result analogous to Lemma~\ref{lemma-ZT:disjoint-projections-HD-parity}. 
To do this, we first construct the "letter game" of $\A$, as introduced in~\cite{HP06}, which is a "Muller" "game" satisfying that a "strategy" for it yields a "resolver@@aut" for $\A$.
The "strategy" that we will use in this "game" is the one obtained by applying McNaughton's algorithm to solve "Muller" "games"~\cite{McNaughton93InfiniteGames} guided by the "Zielonka tree", as presented in~\cite{DJW1997memory}.

Let $\A = (Q,\Sigma, q_0, [d_{\min}, d_{\max}], \DD, \parity)$ be a "parity" "automaton" "recognising" $\Muller{\F}$, and assume that $\SS\cap [d_{\min}, d_{\max}] = \emptyset$.
\AP The ""letter game"" for $\A$ is the "game" $\intro*\letterGame{\A}$ defined as follows:
\begin{itemize}
	\item The set of vertices is $V = Q\sqcup (Q\times \SS)$. "Adam" "controls" vertices in $Q$, and "Eve" "controls" vertices in $Q\times \SS$.
	\item For each letter $a\in \SS$ and each $q\in Q$, there is an edge $q\re{a} (q,a)$.
	\item For each position $(q,a)\in Q\times \SS$, and for each transition $q\re{a:c}q'$ in $\A$, there is an edge $(q,a)\re{c} q'$.
	\item The set of colours is $\GG = \SS\sqcup [d_{\min}, d_{\max}]$, and the "acceptance set" is the "Muller language associated to"
	\[\AP \intro*\impliesMuller{\F}{\parity} = \{ C \subseteq \GG \mid  [C\cap \SS \in \F] \implies [\min (C\cap [d_{\min}, d_{\max}]) \text{ is even}] \}.  \]
\end{itemize}
That is, in the "letter game", "Adam" provides "input letters" one by one, and "Eve" chooses transitions corresponding to those letters in the "automaton" $\A$. "Eve" "wins" this "game" if she manages to build an "accepting run" every time that "Adam" gives as input an infinite word in the language "recognised" by $\A$.

\begin{figure}[ht]
	\centering
	\includegraphics[width=0.85\textwidth]{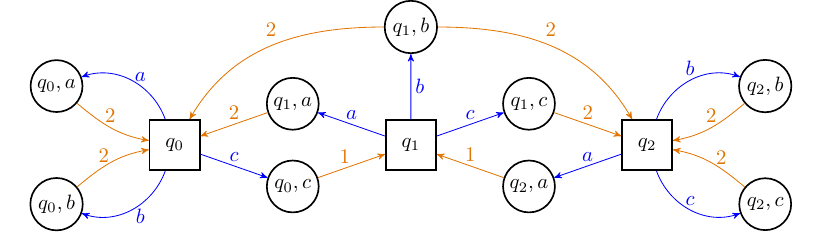}
	\caption{"Letter game" for the "HD automaton" from Figure~\ref{fig-prelim:HD-aut-example}, "recognising" the "Muller language associated to" $\F=\{\{a\},\{b\},\{c\},\{a,b\},\{b,c\}\}$. Squares represent "Adam's" vertices (the states of the automaton) and circles "Eve's" ones. 
	Blue edges correspond to input-letters, and orange edges to the choices of transitions in the automaton after each input letter. 
	The only vertex where "Eve" has a non-trivial choice to make in order to resolve the "non-determinism" of the automaton is $(q_1,b)$. In this example, "Eve" has a "winning strategy" corresponding to the "resolver" described in Example~\ref{ex-prelim:HD-aut}. }
	\label{fig-ZT:letter-game}
\end{figure}

\AP We remark that a "subgraph" of $\letterGame{\A}$ induces a "subautomaton" of $\A$ via the (partial) mapping $\intro*\projLetterGame{\A} \colon \letterGame{\A} \partialF \A$ that sends states of the form $q\in Q$ to $q$ and edges of the form $(q,a)\re{c} q'$ to $q\re{a:c}q'$. 

\begin{remark}\label{rmk-ZT:strat-letterGame-induces-resolver}
	A "strategy" for "Eve" in $\letterGame{\A}$ induces a "resolver@@aut" in $\A$, which is "sound@@resolverAut" if and only if the "strategy" is "winning". 
\end{remark}

\begin{remark}\label{rmk-ZT:disjoint-projections}
	If two subsets of vertices of the "letter game" $S_1, S_2\subseteq V$ are disjoint, then $\projLetterGame{\A}(S_1)\cap \projLetterGame{\A}(S_2)=\emptyset$.
\end{remark}

\begin{remark}\label{rmk-ZT:projections-runs-letter-game}
	If $\rr$ is a "play" in $\letterGame{\A}$, labelled $a_0c_0a_1c_1\dots \in (\SS\cdot [d_{\min}, d_{\max}])^\infty$,  the $\projLetterGame{\A}$-projection of $\rr$  is a "run@@aut" in $\A$ over $a_0a_1\dots\in \SS^\infty$ producing $c_0c_1\dots \in [d_{\min}, d_{\max}]^\infty$ as "output".
\end{remark}

\begin{lemma}[\cite{HP06}]
	A "parity" "automaton" $\A$ is "HD@@aut" if and only if "Eve" "wins" the "letter game" "from@@winGame" some "initial state" of $\A$.
\end{lemma}

\AP For a subset $X$ of vertices or edges of a "game" $\G$, we define Eve's ""attractor to $X$"" as:
\[ \intro*\attr{\G}{X} = \{v\in V \mid \text{ there is a strategy for Eve ensuring to eventually visit } X \text{ from } v\}. \]
For a colour $c\in \GG$ we note $\attr{\G}{c} = \attr{\G}{E_c}$, where $E_c$ is the set of edges coloured $c$.

For the rest of the paragraph, let $\A = (Q,\Sigma, q_0, [0, d], \DD, \parity)$ be a "complete" "history-deterministic@@aut" "parity" "automaton" "recognising" $\Muller{\F}$. We can assume  without loss of generality that the minimal colour that it uses is $0$. 
We let $V$ and $E$ denote the sets of vertices and edges, respectively, of the "letter game" and $\GG = \SS\sqcup [0,d]$ its set of colours. Whenever we use expressions like ``the minimal colour appearing in a play'', it will refer to the restriction of $\GG$ to $[0,d]$.
From the "prefix-independence" of $\Muller{\F}$ we can moreover assume that "Eve" "wins" the "letter game" from any vertex (see Lemma~\ref{lemma-prelims:prefix-indep-automaton}).
We let $n_0$ be the "root" of the "Zielonka tree" of~$\F$ (assumed to be "round", that is $\nu(n_0)\in \F$), let $n_1,\dots, n_k$ be its "children", and  let $\SS_i = \nu(n_i) \subseteq \SS$ (note that $\SS_i\notin \F$ for $i\geq 1$).

Let us examine the condition $\impliesMuller{\F}{\parity}$ used in the "letter game" a bit closer. The first levels of the "Zielonka tree"  of this condition are depicted in Figure~\ref{fig-ZT:zielonka-tree-letter-game}. It is clear that a "strategy" in $\letterGame{\A}$ ensuring to produce colour $0$ infinitely often  is "winning@@strat". It might be the case that "Adam" can prevent "Eve" from doing this, however, since "Eve" "wins" $\letterGame{\A}$, in that case she could ensure to produce infinitely often a set of colours included in some of the "round nodes" below the "root", that is, to either avoid colour $1$, or to produce letters included in some $\SS_i$. We use this idea to define next \emph{attractor decompositions} for $\letterGame{\A}$.

\begin{figure}[ht]
	\centering 
	\includegraphics[width=0.5\textwidth]{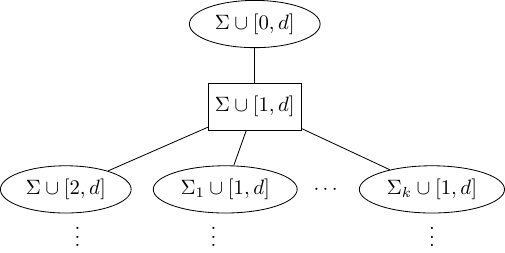}
	\caption{First levels of the "Zielonka tree" of the Muller condition $\impliesMuller{\F}{\parity}$, which is the "winning condition" of the "letter game" $\letterGame{\A}$.}
	\label{fig-ZT:zielonka-tree-letter-game}
\end{figure}

\AP Given a subset of vertices $V'\subseteq V$ we write $\intro*\subgame{V'}$ to denote the subgame of $\letterGame{\A}$ containing the vertices of $V'$ and the edges between them.

\AP Let $x$ be an even integer. For a subgame $\G' = \subgame{V'}$ of $\letterGame{\A}$ with no colour strictly smaller than $x$, we define an ""$x$-attractor decomposition"" of $\G'$ as  a partition of $V'$ into 
\[V' =\attr{\G'}{x} \sqcup V_1 \sqcup A_1 \sqcup \dots\sqcup V_l \sqcup A_l,\] 
satisfying:
\begin{itemize}
	\item \AP $\attr{\G'}{x}$ is Eve's "attractor to $x$" in $\G'$. 
	
	\item For each $V_j$, either (1) there is some $i\in \{1,\dots, k\}$ such that no colour of $\SS\setminus \SS_i$ appears in $\subgame{V_j}$, or (2) "Eve" has a "winning strategy" for $\subgame{V_j}$ (from any vertex) avoiding colour $x+1$; and in both cases, if "Adam" can leave $V_j$ taking an edge $v\re{a} v'$ ($v\in V_j$, $v'\notin V_j$), then $v'\in \attr{\G'}{x} \sqcup V_1 \sqcup A_1 \sqcup \dots  V_{j-1} \sqcup A_{j-1}$.
	In case (1) we say that $V_j$ is a ""$\SS_i$-region of the attractor decomposition"" and in case (2) that $V_j$ is an ""$x+1$-avoiding region"".
	
	\item "Eve" "wins" $\subgame{V_j}$ from every vertex for all $j$.
	
	\item $A_j = \attr{\G_j}{V_j}$, where $\G_j$ is the subgame induced by the subset of vertices given by $V\setminus \left(\attr{\G'}{x} \sqcup V_1 \ab \sqcup \dots \ab\sqcup\ab V_{j-1}\ab \sqcup A_{j-1}\right)$ (we note that this game does not contain edges coloured with $x$).
		
\end{itemize}

If $V_j$ is an "$x+1$-avoiding region", we let $\G_j'$ be the subgame obtained from $\subgame{V_j}$ by removing the transitions labelled $x+1$. 

\AP An ""$x$-recursive attractor decomposition"" of $\G'$ is:
\[ \AP \intro*\attrDec{\G'} = \langle \attr{\G'}{x}, (V_1,A_1, \attrDec{\G_1'}), (V_2,A_2, \attrDec{\G_2'}),\dots , (V_l,A_l, \attrDec{\G_l'})\rangle, \]
where $\attr{\G'}{x} \sqcup V_1 \sqcup A_1 \sqcup \dots\sqcup V_l \sqcup A_l$ is an "$x$-attractor decomposition" of $\G'$, and, if $V_j$ is an "$x+1$-avoiding region", then $\attrDec{\G_j'}$ is an "$x+2$-recursive attractor decomposition" of $\G_j'$. (If $V_j$ is an "$\SS_i$-region@@simple", $\attrDec{\G_j'}$ can be disregarded).

A representation of an attractor decomposition appears in Figure~\ref{fig-ZT:attractor-dec}. 

\begin{figure}[ht]
	\centering 
	\begin{subfigure}[c]{0.45\textwidth}
		\centering
		\includegraphics[width=\textwidth]{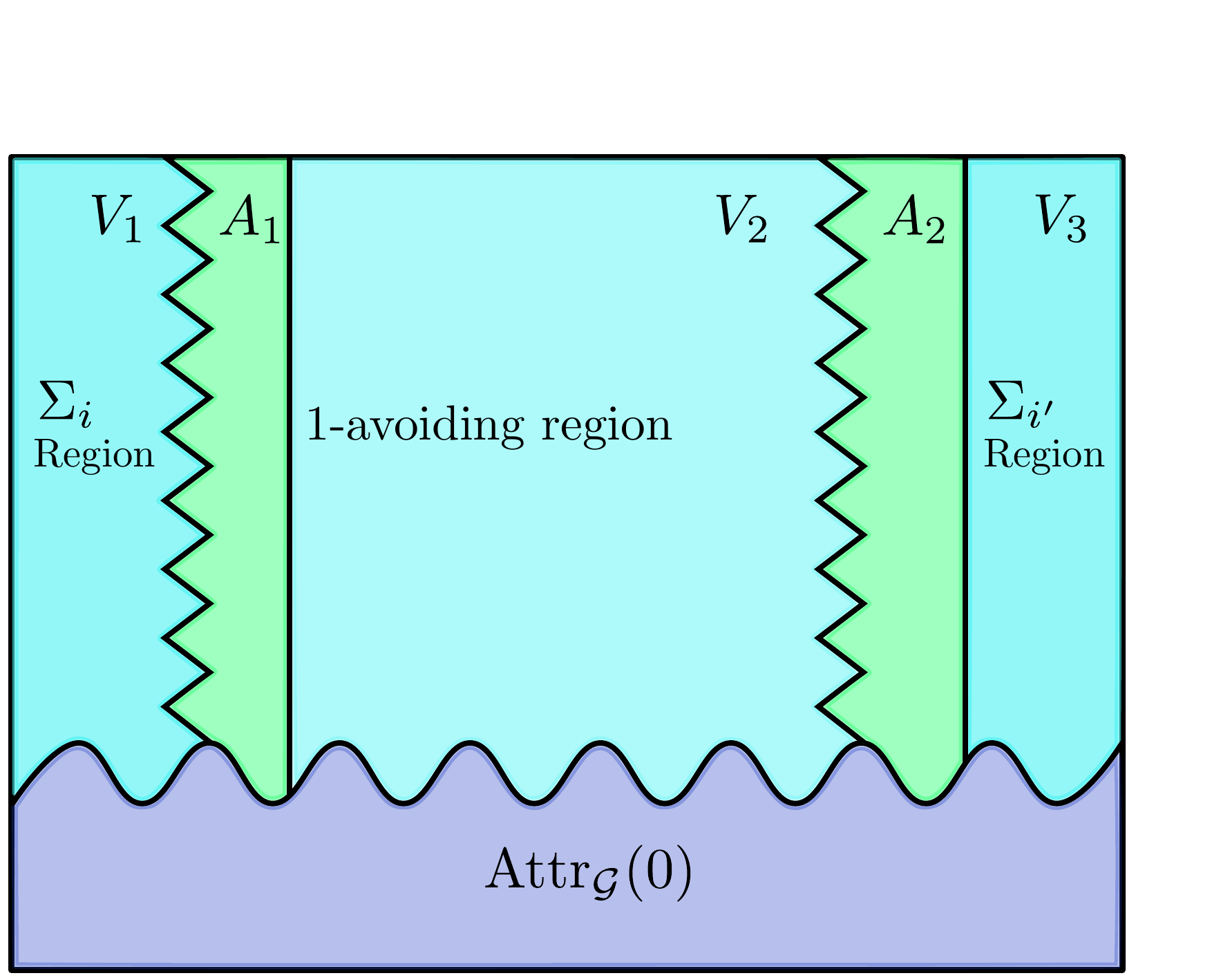}
	\end{subfigure}
	\begin{subfigure}[c]{0.45\textwidth}
		\centering
		\includegraphics[width=\textwidth]{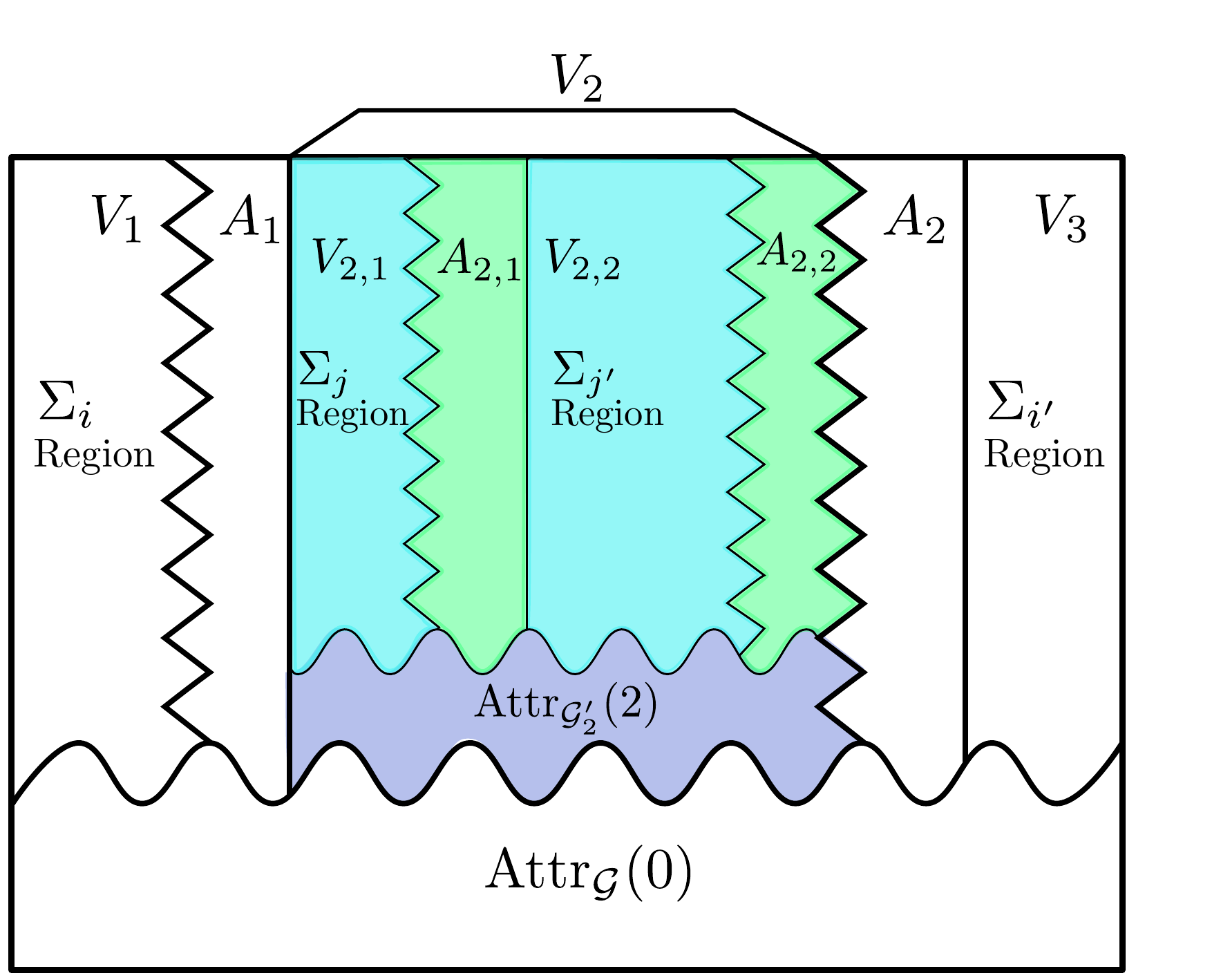}
	\end{subfigure}
	\caption{On the left, a "$0$-attractor decomposition" of a game $\G$.  On the right, the coloured part represents a "$2$-attractor decomposition" of the subgame $\G_2'$ induced by the "$1$-avoiding region" $V_2$. Since no "$3$-avoiding region" appears on it, this is a "full attractor decomposition" of the game $\G$, inducing a partition into three different kinds of regions. 
	The order over the "$\SS_i$-regions@@rec" is given by $V_1\orderAttr V_{2,1} \orderAttr V_{2,2}\orderAttr V_3$.
	Adam can only force to decrease with respect to this order, that is, at each sublevel of the decomposition, "Adam" cannot force to go to the right.}
	\label{fig-ZT:attractor-dec}
\end{figure}
 
\AP We say that a subgame $\S$ of $\G'$ is a ""$\SS_i$-region of $\attrDec{\G'}$"" if it is a "$\SS_i$-region@@simpleAttr" of some of the recursively defined "attractor decompositions". Similarly, for $y>x$ an odd integer, we say that~$\S$ is a ""$y$-avoiding region of $\attrDec{\G'}$"" if it is a  "$y$-avoiding region@@simpleAttr" of some of the recursively defined "attractor decompositions". We say that the full game $\G'$ is an $x-1$-avoiding region (note that $x$ might take the value $0$).
We remark that for any subset $S$ of vertices of $\G'$ there is one and only one minimal "$y$-avoiding region of $\attrDec{\G'}$" containing $S$ (note that $y$ might equal $-1$).

\begin{remark}\label{rmk-ZT:full-attractor-decomposition}
	A "$0$-recursive attractor decomposition" $\attrDec{\letterGame{\A}}$ of $\letterGame{\A}$ induces a partition of the vertices into
	\[ V = S_1 \sqcup \dots \sqcup S_r \; \sqcup \; A_1\sqcup\dots A_r \; \sqcup \;  B_1\sqcup \dots \sqcup B_s,  \]
	such that:
	\begin{itemize}
		\item $S_j$ is a "$\SS_i$-region of $\attrDec{\letterGame{\A}}$", for some $i\in \{1,\dots, k\}$,
		\item $A_j = \attr{\G_j}{S_j}$ for some subgame $\G_j$ appearing at some level of the "decomposition@@attrRecursive",
		\item $B_j = \attr{\G_j'}{x}$ for some even integer $x$ and some "$x-1$-avoiding region@@recursive" $\G_j'$  appearing at some level of the "decomposition@@attrRecursive".
	\end{itemize}

\AP Moreover, such a "decomposition@@attrRecursive" induces a total order over the "$\SS_i$-regions@@recursive": for two sets $S_t$, $S_{t'}$, we write $S_t \intro*\orderAttr S_{t'}$ if there are two regions $V_j$, $V_{j'}$ belonging to the same "attractor decomposition" in $\attrDec{\letterGame{\A}}$ such that $j<j'$, $S_t\subseteq V_j$ and $S_{t'}\subseteq V_{k'}$.

\AP We call such a partition a ""full attractor decomposition"" of $\letterGame{\A}$. We remark that, by definition of an "attractor decomposition", "Eve" "wins" $\subgame{S_j}$ from every vertex for every $S_j$.
See Figure~\ref{fig-ZT:attractor-dec} for an illustration.
\end{remark}

The proof that $\letterGame{\A}$ admits a "full attractor decomposition" uses the ideas appearing in \cite[Section~3]{DJW1997memory}.
\begin{lemma}
	Let $x$ be an even integer. If $\G'$ is a subgame of $\letterGame{\A}$  with no colour smaller than $x$ and such that "Eve" can "win" from every vertex, then it admits an "$x$-attractor decomposition".
	In particular, $\letterGame{\A}$ admits a "full attractor decomposition".
\end{lemma}
\begin{proof}
	We assume without loss of generality that $x=0$.
	Suppose that $V_1, A_1,\dots V_{j-1}, A_{j-1}$ have already been defined and that they verify the desired properties. Suppose that the game $\G_j$ with vertices $V\setminus \left(\attr{\G'}{0} \sqcup V_1 \sqcup \dots  V_{j-1} \sqcup A_{j-1}\right)$ is non-empty. First, note that "Eve" "wins" $\G_j$ from any position. Indeed, "Eve" "wins" $\G'$ from any vertex $v$ in $\G_j$ (as we suppose that she can "win"~$\G'$ starting anywhere); moreover, since $v\notin A_{j'}$ for any $j'<j$, "Adam" has a "strategy" from $v$ forcing to remain in $\G_j$, and "Eve" has to be able to "win" against any such "strategy".
	
	We prove that either (1) there is some $i\in \{1,\dots, k\}$ and $v$ vertex in $\G_j$ such that "Eve" has a "winning strategy" from $v$ forcing to produce no colour in $\SS\setminus \SS_i$, or (2) there is some vertex $v$ in $\G_j$  such that "Eve" has a "winning strategy" from $v$ avoiding colour $1$. Suppose by contradiction that this was not the case. Then, Adam can use the following strategy: first, he forces producing colour $1$, then, a colour not in $\SS_1$, followed by a colour not in $\SS_2$, and continues this pattern until a colour not in $\SS_l$ is produced (and this without producing colour $0$, since no $0$-edge appears in $\G_j$). Afterward, he continues repeating these steps in a round-robin fashion. This allows him to produce a play winning for him (the word produced is in $\Muller{\F}$ while the minimal number produced is $1$), contradicting the fact that "Eve" "wins" $\G_j$ from $v$.
	
	We assume that we are in the case (1) (case (2) is identical), so from some vertices "Eve" can "win" producing no colour in $\SS\setminus \SS_i$. We let $V_j$ be the set of such vertices, and for each of them we fix a "strategy" $\strat_{v}$ that is "winning" in $\G_j$ and avoids colours in $\SS\setminus \SS_i$.  By definition of $V_j$, if $\rr = v\lrp{\phantom{w.}} v'$ is a finite play "consistent with" $\strat_{v}$ in $\G_j$, then $v'\in V_j$ (Eve can still win without producing colours in $\SS\setminus \SS_i$), so "Adam" cannot force leaving $V_j$. This proves that:
	\begin{enumerate}
		\item $\strat_v$ is "winning" in $\subgame{V_j}$ from $v$,
		\item if $v\in V_j$ is controlled by "Adam" and $v\re{}v'$ is an edge in $\letterGame{\A}$, then $v'\in \attr{\G'}{0} \sqcup V_1 \sqcup A_1 \sqcup \dots  V_{j-1} \sqcup A_{j-1} \sqcup V_j$.
	\end{enumerate}	
	Also, if a vertex $v$ controlled by Adam is in $V_j$, no edge $v\re{a\notin \SS_i}v'$ appears in $\G_j$, so no colour of $\SS\setminus \SS_i$ appears in $\subgame{V_j}$.
	
	To finish the proof, we define $A_j$ to be the attractor of $V_j$ in $\G_j$.
 
	The existence of a  "full attractor decomposition" for $\letterGame{\A}$ follows from the fact that any "$x+1$-avoiding region" of an "$x$-attractor decomposition" verifies the hypothesis of the lemma.
\end{proof}

\begin{lemma}[\cite{McNaughton93InfiniteGames}]\label{lemma-ZT:uniform-memory-strat}
	Let $\G$ be a "game" using a "Muller" "acceptance condition" such that "Eve" "wins" $\G$ from every vertex. Then, there is a finite "memory structure"  $(\M,\nextmove)$ over $\G$ "implementing@@mem" a "winning strategy" ""uniformly"", that is, for every vertex $v$ of $\G$ there is a memory state $m_v$ in $\M$ such that the "memory structure" $(\initialTS{\M}{m_v},\nextmove)$ "implements@@strat" a "winning strategy" from $v$.
\end{lemma}

For the rest of the paragraph, we fix a "$0$-recursive attractor decomposition" $\attrDec{\letterGame{\A}}$ for $\letterGame{\A}$ and let $S_1 \orderAttr S_2 \dots \orderAttr S_r$ be the "$\SS_i$-regions@@recursive" of the induced "full attractor decomposition".
For each region $S_j$ we fix a "memory structure" $(\M_j, \ss_j)$ "uniformly@@strat" "implementing" a "winning strategy" for "Eve" in $\subgame{S_j}$  (as given by Lemma~\ref{lemma-ZT:uniform-memory-strat}).
As in the previous paragraph, we can consider the "composition@@memory" $\subgame{S_j} \prodMem{\ss_j} \M_j$ consisting of the "product game@@mem" in which the choices for "Eve" are restricted to those of the form 
$(v,m) \re{c} (v',m')$ if  $\ss_j(v,m) =  e = v\re{c} v'$ and $\mu_j(m,e) = m'$. 
By definition, "Eve" does not have any choice in $\subgame{S_j} \prodMem{\ss_j} \M_j$, and since $(\M_j, \ss_j)$ "implements@@strat" a "winning strategy", any infinite "path" in $\subgame{S_j} \prodMem{\ss_j} \M_j$ produces a set of colours in $\impliesMuller{\F}{\parity}$.
\AP We let $\intro*\piLetterGame\colon \subgame{S_j} \prodMem{\ss_j} \M_j \to \subgame{S_j}$ be the projection into $\subgame{S_j}$.

\AP A "subgraph" $G_j$ of $\subgame{S_j} \prodMem{\ss_j} \M_j$ is ""$X$-Adam-closed@@letterGame"", for a subset $X\subseteq \SS$, if for every vertex $(q,m)$ controlled by Adam and every $a\in X$, the transition $(q,m)\re{a} ((q,a),m')$ remains in $G_j$. 
We say that $G_j$ is an ""$X$-FSCC@@letterGame"" if it is a "final SCC" of the restriction of $\subgame{S_j} \prodMem{\ss_j} \M_j$ to the graph where Adam's choices are restricted to letters in $X$ that is moreover "$X$-Adam-closed@@letterGame". 
\AP We say that a "subgraph" $\G$ of $\letterGame{\A}$ is an ""$X$-closed subgame"" (with respect to the "attractor decomposition" $\attrDec{\letterGame{\A}}$ and a family of "finite memory@@strat" "strategies")  if $\G = \piLetterGame(G_j)$ for $G_j$ some "$X$-Adam-closed@@letterGame" "SCC" of some product $\subgame{S_j} \prodMem{\ss_j} \M_j$.

Intuitively, an "$X$-closed subgame" $\G$ of the "letter game" is a subgame included in a region $S_j$ of the "full attractor decomposition"  such that, if "Adam" only provides letters in $X$ and "Eve" plays according to the strategy defined by the "memory structures" $\M_j$, the play will never leave $\G$. 

\begin{lemma}\label{lemma-ZT:Eve-wins-closed-subgames}
	"Eve" "wins" any "$X$-closed subgame" of $\letterGame{\A}$ (from any vertex).
\end{lemma}
\begin{proof}
	In an "$X$-closed subgame" included in a region $\subgame{S_j}$, "Adam's" moves have been restricted; however, all Eve's moves coming from the "strategy" "implemented by" $(M_j, \ss_j)$ are available. Therefore, this "strategy" is also "winning" in such a subgame, since it is "winning" in the full $\subgame{S_j}$.
\end{proof}

Putting this lemma together with Remark~\ref{rmk-ZT:strat-letterGame-induces-resolver} we obtain:
\begin{lemma}\label{lemma-ZT:Ai-closed-sets-induce-HD-automata}
	Let $X\subseteq \SS$, and let $\G_X \subseteq \letterGame{\A}$ be an "$X$-closed subgame" of $\letterGame{\A}$. The "subautomaton" of $\A$ "induced by@@aut" $\projLetterGame{\A}(\G_X)$ is "HD@@aut" and "recognises" $\MullerC{\restSubsets{\F}{X}}{X}$. 
\end{lemma}

\begin{lemma}\label{lemma-ZT:exist-X-closed-regions}
	If a product $\subgame{S_j} \prodMem{\ss_j} \M_j$ does not contain any "$X$-Adam-closed subgraph@@letterGame", for $X\subseteq \SS$, then from any vertex $(q,m)$ "Adam" can force leaving $S_j$ while playing only letters in $X$. 
	That is, there is a path $(q,m)\lrp{\phantom{w.}} (q',m')$ in $\subgame{S_j} \prodMem{\ss_j} \M_j$ producing exclusively letters in $X$ such that, for some $a\in X$, the edge $q'\re{a} (q',a)$ does not belong to $\subgame{S_j}$. 
\end{lemma}
\begin{proof}
	If this was not the case, the "subgraph" of $\subgame{S_j} \prodMem{\ss_j} \M_j$ consisting of the vertices that can be reachable from $(q,m)$ by reading letters in $X$ would form an "$X$-Adam-closed@@letterGame" subgraph.
\end{proof}

\begin{lemma}\label{lemma-ZT:existance-Ai-closed-regions}
	For each label $\SS_i$ of the "children" of the "root" of $\zielonkaTree{\F}$, $\letterGame{\A}$ admits some "$\SS_i$-closed subgame" contained in a "$\SS_i$-region of $\attrDec{\letterGame{\A}}$". 
\end{lemma}
\begin{proof}
	Assume that the "full attractor decomposition" of $\letterGame{\A}$ induced by $\attrDec{\letterGame{\A}}$ is the following:
	\[ V = S_1 \sqcup \dots \sqcup S_r \; \sqcup \; A_1\sqcup\dots A_r \; \sqcup \;  B_1\sqcup \dots \sqcup B_s,  \]
	
	We fix the following "strategy" $\strat$ for "Eve" in the "letter game": 
	\begin{itemize}
		\item whenever the play lands to $B_j$, where $B_j = \attr{\G_j'}{x}$ for some even colour $x$, she forces producing colour $x$,
		\item whenever the play arrives to some $A_j$, she forces going to $S_j$,
		\item in regions $S_j$ she uses the "strategy" $(\M_j, \ss_j)$. More precisely, let $m_v$ be the state of $\M_j$ such that $(\M_{j,v}, \ss_j)$ "implements" a "winning strategy" for $\subgame{V_j}$ from $(v,m_v)$. Each time that the play arrives to a vertex $v$ in $V_j$ from a different region, "Eve" uses $(\M_{j,v}, \ss_j)$.
	\end{itemize}

	\begin{claim}\label{claim-ZT:strat}
		Let $\rr$ be a "play" consistent with $\strat$ (from any vertex), and let $y\geq -1$ be the maximal odd number such that $\minf(\rr)$ is contained in a "$y$-avoiding region@@rec" $\S$ of $\attrDec{\letterGame{\A}}$. Then, either~$\rr$  eventually stays in a "$\SS_i$-region@@rec" $S_j$ contained in $\S$, or the minimal colour produced infinitely often by $\rr$ is $y+1$.
	\end{claim}
	\begin{subproof}
		Let $\attr{\S}{y+1}\sqcup V_1\sqcup A_1 \sqcup\dots \dots, V_l\sqcup A_l$ be the "attractor decomposition" of $\S$ appearing in $\attrDec{\letterGame{\A}}$.
		By definition of an "attractor decomposition", each time that the play leaves a $V_j$ region, the next vertex is in $v'\in \attr{\S}{y+1} \sqcup V_1 \sqcup A_1 \sqcup \dots  V_{j-1} \sqcup A_{j-1}$.  
		First, if $V_j$ is a "$y+2$-avoiding region@@simple", $\rr$ cannot stay in it (by maximality of $y$).
		Thus, if $\rr$ does not eventually stay in a "$\SS_i$-region@@rec", it leaves regions $V_j$ infinitely often, so it must 
		produce $y+1$ infinitely often too. Since $\S$ is a "$y$-avoiding region@@rec", no colour smaller than $y+1$ is produced.
	\end{subproof}
	We obtain as a consequence that $\strat$ is "winning for Eve" from any initial position: any "play" staying in a "$y$-avoiding region@@rec" and producing infinitely many $y+1$'s is winning, and if a "play" eventually stays in a "$\SS_i$-region@@rec" $S_j$, it has to be "winning@@strat" since the "strategy" "implemented by" $(\M_j, \ss_j)$ is "winning@@strat" in there.

	We remark that we can extract a "$\SS_i$-FSCC@@letterGame"  from any "$\SS_i$-Adam-closed@@letterGame" "subgraph" of $\subgame{S_j} \prodMem{\ss_j} \M_j$, that will be contained in the "$\SS_i$-region@@rec" $S_j$, so it suffices to prove the existence of such "$\SS_i$-Adam-closed@@letterGame" subgraphs. We also recall that in $\subgame{S_j} \prodMem{\ss_j} \M_j$ all choices are left to Adam, so he can choose to produce any path in this product whenever the play arrives to a vertex $v$ in $S_j$.
	
	Suppose by contradiction that no "accessible" "$\SS_i$-Adam-closed@@letterGame" "subgraph" exists in any of the products. We consider a "play" in which "Adam" does the following:
	\begin{enumerate}[label=(\alph*)]
		\item the letters that he gives form a word $w\in\SS^\oo$ such that $\minf(w) = \SS_i$, \label{item:letter-inf-often}
		\item each time that the "play" arrives to a region $S_j$, he exists this region in a finite number of steps.\label{item:Adam-exists-regions}
	\end{enumerate}
	Indeed, he can ensure to exit regions $S_j$ while only producing letters in $\SS_i$ by Lemma~\ref{lemma-ZT:exist-X-closed-regions}. By Claim~\ref{claim-ZT:strat}, the minimal colour produced infinitely often by such a play is even.
	By Remark~\ref{rmk-ZT:projections-runs-letter-game}, we can project such a play in the automaton $\A$, obtaining an "accepting@@run" "run over" $w$. This is a contradiction, since $w\notin \Muller{\F} = \Lang{\A}$ (because $\SS_i\notin \F$).
	We conclude that some $\subgame{S_j} \prodMem{\ss_j} \M_j$ admits a "$\SS_i$-FSCC@@letterGame", and therefore  $\letterGame{\A}$ admits some "$\SS_i$-closed subgame".	  
\end{proof}

We can now infer Proposition~\ref{prop-ZT:existance-disjoint-subautomata} in the case in which the "root" of $\zielonkaTree{\F}$ is "round": from Lemma~\ref{lemma-ZT:existance-Ai-closed-regions}, we obtain "$\SS_i$-closed subgames" in $\letterGame{\A}$ for each $i\in\{1,\dots, k\}$ that are moreover contained in "$\SS_i$-regions". Therefore, their $\projLetterGame{\A}$-projections are disjoint (Remark~\ref{rmk-ZT:disjoint-projections}), and each of these projections "induces@@aut" an "HD@@aut""-subautomaton" recognising $\restSubsets{\F}{\SS_i}$ (Lemma~\ref{lemma-ZT:Ai-closed-sets-induce-HD-automata}).

\subsection{A minimal history-deterministic Rabin automaton}\label{subsec-zt: GFG-Rabin}
In this section, we present the construction of a "history-deterministic" "Rabin" "automaton" $\zielonkaHDAutomaton{\F}$ for a "Muller language" $\Muller{\F}$ using the "Zielonka tree" $\zielonkaTree{\F}$, and prove its minimality (Theorem~\ref{thm-zt:optimality_ZT-HD-Rabin}).
The automaton $\zielonkaHDAutomaton{\F}$ can be seen as a quotient of the "ZT-parity-automaton"; that is, $\zielonkaHDAutomaton{\F}$ is obtained by merging some states of $\zielonkaAutomaton{\F}$. Thus, we replace the complexity in the number of states by complexity in the "acceptance condition".
The size of the "automaton" $\zielonkaHDAutomaton{\F}$ is a well-studied parameter of "Zielonka trees": its "round-branching width", $\memTree{\zielonkaTree{\F}}$. This parameter was introduced by Dziembowski, Jurdziński and Walukiewicz~\cite{DJW1997memory} (under the name of \emph{memory of $\zielonkaTree{\F}$}) and shown to coincide with the memory required by "Eve" to "win" in "games" using $\Muller{\F}$ as an "acceptance condition" (see Proposition~\ref{prop-ZT:mem-ZT-DJW} below). In this paper, we are not concerned with the memory of winning conditions, but we will use the result from~\cite{DJW1997memory} to obtain the minimality of $\zielonkaHDAutomaton{\F}$.

We note that this construction is asymmetric, in the sense that we show it for "Rabin" "automata", but not for "Streett" "automata" (their dual notion). The reason why we cannot dualize the construction is due to the semantics of "non-deterministic" "automata". However, we could use the same idea to obtain a minimal \emph{universal} "history-deterministic" "Streett" "automaton" (we refer to~\cite{BL19GFGFromND} for the definition of universal HD automata).

\subsubsection{The Zielonka-tree-HD-Rabin-automaton}\label{subsubsec-zt-GFG: definition}
\begin{definition}[\cite{DJW1997memory}]
	Let $T$ be a "tree" with nodes partitioned into "round" and "square nodes", and let $T_1, \dots, T_k$ be the "subtrees of $T$ rooted at" the "children" of the "root" of $T$. 
	\AP We define inductively the ""round-branching width"" of $T$, denoted $\intro*\memTree{T}$ as:
	\begin{equation*}
		\memTree{T} =	\begin{cases}
			1 & \text{ if } T \text{ has exactly one "node",}\\[1mm]
			\max \{ \memTree{T_1}, \dots , \memTree{T_k}\} & \text{ if the "root" is "square",}\\[1mm]
			\sum\limits_{i=1}^k \memTree{T_i} & \text{ if the "root" is "round".}
		\end{cases}
	\end{equation*}
\end{definition}

The next lemma directly follows from the definition of $\memTree{T}$.
\begin{lemma}\label{lemma:property_star}
	Let $T = (N = N_{\bigcirc} \sqcup N_\Box, \ancestor)$ be a "tree" with nodes partitioned into "round" and "square" nodes.
	There exists a mapping $\eta\colon \leaves(T) \to \{1,2,\dots, \memTree{T}\}$ satisfying:
	\begin{align}\label{eq:property-star}
		\nonumber&\text{If } n \in N \text{ is a "round node" with children } n_1 \neq n_2 \text{, for any pair}\\
		\tag{$\star$}&\text{of leaves } l_1 \text{ and } l_2 \text{ below } n_1 \text{ and } n_2 \text{, respectively, } \eta(l_1)\neq \eta(l_2).
	\end{align}
\end{lemma}

\begin{example}
	Let $\F=\{ \{a,b\}, \{a,c\}, \{b\}\}$ be the family of subsets considered in Example~\ref{example-ZT:zielonka-tree}. 
	The "round-branching width" of $\zielonkaTree{\F}$ is  $\memTree{\zielonkaTree{\F}} =2$. A labelling $\eta\colon \leaves(\zielonkaTree{\F}) \to \{1,2\}$ satisfying Property~\ref{eq:property-star} is given by $\eta(\theta) = \eta(\xi) = 1$ and $\eta(\zeta) = 2$.
	This labelling is represented in the "Zielonka tree" $\zielonkaTree{\F}$ on the left of Figure~\ref{fig-ZT:zielonka-Rabin automaton}.
\end{example}

\begin{definition}[Zielonka-tree-HD-Rabin-automaton]\label{def-ZT: HD_Rabin_ZielonkaAutomaton}
	\AP Let $\F\subseteq \powplus{\SS}$, let $\zielonkaTree{\F} = (N = \roundnodes \sqcup \squarenodes, \ancestor)$ be its "Zielonka tree" and $\intro*\eta\colon \leaves(\zielonkaTree{\F}) \to \{1,2,\dots, \memTree{\zielonkaTree{\F}}\}$ be a mapping satisfying Property~\eqref{eq:property-star}.
	We define the ""ZT-HD-Rabin-automaton"" $\intro*\zielonkaHDAutomaton{\F}=(Q, \SS, I, \GG, \DD, \RabinC{R}{\GG})$ as a (non-deterministic) "automaton" using a "Rabin" "acceptance condition", where:
	\begin{itemize}
		\item $Q=\{1,2, \dots, \memTree{\zielonkaTree{\F}}\}$, 
		\item $I=Q$,\footnotemark{}
		\item $\GG = N$ (the colours of the "acceptance condition" are the nodes of the "Zielonka tree"),
		\item $\transAut{}(q,a) = \{\big(\jump(l,\supp(l,a)), \supp(l,a)\big) \mid l\in \leaves(\zielonkaTree{\F}) \tst \eta(l) = q\}$,
		\item $R = \{(G_n,R_n)\}_{n\in\roundnodes}$, where $G_n$ and $R_n$ are defined as follows: Let $n$ be a "round node" and $n'$ be any node of $\zielonkaTree{\F}$,
		\begin{equation*}
			\begin{cases}
				n' \in G_n & \text{ if } n'=n,\\
				n' \in R_n & \text{ if }  n'\neq n \text{ and } n \text{ is not an "ancestor" of } n'. 
			\end{cases}
		\end{equation*} 
	\end{itemize}
\end{definition}
\footnotetext{Any non-empty subset of $Q$ can be chosen as the set of "initial states".}

\begin{remark}
	Although we will usually say that $\zielonkaHDAutomaton{\F}$ is \emph{the} "ZT-HD-Rabin-automaton" of $\F$, the structure of this automaton is not unique, it depends on two choices: the "order@@tree" over the nodes of the "Zielonka tree" and the mapping $\eta$.
\end{remark}

The intuition behind this definition is the following. The "automaton" $\zielonkaHDAutomaton{\F}$ has $\memTree{\zielonkaTree{\F}}$ states, and each of them can be associated to a subset of "leaves" of $\zielonkaTree{\F}$ by $\inv{\eta}(q)$. The mapping~$\eta$ is such that the lowest common ancestor of two "leaves" in  $\inv{\eta}(q)$ is a "square node".
As for the "ZT-parity-automaton", for each "leaf" of $l\in \leaves(\zielonkaTree{\F})$ and letter $a\in \SS$, we identify the "deepest" ancestor $n=\supp(l,a)$ containing $a$ in its label, and, using the $\jump$ function, pick a "leaf" $l'$ below the next "child" of $n$. We add a transition $q\re{a:n}q'$ if there are "leaves" $l\in \inv{\eta}(q)$ and $l'\in \inv{\eta}(q')$ giving such a path (we note that the output colour is given by $n=\supp(l,a)$, although this node does not appear as a state of the automaton). This way, we can identify a "run" in the "automaton" $\zielonkaHDAutomaton{\F}$ with a promenade through the nodes of the "Zielonka tree" in which jumps between "leaves" with the same $\eta$-image are allowed. If during this promenade a unique minimal node (for $\ancestor$) is visited infinitely often, it is not difficult to see that the sequence of input colours belongs to $\F$ if and only if the label of this minimal node belongs to $\F$ (it is a "round node"). The "Rabin condition" over the set of nodes of the "Zielonka tree" is devised so that it accepts exactly these sequences of nodes (see Lemma~\ref{lemma-ZT:Rabin-acc-sequences-nodes-ZT} below). 

Another way of presenting the "automaton" $\zielonkaHDAutomaton{\F}$ is as a quotient of the "deterministic" "parity" "automaton" $\zielonkaAutomaton{\F}$. Indeed, the graph structure and the "labelling" by "input letters" of $\zielonkaHDAutomaton{\F}$ is obtained by merging the states of $\zielonkaAutomaton{\F}$ (which are the "leaves" of $\zielonkaTree{\F}$) with the same $\eta$-image, and keeping all the transitions between them. However, a "parity" "acceptance condition" over this smaller structure is no longer sufficient to accept $\Muller{\F}$.

\begin{example}
	The "ZT-HD-Rabin-automaton" $\zielonkaHDAutomaton{\F}$ of the family  $\F=\{ \{a,b\}, \{a,c\}, \{b\}\}$ from Example~\ref{example-ZT:zielonka-tree} is shown on the right of Figure~\ref{fig-ZT:zielonka-Rabin automaton}.
	The "Zielonka tree" $\zielonkaTree{\F}$ appears on the left of the figure, and the labelling $\eta\colon \leaves(\zielonkaTree{\F}) \to \{1,2\}$ is represented by the numbers below its "branches".
	
	The "Rabin" "condition" of this automaton is given by two "Rabin pairs" (corresponding to the "round nodes" of the "Zielonka tree"):\\

	\centering
	\begin{tabular}{l l}
		$G_\bb = \{\bb\}$, & $R_\bb= \{\aa, \lambda, \xi, \zeta\}$,\\[1mm] 
		$G_\lambda = \{\lambda\}$, & $R_\lambda= \{\aa, \bb, \theta\}$. 
	\end{tabular}
 \newline
 \vspace{1mm}
We note that the automaton $\zielonkaHDAutomaton{\F}$ is obtained by merging the states $\theta$ and $\xi$ from the "ZT-parity-automaton" $\zielonkaAutomaton{\F}$ appearing in Figure~\ref{fig-ZT:zielonka-parityAutomaton}, and replacing the "output colours" by suitable nodes from the "Zielonka tree".
\end{example}

\begin{figure}[ht]
		\begin{minipage}[c]{0.45\textwidth} 
			\includegraphics[width=0.7\textwidth]{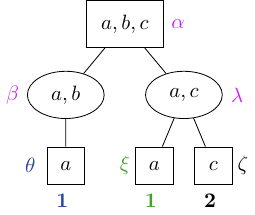}
		\end{minipage} 
		\begin{minipage}[c]{0.45\textwidth} 
			\includegraphics[width=0.95\textwidth]{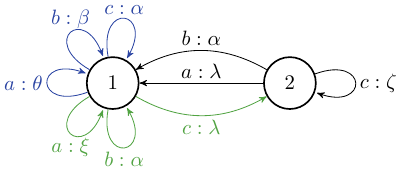}
		\end{minipage} 
		\caption{On the left, the "Zielonka tree" of $\F=\{ \{a,b\}, \{a,c\}, \{b\}\}$. On the right, the "ZT-HD-Rabin- automaton" $\zielonkaHDAutomaton{\F}$. Blue transitions correspond to those coming from "leaf" $\theta$, and green ones to those originating from "leaf" $\xi$.}
		\label{fig-ZT:zielonka-Rabin automaton}
	\end{figure}

\begin{remark}\label{rmk-ZT:simplified-automata}
	We observe that the "automaton" from Figure~\ref{fig-ZT:zielonka-Rabin automaton} presents "duplicated edges", in the sense that there are two transitions $q\re{a:x}q'$ and $q\re{a:y}q'$ between the same pair of states and reading the same "input letter".
	We can always avoid this and remove "duplicated edges" from any "automaton". We provide a proof in Appendix~\ref{sec:appendix-simplifications-automata} (Proposition~\ref{prop-app:simplification_Rabin}). For the  language from the previous example, an equivalent automaton is proposed in Figure~\ref{fig-app:simplified-Rabin}
\end{remark}

\paragraph*{Correctness of the Zielonka-tree-HD-Rabin-automaton.}

\begin{proposition}[Correctness]\label{prop-ZT:correctness_ZT_HD-Rabin}
	Let $\F\subseteq \powplus{\SS}$ be a family of non-empty subsets. Then, 
	\[ \Lang{\zielonkaHDAutomaton{\F}} = \MullerC{\F}{\SS}. \]
	Moreover, the "automaton" $\zielonkaHDAutomaton{\F}$ is "history-deterministic".
\end{proposition}

\begin{lemma}\label{lemma-ZT:Rabin-acc-sequences-nodes-ZT}
	Let $u = n_0n_1n_2\dots \in N^\oo$ be an infinite sequence of nodes of the "Zielonka tree". The word $u$ belongs to $\RabinC{R}{N}$, for $R = \{(G_n,R_n)\}_{n\in\roundnodes}$ the Rabin condition of $\zielonkaHDAutomaton{\F}$, if and only if there is a unique minimal node for the "ancestor relation" in $\minf(u)$ and this minimal node is "round" (recall that the root is the minimal element for $\ancestor$).
\end{lemma}
\begin{proof}
	Assume that there is a unique minimal node in $\minf(u)$, called $n$, and that $n$ is "round". We claim that $u$ is "accepted by the Rabin pair" $(G_n,R_n)$. It is clear that $\minf(u)\cap G_n\neq \emptyset$, because $n\in G_n$. It suffices to show that $\minf(u)\cap R_n = \emptyset$: By minimality, any other node $n'\in \minf(u)$ is a "descendant" of $n$ (equivalently, $n$ is an "ancestor" of $n'$), so $n'\notin R_n$.
	
	Conversely, assume that $u\in \RabinC{R}{N}$. Then, there is some "round node" $n\in \roundnodes$ such that $\minf(u)\cap G_n \neq \emptyset$ and $\minf(u)\cap R_n = \emptyset$. Since $G_n = \{n\}$, we deduce that $n\in \minf(u)$. Moreover, as $\minf(u)\cap R_n = \emptyset$, all nodes in $\minf(u)$ are descendants of $n$. We conclude that $n$ is the unique minimal node in $\minf(u)$, and it is "round".
\end{proof}

\begin{lemma}\label{lemma-ZT:morphism-parity-Rabin}
	There exists a "morphism" "of automata@@morphism" $\pp\colon \zielonkaAutomaton{\F} \to \zielonkaHDAutomaton{\F}$.
\end{lemma}
\begin{proof}
	We define the morphism $\pp$ as follows:
	\begin{itemize}
		\item $\pp_V(l) = \eta(l)$, for $l\in \leaves(\zielonkaAutomaton{\F})$,
		\item for a transition $e = l\re{a:c} l'$ in $\zielonkaAutomaton{\F}$, we let $\pp_E(e) = (\eta(l), a, \supp(l,a), l')$.
	\end{itemize}
	It is clear that $\pp$ is a "weak morphism". We prove that it "preserve the acceptance of runs". Let $\rr = l_0\re{w_0}l_1\re{w_1}l_2\re{w_2}\dots \in \Runs{\zielonkaAutomaton{\F}}$ be an infinite run in $\zielonkaAutomaton{\F}$ (the only "run over" $w_0w_1w_2\dots \in \SS^\oo$), and let $n_i = \supp(l_i,w_i)$. By definition of the morphism, the "output@@run" of the "run" $\rr' = \ppRuns(\rr)$ in $\zielonkaHDAutomaton{\F}$ is $\gg'(\rr') = n_0n_1n_2\dots \in N^\oo$. In the proof of Proposition~\ref{prop-ZT:correctness_ZT_parity}, we proved (Claims~\ref{claim:supp-below-n-CORRECT-ZT} and~\ref{claim:supp-n_inf_often-CORRECT-ZT}) that there exists a unique node $n_w$ appearing infinitely often in $\gg'(\rr')$. Moreover, we proved that $\rr$ is "accepting@@run" in $\zielonkaAutomaton{\F}$ if and only if $n_w$ is "round". Lemma~\ref{lemma-ZT:Rabin-acc-sequences-nodes-ZT} allows us to conclude that $\ppRuns(\rr)$ is "accepting@@run" in $\zielonkaHDAutomaton{\F}$ if and only if $\rr$ is "accepting@@run" in $\zielonkaAutomaton{\F}$.
\end{proof}

\begin{proof}[Proof of Proposition~\ref{prop-ZT:correctness_ZT_HD-Rabin}]
	\textbf{$\Lang{\zielonkaHDAutomaton{\F}} \subseteq \MullerC{\F}{\SS}$:} Let $w\in \Lang{\zielonkaHDAutomaton{\F}}$ and let $u\in N^\oo$ be the sequence of nodes produced as "output" of an "accepting@@run" "run over" $w$ in $\zielonkaHDAutomaton{\F}$. By Lemma~\ref{lemma-ZT:Rabin-acc-sequences-nodes-ZT}, there is a unique  minimal node $n$ for $\ancestor$ appearing infinitely often in $u$ and moreover~$n$ is "round".
	Let $n_1, \dots, n_k$ be an enumeration of the "children" of $n$ ("from left to right"), with labels $\nu(n_i)\subseteq \SS$ (we remark that $\nu(n_i)\notin \F$, for $1\leq i \leq k$). We will prove that $\minf(w)\subseteq \nu(n)$ and $\minf(w)\nsubseteq \nu(n_i)$ for $1\leq i \leq k$. By definition of the "Zielonka tree", as $n$ is "round", this implies that $\minf(w)\in \F$.

	Since eventually all nodes produced as "output" are "descendants" of $n$ (by minimality), $\minf(w)$ must be contained in $\nu(n)$ (by definition of the transitions of $\zielonkaHDAutomaton{\F}$).
	
	We suppose, towards a contradiction, that $\minf(w)\subseteq \nu(n_j)$ for some $1\leq j \leq k$.  
	Let $Q_i=\{\eta(l)\: : \: l \text{ is a "leaf" "below@@tree" } n_i\}$ be the set of states corresponding to "leaves" under $n_i$, for $1\leq i \leq k$. We can assume that the "leaves" corresponding to transitions of an accepting "run over" $w$ are all "below@@tree" $n$, and therefore, transitions of such a run only visit states in $\bigcup_{i=1}^kQ_i$. Indeed, eventually this is going to be the case, because if some leaves $l, l'$ corresponding to a transition $(q,a,n',q')$ are not below $n$, then $n'$ would not be a "descendant" of $n$ (since $n'$ is the least common ancestor of $l$ and $l'$).
	Also, by Property~\eqref{eq:property-star}, we have $Q_i\cap Q_j = \emptyset$, for all $i\neq j$.
	By definition of the transitions of $\zielonkaHDAutomaton{\F}$, if $a\in \SS$ is a letter in $\nu(n)$ but not in $\nu(n_i)$, all transitions from some state in $Q_i$ reading the colour $a$ go to $Q_{i+1}$, for $1\leq i \leq k-1$ (and to $Q_1$ if $i=k$).
	Also, if $a\in \nu(n_i)$, transitions from states in $Q_i$ reading $a$ stay in $Q_i$.	
	We deduce that a run over $w$ will eventually only visit states in $Q_j$, for some $j$ such that $\minf(w)\subseteq \nu(n_j)$. However, the only transitions
	from $Q_j$ that would produce $n$ as output are those corresponding to a colour $a\notin \nu(n_j)$, so the node $n$ is not produced infinitely often, a contradiction.	

	\textbf{$\MullerC{\F}{\SS} \subseteq \Lang{\zielonkaHDAutomaton{\F}}$ and "history-determinism@@aut"}:	
	We claim that the existence of a "morphism" $\pp\colon \zielonkaAutomaton{\F} \to \zielonkaHDAutomaton{\F}$ (Lemma~\ref{lemma-ZT:morphism-parity-Rabin}) and the correctness of $\zielonkaAutomaton{\F}$ (Proposition~\ref{prop-ZT:correctness_ZT_parity}) imply that $\Lang{\zielonkaHDAutomaton{\F}} = \Lang{\zielonkaAutomaton{\F}} = \Muller{\F}$. Indeed, if $\rr$ is an "accepting run" over $w\in \SS^\oo$ in $\zielonkaAutomaton{\F}$, then $\ppRuns(\rr)$  is an "accepting run" over $w$ in $\zielonkaHDAutomaton{\F}$.
	We can moreover use $\zielonkaAutomaton{\F}$ and $\pp$ to define a "sound@@aut" "resolver@@aut" $(r_0,r)$ for $\zielonkaHDAutomaton{\F}$: we let $r_0 = \pp(q_0)$ be the image of the "initial state" of $\zielonkaAutomaton{\F}$. If $\rr_R\in \RunsFin{\zielonkaHDAutomaton{\F}}$ is the image under $\ppRuns$ of some "finite run" $\rr_P\in \RunsFin{\zielonkaAutomaton{\F}}$, we let $r(\rr_R, a) = \pp(e)$, where $e$ is the only $a$-labelled transition from $\mtargetPath(\rr_P)$. We define $r$ arbitrarily in other case. This way, for every $w\in \SS^\oo$, the "run induced by@@aut" $r$ over $w$ is the image of a "run over" $w$ in $\zielonkaAutomaton{\F}$, which must be "accepting@@run" if $w\in \Muller{\F}$.
\end{proof}

\subsubsection{Optimality of the Zielonka-tree-HD-Rabin-automaton}\label{subsubsec-zt-HD: optimality}

We devote this section to the proof of the optimality of $\zielonkaHDAutomaton{\F}$.

\begin{theorem}[Optimality of the "ZT-HD-Rabin-automaton"]\label{thm-zt:optimality_ZT-HD-Rabin}
	Let $\A$ be a "history-deterministic" "Rabin" "automaton" "accepting@@automaton" a "Muller language" $\MullerC{\F}{\SS}$. Then, $ |\zielonkaHDAutomaton{\F}|\leq \size{\A}$.
\end{theorem}

\begin{proposition}[\cite{DJW1997memory}]\label{prop-ZT:mem-ZT-DJW}
	Let $L = \MullerC{\F}{\SS}$ be a "Muller language". 
	\begin{enumerate}
		\item If "Eve" "wins" a "game" with $L$ as "acceptance set" from a position $v$, there is a "winning strategy from" $v$ for her "implemented by@@stratMem" a "memory structure" of size $\memTree{\zielonkaTree{\F}}$.
		\item There exists a "game" $\G$ using $L$ as "acceptance condition" in which "Eve" can win from a position $v$, but there is no "winning strategy from" $v$ for her "implemented by@@stratMem" a "memory structure" of size strictly smaller than $\memTree{\zielonkaTree{\F}}$.
	\end{enumerate}
\end{proposition}

\begin{lemma}[\cite{Klarlund94Determinacy,Zielonka1998infinite}]\label{lemma-ZT:Rabin-games-positional}
	"Rabin languages" are positionally determined, that is, if "Eve" "wins" a "game" using a "Rabin" "acceptance condition" from a position $v$, there is a "winning strategy from" $v$ for her "implemented by@@stratMem" a "memory structure" of size $1$. 
\end{lemma}

\begin{corollary}\label{cor-ZT:HD-Rabin-memory}
	Let $\A$ be a "history-deterministic" "Rabin" "automaton". Then, if "Eve" "wins" a "game" with $\WW = \Lang{\A}$ as "acceptance set" from a position $v$, there is a "winning strategy from" $v$ for her "implemented by@@stratMem" a "memory structure" of size $|\A|$.
\end{corollary}
\begin{proof}
	Let $\G$ be a "game" with $\WW = \Lang{\A}$ as "acceptance set". In order to be able to take the "product@@aut" by $\A$ and obtain an equivalent game, we transform $\G$ into a "game" "suitable for transformations". Let $\tilde{\G}$ be the game obtained from $\G$ in the following way: for every edge $e = v\re{a} v'$ in $\G$, we add a position $(v,e)$ controlled by "Eve" and replace edge $e$ by $v\re{\ee} (v,e) \re{a} v'$. It is clear that "Eve" wins $\G$ from a vertex $v$ if and only if she "wins" $\tilde{\G}$ from that same vertex.
	By Proposition~\ref{prop-prelim:composition_games_HD}, if "Eve" "wins" $\G$ from a vertex $v$, she "wins" $\tilde{\G} \compositionAut \A$ from a vertex $(v,q_0)$, where $q_0$ is an initial vertex of $\A$. Moreover, the "game" $\tilde{\G} \compositionAut \A$ uses the "acceptance set" from $\A$, which is a "Rabin language", so, by Lemma~\ref{lemma-ZT:Rabin-games-positional}, she can win using a "strategy" given by a function $\nextmove\colon \kl{\tilde{V}_\Eve}  \to \tilde{E}$, where $Q$ is the set of states of $\A$ and $\kl{\tilde{V}_\Eve}$ the vertices "controlled by" "Eve" in $\tilde{\G}$ (a subset of $(\VEve \sqcup (V\times E))\times Q$).
	We build a "memory structure" $(\M, \nextmove_\M)$  of size $|Q|$ that projects the "strategy" "implemented by" $\nextmove$ onto $\G$:	
	\begin{itemize}
		\item its set of states is $M=Q$,
		\item the initial state is $q_0$,
		\item the "update function@@mem" $\mu\colon M \times E \to M$ sends $\mu(q,e) = q'$ if $\nextmove((v,e),q) = ((v,e),q) \re{} (v',q')$ is the move chosen by $\nextmove$ from vertex $((v,e),q)$,
		\item for $v\in \VEve$, $q\in M$, we let $\nextmove_\M(v,q) = e$ if $e$ is the move chosen by $\nextmove$ from $(v,q)$, that is, if $\nextmove(v,q) = (v,q) \re{} ((v,q),e)$.
	\end{itemize}
Since $\ss$ "implements" a "winning strategy" in $\tilde{\G}\compositionAut \A$ from $(v,q_0)$, its projection onto $\G$ via the "memory structure" $(\M, \nextmove_\M)$ is a "strategy" that verifies that any play "consistent with@@strat" it produces as "output" a word in $\Lang{\A}$, so it is "winning@@strat".
\end{proof}

Theorem~\ref{thm-zt:optimality_ZT-HD-Rabin} is obtained by combining the fact that $|\zielonkaHDAutomaton{\F}| = \memTree{\zielonkaTree{\F}}$ with Proposition~\ref{prop-ZT:mem-ZT-DJW} (second item) and Corollary~\ref{cor-ZT:HD-Rabin-memory}.

	\section{The alternating cycle decomposition: An optimal approach to Muller transition systems}\label{section:acd}

In Section~\ref{section:zielonka-tree}, we have provided minimal "parity" and "Rabin" "automata" for "Muller languages", using the "Zielonka tree".
We can use these "automata" to transform "Muller" "transition systems", by applying the "product construction".
However, this approach overlooks the structure of the "transition system", meaning it does not take into account the relevant interplay between the "underlying graph" and the "acceptance condition".

In this section, we present our main contributions: optimal transformations of "Muller" "transition systems" into "parity" and "Rabin" ones.
The key novelty is that they precisely capture the way the "transition system" interacts with the "acceptance condition".
This is achieved by generalising "Zielonka trees" from "Muller languages" to "Muller" "transition systems"; we define the "alternating cycle decomposition" (ACD), consisting in a collection of Zielonka-tree-like structures subsuming all the structural information of the "transition system" necessary to determine whether a "run" is "accepting@@run" or not.
More precisely, the "ACD" is a succinct representation of the \emph{alternating chains of loops} of a "Muller" "automaton", in the sense of Wagner~\cite{Wagner1979omega}. The alternating chains of loops of a "DMA" are known to determine the "parity index" of the language it "recognises"~\cite{Wagner1979omega}, and, as we will show, they also capture the essential information to define optimal transformations of automata.

We start with the definition of the "alternating cycle decomposition" in Section~\ref{subsec-acd: acd-definition}. In Section~\ref{subsec-acd: parity-transformation}, we describe the "ACD-parity-transform", turning a "DMA" $\A$ into an equivalent "DPA" $\acdParityTransform{\A}$. Formally, the validity of this transformation is witnessed by a "locally bijective morphism" $\pp\colon \acdParityTransform{\A}\to \A$ (Proposition~\ref{prop-ACD:correctness_ACD-parity_transform}).
In Section~\ref{subsec-acd: HD-Rabin-transformation}, we describe the "ACD-HD-Rabin-transform" that turns a "DMA" $\A$ into an equivalent "history-deterministic" "Rabin" "automaton" $\acdRabinTransform{\A}$. The validity of the transformation is witnessed by an "HD mapping" $\pp\colon \acdRabinTransform{\A}\to \A$ (Proposition~\ref{prop-ACD:correctness_ACD-HD-Rabin-transform}).
These constructions grant strong optimality guarantees. The "automaton" $\acdParityTransform{\A}$ (resp. $\acdRabinTransform{\A}$) has a minimal number of states amongst "parity" (resp. "Rabin") "automata" admitting an "HD mapping" to $\A$ (Theorems~\ref{thm-acd:optimality-size_ACD-parity_transform} and~\ref{thm-acd:optimality_ACD-HD-Rabin-transform}). We note that this implies minimality amongst "automata" admitting a "locally bijective morphism" to $\A$.
Moreover, the "acceptance condition" of $\acdParityTransform{\A}$ uses an optimal number of "colours" (Theorem~\ref{thm-acd:optimality-priorities_ACD-parity_transform}).
The optimality of these constructions is shown in Section~\ref{subsec-acd-parity: optimality}. We are able to prove the optimality of both constructions at the same time, by reducing the problem to an application of the minimality of the "ZT-parity-automaton" and the "ZT-HD-Rabin-automaton".\\

In all this section, we let $\TS = (\underlyingGraph{\TS}, \macc{\TS})$ be a "Muller" "transition system" with "underlying graph"  $\underlyingGraph{\TS} =(V ,E , \msource , \mtarget, I)$ and using a "Muller" "acceptance condition" $\macc{\TS} = (\gg, \GG, \MullerC{\F}{\GG})$.

\subsection{The alternating cycle decomposition}\label{subsec-acd: acd-definition}

\begin{definition}\label{def:tree_alternating_cycles}
	Let $\ell_0 \in \cycles{\TS}$ be a "cycle". 
	\AP We define the ""tree of alternating subcycles"" of~$\ell_0$, denoted $\intro*\altTree{\ell_0} = (N, \ancestor, \intro*\nuAcd\colon N \to \cycles{\TS})$ as a  $\cycles{\TS}$-"labelled tree" with nodes partitioned into ""round nodes@@acd"" and ""square nodes@@acd"", $N= \roundnodes \sqcup \squarenodes$, such that:
	\begin{itemize}
		\item The "root" is labelled $\ell_0$.
		\item If a node is labelled $\ell\in \cycles{\TS}$, and $\ell$ is an "accepting cycle" ($\gg(\ell)\in \F$), then it is a "round node@@acd", and its "children" are labelled exactly with the maximal subcycles $\ell' \subseteq \ell$ such that $\ell'$ is "rejecting@@cycle" ($\gg(\ell') \notin \F$).
		\item If a node is labelled $\ell\in \cycles{\TS}$, and $\ell$ is a "rejecting cycle" ($\gg(\ell)\notin \F$), then it is a "square node@@acd", and its "children" are labelled exactly with the maximal subcycles $\ell' \subseteq \ell$ such that $\ell'$ is "accepting@@cycle" ($\gg(\ell') \in \F$).
	\end{itemize}
\end{definition}

\AP For a $\cycles{\TS}$-"labelled tree" $T = (N, \ancestor, \nuAcd\colon N \to \cycles{\TS})$ and $n\in N$, we let $\intro*\nuStates(n) = \states{\nuAcd(n)}$ be the "set of states@@cycle" of the "cycle" labelling $n$. 

\begin{remark}\label{rmk-acd:union-changes-acceptance}
	Let $n$ be a node of $\altTree{\ell_0}$ and let $n_1$ be a "child" of it. If $\ell'$ is a "cycle" such that $\nuAcd(n_1) \subsetneq \ell' \subseteq \nu(n)$, then $\nu(n_1) \text{ is "accepting@@cycle" } \iff \ell' \text{ is "rejecting@@cycle" } \iff \nu(n) \text{ is "rejecting@@cycle"}$. 
\end{remark}

\begin{definition}[Alternating cycle decomposition]\label{def:acd}
	\AP Let $\TS$ be a "transition system", and let $\ell_1, \ell_2, \dots, \ell_k$ be an enumeration of its maximal "cycles" (that is, the edge set of its "SCCs"). We define the ""alternating cycle decomposition"" of $\TS$ as  the "forest" $\intro*\acd{\TS} = \{\altTree{\ell_1},\dots, \altTree{\ell_k}\}$.
	
	We let $\intro*\nodesAcdCycle{\ell_i}$ be the set of nodes of $\altTree{\ell_i}$, and $n_{\ell_i}$ its "root". We will assume that $\nodesAcdCycle{\ell_i} \cap \nodesAcdCycle{\ell_j} = \emptyset$ if $i\neq j$.
\end{definition}

\AP We define the ""set of nodes of $\acd{\TS}$"" to be $\intro*\nodesAcd{\TS} = \bigcup_{i = 1}^k N_{\ell_i}$, and we let $\intro*\nodesAcdRound{\TS}$ (resp. $\intro*\nodesAcdSquare{\TS}$) be the subset of "round@@acd" (resp. "square@@acd") nodes. 
As for "Zielonka trees", from now on we equip the "trees" of $\acd{\TS}$ with an arbitrary "order@@tree" making them "ordered trees", without explicitly mentioning it.

We remark that for a "recurrent" vertex $v$ of $\TS$, there is one and only one "tree" $\altTree{\ell_i}$ in $\acd{\TS}$ such that $v\in \nuStates(n_{\ell_i})$. On the other hand, "transient" vertices do not appear in the trees of $\acd{\TS}$.

\AP If $v$ is a "recurrent" vertex of $\TS$, we define the ""local subtree at $v$"", noted $\intro*\treeVertex{v}$, as the "subtree" of $\altTree{\ell_i}$ containing the nodes 
$ \intro*\nodesTreeVertex{v} = \{ n \in N_{\ell_i} \mid v\in \nuStates(n)\}$.
 If $v$ is a "transient" vertex, we define $\treeVertex{v}$ to be a "tree" with a single node.

For $v$ "recurrent", as $\nodesTreeVertex{v}$ is a subset of the nodes of $\altTree{\ell_i}$, the "tree" $\treeVertex{v}$ inherits the "order@@tree" from $\altTree{\ell_i}$, as well as its partition into "round@@acd" and "square@@acd" nodes, $\nodesTreeVertex{v} = \roundnodesv \sqcup \squarenodesv$. Also, it inherits the labelling given by the mapping $\nuAcd$, whose restriction to $\treeVertex{v}$ has an image in $\cyclesState{\TS}{v}$.

\begin{remark}\label{rmk-acd:treeVertex-closed-by-ancestor}
	Let $v\in \nuStates(n_{\ell_i})$. If $n\in \nodesTreeVertex{v}$ and $n'$ is an "ancestor" of $n$ in $\altTree{\ell_i}$, then $n'\in \nodesTreeVertex{v}$. In particular, $\treeVertex{v}$ is indeed a "subtree" of $\altTree{\ell_i}$. Also, we note that the "root" of $\treeVertex{v}$ is $n_{\ell_i}$.
\end{remark}

\AP For a node $n\in N_{\ell_i}$ and an edge $e\in \ell_i$ we define $\intro*\suppAcd(n,e) = n'$ to be the "deepest" "ancestor" of $n$ such that $e\in \nuAcd(n')$. We remark that if $e = v\re{} v'$, then $\suppAcd(n,e)$ is a node in both $\treeVertex{v}$ and $\treeVertex{v'}$.

\begin{example}\label{ex-acd:example-acd}
	We will use the "transition system" $\TS$ from Figure~\ref{fig-acd:transition-system} as a running example. 
	We have named the edges of $\TS$ with letters from $a$ to $l$, that are also used as the "output colours" of the "acceptance condition". The "acceptance set" of $\TS$ is the "Muller language associated to":
		\begin{flalign*}
			\F=\{\{c,d,e \},\{e \},
			\{ g,h,i \},\{l \},
			\{h,i,j,k \},\{j,k \}	\}.
		\end{flalign*}
	
	The "initial vertex" of $\TS$, $v_0$, is its only "transient vertex", all the others vertices are "recurrent". 	
	$\TS$ has $2$ strongly connected components, corresponding to "cycles" $\ell_1$ and $\ell_2$.
	
	\begin{figure}[ht]
		\centering 
		\includegraphics[width=0.6\textwidth]{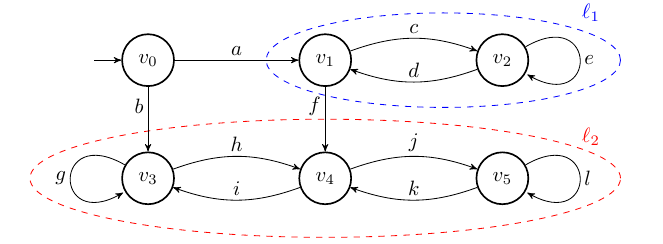}
		\caption{"Transition system" $\TS$ using a "Muller" "acceptance condition" given by $\F=\{\{c,d,e \},\{e \},\{ g,h,i \},$ $\{l \},	\{h,i,j,k \},\{j,k \}\}$. The two maximal "cycles", $\ell_1$ and $\ell_2$, are encircled by blue and red dashed lines, respectively.}
		\label{fig-acd:transition-system}
	\end{figure}

	The "alternating cycle decomposition" of $\TS$ is shown in Figure~\ref{fig-acd:acd}. It consists of two "trees", $\altTree{\ell_1}$ and $\altTree{\ell_2}$.
	We use Greek letters (in pink) to name the nodes of the tree.
	Inside each node we indicate both its label $\nuAcd(n)$ and the "set of states@@cycle" of it. For example, $\nuAcd(\kappa)=\{g,h,i\}$ and $\nuStates(\kappa)=\{v_3,v_4\}$.
	We have that $\suppAcd(\tau,g)=\kappa$ and $\suppAcd(\tau,j)=\lambda$.
	We highlight in bold orange the "local subtree at $v_4$",  $\treeVertex{v_4}$. The tree $\treeVertex{v_0}$, consisting in a single node, does not appear in the figure.
	The numbering on the right of the trees will be used in the next section.\begin{figure}[ht]
		\centering 
		\includegraphics[width=0.85\textwidth]{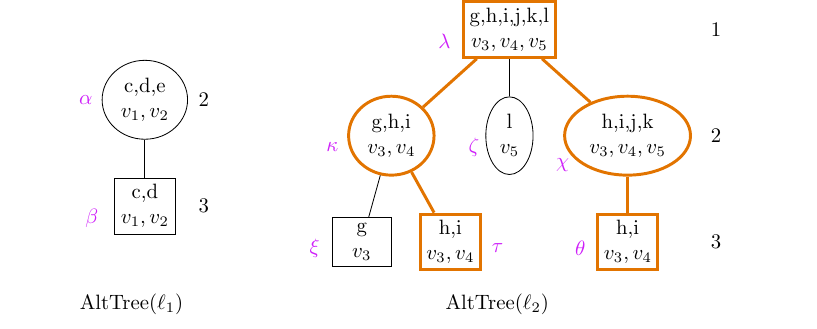}
		\caption{"Alternating cycle decomposition" of $\TS$. In bold orange, the "local subtree at" $v_4$, $\treeVertex{v_4}$.}
		\label{fig-acd:acd}
	\end{figure}
\end{example} 

\begin{remark}\label{rmk-acd:complementation-ACD}
	Let $\TS$ be a "Muller" "TS" using as "acceptance set" $\WW = \MullerC{\F}{\GG}$, and let $\intro*\complTS{\TS}$ be the "TS" obtained by replacing $\WW$ with $\intro*\complSet{\WW} = \GG^\oo \setminus \WW$ (which is a "Muller language"). Then, the "ACD" of $\complTS{\TS}$ coincides with that of $\TS$, with the only difference that the partition into "round@@acd" and "square nodes@@acd" is inverted: $\nodesAcdRound{\complTS{\TS}} = \nodesAcdSquare{\TS}$ and $\nodesAcdSquare{\complTS{\TS}} = \nodesAcdRound{\TS}$.
	
	We note that if $\A$ is a "DMA" "recognising" $L\subseteq \SS^\oo$, the "automaton" $\complTS{\A}$ is a "DMA" "recognising"~$\SS^\oo\setminus L$.
\end{remark}

\begin{remark}\label{rmk-acd:ZT-as-ACD}
	The "Zielonka tree" can be seen as a special case of the "alternating cycle decomposition". Indeed, a "Muller language" $\MullerC{\F}{\SS}$ can be trivially "recognised" by a "DMA" $\A$ with a single state $q$ and self-loops $q\re{a:a}q$. The "ACD" of this "automaton" is exactly the "Zielonka tree" of~$\F$.
\end{remark}

\begin{remark}[Size and computation of the ACD]
    Let $\TS$ be a Muller "TS" and let $\zielonkaTree{\F}$ be the "Zielonka tree" of its "acceptance set". It can be shown that for each vertex $v$ of $\TS$ we have that $|\treeVertex{v}| \leq |\zielonkaTree{\F}|$, and therefore the size of $\acd{\TS}$ is polynomial in $\zielonkaTree{\F}$. 
    This, and the question of the complexity of computing the "ACD" is the subject of an independent work~\cite{CM24Simplifying}.
\end{remark}

\subparagraph{Local Muller languages.}\label{subsect-prelim:local-Muller}
\AP For a "recurrent" state $v$ of $\TS$, we define the ""local Muller language of $\TS$ at $v$"" as the "Muller language" defined over the alphabet $\GG_v = \cyclesState{\TS}{v}$ associated to:
\[ \intro*\localMuller{v}{\TS} = \{ \C \subseteq \cyclesState{\TS}{v} \mid \bigcup_{\ell\in \C} \ell \; \text{ is an "accepting@@cycle" "cycle"}\}. \] 

We note that $\localMuller{v}{\TS}$ is determined by singletons ($\C \in \localMuller{v}{\TS}$ if and only if $\{\bigcup_{\ell \in \C}\ell\} \in \localMuller{v}{\TS}$). For simplicity, and by a slight abuse of notation, we will work as if $\localMuller{v}{\TS} \subseteq \cyclesState{\TS}{v}$.
Also, to lighten notations, we will just write $\localMuller{v}{\TS}$ to denote $\Muller{\localMuller{v}{\TS}}$ whenever no confusion arises.

The following lemma directly follows from the definition of $\treeVertex{v}$ and that of the "Zielonka tree". It provides insight in the structure of the "trees" $\treeVertex{v}$, and it will be a key ingredient in the proof of the optimality of the transformations based on the "alternating cycle decomposition".

\begin{lemma}\label{lemma-acd:tree_q_ZT_local_Muller}
	The "tree" $\treeVertex{v}$
	 is the "Zielonka tree" of the family $\localMuller{v}{\TS}$,\footnotemark{} for any recurrent vertex $v$.
\end{lemma}
\footnotetext{Formally, the labelling $\nu$ of $\treeVertex{v}$ goes to $\cyclesState{\TS}{v}$, and not to $\powplus{\cyclesState{\TS}{v}}$, as required by the definition of the "Zielonka tree". To obtain a proper "Zielonka tree" with a labelling of nodes $\nu'\colon \nodesTreeVertex{v} \to \powplus{\cyclesState{\TS}{v}}$, we would have to define $\nu'(n) = \{\ell'\in \cyclesState{\TS}{v} \mid \ell' \subseteq \nu(n)\}$.}

\subsection{An optimal transformation to parity transition systems}\label{subsec-acd: parity-transformation}
We now define the "ACD-parity-transform", an optimal transformation turning a "Muller" "TS" into a "parity" "TS" while preserving "determinism". In order to obtain the optimality in the number of "output colours", we need to pay attention to the parity of the minimal colour used in different "SCCs". To incorporate this parameter in the transformation, we define "positive@@acd" and "negative@@acd" "ACDs".\\

Let $\TS$ be a "Muller" "transition system" and let $\acd{\TS} = \{\altTree{\ell_1},\dots, \altTree{\ell_k}\}$ be its "alternating cycle decomposition".

\AP We say that a "tree" $\altTree{\ell_i}\in \acd{\TS}$ is ""positive@@tree"" if $\ell_i$ is an "accepting cycle", and that it is ""negative@@tree"" otherwise.
\AP We say that the "alternating cycle decomposition" of $\TS$ is ""positive@@acd"" if all the "trees" of maximal "height" of $\acd{\TS}$ are "positive@@tree", that it is ""negative@@acd"" if all "trees" of maximal "height" are "negative@@tree", and that it is ""equidistant@@acd"" if there are "positive@@tree" and "negative@@tree" trees of maximal "height".

\AP As for the "Zielonka tree", we associate a non-negative integer to each level of the trees of $\acd{\TS}$ via a function $\intro*\parityNodesAcd(n)\colon \nodesAcd{\TS} \to \NN$. Let $\ell_i$ be a maximal "cycle" of $\TS$ and $n\in N_{\ell_i}$. 
\begin{itemize}
	\item If $\acd{\TS}$ is "positive@@acd" or "equidistant@@acd": 
	\begin{itemize}
		\item $\parityNodesAcd(n) = \depth(n)$, if $\ell_i$ is "accepting@@cycle",
		\item $\parityNodesAcd(n) = \depth(n) + 1$, if $\ell_i$ is "rejecting@@cycle".
	\end{itemize}
	\item If $\acd{\TS}$ is "negative@@acd":
	\begin{itemize}
		\item $\parityNodesAcd(n) = \depth(n)+2$, if $\ell_i$ is "accepting@@cycle",
		\item $\parityNodesAcd(n) = \depth(n) + 1$, if $\ell_i$ is "rejecting@@cycle".
	\end{itemize}
\end{itemize}
\AP We let $\intro*\minparityAcd{\TS}$ (resp. $\intro*\maxparityAcd{\TS}$) be the minimum (resp. maximum) value taken by the function~$\parityNodesAcd$. 

\begin{remark}\label{rmk-acd:parity_levels_ACD}
	A node $n$ in $\altTree{\ell_i}$ verifies that $\parityNodesAcd(n)$ is even if and only if $\nuAcd(n)$ is an "accepting cycle" (that is, if $n$ is a "round node@@acd").
\end{remark}

\begin{remark}
	It is satisfied:
	\begin{itemize}
		\item $\minparityAcd{\TS}=0$ if $\acd{\TS}$ is "positive@@acd" or "equidistant@@acd",
		\item $\minparityAcd{\TS}=1$ if $\acd{\TS}$ is "negative@@acd".
	\end{itemize}	 
\end{remark}

\begin{example}
	In the previous Example~\ref{ex-acd:example-acd}, $\altTree{\ell_1}$ is a "positive tree" and $\altTree{\ell_2}$ is  "negative@@tree". As $\altTree{\ell_2}$ is the tree of maximal "height",  $\acd{\TS}$ is "negative@@acd". 
	The function $\parityNodesAcd$ is represented in Figure~\ref{fig-acd:acd} by the integers on the right of each tree. It takes values $2$ and $3$ over $\altTree{\ell_1}$ ($\parityNodesAcd(\aa)=2$ and $\parityNodesAcd(\bb)=3$), because $\acd{\TS}$ is "negative@@acd". 
	In this example, $\minparityAcd{\TS}=1$ and $\maxparityAcd{\TS}=3$.
	We note that if we had associated integers $0$ and $1$ to the levels of $\altTree{\ell_1}$, we would have used $4$ integers in total, instead of just $3$ of them.		
\end{example}

\begin{definition}[ACD-parity-transform]\label{def:parityTransformationACD}
	 Let $\TS$ be a "Muller" "TS" with  $\acd{\TS} = \{\altTree{\ell_1}\ab,\ab \dots,\ab \altTree{\ell_k}\}$. 
	 \AP We define the ""ACD-parity-transform"" of $\TS$  be the "parity" "TS" $\intro*\acdParityTransform{\TS} = (\underlyingGraph{}', \macc{}')$, with $\underlyingGraph{}' = (V', E', \msource', \mtarget', I')$, and $\macc{}' = (\gg', \ab[\minparityAcd{\TS},\ab \maxparityAcd{\TS}], \parity)$ defined as follows.
	 \begin{description}
	 	\setlength\itemsep{2mm}
	 	\item[Vertices.] The set of vertices is \[V'= \bigcup_{v\in V} \left( \{v\}\times \leaves(\treeVertex{v}) \right).\]
	 	
	 	\item[Initial vertices.] 
	 	$I' = \{ (v_0, n) \mid v_0\in I \tand n \text{ is the "leftmost" "leaf" in } \treeVertex{v_0} \}$.
	 	
	 	\item[Edges and output colours.] For each $(v,n)\in V'$ and each edge $e = v\re{}v'\in \mout(v)$ in $\TS$ we define an edge $e_n = (v,n)\re{\gg'(e_n)} (v',n')$. Formally, 
	 	\[E'= \bigcup_{e\in E} \left( \{e\}\times \leaves(\treeVertex{\msource(e)}) \right).\]
	 	 If $v$ and $v'$ are not in the same "SCC", we let $n'$ be the "leftmost" "leaf" in $\treeVertex{v'}$ and $\gg'(e_n) = \minparityAcd{\TS}$.\footnotemark{} If $v$ and $v'$ belong to the same "SCC", we let:
	 	\begin{itemize}
	 		\item $n' = \jump_{\treeVertex{v'}}(n, \suppAcd(n,e))$,
	 		\item $\gg'(e_n) = \parityNodesAcd(\suppAcd(n,e))$.
	 	\end{itemize}
 	\footnotetext{The colours associated to transitions changing of "SCC" are \emph{almost} arbitrary (we could even leave them "uncoloured"). We define them to be the minimal colour used so that the obtained "transition system" is "normalised" in the sense of Section~\ref{subsec-corollaries: normal-form}.}
 	
 		\item[Labellings.]If $\TS$ is a "labelled transition system", with labels $l_V\colon V \to L_V$ and $l_E\colon E \to L_E$, we label $\acdParityTransform{\TS}$ by $l'_{V'}(v,n) = l_V(v)$ and $l'_{E'}(e_n) = l_E(e)$.
	 \end{description}
\end{definition}

Intuitively, a "run" in the "transition system" $\acdParityTransform{\TS}$ follows a "run" in $\TS$ with some extra information, updated in the same manner as it was the case with the "ZT-parity-automaton". 
To define transitions in $\acdParityTransform{\TS}$, we move simultaneously in $\TS$ and in $\acd{\TS}$. When
we take a transition $e$ in $\TS$ that goes from $v$ to $v'$, while being in a node $n$ in the "ACD", we climb the
branch of $n$ searching the "lowest@@tree" node $\tilde{n}$ with $e$ and $v'$ in its label ($\tilde{n}=\suppAcd(n,e)$). We produce as output the
colour corresponding to the level reached. If no such node exists in the current tree (this occurs if we change of "SCC"), we jump to the "root" of the tree containing $v'$. After having reached the node $\tilde{n}$, we move to its next child in the tree $\treeVertex{v'}$ (in a cyclic way), and we pick the leftmost "leaf" under it. 

\begin{example}\label{ex-acd:ACD-parity_transform}
	We show in Figure~\ref{fig-acd:ACD-parity_transform} the "ACD-parity-transform" $\acdParityTransform{\TS}$  of the "transition system" $\TS$ from Figure~\ref{fig-acd:transition-system} (Example~\ref{ex-acd:example-acd}).
	For each vertex $v$ in $\TS$, we make as many copies as "leaves" of the tree $\treeVertex{v}$. We note that, as $v_0$ is "transient",  the tree $\treeVertex{v_0}$ consists of a single node (by definition), that we name~$\iota$.
	Transitions are of the form $(e,l)$, for $e$ a transition from $\TS$ and $l$ a "leaf" of some "local subtree"; these are denoted $e_l$ in the figure for the sake of space convenience. These labels simply indicate the names of the edges, they should not be interpreted as "input letters" ($\acdParityTransform{\TS}$ is not an "automaton").
	
	We observe that the mappings $\pp_V(v,l) = v$ and $\pp_E(e_l)=e$ define a "locally bijective" "morphism of transition systems" from $\acdParityTransform{\TS}$ to $\TS$. 

	\begin{figure}[ht]
		\centering 
		\includegraphics[width=0.7\textwidth]{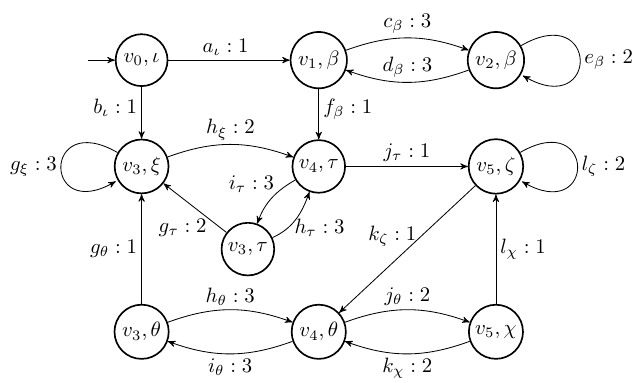}
		\caption{"ACD-parity-transform" $\acdParityTransform{\TS}$  of the "transition system" $\TS$ from Figure~\ref{fig-acd:transition-system}.}
		\label{fig-acd:ACD-parity_transform}
	\end{figure}
	
	Another example can be found in Figure~\ref{fig-acd:ACD-transform-doesNotPreserveMinimality}.
\end{example}

\begin{remark}\label{rmk-acd:size_ACD-parity-transformation}
	The size of the "ACD-parity-transformation" of $\TS$ is:
	\[|\acdParityTransform{\TS}| = \sum_{v\in V} |  \leaves(\treeVertex{v}) | = \sum_{v\in \Vrec}|  \leaves(\treeVertex{v}) | + |\Vtrans|,\]
	where $\Vrec$ and $\Vtrans$ are the sets of "recurrent" and "transient" vertices of $\TS$, respectively.
\end{remark}

\begin{remark}\label{rmk-acd:acd-transform-parity-is-parity}
	We remark that if $\TS = (\underlyingGraph{\TS}, \macc{\TS})$ is already a "parity" "TS", then the underlying graphs of $\acdParityTransform{\TS}$ and $\TS$ are isomorphic. In fact, by Proposition~\ref{prop-ACD:correctness_ACD-parity_transform}, $\acdParityTransform{\TS}$ and $\TS$ will also be "isomorphic@@TS" as "transition systems".
	In this case, the construction of $\acdParityTransform{\TS}$ boils down to the application of the procedure described by Carton and Maceiras~\cite{CartonMaceiras99RabinIndex}.
\end{remark}

\begin{remark}\label{rmk-acd:oblivious-to-representation}
	The "ACD-parity-transform" is oblivious to the labelling $\gg$ of the "acceptance condition" of $\TS$; the only information taken into account to define the graph of $\acdParityTransform{\TS}$ and its "output colours" is the structure of the "trees" of $\acd{\TS}$.
	That is, the definition of this transformation is independent of the actual representation of the "acceptance condition" of $\TS$ (whether it is Emerson-Lei, "Muller", "Rabin"...), and we only use that any such representation induces a mapping $f\colon \cycles{\TS} \to \{\mathrm{Accept}, \mathrm{Reject} \}$.
\end{remark}

\begin{remark}
	The "ZT-parity-automaton" can be seen as a special case of the "ACD-parity-transform", as $\zielonkaAutomaton{\F}$ coincides with the "DPA" $\acdParityTransform{\A}$, where $\A$ is the "DMA" with a single state "recognising" $\Muller{\F}$ (see Remark~\ref{rmk-acd:ZT-as-ACD}).
\end{remark}

\paragraph*{Correctness of the ACD-parity-transform.}

\begin{proposition}[Correctness of the "ACD-parity-transform"]\label{prop-ACD:correctness_ACD-parity_transform}
	Let $\TS$ be a "(labelled)" "Muller" "TS" and let $\acdParityTransform{\TS}$ be its "ACD-parity-transform". There is a "locally bijective morphism" of "(labelled)" "transition systems" $\pp\colon \acdParityTransform{\TS} \to \TS$.
\end{proposition}

The following lemma, analogous to Lemma~\ref{lemma-ZT:transition_ZT-parity-automaton} from Section~\ref{subsec-zt: parity automaton}, follows from the definition of the "ACD-parity-transform".
\begin{lemma}\label{lemma-acd:transition_ACD-parity_transform}
	Let $n$ be a node of $\altTree{\ell_i}$, let $\tilde{n}$ be an "ancestor" of $n$ and let $e= v\re{}v'$ be an edge in $\ell_i$. Then, $\suppAcd(n, e)$ is a "descendant" of $\tilde{n}$ if and only if $e\in \nuAcd(\tilde{n})$, and in this case, if $e_n = (v,n)\re{} (v',n')$ is an edge of $\acdParityTransform{\TS}$, then $n'$ is a "descendant" of  $\tilde{n}$ too.
\end{lemma}

\begin{proof}[Proof of Proposition~\ref{prop-ACD:correctness_ACD-parity_transform}]
	We consider the mapping $\pp = (\pp_V, \pp_E)$ naturally defined by $\pp_V(v,n) = v$ and $\pp_E(e_n) = e$. 
	It is immediate to check that $\pp$ is a "weak morphism of transition systems" (it preserves "initial states" and transitions). Also, it is easy to see that it is "locally bijective": for each "initial state" $v_0\in I$, there is exactly one node in $I'$ of the form $(v_0, n)$: the node where $n$ is the "leftmost" "leaf" of $\treeVertex{v}$; and for each vertex $(v,n)$ and edge $e\in \mout(v)$ of $\TS$, we have defined exactly one edge outgoing from $(v,n)$ corresponding to $e$.
	
	We prove that $\pp$ "preserves the acceptance" of "runs", following the proof scheme from Proposition~\ref{prop-ZT:correctness_ZT_parity}. We can assume w.l.o.g. (see Remarks~\ref{rmk:colours=edges} and~\ref{rmk-acd:oblivious-to-representation}) that the set of "output colours" used by $\TS$ is its set of edges $E$.
	Let $\rr\in \Runs{\acdParityTransform{\TS}}$ be an infinite "run" in $\acdParityTransform{\TS}$. Eventually, $\rr$ will remain in one "SCC", and $\minf(\rr)$ will form a "cycle"  that is "accepting@@cycle" if and only if $\rr$ is an "accepting run".
	We will assume that all the edges in $\rr$ appear infinitely often and belong to this "cycle" (we can do it by using a similar argument as the one presented in the proof of Proposition~\ref{prop-ZT:correctness_ZT_parity}), and we let:
	\[ \rr = (v_0, n_0) \re{x_0} (v_1, n_1) \re{x_1} (v_2, n_2) \re{x_3}\dots .  \]
	The projection of $\rr$ under $\pp$ is:
	\[ \ppRuns(\rr) = v_0\re{e_0} v_1 \re{e_1} v_2 \re{e_3}\dots. \]
	We note that the edges $\{e_0, e_1, \dots \}$ form a "cycle" in $\TS$, that we will call $\ell_\rr$. In particular, $\ell_\rr$ is contained in some maximal "cycle" $\ell_{\max}$, and all the nodes $n_i$ belong to the same tree $\altTree{\ell_{\max}}$ of the "ACD".
	Our objective is to show that $\ell_\rr$ is an "accepting cycle"  in $\TS$ if and only if $\min \{x_0, x_1, x_2,\dots\}$ is even.
	We let $\tilde{n}_i = \suppAcd(n_i, e_i)$ be the node of $\acd{\TS}$ determining the i\ts{th} transition of $\rr$, so we have that $x_i = \parityNodesAcd(\tilde{n}_i)$. 
	Finally, let $n_\rr$ be the "deepest" "ancestor" of $n_0$ such that $\ell_\rr \subseteq \nuAcd(n_\rr)$.
	
	\begin{claim}\label{claim-ACD:supp-below-n-CORRECT-ACD}
		For all $i \geq 0$, $n_i \descendant n_\rr$ and $\tilde{n}_i \descendant n_\rr$ (that is, all nodes appearing in $\rr$ are "below" $n_\rr$). In particular, $x_i \geq \parityNodesAcd(n_\rr)$.
	\end{claim}
	\begin{subproof}
		The claim follows from Lemma~\ref{lemma-acd:transition_ACD-parity_transform} and induction.
	\end{subproof}
	\begin{claim}\label{claim-ACD:supp-n_inf_often-CORRECT-ACD}
		Let $n_{\rr,1},\dots, n_{\rr,s}$ be an enumeration of $\children_{\altTree{\ell_{\max}}}(n_\rr)$. It holds that:
		\begin{enumerate}
			\item $\suppAcd(n_i,e_i) = n_\rr$ infinitely often. In particular, $x_i = \parityNodesAcd(n_\rr)$ for infinitely many~$i$'s.
			\item There is no $n_{\rr,k}\in \children(n_\rr)$ such that $\ell_\rr \subseteq \nuAcd(n_{\rr,k})$.
		\end{enumerate}	
	\end{claim}
	\begin{subproof}
		The proof is identical to that of Claim~\ref{claim:supp-n_inf_often-CORRECT-ZT}, from Proposition~\ref{prop-ZT:correctness_ZT_parity}.
	\end{subproof}
	
	 We conclude that $\min\{x_0, x_1, x_2,\dots\} = \parityNodesAcd(n_\rr)$, which is even if and only if $\ell_\rr$ is an "accepting cycle", by Remarks~\ref{rmk-acd:union-changes-acceptance} and~\ref{rmk-acd:parity_levels_ACD}.	 
\end{proof}

\begin{remark}
	We can give an alternative interpretation of the previous proof. Given a "run" $\rr$ in $\TS$ and a vertex $v$ appearing infinitely often in $\rr$, we can decompose the "run" into:
	\[ \lrp{\rr_0} \; v \; \lrp{\rr_1} \; v \;\lrp{\rr_2} \; v \; \lrp{\rr_3} \; v \; \lrp{\rr_4} \dots , \]
	where the finite runs $\rr_i$ are "cycles" "over $v$@@cycle", for $i>0$. Therefore, the sequence of these cycles can be processed by the "ZT-parity-automaton" corresponding to the "local Muller condition" $\localMuller{v}{\TS}$. By Lemma~\ref{lemma-acd:tree_q_ZT_local_Muller} and the correctness of the "ZT-parity-automaton", the minimal colour produced by a "run over" this sequence of cycles in $\zielonkaAutomaton{\localMuller{v}{\TS}}$ coincides with the minimal "output colour" produced by the run $\inv{\ppRuns}(\rr)$ in the "ACD-parity-transform" $\acdParityTransform{\TS}$ (disregarding the initial path $\rr_0$). This colour is exactly the one corresponding to the "deepest@@tree" node in $\treeVertex{v}$ above the "leftmost" "leaf" containing $\minf(\rr)$.
\end{remark}

The "locally bijective morphism" given by Proposition~\ref{prop-ACD:correctness_ACD-parity_transform} witnesses that $\acdParityTransform{\TS}$ shares the same semantic properties as $\TS$. The next corollaries follow from Proposition~\ref{prop-morph:HD mappings-preserve-languages} and Corollary~\ref{cor-morph:HD-map-preserve-full-WR} (and the fact that the choice of initial vertices in $\acdParityTransform{\TS}$ is arbitrary).

\begin{corollary}
	Let $\A$ be a "Muller" "automaton" and let $\acdParityTransform{\A}$ be its "ACD-parity-transform". Then, $\Lang{\A} = \Lang{\acdParityTransform{\A}}$, and $\A$ is "deterministic" (resp. "history-deterministic") if and only if $\acdParityTransform{\A}$ is "deterministic" (resp. "history-deterministic").
\end{corollary}

\begin{corollary}
	Let $\G$ be a "Muller" "game" and let $\acdParityTransform{\G}$ be its "ACD-parity-transform". "Eve" "wins" $\acdParityTransform{\G}$ from a vertex	of the form $(v,n)$ if and only if she "wins" $\G$ from $v$.
\end{corollary}

\subsection{An optimal history-deterministic transformation to Rabin transition systems}\label{subsec-acd: HD-Rabin-transformation}

In this section we describe the "ACD-HD-Rabin-transform", an optimal transformation of "Muller" "TS" to "Rabin" "TS" preserving "history-determinism". This construction generalises that from Section~\ref{subsec-zt: GFG-Rabin}.

\begin{definition}[ACD-HD-Rabin-transform]\label{def-ACD:RabinHDTransformationACD}
	Let $\TS$ be a "Muller" "TS". 
	For each vertex $v\in V$ we let $\eta_v \colon \leaves(\treeVertex{v}) \to \{1,\dots, \memTree{\treeVertex{v}}\}$ be a mapping satisfying Property~\eqref{eq:property-star} from Lemma~\ref{lemma:property_star}.
	
	\AP We define the ""ACD-HD-Rabin-transform"" of $\TS$ to be the "Rabin" "TS"	 $\intro*\acdRabinTransform{\TS} = (\underlyingGraph{}', \macc{}')$, with $\underlyingGraph{}' = (V', E', \msource', \mtarget', I')$, and $\macc{}' = (\gg',\ab \nodesAcd{\TS}, \ab \Rabin{R})$ defined as follows.
	\begin{description}
		  \setlength\itemsep{2mm}
		\item[Vertices.] The set of vertices is \[V' = \bigcup_{v\in V} \left( \{v\} \times \{1,\dots, \memTree{\treeVertex{v}}\} \right),\]
		where $\memTree{\treeVertex{v}}$ is the "round-branching width" of $\treeVertex{v}$.
		
		\item[Initial vertices.] $I' = \{ (v_0, x) \mid v_0\in I \tand x \in \{1,\dots, \memTree{\treeVertex{v_0}}\} \}$.
		
		\item[Edges and output colours.] We let \[E' = \bigcup_{e\in E} \left( \{e\} \times \leaves(\treeVertex{\msource(e)})\right).\]
		
		For each edge $e = v\re{} v'\in E$ in $\TS$ and $x\in \{1,\dots, \memTree{\treeVertex{v}}\}$, we will place one edge from $(v, x)$ for each leaf $l$ of $\treeVertex{v}$ such that $\eta_v(l) = x$.
		More precisely, we let $(v, x) \re{n} (v',x')\in E'$ if either 
		\begin{itemize}
			\item $v$ and $v'$ are not in the same "SCC" (in this case the output colour $n$ is irrelevant), or
			\item $v$ and $v'$ are in the same "SCC" and there are "leaves" $l$ and $l'$ of $\treeVertex{v}$ and $\treeVertex{v'}$, respectively, such that:
			\begin{itemize}
				\item $\eta_v(l) = x$, $\eta_{v'}(l') = x'$,
				\item $l' = \jump_{\treeVertex{v'}}(l, \suppAcd(l, e))$,
				\item $n = \suppAcd(l,e)$.
			\end{itemize}
		\end{itemize}
		
		\item[Rabin condition.] $R = \{(G_n, R_n)\}_{n\in \nodesAcdRound{\TS}}$, where $G_n$ and $R_n$ are defined as follows: Let $n$ be a "round@@acd" node, and let $n'$ be any node in $\nodesAcd{\TS}$,
		\begin{equation*}
			\begin{cases}
				n' \in G_n & \text{ if } n'=n,\\
				n' \in R_n & \text{ if }  n'\neq n \text{ and } n \text{ is not an "ancestor" of } n'. 
			\end{cases}
		\end{equation*} 	
		
		\item[Labellings.] If $\TS$ is a "labelled transition system", with labels $l_V\colon V \to L_V$ and $l_E\colon E \to L_E$, we label $\acdRabinTransform{\TS}$ by $l_{V'}'(v,x) = l_V(v)$ and $l'_{E'}(e') = l_E(e)$, if $e'\in E'(e)$.
	\end{description}
\end{definition}

This construction generalises the "ZT-HD-Rabin-automaton" in the same way as the "ACD-parity-transform" generalises the "ZT-parity-automaton". Intuitively, a "run" in $\acdRabinTransform{\TS}$ can be identified with a promenade through the nodes of the "ACD", which are used as the "output colours" to define the "Rabin" "acceptance condition".

\begin{remark}\label{rmk-acd:size_acd-RabinHD-transformation}
	The size of the "ACD-HD-Rabin-transform" of $\TS$ is:
	\[|\acdRabinTransform{\TS}| = \sum_{v\in V} \memTree{\treeVertex{v}} = \sum_{v\in \Vrec}\memTree{\treeVertex{v}} + |\Vtrans|,\]
	where $\Vrec$ and $\Vtrans$ are the sets of "recurrent" and "transient" vertices of $\TS$, respectively.
\end{remark}

\paragraph*{Correctness of the ACD-HD-Rabin-transform.} To obtain the correctness of the "ACD-HD-Rabin-transform", we follow the same steps as in the proof of the correctness of the "ZT-HD-Rabin-automaton" (Proposition~\ref{prop-ZT:correctness_ZT_HD-Rabin}).

\begin{proposition}[Correctness of the "ACD-HD-Rabin-transform"]\label{prop-ACD:correctness_ACD-HD-Rabin-transform}
	Let $\TS$ be a "(labelled)" "Muller" "TS" and let $\acdRabinTransform{\TS}$ be its "ACD-HD-Rabin-transform". There is an "HD mapping" of "(labelled)" "transition systems" $\pp\colon \acdRabinTransform{\TS} \to \TS$.
\end{proposition}

The proof of the next two lemmas are completely analogous to those of Lemmas~\ref{lemma-ZT:Rabin-acc-sequences-nodes-ZT} and~\ref{lemma-ZT:morphism-parity-Rabin}.
\begin{lemma}\label{lemma-acd:Rabin-acc-sequences-nodes-ACD}
	Let $u = n_0n_1n_2\dots \in \nodesAcd{\TS}^\oo$ be an infinite sequence of nodes of the "ACD" of $\TS$. The word $u$ belongs to $\Rabin{R}$, for $R = \{(G_n,R_n)\}_{n\in\nodesAcdRound{\TS}}$ the Rabin condition of $\acdRabinTransform{\TS}$, if and only if there is a unique minimal node for the "ancestor relation" in $\minf(u)$ and this minimal node is "round".
\end{lemma}

\begin{lemma}\label{lemma-acd:morphism-Parity-Rabin}
	There exists a "morphism of transition systems" $\pp\colon \acdParityTransform{\TS} \to \acdRabinTransform{\TS}$.
\end{lemma}

Using these lemmas we can prove Proposition~\ref{prop-ACD:correctness_ACD-HD-Rabin-transform}.

\begin{proof}[Proof of Proposition~\ref{prop-ACD:correctness_ACD-HD-Rabin-transform}]
	We define the mapping $\pp\colon \acdRabinTransform{\TS} \to \TS$ in the natural way: $\pp_V(v,x) = v$ and $\pp_E(e,l) = e$. It is immediate to check that $\pp$ is a "weak morphism". The fact that $\pp$ preserves "accepting runs" can be proven analogously to the fact that $\Lang{\zielonkaHDAutomaton{\F}} \subseteq \MullerC{\F}{\SS}$ in Proposition~\ref{prop-ZT:correctness_ZT_HD-Rabin} (by using Lemma~\ref{lemma-acd:Rabin-acc-sequences-nodes-ACD}).
	
	Definition of a "sound resolver@@morph" for $\pp$: In order to show how to simulate runs of $\TS$ in $\acdRabinTransform{\TS}$, we use the fact that we can see $\acdRabinTransform{\TS}$ as a quotient of $\acdParityTransform{\TS}$ (Lemma~\ref{lemma-acd:morphism-Parity-Rabin}). Let $\tilde{\pp}\colon \acdParityTransform{\TS} \to \TS$ be the "locally bijective morphism" given by Proposition~\ref{prop-ACD:correctness_ACD-parity_transform}, and let $\hat{\pp}\colon \acdParityTransform{\TS} \to \acdRabinTransform{\TS}$ be the "morphism" given by Lemma~\ref{lemma-acd:morphism-Parity-Rabin}.
	Since $\tilde{\pp}$ is "locally bijective", $\ppRunsP{\tilde{\pp}}$ is a bijection between the runs of the transitions systems $\acdParityTransform{\TS}$ and $\TS$, admitting an inverse $\inv{\ppRunsP{\tilde{\pp}}}$. Composing this mapping with $\hat{\pp}$, we obtain a way to simulate the runs from $\TS$ in $\acdRabinTransform{\TS}$:
	\[ \ppRunsP{\hat{\pp}} \circ \inv{\ppRunsP{\tilde{\pp}}} \colon \RunsInfty{\TS} \to \RunsInfty{\acdRabinTransform{\TS}}. \]
	
	This composition of mappings provides a "sound@@morph" "resolver@@morph" simulating $\pp$. Formally, let $(\rInit, r)$ be the "resolver@@morph" defined as follows. The choice of initial vertices $\rInit\colon I \to I'$ is given by $\rInit(v_0,x) = v_0$. The function $r\colon E'^* \times E \to E'$ associates to a finite run $\rr\in E'^*$ and $e\in E$ the last edge of the run $\hat{\pp}(\inv{\tilde{\pp}}(\pp(\rr)e))$ (subscripts have been omitted for legibility). It is easy to check that $(\rInit, r)$ indeed defines a "resolver simulating $\pp$". Its "soundness@@morph" follows from the fact that $\tilde{\pp}$ and $\hat{\pp}$ "preserve the acceptance of runs".
\end{proof}

From Proposition~\ref{prop-morph:HD mappings-preserve-languages} we obtain:

\begin{corollary}
	Let $\A$ be a "Muller" "automaton" and let $\acdRabinTransform{\A}$ be its "ACD-HD-Rabin-transform". Then, $\Lang{\acdRabinTransform{\A}} = \Lang{\A}$. Moreover, $\acdRabinTransform{\A}$  is "history-deterministic" if and only if $\A$ is "history-deterministic".
\end{corollary}

\paragraph*{ACD-HD-Rabin-transform-for-games.}

In Section~\ref{subsect-prelim:transition-systems}, we discussed some technical difficulties appearing when we wanted to define the "composition" of a "game" $\G$ and an "HD automaton": as the output of such operation, we would like to obtain a "game" in which "Eve" always chooses the transitions taken in the automaton, even if it is Adam who makes a move in the "game", which is not the case if $\G$ is an arbitrary game. Also, in Section~\ref{subsect-morp: hist-det-morphisms} we had to introduce "HD-for-games mappings" in order to formalise correct transformations of "games".
A similar difficulty appears in the context of the "ACD-HD-Rabin-transform"; we can see the "ACD-HD-Rabin-transform" of a "game" $\G$ as a "game" in which, at each moment, first, a move takes place in $\G$, and then a choice is made to update the current node in $\acd{\G}$. With the current definition of $\acdRabinTransform{\G}$, it is the player who makes the move in the game component who chooses how to update the node in $\acd{\G}$. This is potentially a problem, as in order to obtain an equivalent game we would like that "Eve" had full control to decide how to update the nodes in $\acd{\G}$, even when it was Adam who moved in the game component (we note that in Proposition~\ref{prop-ACD:correctness_ACD-HD-Rabin-transform} we did not claim that there is an "HD-for-games mapping" $\pp\colon \acdRabinTransform{\TS} \to \TS$).
In order to obtain a transformation working for games, we need to slightly modify the definition of the "ACD-HD-Rabin-transform". 

\AP  For a "Muller" "game" $\G$ "suitable for transformations", we define its "ACD-HD-Rabin-transform-for-games", written $\acdRabinTransformGFG{\G}$. The idea is simply to take from "Adam" the power to update the $\acd{\G}$-component of vertices. 
The update of this information is delayed of one transition, so it is "Eve" who makes the choice of how to move in the "ACD".
To do this, we need to introduce some additional vertices controlled by "Eve".
The formal details of this construction and the proof of correctness can be found in Appendix~\ref{sec:appendix-games-transformations}.

\begin{restatable}[Correctness of the ACD-HD-Rabin-transform-for-games]{proposition}{correctnessAcdRabinTransformGames}\label{prop-ACD:correctness_ACD-HD-Rabin-transform-games}
	Let $\G$ be a "Muller" "game" "suitable for transformations", and let $\acdRabinTransformGFG{\G}$ be its "ACD-HD-Rabin-transform-for-games". Then, there is an "HD-for-games mapping" $\pp\colon \acdRabinTransformGFG{\G} \to \TS$.
\end{restatable}
\begin{corollary}
	Let $\G$ be a "Muller" "game" "suitable for transformations", and let $\acdRabinTransformGFG{\G}$ be its "ACD-HD-Rabin-transform-for-games". Then, "Eve's" "full winning region" in $\G$ is the projection of her "full winning region" in $\acdRabinTransformGFG{\G}$.
\end{corollary}

\subsection{Optimality of the ACD-transforms}\label{subsec-acd-parity: optimality}
We now state and prove the optimality of both the "ACD-parity-transform" (Theorems~\ref{thm-acd:optimality-priorities_ACD-parity_transform} and~\ref{thm-acd:optimality-size_ACD-parity_transform}) and the "ACD-HD-Rabin-transform" (Theorem~\ref{thm-acd:optimality_ACD-HD-Rabin-transform}).
The proofs of these results will use the optimality of the "automata" based on the "Zielonka tree" (c.f. Section~\ref{section:zielonka-tree}) as a black-box, which will allow us to prove the optimality of both transformations at the same time.
The key idea is that if $\pp\colon \TS \to \TS'$ is an "HD mapping", we can see $\TS$ as an "HD automaton" "recognising" the "accepting runs" of $\TS'$. We can then use "local Muller conditions" at vertices of $\TS'$ to reduce the problem to "automata" "recognising" "Muller languages".

\subsubsection{Statement of the optimality results}
We state the optimality of the transformations based on the "ACD". All the results below apply to "labelled transition systems" too.  For technical reasons, we need to suppose that all the states of "transition systems" under consideration are "accessible", an hypothesis that can always be made without loss of generality. We recall that "HD mappings" are in particular "locally bijective morphisms" and "HD-for-games mappings" (c.f. Figure~\ref{fig-morph:types-morphisms}).

\begin{theorem}\label{thm-acd:optimality-priorities_ACD-parity_transform}
	Let $\TS$ be a  "Muller" "TS" whose states are "accessible" and let $\widetilde{\TS}$ be a  "parity" "TS". If $\widetilde{\TS}$ admits an "HD mapping" $\pp\colon \widetilde{\TS} \to \TS$, then, its "acceptance condition" uses at least as many colours as that of $\acdParityTransform{\TS}$.
\end{theorem}

\begin{theorem}\label{thm-acd:optimality-size_ACD-parity_transform}
	Let $\TS$ be a  "Muller" "TS" whose states are "accessible" and let $\widetilde{\TS}$ be a  "parity" "TS". If $\widetilde{\TS}$ admits an "HD mapping" $\pp\colon \widetilde{\TS} \to \TS$, then, 
	$|\acdParityTransform{\TS}| \leq \size{\widetilde{\TS}}$.
\end{theorem}

\begin{theorem}\label{thm-acd:optimality_ACD-HD-Rabin-transform}
	Let $\TS$ be a  "Muller" "TS" whose states are "accessible" and let $\widetilde{\TS}$ be a  "Rabin" "TS". If $\widetilde{\TS}$ admits an "HD mapping" $\pp\colon \widetilde{\TS} \to \TS$, then,
	$|\acdRabinTransform{\TS}| \leq \size{\widetilde{\TS}}$.
\end{theorem}

We obtain an analogous optimality result for the "ACD-HD-Rabin-transform-for-games". In this case, the bound is not tight due to the additional vertices that are added to $\acdRabinTransformGFG{\G}$ (see Appendix~\ref{sec:appendix-games-transformations} for details).
\begin{restatable}{corollary}{optimalityAcdRabinTransformGames}
	Let $\G$ be a "Muller" "game" "suitable for transformations" whose states are "accessible" and let $\widetilde{\G}$ be a  "Rabin" "game". If $\widetilde{\G}$ admits an "HD-for-games mapping" $\pp\colon \widetilde{\G} \to \G$, then,
	$|\acdRabinTransformGFG{\G}| \leq 2\size{\widetilde{\G}}$.
\end{restatable}


\subsubsection{Discussion: Limits on the applicability of HD automata and preservation of minimality}

Before presenting the proofs of the optimality theorems, we discuss some consequences and limitations of our results.

\subparagraph*{Difficulty of finding succinct history-deterministic automata.}
As mentioned in the introduction, several years had to pass after the introduction of "history-deterministic automata"~\cite{HP06}  before finding "HD automata" that were actually smaller than equivalent "deterministic" ones~\cite{KS15DeterminisationGFG}.
As of today, we only know a handful of examples of "$\oo$-regular languages" admitting succinct "HD automata"~\cite{AK22MinimizingGFG, KS15DeterminisationGFG,CCL22SizeGFG}, and their applicability in practice has yet to be fully determined. We assert that we can derive from our results some enlightening explanations on the difficulty of finding succinct "HD" "parity" "automata", and set some limits in their usefulness in practical scenarios such as LTL synthesis.

First, Corollary~\ref{cor-zt:no-small-HD-for-Muller} already sets the impossibility of the existence of small "HD" "parity" "automata" "recognising" "Muller languages". Corollary~\ref{cor-acd:HD-transformations-are-big} states that if an "HD" "parity" "automaton"~$\A$ has been obtained as a transformation of a "DMA" $\B$, then $\A$ is not strictly smaller than a minimal "deterministic" "parity" "automaton" for~$\Lang{\A}$.

\begin{corollary}
	Let $\TS$ be a  "Muller" "TS". A minimal "parity" "TS" admitting an "HD mapping" to $\TS$ has the same size than a minimal "parity" "TS" admitting a "locally bijective morphism" to $\TS$.
\end{corollary}

\begin{corollary}\label{cor-acd:HD-transformations-are-big}
	Let $\A$ be a "history-deterministic" "parity" "automaton". Assume that there exists a "DMA" $\B$ such that $\A$ admits an "HD mapping" to $\B$. Then, there exists a "DPA" $\A'$ "recognising" $\Lang{\A}$ such that $\size{\A'}\leq \size{\A}$.
\end{corollary}

Both corollaries follow from an immediate application of Theorem~\ref{thm-acd:optimality-size_ACD-parity_transform}.

\subparagraph*{The ACD-transform does not preserve minimality.}

A natural question is whether the "ACD-parity-transform" preserves minimality of automata, that is, given a "DMA" $\A$ with a minimal number of states for the language it "recognises", is $\acdParityTransform{\A}$ minimal amongst "DPAs" "recognising" $\Lang{\A}$?\footnotemark{} The answer to this question is negative, as we show now. 
\footnotetext{This question was left open as a conjecture in the conference version of this paper~\cite{CCF21Optimal}.}

\begin{proposition}
	There exists a "DMA" $\A$ that is minimal amongst "DMAs" "recognising" $\Lang{\A}$, but such that its "ACD-parity-transform" $\acdParityTransform{\A}$ is not a minimal "DPA".
\end{proposition}

We consider the alphabet $\SS=\{a,b,c\}$ and the language 
\[ L = \{ w\in \SS^\oo \mid c\in \minf(w) \tand w \text{ contains infinitely often the factor } ab\}. \]

A minimal "DMA" for $L$ is depicted in Figure~\ref{fig-acd:muller-c-ab}. Its minimality follows simply from the fact that, as $L$ is not a "Muller language" ($(abc)^\oo\in L$ but $(bac)^\oo\notin L$, c.f. Remark~\ref{rmk-prelim:characterisation_Muller_Languages}), a "DMA" with just one state cannot "recognise" $L$. In Figure~\ref{fig-acd:ACD-transform-doesNotPreserveMinimality} we show its "alternating cycle decomposition" and its "ACD-parity-transform" that has $4$ states. However, we can find a "DPA" with just $3$ states "recognising" $L$, as shown in Figure~\ref{fig-acd:parity-c-ab}.

\begin{figure}[ht]
	\begin{subfigure}[t]{0.475\textwidth}
		\centering
		\includegraphics[scale=1.2]{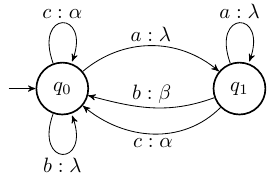}
		\caption{A "Muller" "automaton" with "acceptance set" given by $\F =\{\{\aa,\bb, \lambda\}\}$.}
		\label{fig-acd:muller-c-ab}
	\end{subfigure}\hfill
	\begin{subfigure}[t]{0.475\textwidth}
		\centering
		\includegraphics[scale=1.2]{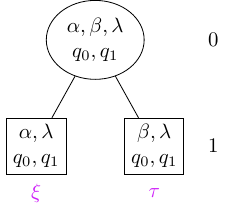}
		\caption{"Alternating cycle decomposition" of $\A$. To indicate the labels of the nodes of this "ACD", we include just the colours of the corresponding edges.}
		\label{fig-acd:acd-muller-c-ab}
	\end{subfigure}
\begin{subfigure}[t]{0.475\textwidth}
	\centering
	\includegraphics[scale=1.2]{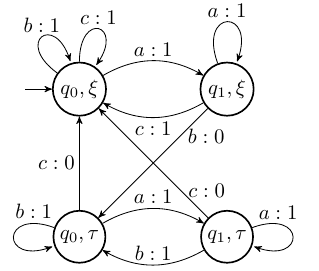}
	\caption{"ACD-parity-transform" of $\A$, $\acdParityTransform{\A}$, with $4$ states.}
	\label{fig-acd:acd-transform-c-ab}
\end{subfigure}\hfill
\begin{subfigure}[t]{0.475\textwidth}
	\centering
	\includegraphics[scale=1.2]{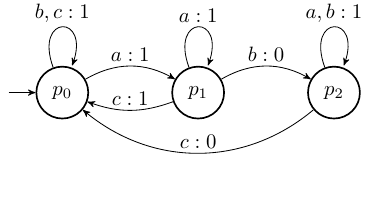}
	\caption{A "parity" "automaton" "recognising" $L$ with only $3$ states .}
	\label{fig-acd:parity-c-ab}
\end{subfigure}
	\caption{A minimal "DMA" whose "ACD-parity-transform" is not a minimal "DPA".} 
	\label{fig-acd:ACD-transform-doesNotPreserveMinimality}
\end{figure}

\subsubsection{Optimality of the parity condition of \texorpdfstring{$\acdParityTransform{\TS}$}{ACD-Parity(TS)}}\label{subsubsec-acd:optimality-parity-condition}
We show next the proof of Theorem~\ref{thm-acd:optimality-priorities_ACD-parity_transform}. To prove this result, we would like to use the Flower Lemma~\ref{lemma:flower-lemma}, 
however, the statement of Theorem~\ref{thm-acd:optimality-priorities_ACD-parity_transform} does not involve "$\oo$-regular languages". In order to set up a context in which apply the Flower Lemma, we show that, whenever we have a "morphism" $\pp\colon \TS \to \TS'$, $\TS$ can be seen as an "automaton" reading the "runs" of $\TS'$.

\AP Let $\TS = (\underlyingGraph{}, \macc{})$ and $\TS' = (\underlyingGraph{}', \macc{}')$ be "transition systems" with "underlying graphs" $\underlyingGraph{} =(V ,E , \msource , \mtarget, I)$ and $\underlyingGraph{}' =(V' ,E' , \msource' , \mtarget', I')$, and "acceptance conditions" $\macc{} = (\gg, \GG, \WW)$ and  $\macc{}' = (\gg', \GG', \WW')$. 
A "weak morphism of transition systems" $\pp\colon \TS \to \TS'$ provides a labelling of the edges of $\TS$ by $\pp_E\colon E \to E'$. Therefore, we can see $\TS$ as an "automaton" with "input alphabet"~$E'$, inheriting the "underlying graph" and "acceptance condition" from $\TS$. 
We say that this is the ""automaton of morphism $\pp$"" and denote it by~$\intro*\autMorphism{\pp}$.

\AP We define the ""language of accepting runs of"" a "transition system" $\TS$ as:
\[ \intro*\LangTS{\TS} = \{\rr\in E^\oo \mid \rr \text{ is an "accepting run" in } \TS \}. \]

\begin{lemma}\label{lemma-acd:morphism_HD-iff-automaton-HD}
	Let $\TS$ and $\TS'$ be "transition systems" with a single "initial state", let $\pp\colon \TS \to \TS'$ be a "weak morphism of transition systems", and let $\autMorphism{\pp}$ be its "automaton@@morphism".
	Then, $\pp$ is an "HD mapping" if and only if the automaton $\autMorphism{\pp}$ is "history-deterministic", and, in this case, \[\Lang{\autMorphism{\pp}} = \LangTS{\TS'}.\]
\end{lemma}
\begin{proof}
	We first note that a "resolver@@aut" for $\autMorphism{\pp}$ (in the sense of "HD automata") is a mapping of the form $r\colon E^* \times E' \to E$, as $E'$ is the "input alphabet" of this "automaton". A "resolver simulating $\pp$" (in the sense of "HD mappings") is a mapping of the same form.
	It is straightforward to check that $(q_0, r)$ is a "sound@@aut" "resolver@@aut" for $\autMorphism{\pp}$ if and only if $(\rInit, r)$ is a "sound@@morph" "resolver simulating $\pp$" (where $\rInit(q_0') = q_0$ is the only possible choice of initial vertex).

	We prove that $\Lang{\autMorphism{\pp}} = \{\rr'\in \Runs{\TS'} \mid \rr' \text{ is an "accepting run"}\}$. 
	First, we remark that if $\rr$ is a "run@@automaton" in $\autMorphism{\pp}$ over $\rr'\in\Runs{\TS'}$, then $\rr' = \ppRuns(\rr)$, since the "labelling" of $\autMorphism{\pp}$ by "input letters" is given exactly by $\pp$ itself. Therefore, if $\rr'\in \Lang{\autMorphism{\pp}}$, there exists an "accepting@@run" "run $\rr$ over"~$\rr'$, and since $\pp$ "preserves accepting runs",  $\rr' = \ppRuns(\rr)$ is "accepting@@run" in $\TS'$, proving the inclusion from left to right. For the other inclusion, we let $(\rInit, r)$ be a "sound@@morph" "resolver simulating" $\pp$. If $\rr'$ is an "accepting run" in $\TS'$, then $\rRuns(\rr')$ is an "accepting@@run" "run over" $\rr'$ in~$\autMorphism{\pp}$.
\end{proof}

We recall that $[\minparityAcd{\TS}, \maxparityAcd{\TS}]$ are the colours used by the "ACD-parity-transform" of $\TS$, which coincides with the maximal "height" of a "tree" in $\acd{\TS}$.
We also recall that $\minparityAcd{\TS}=0$ if $\acd{\TS}$ is "positive@@acd" or "equidistant@@acd", and that $\minparityAcd{\TS}=1$ if $\acd{\TS}$ is "negative@@acd".

\begin{lemma}\label{lemma-acd:branch-ACD-induces-flower}
	Let $\TS$ be a "Muller" "TS", and let $\altTree{\ell}\in \acd{\TS}$ be a "positive@@tree" (resp. "negative@@tree") tree of the "ACD" of $\TS$ of "height" $d$. Then, $\TS$ admits a "positive@@flower" (resp. "negative@@flower") "$d$-flower".
\end{lemma}
\begin{proof}
	We use the same argument as the one used in the proof of Theorem~\ref{thm-zt:optimality_ZTparity-priorities}. Let $n_1\ancestor n_2\ancestor \dots \ancestor n_d$ be a "branch" of length $d$ of $\altTree{\ell}$ (where $n_1$ is the "root" and $n_d$ is a "leaf" of the tree). Let $v\in \nuStates(n_d)$ be a vertex appearing in the "leaf". Then, the whole "branch" is contained in $\treeVertex{v}$ (by Remark~\ref{rmk-acd:treeVertex-closed-by-ancestor}), that is, $\nuAcd(n_i)\in \cyclesState{v}{\TS}$. Moreover, $\nuAcd(n_1)\supsetneq \nuAcd(n_2)\supsetneq \dots \nuAcd(n_d)$ is a chain that alternates "accepting@@cycle" and "rejecting@@cycle" "cycles", so it is a "$d$-flower" that is "positive@@flower" if and only if $\nuAcd(n_1) = \ell$ is an "accepting cycle", that is, if $\altTree{\ell}$ is "positive@@tree".
\end{proof}

\begin{lemma}\label{lemma-acd:parity index-muller-TS}
	Let $\TS$ be a "Muller" "TS" with a single initial vertex and whose vertices are all "accessible". Then, the "parity index" of $\LangTS{\TS}$ is:
	\begin{itemize}
		\item $[\minparityAcd{\TS}, \maxparityAcd{\TS}]$ if $\acd{\TS}$ is "positive@@acd" or "negative@@acd",
		\item $\WeakIndex{\maxparityAcd{\TS}}$ if $\acd{\TS}$ is "equidistant@@acd".
	\end{itemize}
\end{lemma}
\begin{proof}
	We consider the identity "morphism" $\Id{\TS}\colon \TS \to \TS$ and its automaton $\autMorphism{\Id{\TS}}$, which is a "deterministic" automaton trivially "recognising@@aut" $\LangTS{\TS}$ (that is, we see $\TS$ as an automaton reading its own edges as "input letters").
	The result follows from the Flower Lemma~\ref{lemma:flower-lemma} and the fact that a "tree" $\altTree{\ell}\in \acd{\TS}$ of "height" $d$ provides a "$d$-flower" that is "positive@@flower" if $\ell$ is "accepting@@cycle" and "negative@@flower" if  $\ell$ is "rejecting@@cycle" (Lemma~\ref{lemma-acd:branch-ACD-induces-flower}). These "flowers" are "accessible" as we have supposed that all the vertices of $\TS$ are "accessible".
\end{proof}

The previous lemmas allow us to obtain Theorem~\ref{thm-acd:optimality-priorities_ACD-parity_transform} for "transition systems" with a single "initial vertex". We introduce some further notations to deal with the general case.

\AP For a "Muller" "TS" $\TS$ and a vertex $v$, we let $\intro*\acdVertex{\TS}{v}$ be the "alternating cycle decomposition" of the "accessible part of $\TS$ from $v$". We note that the "trees" of $\acdVertex{\TS}{v}$ are a subset of the trees of $\acd{\TS}$: a "tree" $\altTree{\ell_i}\in \acd{\TS}$ appears in $\acdVertex{\TS}{v}$ if and only if the "cycle" $\ell_i$ is "accessible from" $v$.
\AP Accordingly, for each vertex $v$ of $\TS$ we let $\intro*\minparityAcdVertex{\TS}{v}$ (resp. $\intro*\maxparityAcdVertex{\TS}{v}$) be the minimum (resp. maximum) value taken by the function~$\parityNodesAcd$ when restricted to the "trees" of $\acdVertex{\TS}{v}$.

\begin{remark}\label{rmk-acd:priorities-accessible-from-vertex}
	For every "transition system" $\TS$, one of the two following statements holds:
	\begin{itemize}
		\item There is some vertex $v$ such that $[\minparityAcd{\TS}, \maxparityAcd{\TS}] = [\minparityAcdVertex{\TS}{v}, \maxparityAcdVertex{\TS}{v}]$.
		\item There are two vertices $v_0$ and $v_1$ such that $\minparityAcdVertex{\TS}{v_0}=0$, $\maxparityAcdVertex{\TS}{v_0} = \maxparityAcd{\TS}-1$ and $\minparityAcdVertex{\TS}{v_1}=1$, $\maxparityAcdVertex{\TS}{v_1} = \maxparityAcd{\TS}$.
	\end{itemize}
Moreover, if all the states of $\TS$ are "accessible", we can choose $v$ (resp. $v_0$ and $v_1$) to be an initial vertex.
\end{remark}

We can finally deduce Theorem~\ref{thm-acd:optimality-priorities_ACD-parity_transform} from the preceding lemmas.
\begin{proof}[Proof of Theorem~\ref{thm-acd:optimality-priorities_ACD-parity_transform}]
	We assume that we are in the first case of Remark~\ref{rmk-acd:priorities-accessible-from-vertex} (a proof for the second case follows easily). First, we show that we can suppose that $\widetilde{\TS}$ and $\TS$ have a single initial vertex. Let $v$ be an "initial vertex" of $\TS$ such that $[\minparityAcd{\TS}, \maxparityAcd{\TS}] = [\minparityAcdVertex{\TS}{v}, \maxparityAcdVertex{\TS}{v}]$. 
	Let $\pp\colon \widetilde{\TS} \to \TS$ be an "HD mapping", and let $(\rInit, r)$ be a "sound@@morph" "resolver@@morph" simulating it. We let $\tilde{v} = \rInit(v)$ be the initial vertex in $\widetilde{\TS}$ chosen by the resolver. It suffices then to prove the result for the "accessible part of $\widetilde{\TS}$ from $\tilde{v}$", the "transition system" $\initialTS{\TS}{v}$,  and the restriction of $\pp$ to these transition systems. 
	
	From now on, we assume that both $\widetilde{\TS}$ and $\TS$ have a single initial vertex.
	By Lemma~\ref{lemma-acd:parity index-muller-TS} and Proposition~\ref{prop-prelim:parity_index-HDAutomata}, a "parity" "history-deterministic automaton" recognising $\LangTS{\TS}$ uses at least $|[\minparityAcd{\TS}, \maxparityAcd{\TS}]|$ colours. By Lemma~\ref{lemma-acd:morphism_HD-iff-automaton-HD}, the "automaton@@autMorphism" $\autMorphism{\pp}$ "of the morphism@@aut" $\pp$ is a "parity" "history-deterministic automaton" recognising $\LangTS{\TS}$, and therefore uses at least $|[\minparityAcd{\TS}, \maxparityAcd{\TS}]|$ colours. Since the "acceptance condition" of $\widetilde{\TS}$ is exactly the same as that of $\autMorphism{\pp}$, we can conclude.
\end{proof}

\subsubsection{Optimality of the sizes \texorpdfstring{of $\acdParityTransform{\TS}$ and $\acdRabinTransform{\TS}$}{}}
We prove now Theorems~\ref{thm-acd:optimality-size_ACD-parity_transform} and~\ref{thm-acd:optimality_ACD-HD-Rabin-transform}. 

\begin{proof}[Sketch of the proof]
	Let $\pp\colon \widetilde{\TS}\to \TS$ be an "HD mapping", and let $v$ be a vertex in $\TS$. We can see the set $\inv{\pp}(v)$ as the states of an "HD automaton" reading finite "runs" in $\TS$ looping around $v$. This allows to define an "HD automaton" having $\inv{\pp}(v)$ as set of states and "recognising" $\localMuller{v}{\TS}$. As the "Zielonka tree" of $\localMuller{v}{\TS}$ is the tree $\treeVertex{v}$, by optimality of the "ZT-parity-automaton" (resp. the "ZT-HD-Rabin-automaton"), we deduce that $|\inv{\pp}(v)|\geq |\leaves(\treeVertex{v})|$ (resp. $|\inv{\pp}(v)|\geq \memTree{\treeVertex{v}}$).
\end{proof}

\begin{definition}\label{def:cycle-preimage-automaton}
	Let $\TS$ and $\TS'$ be two "transition systems", and let $(\gg, \GG, \MullerC{\F}{\GG})$ be the "acceptance condition" of $\TS$.
	Let $\pp \colon \TS \to \TS'$ be a "weak morphism" of  "transition systems" that is "locally surjective", and let $v'$ be an "accessible" "recurrent" state of $\TS'$. 
	For each $\ell' \in \cyclesState{\TS'}{v'}$ we let $\rr_{\ell'}$ be a finite "path" starting and ending in $v'$ visiting exactly the edges of~$\ell'$.
	\AP We define the ""cycle-preimage-automaton at $v'$"" to be the "Muller" "automaton" $\intro*\autCyclePreimage{v'} = (Q_{v'}, \cyclesState{\TS'}{v'}, Q_{v'}, \powplus{\GG}, \transAut{}, \MullerC{\tilde{F}}{\powplus{\GG}})$ over the input alphabet $\cyclesState{v'}{\TS'}$ defined as:
	\begin{itemize}
		\item the set of states is $Q_{v'} = \inv{\pp}(v')$,
		\item all the states are initial,
		\item the "output colours" are non-empty subsets of the colours used by $\TS$,
		\item $(q_2, C)\in \transAut{}(q_1,\ell')$  if there is a finite "path" $\rr \in \PathSetFin{\TS}{q_1}$ from $q_1$ to $q_2$ such that $\pp(\rr) = \rr_{\ell'}$ producing as "output@@run" the colours in $C\subseteq \GG$, that is $\gg(\rr) = C$. If $C$ is empty, this corresponds to an "uncoloured" edge $q_1\re{\ell':\ee}q_2$.
		We remark that, since $\pp$ is assumed "locally surjective", there is at least one such "path" $\rr$.
		\item $\{C_1,\dots, C_k\} \in \tilde{F}$ if and only if $\cup_{i=1}^k C_i \in \F$.
	\end{itemize}
\end{definition}


\AP We remark that a transition $e=q_1 \re{\ell':C }q_2$ in $\autCyclePreimage{v'}$ induces a finite "path"  $\intro*\unfold(e) = q_1\lrp{C}q_2$ in $\TS$ called the ""unfolding"" of $e$, producing as output the set of colours $C$ and such that $\pp(\unfold(e)) = \ell'$. In particular, a run $\rr$ in $\autCyclePreimage{v'}$ is "accepting@@run" if and only if $\unfold(\rr)$ is "accepting@@run".

\begin{lemma}\label{lemma-acd:same_kind_condition-automaton-precycles}
	If $L = \MullerC{\F}{\GG}$ is a "parity@@language" (resp. "Rabin@@language") language, then, so is the language $\tilde{L} = \MullerC{\tilde{F}}{\powplus{\GG}}$ used by the "acceptance condition" of $\autCyclePreimage{v'}$. 
\end{lemma}
\begin{proof}
	Assume that $L$ is a "parity language", that is, there are $d_{\min}\leq d_{\max}$ and $\phi\colon \GG \to [d_{\min}, d_{\max}]$ such that for any non-empty subset $C\subseteq \GG$, $C\in \F$ if and only if $\min \phi(C)$ is even. We define $\tilde{\phi}\colon \powplus{\GG} \to [d_{\min}, d_{\max}]$ as: $\tilde{\phi}(C) = \min \phi(C)$. It is immediate to see that $\{C_1,\dots, C_k\} \in \tilde{F}$ if and only if $\min \tilde{\phi}(\{C_1,\dots, C_k\})$ is even.
	
	Assume now that $L$ is a "Rabin language" represented by the "Rabin pairs" $\{(G_1, R_1), \dots, \ab(G_r, R_r)\}$. We define a family of "Rabin pairs" $\tilde{R} = \{(\tilde{G}_1, \tilde{R}_1), \dots, (\tilde{G}_r, \tilde{R}_r)\}$ for $\tilde{L}$ as: $\{C_1, \dots, C_k\} \in \tilde{G_i}$ (resp. $\in\tilde{R_i}$) if $\cup_{i=1}^k C_i \in G_i$ (resp. $\in\tilde{R_i}$). It is immediate to see that $\tilde{L} = \RabinC{\tilde{R}}{\powplus{\GG}}$.
\end{proof}


\begin{lemma}\label{lemma-acd:automaton-precycles}
	Let $\TS$ and $\TS'$ be two "Muller" "TS", $\pp \colon \TS \to \TS'$ a "weak morphism of TS", and~$v'$ an "accessible" "recurrent" state of $\TS'$.
	If $\pp$ is an "HD mapping", then the "automaton"  $\autCyclePreimage{v'}$ is "history-deterministic" and "recognises" the "local Muller condition of $\TS'$ at $v'$".
\end{lemma}
\begin{proof}
	\textbf{$\Lang{\autCyclePreimage{v'}} \subseteq \localMuller{\TS}{v'}$:} Let $\ell_1'\ell_2'\dots\in \cyclesState{\TS}{v}^\oo$ be a sequence of "cycles" accepted by $\autCyclePreimage{v'}$. By "prefix-independence" of "Muller languages" we can assume that all the "cycles" $\ell_i'$ appear infinitely often.
	 Let $\rr = q_0\re{\ell_1':C_1}q_1\re{\ell_2'}q_2\re{} \dots$ be an "accepting run" in $\autCyclePreimage{v'}$ over $\ell_1'\ell_2'\dots$, and let $\unfold(\rr)$ be its unfolding. As $\rr$ is an "accepting run", so is $\unfold(\rr)$, and since $\pp$ "preserves accepting runs", $\pp(\unfold(\rr))$ is an "accepting run" in $\TS'$. The edges visited by $\pp(\unfold(\rr))$ form the "cycle" $\cup_{i\geq 1} \ell_i'$, which is therefore an "accepting cycle", so $\ell_1'\ell_2'\dots\in \localMuller{\TS}{v'}$ by definition of "local Muller condition".

\textbf{$\localMuller{v'}{\TS'}\subseteq \Lang{\autCyclePreimage{v'}}$ and "history-determinism@@aut":} 	
 Let $r_\pp:E^*\times E' \to E$ be a "sound@@morph" "resolver simulating $\pp$". We will transfer the strategy given by $r_\pp$ to define a "resolver@@aut" $r_\A:\DD^*\times \cyclesState{\TS'}{v'} \to \DD$ for $\autCyclePreimage{v'}$, where $\DD$ is the set of transitions of the automaton.  
 Let $\rr_0'\in \RunsFin{\TS'}$ be a finite run reaching $v'$, and let $\rr_0 = \rRunsOption{\pp}(\rr_0')$ the preimage given by the resolver, ending in some $q_0\in Q_{v'}$ that is going to by used as "initial state" for $\autCyclePreimage{v'}$. For a sequence $e_1e_2\dots e_k\in \DD^*$ and $\ell'\in \cyclesState{\TS'}{v'}$, we let
 \[ r_\A(e_1e_2\dots e_k, \ell') = r_\pp(\rr_0'\rr_1'\dots \rr_k', \rr_{\ell'}),\footnotemark{} \]
 \footnotetext{Here we use a slight abuse of notation, since, formally, $r_\pp$ takes as input elements in $E^*\times E'$, but $\rr_{\ell'}\in {E'}^*$. We can naturally extend $r_\pp$ to ${E'}^*$ by induction. Equivalently, we can say that $r_\A(e_1e_2\dots e_k, \ell')$ is a suffix of $\rRunsOption{\pp}(\rr_0'\rr_1'\dots \rr_k'\rr_{\ell'})$.}
 where $\rr_j' = \pp(\unfold(e_j))$ and $v'\lrp{\rr_{\ell'}}v'$ is the finite run corresponding to $\ell'$ fixed in the definition of $\autCyclePreimage{v'}$.
 By definition, the obtained resolver satisfies the following property: 
 \begin{align}
 	\nonumber&\text{If } e_1e_2\dots \in \Delta^\oo \text{ is the "run induced by@@aut" } r_\A  \text{ over } \ell_1'\ell_2'\dots \in \cyclesState{\TS'}{v'}^\oo,\\
 	\nonumber&\text{then } \rr_0\unfold(e_1e_2\dots) = \rRunsOption{\pp}(\rr_0'\rr_1'\rr_2'\dots).
 \end{align}
 This gives us:
 \begin{multline}
 	\nonumber \bigcup\minf(\ell_1',\ell_2',\dots) \text{ is "accepting cycle" in } \TS' \iff \rr_0'\rr_1'\rr_2'\dots \text{ is "accepting run" in } \TS' \implies \\
 	\implies \rr_0\unfold(e_1e_2\dots) \text{ "accepting run" in }  \TS  \iff e_1e_2\dots \text{ "accepting run" in }  \autCyclePreimage{v'}. 
 \end{multline}
 Which allows us to conclude that the $\autCyclePreimage{v'}$ "recognises@@aut" $\localMuller{\TS}{v'}$ and that $r_\A$ is a "sound@@aut" "resolver@@aut".  
\end{proof}

\begin{corollary}
	Let $\TS$ and $\widetilde{\TS}$ be a "Muller" and a "parity" "transition system", respectively, and let $\pp \colon \widetilde{\TS} \to \TS$ be an "HD mapping". Let $v$ be an "accessible" "recurrent" state of $\TS$. Then,
	\[ |\inv{\pp}(v)| \geq |\leaves(\zielonkaTree{\localMuller{v}{\TS}})| = |\leaves(\treeVertex{v})|.\]
\end{corollary}

\begin{proof}
	By Lemma~\ref{lemma-acd:automaton-precycles}, the "automaton" $\autCyclePreimage{v}$ is a "history-deterministic" automaton "recognising@@aut" $\localMuller{\TS}{v}$ of size $|\inv{\pp}(v)|$, and by Lemma~\ref{lemma-acd:same_kind_condition-automaton-precycles}, it is a "parity" "automaton". The optimality of the "ZT-parity-automaton" (Theorem~\ref{thm-zt:strong_optimality_ZTparity}) gives us the first inequality. The second equality follows from the fact that $\treeVertex{v}$ is the "Zielonka tree" of $\localMuller{v}{\TS}$ (Lemma~\ref{lemma-acd:tree_q_ZT_local_Muller}).
\end{proof}

The next corollary admits an identical proof, using the optimality of the "ZT-HD-Rabin-automaton" (Theorem~\ref{thm-zt:optimality_ZT-HD-Rabin}).
\begin{corollary}
	Let $\TS$ and $\widetilde{\TS}$ be a "Muller" and a "Rabin" "transition system", respectively, and let $\pp \colon \widetilde{\TS} \to \TS$ be an "HD mapping". Let $v$ be an "accessible" "recurrent" state of $\TS$. Then,
	\[ |\inv{\pp}(v)| \geq |\memTree{\zielonkaTree{\localMuller{v}{\TS}}}| = |\memTree{\treeVertex{v}}|.\]
\end{corollary}

Theorems~\ref{thm-acd:optimality-size_ACD-parity_transform} and~\ref{thm-acd:optimality_ACD-HD-Rabin-transform} follow from these two corollaries, the formulas for the size of the "ACD-transforms" (Remarks~\ref{rmk-acd:size_ACD-parity-transformation} and~\ref{rmk-acd:size_acd-RabinHD-transformation}) and the fact that a "locally surjective morphism" $\pp:\widetilde{\TS} \to \TS$  is "surjective" if all vertices of $\TS$ are "accessible" (Lemma~\ref{lemma-morph:locSurjMorph_OntoAccessible}).

	\section{Corollaries}\label{section:corollaries}
	In this section, we discuss some further applications of the "Zielonka tree" and the "alternating cycle decomposition".
In Section~\ref{subsec-corollaries:typeness}, we use the insights gained from the "ACD" to conduct a comprehensive study of "typeness" results for "deterministic" "Muller" "automata" (that is, when can we relabel a "DMA" with an "equivalent@@accCond" and simpler "acceptance condition").
In Section~\ref{subsec-corollaries: normal-form} we present a "normal form" for "parity" "transition systems" and prove the main properties exhibited by "TS" in this form.
In Section~\ref{subsec-corollaries: minimisation-parity}, we provide a polynomial-time algorithm minimising "DPA" "recognising" "Muller languages".

\subsection{Typeness results}\label{subsec-corollaries:typeness}

As we have seen, there are many different types of "acceptance conditions" for $\oo$-regular automata. An important question is the following:

\begin{center}
	Question: Given a "Muller" "automaton" $\A$, can we define a simpler "acceptance condition" over the "underlying graph" of $\A$ obtaining an "equivalent@@aut" automaton $\A'$?
\end{center}

This question was first studied (in the context of "automata" using state-based acceptance) by Krishnan, Puri and Brayton~\cite{KPB94DetOmega,KPB95Structural}, who showed how to determine if a "DMA" can be relabelled with an "equivalent@@accCond" "B\"uchi" "condition@@acc". Their work was generalised to "parity" "automata" by Boker, Kupferman and Steinitz~\cite{BKS10Paritizing}, and related questions about "typeness" were studied for "non-deterministic" "automata" by Kupferman, Morgenstern and Murano~\cite{KMM06Typeness}, and for "history-deterministic" automata by Boker, Kupferman and Skrzypczak~\cite{BKS17HowDeterministicGFG}.

In this section, we provide new general characterisations of "typeness" for "Muller" "transition systems". The main contributions of this section appear in Propositions~\ref{prop-typ:Rabin ACD type},~\ref{prop-typ:Streett ACD type} and~\ref{prop-typ:parity ACD type}, which characterise when a "Muller" "TS" can be relabelled with "equivalent@@acc" "parity", "Rabin", or "Streett" "conditions@@acc" in terms of properties of the "cycles" of the "TS". For instance, Proposition~\ref{prop-typ:Rabin ACD type} states that a "Muller" "TS" can be relabelled with an "equivalent@@acc" "Rabin" "condition@@acc" if and only if its "rejecting@@cycle" "cycles" are closed under union. The ``only if'' part of these results was already known~\cite{Loding1999Optimal}, but the fact that this is indeed a characterisation is a novel result, for which the use of the "ACD" is essential. 
These characterisations directly imply the results from~\cite{BKS10Paritizing,KPB94DetOmega,KPB95Structural}.
We also show how to use the "ACD" to determine the "parity index" of the "language recognised" by a "DMA" (Proposition~\ref{prop-typ:parity index-Muller-automaton}), which can be seen as a simplification of the results from~\cite[Section~3.2]{KPB95Structural}. 
Further results concerning "generalised B\"uchi languages" and "weak automata" can be found in Appendix~\ref{sec:appendix-weak-conditions}.

\subsubsection{Typeness for Muller languages}\label{subsubsec-typeness:Typeness_Conditions}

We first present some results proven by Zielonka~\cite[Section~5]{Zielonka1998infinite} that show how we can use the "Zielonka tree" to deduce if a "Muller language" is a "Rabin@@language", a "Streett@@language" or a "parity language". These results are generalised to "transition systems" in the next subsection. A study of further types of "Muller languages" can be found in Appendix~\ref{sec:appendix-weak-conditions}.

We do not include the proofs of the results of this section in the main body of the paper, as they are known results~\cite[Section~5]{Zielonka1998infinite} and they are special cases of the proofs in Section~\ref{subsubsec-typeness:Typeness_TS}. Nevertheless, we include them in Appendix~\ref{sec:appendix-typeness-proofs-Muller-conditions}.

We first introduce some definitions. The terminology will be justified by the upcoming results.

\begin{definition}\label{def:ShapesOfTrees}
	Let $T$ be a "tree" with nodes partitioned into "round" nodes and "square" nodes. We say that $T$ has:
	\begin{itemize}
		\item  \AP""Rabin shape"" if every "round node" has at most one "child".
		
		\item  \AP""Streett shape"" if every "square node" has at most one "child".
		
		\item  \AP""Parity shape"" if every node has at most one "child".
	\end{itemize}
\end{definition}

\begin{proposition}\label{prop-typ:RabinZielonkaShape}
	Let $\F\subseteq \powplus{\Gamma}$ be a family of non-empty subsets. The following conditions are equivalent:
	\begin{enumerate}
		\item $\MullerC{\F}{\GG}$ is a "Rabin language".
		\item $\powplus{\GG}\setminus \F$ is closed under union: If 
		$C_1\notin \F$ and $C_2\notin \F$, then $C_1\cup C_2\notin\F$.
		\item $\zielonkaTree{\F}$ has "Rabin shape". 
	\end{enumerate}
\end{proposition}

\begin{proposition}\label{prop-typ:StreettZielonkaShape}
	Let $\F\subseteq \powplus{\Gamma}$ be a family of non-empty subsets. The following conditions are equivalent:
	\begin{enumerate}
		\item $\MullerC{\F}{\GG}$ is a "Streett language".
		\item The family $\F$ is closed under union.
		\item $\zielonkaTree{\F}$ has "Streett shape". 
	\end{enumerate}
	
\end{proposition}

\begin{proposition}\label{prop-typ:parityZielonkaShape}
	Let $\F\subseteq \powplus{\Gamma}$ be a family of non-empty subsets. The following conditions are equivalent:
	\begin{enumerate}
		\item $\MullerC{\F}{\GG}$ is a "parity language".
		\item Both $\F$ and $\powplus{\GG}\setminus \F$ are closed under union: If 
		$C_1\in \F \iff C_2\in \F$, then, $C_1\cup C_2\in \F \iff C_1\in \F$.
		\item $\zielonkaTree{\F}$ has "parity shape". 
	\end{enumerate}
	Moreover, if some of these conditions is satisfied, $\MullerC{\F}{\GG}$ is a "$[\minparityZ{\F},\maxparityZ{\F}]$-parity language".
\end{proposition}

\begin{corollary}
	A "Muller language" $L \subseteq \GG^\oo$ is a "parity language" if and only if it is both a "Rabin@@language" and  a "Streett language".
\end{corollary}

\subsubsection{Typeness for  Muller transition systems and deterministic automata}\label{subsubsec-typeness:Typeness_TS}
We start this subsection by introducing the necessary definitions about "equivalence of acceptance conditions" and "typeness". Then, we state and prove our main contributions concerning "typeness" of "transition systems".

\paragraph*{Equivalence of acceptance conditions and typeness.}
 Let $\TS_1=(G, \macc{1})$ and $\TS_2=(G, \macc{2})$ be two "transitions systems" over the same "underlying graph" $G$, with "acceptance conditions" $\macc{i} = (\gg_i, \GG_i, \WW_i)$, for $i \in \{1,2\}$.
\AP We say that $\macc{1}$ and $\macc{2}$ are ""equivalent over $G$"", written $\macc{1} \intro*\equivCond{G} \macc{2}$, if for all "runs" $\rr\in \Runs{G}$, $\rr$ is "accepting@@run" for $\TS_1$ if and only if it is "accepting@@run" for $\TS_2$; that is, $\gg_{1}(\rr)\in \WW_1 \iff \gg_{2}(\rr)\in \WW_2$.

\AP We write $\TS_1 \intro*\equivTrans \TS_2$ if $\TS_1$ and $\TS_2$ are "isomorphic@@TS". We recall that two "transition systems" are "isomorphic@@TS" if there is a "morphism of transition systems" $\pp\colon \TS_1\to \TS_2$ whose inverse is also a "morphism@@TS", that is, $\pp$ and $\inv{\pp}$ "preserve the acceptance of runs".

\begin{remark}\label{rmk-typ:condition-transferred-by-isomorphism}
	If $\pp\colon \TS_1 \to \TS_2$ is an "isomorphism@@TS", then $(\gg_2\circ \pp, \GG_2, \WW_2)$ is an "acceptance condition" over the "underlying graph" of $\TS_1$ that is equivalent to $(\gg_1, \GG_1, \WW_1)$ over this graph.
	
	Conversely, if two "acceptance conditions" $\macc{1}$ and $\macc{2}$ are "equivalent over a same graph~$G$", then the identity function is an "isomorphism@@TS" between $\TS_1=(G, \macc{1})$ and $\TS_2=(G, \macc{2})$.
\end{remark}

\AP For $X$ one of types of languages defined in Section~\ref{subsect-prelim:acceptance-conditions} (B\"uchi, parity, Muller, etc...), we say that a "transition system" $\TS$ is ""$X$ type"" if there exists an "isomorphic@@TS" "transition system" $\TS'\equivTrans \TS$ using an "$X$ acceptance condition". We note that, by the previous remark, in that case an "$X$ acceptance condition" can be defined directly over the "underlying graph" of $\TS$.

We remark that, given a "pointed graph" $G$ (whose states are "accessible"), the equivalence classes of "Muller" "acceptance conditions" for the relation $\equivCond{G}$ are given exactly by the mappings $f\colon \cycles{G} \to \{\mathrm{Accept}, \mathrm{Reject} \}$.

\paragraph*{The ACD determines the type of transition systems.}
\begin{definition}\label{def-cor:ShapesOfACD}
	Let $\TS$ be a "Muller" "transition system" with a set of states $V$. We say that its "alternating cycle decomposition" $\acd{\TS}$ is a:
	\begin{itemize}
		\item  \AP""Rabin ACD"" if for every state $v\in V$, the tree $\treeVertex{v}$ has "Rabin shape".
		
		\item  \AP""Streett ACD"" if for every state $v\in V$, the tree $\treeVertex{v}$ has "Streett shape".
		
		\item  \AP""Parity ACD"" if for every state $v\in V$, the tree $\treeVertex{v}$ has "parity shape".
		\item \AP ""$[0,d-1]$-parity ACD"" (resp. "$[1,d]$-parity ACD") if it is a "parity ACD", "trees" of $\acd{\TS}$ have "height" at most $d$ and "trees" of "height" $d$ are "positive@@tree" (resp. "negative@@tree").
	\end{itemize}
\end{definition}

\begin{remark}\label{rmk-typ:shape-parity-Street-Rabin-ACD}
 $\acd{\TS}$ is a "parity ACD" if and only if it is both a "Rabin@@ACD" and a "Streett ACD". 
\end{remark}

\begin{proposition}\label{prop-typ:Rabin ACD type}
	Let $\TS = (\underlyingGraph{\TS}, \macc{\TS})$ be a "Muller" "transition system" whose states are "accessible". The following conditions are equivalent:
	\begin{enumerate}
		\item $\TS$ is "Rabin type".
		\item For every pair of "rejecting cycles" $\ell_1, \ell_2\in \cycles{\TS}$  with some "state in common",  $\ell_1\cup \ell_2$ is a "rejecting cycle".
		\item $\acd{\TS}$ is a "Rabin ACD". 
	\end{enumerate}
\end{proposition}
\begin{proof}
	\begin{description}
		\item[($1 \Rightarrow 2$)] 
		Let $\macc{R} = (\gg, \GG, \Rabin{R}$ be the "Rabin" "acceptance condition" "equivalent to@@condition" $\macc{\TS}$, and let $R = (G_1,R_1),\dots,(G_r,R_r)$ be its "Rabin pairs". Let $\ell_1$ and $\ell_2$ be two "cycles" with a state in common, and suppose that $\ell_1\cup \ell_2$ is "accepting@@cycle"; we show that either $\ell_1$ or $\ell_2$ is "accepting@@cycle". The "cycle" $\ell_1\cup \ell_2$ is "accepted by@@Rabinpair" some "Rabin pair" $(G_j,R_j)$, so for all edges $e\in \ell_1\cup \ell_2$, $\gg(e)\notin R_j$, and there is some $e_0\in \ell_1\cup \ell_2$ such that $\gg(e_0)\in G_j$. If $e_0$ belongs to $\ell_1$, then $\ell_1$ is "accepted by@@Rabinpair" the "Rabin pair" $(G_j,R_j)$, and if $e_0\in \ell_2$, then $\ell_2$ is "accepted@@Rabinpair" by it.
		
		\item[($2 \Rightarrow 3$)] Let $v$ be a vertex of $\TS$ and $\treeVertex{v}$ the "local subtree at $v$". Suppose that there is a "round node@@acd" $n\in \treeVertex{v}$ with two different children $n_1$ and $n_2$. The "cycles" $\nuAcd(n_1)$ and $\nuAcd(n_2)$ are "rejecting cycles" over $v$, but their union is an "accepting cycle" (by Remark~\ref{rmk-acd:union-changes-acceptance}).
		
		\item[($3 \Rightarrow 1$)] We observe that $\acd{\TS}$ is a "Rabin ACD" if and only if $\memTree{\treeVertex{v}} = 1$ for all vertices $v$ of $\TS$. In particular, the "ACD-HD-Rabin-transform" of $\TS$ does not add any state to $\TS$. It is immediate to check that the "morphism" $\pp\colon \acdRabinTransform{\TS} \to \TS$ given by $\pp_V(v,x)=v$, $\pp_E(e,l)=e$ defined in the proof of Proposition~\ref{prop-ACD:correctness_ACD-HD-Rabin-transform} is an "isomorphism", and $\TS$ uses a "Rabin" "acceptance condition".\qedhere
	\end{description}
\end{proof}

\begin{proposition}\label{prop-typ:Streett ACD type}
	Let $\TS = (\underlyingGraph{\TS}, \macc{\TS})$ be a "Muller" "transition system". The following conditions are equivalent:
	\begin{enumerate}
		\item $\TS$ is "Streett type".
		\item For every pair of "accepting cycles" $\ell_1, \ell_2\in \cycles{\TS}$  with some "state in common", $\ell_1\cup \ell_2$ is an "accepting cycle".\footnotemark{}
		\item $\acd{\TS}$ is a "Streett ACD". 
	\end{enumerate}
\footnotetext{This property was introduced by Le Saëc under the name \emph{cyclically closed automata}~\cite{Saec90Saturating}. We point out that the ``if'' direction of the result stated in~\cite[Theorem~5.2]{Saec90Saturating} does not hold. That statement can be rephrased as: If a "DMA" $\A$ is cyclically closed, then the "parity index" of $\A$ is $[0,1]$. We refer to Proposition~\ref{prop-typ:parity index-Muller-automaton} for a correct characterisation.}
\end{proposition}
\begin{proof}
	Implications ($1 \Rightarrow 2$) and ($2 \Rightarrow 3$) are analogous to those from Proposition~\ref{prop-typ:Rabin ACD type}.
	\begin{description}		
		\item[($3 \Rightarrow 1$)] We consider the "transition system" $\complTS{\TS}$ obtained by complementing the "acceptance set" of $\macc{\TS}$. By Remark~\ref{rmk-acd:complementation-ACD}, the "ACD" of $\complTS{\TS}$ is obtained from $\acd{\TS}$ by turning "round nodes@@acd" into "square nodes@@acd" and vice-versa. Thus, the "ACD" of $\complTS{\TS}$ is a "Rabin ACD", and by applying the previous proposition we can define a "Rabin condition" $\macc{R} = (\gg, \GG, \RabinC{R}{\GG})$ such that the  "transition system" $(\underlyingGraph{\TS}, \macc{R})$ is "isomorphic@@TS" to $\complTS{\TS}$. Since $\StreettC{R}{\GG}$  is the complement language of $\RabinC{R}{\GG}$, we obtain that $\macc{S} = (\gg, \GG, \StreettC{R}{\GG})$ is a "Streett" "acceptance condition" "equivalent to@@accCond" $\macc{\TS}$ over $\underlyingGraph{\TS}$.\qedhere
	\end{description}
\end{proof} 

\begin{proposition}\label{prop-typ:parity ACD type}
	Let $\TS = (\underlyingGraph{\TS}, \macc{\TS})$ be a "Muller" "transition system". The following conditions are equivalent:
	\begin{enumerate}
		\item $\TS$ is "parity type".
		\item For every pair of "accepting@@cycle" (resp. "rejecting@@cycle") "cycles" $\ell_1, \ell_2\in \cycles{\TS}$  with some "state in common",  $\ell_1\cup \ell_2$ is an "accepting@@cycle" (resp. "rejecting@@cycle") "cycle".
		\item $\acd{\TS}$ is a "parity ACD". 
	\end{enumerate}
	Moreover, if some condition is satisfied, $\TS$ is "$[0,d-1]$ (resp. $[1,d]$)-parity type" if and only if $\acd{\TS}$ is a "$[0,d-1]$(resp. $[1,d]$)-parity ACD".
\end{proposition}
\begin{proof}
	\begin{description}
		\item[($1 \Rightarrow 2$)] Proven in Lemma~\ref{lemma-ZT:union_cycles_parity-TS}.
		
		\item[($2 \Rightarrow 3$)] Admits an analogous proof to the corresponding implication in Proposition~\ref{prop-typ:Rabin ACD type}.
		
		\item[($3 \Rightarrow 1$)] By definition, $\acd{\TS}$ is a "parity ACD" if and only if $\leaves(\treeVertex{v})$ is a singleton for each vertex $v$ of $\TS$. In particular, the "ACD-parity-transform" of $\TS$ does not add any state to $\TS$. It is immediate to check that the "morphism" $\pp\colon \acdParityTransform{\TS} \to \TS$ defined in the proof of Proposition~\ref{prop-ACD:correctness_ACD-parity_transform} is an "isomorphism". Therefore, $\TS$ and $\acdParityTransform{\TS}$ are "isomorphic@@TS" "transition systems", and the latter uses a "parity" "acceptance condition" that is a $[0,d-1]$ (resp. $[1,d]$)-parity condition if $\acd{\TS}$ is a "$[0,d-1]$ (resp. $[1,d]$)-parity ACD". If $\acd{\TS}$ is not a "$[0,d-1]$(resp. $[1,d]$)-parity ACD", then the number of colours cannot be reduced by the optimality of the number of colours of the "ACD-parity-transform" (Theorem~\ref{thm-acd:optimality-priorities_ACD-parity_transform}).\qedhere
	\end{description}
\end{proof}

\begin{corollary}\label{cor-cor:parityTS-iff-RabinAndStreettTS}
	A "Muller" "transition system" is "parity type" if and only if it is both "Rabin@@type" and "Streett type". 
\end{corollary}

\paragraph*{The ACD and the parity index of $\oo$-regular languages.}
\begin{proposition}\label{prop-typ:parity index-Muller-automaton}
	Let $\A$ be a "deterministic" "Muller" "automaton" whose states are "accessible". Then, the "parity index" of $\Lang{\A}$ is:
	\begin{itemize}
		\item $[0,d-1]$ (resp. $[1,d]$) if and only if:
		\begin{itemize}
			\item "trees" of $\acd{\A}$ have "height" at most $d$, 
			\item there is at least one "tree" of "height" $d$, and
			\item "trees" of "height" $d$ are "positive@@tree" (resp. "negative@@tree").
		\end{itemize} 
		\item $\WeakIndex{d}$ if and only if: 
		\begin{itemize}
			\item "trees" of $\acd{\A}$ have "height" at most $d$, 
			\item there is at least one "positive@@tree" "tree" of "height" $d$, and
			\item there is at least one "negative@@tree" "tree" of "height" $d$.
		\end{itemize} 
	\end{itemize}
\end{proposition}
\begin{proof}
	We prove the right-to-left implication for the case $\WeakIndex{d}$.
	Assume that $\acd{\A}$ verifies the previous list of conditions (in particular, it is "equidistant@@acd"). Then, the "ACD-parity-transform" $\acdParityTransform{\A}$ is a "DPA" "recognising" $\Lang{\A}$ using colours in $[0,d]$. In order to obtain a "DPA" for $\Lang{\A}$ with colours in $[1,d+1]$ we need to introduce a small modification to the function $\parityNodesAcd$. For $\ell_i$ a maximal "cycle" of $\A$ and $n\in N_{\ell_i}$ we define: \begin{itemize}
		\item $\parityNodesAcd'(n) = \depth(n)+2$, if $\ell_i$ is "accepting@@cycle",
		\item $\parityNodesAcd'(n) = \depth(n) + 1$, if $\ell_i$ is "rejecting@@cycle".
	\end{itemize}
	It is a routine check to see that the version of the "ACD-parity-transform" using $\parityNodesAcd'$ is indeed a correct "parity" "automaton" using colours in $[1,d+1]$.
	
	To prove that no "DPA" "recognising" $L$ uses less than $d$ colours, it suffices to use the Flower Lemma~\ref{lemma:flower-lemma} and the fact that a branch of length $d$ in a "tree" of the "ACD" induces a "$d$-flower" in $\A$, which is "positive@@flower" if and only if the corresponding tree is "positive@@tree" (Lemma~\ref{lemma-acd:branch-ACD-induces-flower}).
	
	This is indeed a complete characterisation, since for any "ACD" there is a minimal $d$ such that $\acd{\A}$ lies in one and only one of the classes specified in the statement of the proposition.
\end{proof}

\begin{proposition}\label{prop-typ:min-colours-Muller-automata}
	Let $L\subseteq \SS^\oo$ be an "$\oo$-regular language" of "parity index at least" $[0,d-1]$ (resp. $[1,d]$). Any "history-deterministic" "Muller" "automaton" "recognising" $L$ uses an "acceptance condition" with at least $d$ different "output colours". 
\end{proposition}
\begin{proof}
	We first prove the result for "deterministic" automata. Let $\A$ be a "DMA" "recognising" $L$ using the "acceptance condition" $(\gg, \GG, \MullerC{\F}{\GG})$. By Proposition~\ref{prop-typ:parity index-Muller-automaton}, there is a "tree" $\altTree{\ell_i}$ in the "ACD" of $\A$ of "height" at least $d$. 
	\AP We define $\intro*\coloursNodesAcd\colon \nodesAcd{\TS} \to \GG$ to be the function that assigns to each node of the "ACD" the colours appearing in it, that is: $\coloursNodesAcd(n) = \gg(\nuAcd(n))$. We remark that if $n'$ is a "descendant" of $n$ then $\coloursNodesAcd(n')\subseteq \coloursNodesAcd(n)$, and that a node $n$ is "round@@acd" if and only if $\coloursNodesAcd(n)\in \F$. Therefore, by the alternation of "round@@acd" and "square@@acd" nodes, if $n'$ is a strict "descendent" of $n$, $\coloursNodesAcd(n')\subsetneq \coloursNodesAcd(n)$. We conclude that the root of $\altTree{\ell_i}$ must contain at least $d$ different colours.
	
	In order to obtain the result for "history-deterministic" automata we use "finite-memory resolvers" as defined in Section~\ref{subsec-zt: parity automaton}. If $\A$ is a "history-deterministic" "Muller" "automaton", it admits a "sound@@aut" "resolver@@aut" "implemented by@@aut" a finite "memory structure" $(\M, \nextmoveResolver)$ (Lemma~\ref{lemma-ZT:finite-memory-resolvers}). Then, the "composition@@memory" $\A \prodMem{\ss} \M$ is a "DMA" using the same number of colours, that has to be at least~$d$.
\end{proof}

The following result (which was already known, as it is a consequence of the construction by Carton and Maceiras~\cite{CartonMaceiras99RabinIndex}), is refined and proven in Appendix~\ref{sec:appendix-weak-conditions} (Corollary~\ref{cor-app-weak:index-parity-type}).
\begin{proposition}
	Let $\A$ be a "deterministic" "parity" "automaton" such that all its states are "accessible" and the "parity index" of $\Lang{\A}$ is $[0,d-1]$ (resp. $[1,d]$). 
	Then, $\A$ is $[0,d-1]$ (resp. $[1,d]$)-"parity type". 
\end{proposition}

The previous result does not hold for "history-deterministic" automata, as we could artificially add transitions augmenting the complexity of the structure of the automaton (enlarging the "flowers" of the automaton) without modifying the language it "recognises". Nevertheless, some analogous results applying to "HD automata" can be obtained.  Boker, Kupferman and Skrzypczak proved that any "HD" "parity" "automaton" "recognising" a language of "parity index" $[0,1]$ (resp. $[1,2]$) admits an "equivalent@@aut" "HD@@aut" "subautomaton" using a "B\"uchi" (resp. "coB\"uchi") "condition@@acc"~\cite[Theorems~10 and~13]{BKS17HowDeterministicGFG}. We do not know whether the result holds for languages of arbitrary "parity index".

\paragraph*{Typeness for deterministic automata.}
Two "automata" $\A_1$ and $\A_2$ such that $\A_1 \equivTrans \A_2$ recognise the same language: $\Lang{\A_1} = \Lang{\A_2}$. However, the converse only holds for "deterministic" automata.

\begin{lemma}\label{lemma-typ:equivalece_automata_languages}
	Let $\A_1$ and $\A_2$ be two "deterministic automata" over the same "underlying graph" and with the same "labelling" by "input letters". Then, $\Lang{\A_1} = \Lang{\A_2}$ if and only if $\A_1 \equivTrans \A_2$.
\end{lemma}
\begin{proof}
	The implication from right to left is trivial. For the other implication, suppose that $\Lang{\A_1} = \Lang{\A_2}$, and let $\rr \in \Runs{\A_1}=\Runs{\A_2}$ be an infinite "run" over the "underlying graph" of $\A_1$. Let $w\in \SS^\oo$ be the word over the "input alphabet" $\SS$ labelling the run $\rr$. Since $\A_1$ and $\A_2$ are "deterministic", $\rr$ is the only "run over" $w$, and therefore:
	\[ \rr \text{ is "accepting@@run" for } \A_1 \iff w\in \Lang{\A_1}=\Lang{\A_2} \iff \rr \text{ is "accepting@@run" for } \A_2. \qedhere\]
\end{proof}

\begin{corollary}[First proven in~{\cite[Theorem~7]{BKS10Paritizing}}]
	Let $\underlyingGraph{\A}$ be the "underlying graph" of a "deterministic automaton". Then, there are "Rabin" and "Streett" conditions $\macc{R}$ and $\macc{S}$ such that $\Lang{\underlyingGraph{\A},\macc{R}} = \Lang{\underlyingGraph{\A},\macc{S}}$ if and only if there is a "parity condition" $\macc{p}$ such that $\Lang{\underlyingGraph{\A},\macc{p}}=\Lang{\underlyingGraph{\A},\macc{R}} = \Lang{\underlyingGraph{\A},\macc{S}}$.
\end{corollary}

We remark that the hypothesis of "determinism" in the previous corollary is necessary, as it has been shown that an analogous result does not hold for "non-deterministic" "automata"~\cite{BKS10Paritizing}.

\begin{proposition}[First proven in~{\cite[Theorem~15]{KPB94DetOmega}}]
	Let $\A$ be a "deterministic" "Rabin" (resp. Streett) "automaton", and assume that $\Lang{\A}$ can be "recognised" by a "deterministic" "B\"uchi" (resp. "coB\"uchi") "automaton"; that is, the "parity index" of $\Lang{\A}$ is at most $[0,1]$ (resp. at most $[1,2]$). Then,~$\A$ is "B\"uchi type" (resp. "coB\"uchi type").
\end{proposition}
\begin{proof}
	We do the proof for the case Rabin-B\"uchi. We can assume that all the states of $\A$ are "accessible", as we can define a trivial acceptance condition in the part of $\A$ that is not accessible.  Since $\Lang{\A}$ has "parity index" at most $[0,1]$, the trees of the "ACD" of $\A$ have "height" at most $2$, and trees of "height" $2$ are "positive@@tree" (the "root" is a "round node@@acd"), by Proposition~\ref{prop-typ:parity index-Muller-automaton}. As $\A$ is a "Rabin" "automaton", its "ACD" has "Rabin shape" (Proposition~\ref{prop-typ:Rabin ACD type}), so "round nodes@@acd" have at most one "child". We conclude that the trees of the "ACD" of $\A$ have a single branch, so it is a "$[0,1]$-parity ACD", and by Proposition~\ref{prop-typ:parity ACD type}, $\A$ is "B\"uchi type".
\end{proof}

\subsection{A normal form for parity transition systems}\label{subsec-corollaries: normal-form}
In this section, we propose a definition of a "normal form" of "parity" "automata". This is exactly the form of "automata" resulting by applying the procedure defined by Carton and Maceiras~\cite{CartonMaceiras99RabinIndex}, or, equivalently, of "automata" resulting from the "ACD-parity-transform" (Corollary~\ref{cor-normF:ACD-transform_normal_form}). These automata satisfy that they are "parity-index-tight", that is, their "acceptance condition" uses the minimal possible number of "colours". But they offer some further convenient properties, stated in Propositions~\ref{prop-normF:paths_form_parity} and~\ref{prop-normF:flowers-in-normal_form}, which make them particularly well-suited for reasoning about "deterministic" "parity" "automata".

This "normal form", or partial versions of it, have already been used in the literature to prove results about "parity" automata in different contexts, such as "history-deterministic" "coB\"uchi" "automata"~\cite{AK22MinimizingGFG,EhlersSchewe22NaturalColors,KS15DeterminisationGFG}, positionality of languages defined by "deterministic" "B\"uchi" "automata"~\cite{BCRV22HalfPosBuchi} or learning of "DPAs"~\cite{BohnLoding23DetParityFromExamples}. The "normalisation" of "transition systems" also facilitates solving "parity" "games" in practice~\cite{FriedmannM09ParityGames}.
However, the application of this "normal form" in the literature is limited to specific cases, and no prior works have provided a formal and systematic study of it.

From our results we obtain three equivalent ways of defining the "normal form" of a "parity" "transition system" $\TS$. Informally, they can be stated as:
\begin{enumerate}
	\item $\TS$ use the smallest possible "colour" in each of its transitions (Definition~\ref{def-normF:normal_form_parity}).
	\item $\TS = \acdParityTransform{\TS}$ (Corollary~\ref{cor-normF:ACD-transform_normal_form}).
	\item Paths in $\TS$ producing "colour" $d>0$ can be closed into a "cycle" producing $d'$ as minimal "colour", for all $d'\leq d$ (Theorem~\ref{th-normF:charactNormForm}).
\end{enumerate}



\begin{remark}
	If $\macc{} = (\gg, [d_{\min}, d_{\max}], \parity)$ is a "parity" "acceptance condition" over a "pointed graph"~$\underlyingGraph{}$, we can always assume that $d_{\min}$ is $0$ or $1$. Indeed, define $\chi = d_{\min}$ if $d_{\min}$ is even, and $\chi = d_{\min} - 1$ if  $d_{\min}$ is odd. The "parity" "acceptance condition" $(\gg', [d_{\min} - \chi, d_{\max} - \chi], \parity)$ defined as $\gg'(e) = \gg(e) - \chi$ is "equivalent@@acceptCond" to $\macc{}$ over $\underlyingGraph{}$.
\end{remark}

\paragraph*{Definition of the normal form.}

Just as in the definition of the "ACD-parity-transform" we had to define "positive@@acd" and "negative@@acd" "ACDs" to obtain an accurate optimality result in the number of "colours", we need now to take care of a small technical detail so that "TS" in "normal form" are "parity-index-tight".

\AP We say that a "transition system" $\TS$ is ""negative@@TS"" if $\acd{\TS}$ is "negative@@acd", that is, if for some~$d$ $\TS$ contains a "negative@@flower" "$d$-flower" but contains no "positive@@flower" "$d$-flower". Intuitively, a "parity" "TS" is "negative@@TS" if and only if the minimal colour used by a "parity" "acceptance condition" using an optimal number of "colours" is $1$.

\begin{definition}[Normal form]\label{def-normF:normal_form_parity}
	\AP Let $\TS = (\underlyingGraph{\TS}, \macc{\TS})$ be a "parity" "transition system" using a "colouring function" $\gg$. If $\TS$ is "not negative@@TS", we say that $\TS$ is in ""normal form"" if any other "parity" "acceptance condition" "equivalent to@@accCond" $\macc{\TS}$ over $\underlyingGraph{\TS}$ using a "colouring function" $\gg'$ satisfies that for every edge $e$:
	\[ \gg(e)\leq \gg'(e).  \]
	
	If $\TS$ is "negative@@TS", we say that it is in ""normal form@@neg"" if any other "equivalent@@accCond"  "colouring" $\gg'$ not using colour $0$ satisfies that for any edge $e$:
	\[  1\leq\gg(e)\leq \gg'(e).  \]
	If $\TS$ is in "normal form", we will also say that its "acceptance condition" or the "colouring function" it uses are in normal form.
\end{definition}

\begin{example}\label{ex-cor:normal-form}
	"Parity" "transition systems" from Figures~\ref{fig-prelim:multiple-aut-ex},~\ref{fig-ZT:zielonka-parityAutomaton},~\ref{fig-acd:ACD-parity_transform} and~\ref{fig-acd:ACD-transform-doesNotPreserveMinimality} are all in "normal form". Parity automata appearing in Figures~\ref{fig-ZT:zielonka-parityAutomaton} and~\ref{fig-acd:ACD-parity_transform} are "negative@@TS" (the minimal colour used by an optimal "acceptance condition" is odd), whereas "parity" "automata" in Figure~\ref{fig-acd:ACD-transform-doesNotPreserveMinimality} are not.
	
	On the other hand, the "automaton" from Figure~\ref{fig-prelim:HD-aut-example} is not in "normal form" (even if it uses an optimal number of colours). We can put it in "normal form" by assigning colour $1$ to transitions $q_1\re{a,b}q_0$ and $q_1\re{b,c}q_2$. The "automaton" obtained in this way "recognises" the same language.
\end{example}

\begin{proposition}\label{prop-normF:existance+uniqueness_normal_form}
	Let $\TS = (\underlyingGraph{\TS}, \macc{\TS})$ be a "parity" "transition system" with a "colouring function" $\gg$.
	There is a unique "parity" "acceptance condition"  
	"equivalent@@condition" to $\macc{\TS}$ over $\underlyingGraph{\TS}$ in "normal form".
	Moreover, this "acceptance condition" is exactly the "parity condition" of the "ACD-parity-transform" of $\TS$.
\end{proposition}

Before showing the proof of Proposition~\ref{prop-normF:existance+uniqueness_normal_form}, we prove a useful technical lemma.
\begin{lemma}\label{lemma-normF:flowers-lower-bound-in-parity}
	Let $\TS$ be a "parity" "transition system" with "colouring function" $\gg$. If $\ell_1\supsetneq \ell_2 \supsetneq \dots \supsetneq \ell_{k}$ is a "positive@@flower" (resp. "negative@@flower") "$k$-flower" of $\TS$, then $\min \gg(\ell_k)\geq k-1$ (resp. $\min \gg(\ell_k)\geq k$).
\end{lemma}
\begin{proof}
	We show the result for "negative flowers". Let $d_i = \min \gg(\ell_i)$. We show that $d_i\geq i$ by induction. Since $\ell_i$ is an "accepting cycle" if and only if $i$ is even, we have that $d_i$ is even if and only if $i$ is even. Clearly, $d_1\geq 1$, as $1$ is the least odd number. Also, $d_{i+1}\geq d_i$, since $ \ell_{i+1}\subseteq \ell_i$, and the inequality is strict by the alternation of the parity, concluding the proof.
\end{proof}

\begin{proof}[Proof of Proposition~\ref{prop-normF:existance+uniqueness_normal_form}]
	We first remark that the uniqueness is directly implied by the definition of "normal form".
	
	We prove that the "acceptance condition" of the "ACD-parity-transform"  is in "normal form". We note its "colouring function" by $\gg_{\acdNoP}$.
	The transitions not belonging to any "SCC" are coloured~$0$ if $\TS$ is "not negative@@TS" and $1$ if $\TS$ is "negative@@TS", as desired. It suffices to prove the result for edges in "SCCs".
	
	We assume that $\TS$ is "not negative@@TS" and 
	we let $\S$ be an "accepting SCC" of $\TS$ (the proof is similar for $\TS$ "negative@@TS" and a "rejecting SCC"). 
	Let $e=v\re{}v'$ be an edge in $\S$, and let $\treeVertex{v}$ be the "local subtree at $v$", which is composed of a single "branch" (see Proposition~\ref{prop-typ:parity ACD type}). We let $n_0\ancestor n_1 \ancestor\dots \ancestor n_r$ be that "branch", where $n_0$ is the "root" and $n_r$ the "leaf". Let $n_k$ be the "deepest" node of $\treeVertex{v}$ such that $e\in \nuAcd(n_k)$. By definition of the "ACD-parity-transform", $\gg_{\acdNoP}(e) = \parityNodesAcd(e) = k$. Also, $\nuAcd(n_0)\ancestor \nuAcd(n_1) \ancestor\dots \ancestor \nuAcd(n_k)$ is a "positive@@flower" "$k+1$-flower" (by Lemma~\ref{lemma-acd:branch-ACD-induces-flower}).
	Lemma~\ref{lemma-normF:flowers-lower-bound-in-parity} implies then that any "equivalent@@accCond" "parity" "condition" using a "colouring function" $\gg'$ verifies $\gg'(e)\geq \gg_{\acdNoP}(e)=k$. 
\end{proof}

\begin{corollary}\label{cor-normF:ACD-transform_normal_form}
	The "ACD-parity-transform" $\acdParityTransform{\TS}$ of any "Muller" "transition system" $\TS$ is in "normal form". 
\end{corollary}

\paragraph*{Fundamental properties of the normal form.}

We now state what we consider to be the two fundamental properties of "parity" "transition systems" in "normal form". Theorem~\ref{th-normF:charactNormForm} states that these properties characterise the "normal form".

\begin{proposition}\label{prop-normF:paths_form_parity}
	Let $\TS$ be a "parity" "transition system" in "normal form". If there is a path $v\lrpE v'$ producing $d$ as minimal colour, then, either:
	\begin{itemize}	
		\item $v$ and $v'$ are in different "SCCs" (and in this case $d\in \{0,1\}$), or
		\item there is a path  $v'\lrpE v$ producing no colour strictly smaller than $d$.
	\end{itemize} 
\end{proposition}
\begin{proof}
	By Proposition~\ref{prop-normF:existance+uniqueness_normal_form}, we know that the "colouring" of $\TS$ is the one given by its "ACD-transform", that we note $\gg_{\acdNoP}$. 
	If $v$ and $v'$ are in different "SCCs" the result is trivial.  
	Let $v$ and $v'$ be in the same "SCC", that we suppose to be an "accepting SCC" without loss of generality. 
	Let $\rr = v\re{e_1}\dots \re{e_k} v'$ be a path from $v$ to $v'$ producing $\min \gg_{\acdNoP}(\rr) = d$ as minimal colour.
	We remark that, as $\acd{\TS}$ is a "parity ACD", each edge $e$ appears in one and only one branch of $\acd{\TS}$, and that $\gg_{\acdNoP}(e)$ equals the "depth" of the "deepest" node containing $e$. In particular, if $e\in\nuAcd(n)$ for some node $e$, $\gg_{\acdNoP}(e)\geq \depth(n)$. Our objective is to show that a similar result holds for the path $\rr$ as a set of edges:
	\begin{claim}\label{claim-normF:priority-path}
		Let $N_{\rr}$ be the set of nodes of $\acd{\TS}$ containing the edges of the path $\rr$ in their label, that is, $N_{\rr} = \{n\in \nodesAcd{\TS} \mid \{e_1,\dots, e_k\}\subseteq \nuAcd(n)\}$.
		Then, $\min (\gg_{\acdNoP}(\rr))$ equals the "depth" of a node of maximal "depth" of $N_\rr$.\footnotemark{}
	\end{claim}
	\footnotetext{In fact, the nodes of $N_\rr$ are totally ordered by the "ancestor" relation, so there is a unique node of maximal "depth" in $N_{\rr}$. This fact is not used in our proof.}
	This claim allows us to conclude. Indeed, let $n$ be a node of maximal "depth" of $N_\rr$, verifying $\depth(n)=d$. Then, $\nuAcd(n)$ is a "cycle" containing the vertices $v$ and $v'$, and for all the edges $e\in \nuAcd(n)$, $\gg_{\acdNoP}(e)\geq \depth(n)=d$. This provides the desired path from $v'$ to $v$.
	
	\begin{subproof}[Proof of Claim~\ref{claim-normF:priority-path}]
		First, we remark that if $\ell_1, \ell_2,\dots, \ell_k$ are "cycles" such that $\ell_i$ and $\ell_{i+1}$ have some state in common, then $\cup_{i=1}^k \ell_i$ is a "cycle". Let $n$ be a node of maximal "depth" in $N_\rr$. By the previous remarks, $\gg_{\acdNoP}(e)\geq \depth(n)$. Suppose by contradiction that $\gg_{\acdNoP}(\rr)> \depth(n)$. Then, each edge $e_i$ of $\rr$ would appear in some strict descendant $n_i$ of $n$ (we can assume that $n_i$ is a "child" of $n$). Then, $\nuAcd(n_1), \dots, \nuAcd(n_k)$ would be "cycles" such that $\nuAcd(n_i)$ and $\nuAcd(n_{i+1})$ have some state in common (namely, $\mtarget(e_i) = \msource(e_{i+1})$), so their union is a "cycle". However, this is not possible in a "parity" "transition system", as $\nuAcd(n)$ is "accepting@@cycle" if and only if each of the $\nuAcd(n_i)$ is "rejecting@@cycle" (see Lemma~\ref{lemma-ZT:union_cycles_parity-TS}).
	\end{subproof} 
  This completes the proof of Proposition~\ref{prop-normF:paths_form_parity}.
\end{proof}

\begin{proposition}[Normal flowers do not lack petals]\label{prop-normF:flowers-in-normal_form}
	Let $v$ be a state of a "parity" "transition system" in "normal form" belonging to an "accepting@@SCC" (resp. "rejecting@@SCC") "SCC". Let  $\ell\in \cyclesState{\TS}{v}$ be a "cycle" over $v$ and let $d_\ell$ be the minimal colour appearing in it. 
	\begin{itemize}
		\item If $\TS$ is "not negative@@TS", for each $x\in [0, d_\ell]$ (resp. $x\in [1, d_\ell]$) there is a cycle $\ell_x \in \cyclesState{\TS}{v}$ producing $x$ as minimal colour.
		\item If $\TS$ is "negative@@TS", for each $x\in [2, d_\ell]$ (resp. $x\in [1, d_\ell]$) there is a cycle $\ell_x \in \cyclesState{\TS}{v}$ producing $x$ as minimal colour.
	\end{itemize}
	
\end{proposition}
\begin{proof}
	We do the proof for the case in which $\TS$ is "not negative@@TS" and $v$ belongs to an "accepting SCC".
	By Proposition~\ref{prop-normF:existance+uniqueness_normal_form}, the "colouring" of $\TS$ is the one given by its "ACD-transform", noted $\gg_{\acdNoP}$. Consider the "local subtree at $v$", $\treeVertex{v}$, consisting in a single "branch", as it has "parity shape" (Proposition~\ref{prop-typ:parity ACD type}). Let $n_0\ancestor\dots \ancestor n_k$ be that branch, and let $n_i$ be the "deepest" node such that $\ell\subseteq \nuAcd(n_i)$. We remark that, by definition of $\gg_{\acdNoP}$, $d_\ell = \depth(n_i) = i$. The desired "cycles" are obtained by taking $\ell_x = \nuAcd(n_x)$, for $x\in [0, d_\ell]$.
\end{proof}

The next theorem states a simple characterisation of "transition systems" in "normal form". It provides a useful tool to show "normality" of "parity" "TS" in many proofs. In essence, it shows that the two previous propositions characterise the "normal form". We state it for "non-negative@@TS" "transition systems" for simplicity; a similar characterisation for "negative@@TS" "transition systems" is immediate.

\AP We say that an "SCC" of a parity "TS" is ""positive@@SCC"" if the minimal "colour" appearing on it is even, and that it is ""negative@@SCC"" if this minimal "colour" is odd. 

\begin{theorem}\label{th-normF:charactNormForm}
	A "non-negative@@TS" "parity" "transition system" is in "normal form" if and only if:
	\begin{itemize}
		\item transitions changing of "SCCs" are coloured  $0$, and
		\item if $v$ and $v'$ belong to a same "positive@@SCC" (resp. "negative@@SCC") "SCC" and there is a transition $v\re{d} v'$ producing "colour" $d>0$ (resp. $d>1$), then there are two paths  $v'\lrpE v$ producing as minimal "colour" $d$ and $d-1$, respectively.
	\end{itemize}
\end{theorem}
\begin{proof}
	The fact that a "TS" in "normal form" satisfies these properties follows from the previous propositions.
	
	Let $\TS$ be a "TS" satisfying these properties and using $\gg$ as "colouring function@@TS". Let $e=v\re{d}v'$ be an edge with $\gg(e)=d$. We will show that for any other "equivalent@@condition" "colouring@@TS" $\gg'$, we have $\gg'(e)\geq d$. 
	This is trivial if $d=0$.
	If $d>0$, $v$ and $v'$ must be in the same "SCC", that we assume "positive@@SCC" without loss of generality.
	By hypothesis, we can close "cycles" $\ell_d$ and $\ell_{d-1}$ over $v$ producing $d$ and $d-1$ as minimal "colour", respectively. 
	Cycle $\ell_{d-1}$ can be decomposed in $v\re{} v'\lrpE v_1 \re{d-1}v_1' \lrpE v$. Applying the hypothesis over the edge  $v_1 \re{d-1}v_1'$ gives a path $v_1'\lrpE v_1$ producing $d-2$ as minimal "colour", which can be merged with 
	$\ell_{d-1}$ to produce a cycle $\ell_{d-2}$ over $v$ producing $d-2$ as minimal "colour".
	Iterating this process, we can find "cycles" $\ell_0\supseteq \ell_1\supseteq\dots, \supseteq \ell_{d}$ over $v$ such that $\ell_i$ produces $i$ as minimal "colour". 
	Taking $\ell'_i = \cup_{j=i}^d\ell_i$, we obtain a "positive@@flower" "$(d+1)$-flower" $\ell_0'\supsetneq\ell'_1\supsetneq\dots \supsetneq \ell'_d$, so by Lemma~\ref{lemma-normF:flowers-lower-bound-in-parity} we conclude that $\gg'(e)\geq d$.
\end{proof}

\paragraph*{Parity index from automata in normal form.}

The next definition constitutes a syntactic version of the "parity index", defined at the level of "parity" "transition systems". 
The following results establish the tight relation between the semantic notion of "parity index" and its syntactic counterpart, and state that the "parity index" of a language can be directly read from a "DPA" in "normal form".


\begin{definition}\label{def:parity_tight_TS}
	\AP We say that a "parity" "transition system" $\TS = (\underlyingGraph{\TS}, \macc{\TS})$ is ""parity-index-tight"" if any other "parity condition" $\macc{}'$ over $\underlyingGraph{\TS}$ such that $\macc{}' \equivCond{\underlyingGraph{\TS}} \macc{\TS}$ uses at least as many "colours" as $\macc{\TS}$.
\end{definition}

We have shown in Corollary~\ref{cor-normF:ACD-transform_normal_form} that the "ACD-parity-transform" is always in "normal form". Therefore, the optimality properties of the "colouring" of $\acdParityTransform{\TS}$ (Theorem~\ref{thm-acd:optimality-priorities_ACD-parity_transform}) transfer to "parity" "transition systems" in "normal form".

\begin{corollary}\label{lemma-normF:normal_form-implies-strong_parity_index}
	A "parity" "transition system" in "normal form" is "parity-index-tight".
\end{corollary}

Moreover, the "parity index" of an "$\oo$-regular language" can be read from any "DPA" in "normal form" "recognising" it.

\begin{corollary}
	Let $\A$ be a "deterministic" "parity" "automaton" in "normal form" such that all its states are "accessible". 
	If $\A$ uses "colours" in $[0,d-1]$ (resp. $[1,d]$), then the "parity index" of $\Lang{\A}$ is $\WeakIndex{d-1}$ or $[0,d-1]$ (resp. $\WeakIndex{d-1}$ or $[1,d]$).
\end{corollary}

We prove this result in Appendix~\ref{sec:appendix-weak-conditions} (Corollary~\ref{cor-app-weak:normal-form-index-par}), and we provide there a refined characterisation by using "generalised weak automata".

\subsection{Minimisation of deterministic parity automata recognising Muller languages}\label{subsec-corollaries: minimisation-parity}
The minimisation of $\oo$-automata is a fundamental problem of an intriguing complexity. In 2010, Schewe showed that the minimisation of "deterministic" "B\"uchi" and "parity" "automata" is $\NPcomplete$, if the "acceptance condition" is defined over the states~\cite{Schewe10MinimisingNPComplete}. However, the reduction of $\NP$-hardness does not generalise to automata with edge-based acceptance. A surprising positive result was obtained in 2019 by Abu Radi and Kupferman: we can minimise in polynomial time "HD" "coB\"uchi" automata using transition-based acceptance~\cite{AK22MinimizingGFG}. Schewe showed that the minimisation was again $\NP$-hard for "HD" automata with state-based acceptance~\cite{Schewe20MinimisingGFG}. To the best of our knowledge, the only existing hardness result applying to transition-based automata is Casares' result about the $\NP$-completeness of the minimisation of "deterministic" "Rabin" "automata"~\cite{Casares2021Chromatic}. In fact, in~\cite{Casares2021Chromatic} a stronger result is proven: it is $\NP$-hard to minimise "deterministic" "Rabin" "automata" "recognising" "Muller languages".

In this section, we provide a polynomial-time algorithm for the minimisation of "DPA" "recognising" "Muller languages" (with "acceptance condition" over transitions). By Proposition~\ref{prop-ZT:correctness_ZT_parity} and Theorem~\ref{thm-zt:strong_optimality_ZTparity}, we know that a minimal "(history-)deterministic" "parity" "automaton" "recognising" a "Muller language" $L = \MullerC{\F}{\SS}$ can be constructed in linear time from the "Zielonka tree" $\zielonkaTree{\F}$. 
We will therefore provide a polynomial-time algorithm computing this "Zielonka tree" from a "DPA" "recognising"~$L$.

\begin{theorem}\label{thm-min:minimisation_parity_automata}
	Let $\A$ 
	be a "DPA"  "recognising@@automaton" a "Muller language" $L = \MullerC{\F}{\SS}$. We can find a minimal "deterministic" (resp. "history-deterministic") "parity" "automaton" recognising $L$ in polynomial time in the size of the representation of $\A$.\footnotemark{}
	\footnotetext{We can assume that the representation of $\A$ has size polynomial in $|Q|+|\SS|$, where $Q$ and $\SS$ are the set of states and the "input alphabet" of $\A$. Indeed, as $\A$ is "deterministic" the number of transitions is at most $|Q|\cdot|\SS|$, and we can assume that $\A$ has no more "output colours" than transitions.}
\end{theorem}

\subparagraph*{Description of the algorithm.}
Let $\A=(Q,\SS,q_0,\GG,\DD,\parity)$  be a "DPA"  "recognising@@automaton" $L = \MullerC{\F}{\SS}$. We outline a recursive algorithm  building $\zielonkaTree{\F} = (N, \ancestor, \nu)$ in a top-down fashion; it starts from the "root" of the "tree" (which is always labelled $\SS$), and each time that some node is added to $N$, we compute its "children".
If we have built $\zielonkaTree{\F}$ up to a node $n$, we compute the "children" of $n$ by using the procedure $\AlternatingSets$ described in Algorithm~\ref{algo: ChildrenNode}, which we disclose next. 

We assume without loss of generality that $n$ is "round", that is, $\nu(n)\in \F$. First, we take the restriction of $\A$ to transitions labelled with letters in $\nu(n)$ and pick a "final SCC" on it. Such "final SCC" induces a "subautomaton" $\A'$ of $\A$ "recognising" $\MullerC{\restSubsets{\F}{\nu(n)}}{\nu(n)}$ (see also Lemma~\ref{lemma-ZT:X-FSCC-induce-automata}).
Our objective is to find the maximal "subautomata" of $\A'$ using as "input letters" sets $X\subseteq \nu(n)$ such that $X\notin \F$. We will keep all such subsets $X$ in a list $\mathsf{altSets}$. The labels of the "children" of $n$ will then correspond to the maximal sets appearing  in this list, which are returned by the algorithm $\AlternatingSets$ (Line~\ref{line AS : Return}).
In order to find them, we remove the transitions using the minimal colour in $\A'$ (that is even, since $\nu(n)\in \F$) and compute a decomposition in "strongly connected components" of the obtained graph.
Let $\S$ be a component of this decomposition and let $\SS_\S\subseteq \nu(n)$ be the "input letters" appearing in it. Then, $\SS_\S\notin \F$ if and only if  the minimal "output colour" in $\S$ is odd (see Lemma~\ref{lemma-minm:SCC-determines-acceptance} below). In this case, we add $\SS_\S$ to $\mathsf{altSets}$. On the contrary, we remove the minimal (even) colour from $\S$, and we start again finding a decomposition in "SCCs" of the obtained graph.

We include the pseudocode for the procedure $\AlternatingSets$ in Algorithm~\ref{algo: ChildrenNode}. We use the following notations:
\begin{itemize}
	\item \AP $\intro*\mletters(\S)$ is the set of "input letters" appearing in $\S$,
	\item \AP $\intro*\minCol(\S)$ is the minimal "output colour" appearing in  $\S$ (which determines whether $\mletters(\S)\in \F$, if $\S$ is "strongly connected"),
	\item \AP $\intro*\SCCDec(\A)$ outputs a list of the "strongly connected components" of~$\A$. If $\A$ is empty, it outputs an empty list.
	\item \AP $\intro*\maxInclusion(\mathsf{lst})$ returns the list of the maximal subsets in $\mathsf{lst}$.
\end{itemize}

\begin{algorithm}[ht]
	\caption{ $\intro*\AlternatingSets(\A)$: Computing the children of a node.}
	\label{algo: ChildrenNode}
		\KwIn{A "strongly connected" "automaton" $\A$ over $\SS$; $\Lang{\A}=\Muller{\F}$} 
		\KwOut{Maximal subsets $\SS_1,\dots, \SS_k\subseteq \SS$ such that $\SS_i\in \F \iff \SS\notin \F$}
		$d \leftarrow \minCol(\A)$\;
		 $\A_{>d} \leftarrow $ restriction of $\A$ to transitions $\DD_{>d}=\{ q\re{a:x}q'\in \DD \mid x>d\}$\;
		 $\langle \S_1, \dots, \S_r\rangle \leftarrow \SCCDec(\A_{>d})$ \;
		 $\mathsf{altSets}\leftarrow \{ \}$ \;
		\For{$i = 1,\dots, r$}{\label{line AS : } 
		\If{$\minCol(\S_i)$ is odd if and only if $d$ is even}{
                 $\mathsf{altSets} \leftarrow  \mathsf{altSets} \, \cup \{ \mletters(\S_i) \}$\; \label{line AS : addSubset}
                  }
		\Else{
                 $\mathsf{altSets} \leftarrow  \mathsf{altSets} \, \cup \, \AlternatingSets(\S_i)$\; \label{line AS : RecursiveAltenation}
                }
                }
		 $\mathsf{maxAltSets} \leftarrow \maxInclusion(\mathsf{altSets})$ \label{line AS : MaxLetters}\;
		 \Return{$\mathsf{maxAltSets}$} \label{line AS : Return}\;
\end{algorithm}

\subparagraph*{Correctness of the algorithm.}
Let $n$ be a node of the "Zielonka tree" of $\F$ labelled with $\nu(n)$, and let $\A_n$ be an "accessible" "subautomaton" of $\A$ over $\nu(n)$ "recognising" $\MullerC{\restSubsets{\F}{\nu(n)}}{\nu(n)}$.
We prove that $\AlternatingSets(\A_n)$ returns a list of sets corresponding to the labels of the "children" of~$n$ in $\zielonkaTree{\F}$. We assume without loss of generality that $\nu(n)\in \F$ and therefore the minimal colour~$d$ in $\A_n$ is even.

First, we observe that if $X\subseteq \SS$ is added to $\mathsf{altSets}$ during the execution of the procedure $\AlternatingSets$, then $X$ is the set of input letters appearing in a "cycle" whose minimal colour is odd.
Next lemma implies that in this case, $X\notin \F$.
In particular, no subset is added if $n$ is a "leaf" of $\zielonkaTree{\F}$.

\begin{lemma}\label{lemma-minm:SCC-determines-acceptance}
	Let $\A$ be a "DPA" such that $\Lang{\A} = \MullerC{\F}{\SS}$.
	Let $\ell\in \cycles{\A}$ be an "accessible" "cycle" of $\A$. Let $\SS_\ell\subseteq \SS$ be the "input letters" appearing in $\ell$, and let $d_\ell$ be the minimal colour on $\ell$. Then, $\SS_\ell \in \F$ if and only if $d_\ell$ is even. 
\end{lemma}
\begin{proof}
	Since $\ell$ is an "accessible" "cycle", there is a word $w\in \SS^\oo$ such that $\minf(w) = \SS_\ell$ and verifying that the edges visited infinitely infinitely often by the (only) "run over" $w$ in $\A$ are the edges of $\ell$. Therefore $w\in \Lang{\A}$ if and only if $d_\ell$ is even, and since $\Lang{\A}$ is a "Muller language", $w\in \Lang{\A}$ if and only if $\minf(w) = \SS_\ell \in \F$.
\end{proof}

As the final output of the algorithm consists solely on the maximal subsets in $\mathsf{altSets}$, and no accepting set is added to this list, it suffices to show that each maximal rejecting subset $\SS_{\max}\subseteq \nu(n)$ is added to $\mathsf{altSets}$ at some point.

Let $\SS_{\max}\subseteq \nu(n)$ be one of the maximal rejecting subsets of $\nu(n)$. Let $\S$ be a "final SCC" of the restriction of $\A_n$ to transitions labelled with letters in $\SS_{\max}$ (by the previous lemma, $\minCol(\S)$ is odd).  We show that $\SS_{\max}$ will eventually be considered by the recursive procedure $\AlternatingSets$, and therefore $\SS_{\max}$ will be added to $\mathsf{altSets}$ 
We use of the following remark:

\begin{claim}
	If $\S'$ is a "strongly connected" "subautomaton" of $\A_n$ such that $\S\subsetneq \S' \subseteq \A_n$, then the minimal colour in $\S'$ is even.
\end{claim}
\begin{subproof}
	Let $\SS'$ be the "input letters" appearing in $\S'$.
	As $\S\subsetneq \S'$ and no transition labelled with a letter in $\SS_{\max}$ leaves $\S$, we must have $\SS_{\max}\subsetneq \SS'$. The claim follows from Lemma~\ref{lemma-minm:SCC-determines-acceptance}.
\end{subproof}

Therefore, either $\S$ is one of the "SCCs" of $\A_{>d}$ (in this case, $\SS_{\max}$ is added to $\mathsf{altSets}$ in Line~\ref{line AS : addSubset}), or it is contained in one "SCC" of $\A_{>d}$ whose minimal colour is even, and we can conclude by induction.

\subparagraph{Complexity analysis.}
We will show that the proposed algorithm works in time $\O(|Q|^3|\SS|^2|\GG|)$, where $Q$, $\SS$ and $\GG$ are the states, set of "input letters" and set of "output colours" of the automaton, respectively.
We remark that, since $\A$ is "deterministic", $|\DD|\leq |Q||\SS|$.

First, we study the complexity of the procedure $\AlternatingSets(\A)$. At each recursive call, at least one edge is removed from $\DD$, and a decomposition in "strongly connected components" of the "automaton" is performed, which can be done in $\O(|Q||\SS|)$~\cite{Tarjan72DepthFirst}. Therefore, the "children" of a node of the "Zielonka tree" can be computed in $\O(|Q|^2|\SS|^2)$.

We perform this operation for each node of the "Zielonka tree". By the optimality of the "ZT-parity-automaton" (Theorems~\ref{thm-zt:optimality_ZTparity-priorities} and~\ref{thm-zt:strong_optimality_ZTparity}), we know that $|Q| \geq |\leaves(\zielonkaTree{\F})|$ and that the "height" of $\zielonkaTree{\F}$ is at most $|\GG|$. Therefore, $|\zielonkaTree{\F}| \leq |Q||\GG|$, and the procedure $\AlternatingSets$ is called at most $|Q||\GG|$ times. We conclude that the proposed algorithm works in time $\O(|Q|^3|\SS|^2|\GG|)$.

\begin{remark}[State-based automata]
    The acceptance condition of the parity automaton obtained from the "Zielonka tree" appears naturally over the transitions of the automaton. In order to make it a state-based automaton, we would need to add one state per "colour" it uses. It turns out that, in this specific case, this is optimal, and the state-based parity automaton we obtain is minimal. Therefore, we can also minimise in polynomial time state-based parity automata "recognising" "Muller languages".
    However, it is no longer possible to obtain optimal transformations towards state-based parity automata based on the "ACD" (see~\cite[Section~5.3]{CDMRS22Tacas} and \cite[Section~I.8]{Casares23Thesis} for further details).
\end{remark}

	\section{Conclusion}\label{section:conclusion}
	In this work, we have carried out an extensive study of transformations of "automata" and "games" that use "Muller" "acceptance conditions". We have proposed different types of "morphisms@@TS" to formalise the idea of valid transformations of "transition systems", which distil the central features of existing transformations. Our main contribution resides in the introduction of a new structure, the "alternating cycle decomposition", which is a succinct representation of the \emph{alternating chains of loops} of a "Muller" "automaton" -- in the sense of Wagner~\cite{Wagner1979omega} -- and provides the necessary information to understand the interplay between its "acceptance condition" and its "underlying graph".


\subparagraph*{Optimal and practical transformations of automata.} We have presented a transformation that, given a "deterministic" "Muller" "automaton",  provides an equivalent "deterministic" "parity" "automaton", and another that provides an equivalent "history-deterministic" "Rabin" "automaton". These transformations are optimal in a strong sense; the obtained automata have a minimal number of states amongst those which accept a "history-deterministic mapping" to the original "Muller" "automaton".
The first of these transformations has been implemented in the open-source tools Spot~2.10~\cite{Spot2.10CAV22} and Owl~21.0~\cite{KMS18Owl}, and it has been shown to perform extremely well in practice~\cite{CDMRS22Tacas}, as the natural definition of the "ACD" provides a fairly efficient way to compute the transformation, while its optimality guarantees to produce "automata" as small as possible.

\subparagraph*{Understanding the limitations of history-deterministic automata.} As a corollary of our results, we have obtained that minimal "deterministic" and "history-deterministic" "parity" "automata" "recognising" "Muller languages" have the same size (Corollary~\ref{cor-zt:no-small-HD-for-Muller}). Moreover, we have shown that "HD" "parity" "automata" that are strictly smaller than equivalent "deterministic" ones cannot come from a "deterministic" "Muller" "automaton" (Corollary~\ref{cor-acd:HD-transformations-are-big}). This provides a partial explanation on the difficulty to find succinct "HD" "parity" "automata", as we could argue that a simple way to conceptualise "$\oo$-regular languages" is through "deterministic" "Muller" "automata". 
Maybe most importantly, this sets a limitation in the usefulness of "history-determinism" in practice, as procedures that use a "DMA" as an intermediate step -- as the ones from the tools Strix~\cite{LMS20SynthesisLTL} and \texttt{ltlsynt}~\cite{MichaudColange18Synt}, or  automata determinisation~\cite{Piterman2006fromNDBuchi,Schewe2009tighter,LodingP19} -- cannot benefit from the succinctness of "HD automata".

On the other hand, we have shown that, if our objective is to obtain "Rabin" "automata" as output, the "ACD-HD-Rabin-transform" allows us to benefit from succinct "HD automata". In this case, it has been shown that these automata can be exponentially smaller than equivalent "deterministic" ones~\cite[Theorem~21]{CCL22SizeGFG}.

\subparagraph*{Disclosing the structure of $\oo$-automata.} As an application of the insights gained from the "alternating cycle decomposition", we have derived results concerning "typeness" of "automata". In particular, we have characterised when we can define a "parity", "Rabin" or "Streett" "condition@@acc" on top of a "Muller" "automaton", obtaining an "equivalent@@aut" "automaton" (Propositions~\ref{prop-typ:Rabin ACD type}, \ref{prop-typ:Streett ACD type} and~\ref{prop-typ:parity ACD type}). These characterisations have already been proven instrumental in works about the "memory@@games" for "games"~\cite{Casares2021Chromatic}, and to obtain lower bounds on the size of "deterministic" "Rabin" automata~\cite{CCL22SizeGFG}.

We have also employed the "ACD" to present a "normal form" for "parity" "transition systems" and systematically proved the most important properties that make this form a valuable tool for manipulating "parity" "automata". We believe that this "normal form" will be useful to extend existing results about "B\"uchi" and "coB\"uchi" automata (as the ones in~\cite{AK22MinimizingGFG,BKS17HowDeterministicGFG, BCRV22HalfPosBuchi}) to "parity" automata.
 	
    \printbibliography
	
	\newpage
	\appendix
	\section{Generalised classes of acceptance conditions}
	\label{sec:appendix-weak-conditions}
	\paragraph*{Further acceptance conditions.}

\begin{description}	
	\item[Generalised B\"uchi.] \AP Given $k$ non-empty subsets $B_1,\dots, B_k\subseteq \GG$, we define the ""generalised Büchi language associated to"" $B=\{B_1,\dots,B_k\}$ as
	\[ \intro*\genBuchiC{B}{\GG}  = \{w\in \GG^\oo \mid \minf(w) \cap B_i \neq \emptyset \text{ for all } i\in\{1,\dots, k\}\}. \]
	\AP We say that a language $L\subseteq \GG^\oo$ is a ""generalised B\"uchi language"" if there is a family of sets $B=\{B_1,\dots,B_k\}$ such that $L = \genBuchiC{B}{\GG}$.
	\item[Generalised coB\"uchi.] \AP Given $k$ non-empty subsets $B_1,\dots, B_k\subseteq \GG$, we define the ""generalised coB\"uchi language associated to"" $B=\{B_1,\dots,B_k\}$ as
	\[ \intro*\gencoBuchiC{B}{\GG}  = \{w\in \GG^\oo \mid \minf(w) \cap B_i = \emptyset \text{ for some } i\in\{1,\dots, k\}\}. \]
	\AP We say that a language $L\subseteq \GG^\oo$ is a ""generalised coB\"uchi language"" if there is a family of sets $B=\{B_1,\dots,B_k\}$ such that $L = \gencoBuchiC{B}{\GG}$.
\end{description}

\begin{remark}\label{rmk-app:genBuchi-expressive-power}
	"Deterministic" "generalised B\"uchi" (resp. "generalised coB\"uchi") "automata" have the same expressive power than "deterministic" "B\"uchi" (resp. "coB\"uchi") "automata": they "recognise" languages of "parity index" at most $[0,1]$ (resp. $[1,2]$).
\end{remark}

\AP We will also define ""conditions that depend on the structure"" of the "transition system" and not only on the set of colours.
\begin{description}	
	\item[Generalised weak transition systems.]\AP Let $\TS =(\underlyingGraph{\TS}, \macc{\TS})$ be a "transition system" using a "parity condition" $\macc{\TS}=(\gg, [d_{\min},d_{\max}], \parity)$. We say that $\TS$ is $\intro*\Weak{d}$ if in each "strongly connected component" $\S \subseteq \underlyingGraph{\TS}$ there are at most $d$ different colours that appear, that is, $|\gg(E_\S)| \leq d$, where $E_\S$ is the set of edges of $\S$.	
\end{description}

\AP As for the rest of conditions, we say that a "transition system" $\TS$ is ""$\Weak{d}$ type"" if there exists an "isomorphic@@TS" "parity" "transition system" $\TS'\equivTrans \TS$ that is $\Weak{d}$.

The adjective \textit{Weak} has typically been used to refer to the condition corresponding to a partition of $\TS$ into accepting and rejecting "SCC". A "run" will be "accepting@@run" if the component it finally stays in is accepting. It corresponds to $\Weak{1}$ with our notation.

As we will show (Corollary~\ref{cor-app:weak-automata-weaktype}), the notation is justified by the fact that an "$\oo$-regular language" of "parity index" $\WeakIndex{d}$ can be "recognised" by a "deterministic" $\Weak{d}$ automaton.

\paragraph*{The Zielonka tree of generalised acceptance conditions.}
\begin{definition}
	Let $T$ be a "tree" $T$ with nodes partitioned into "round" nodes and "square" nodes. We say that $T$ has:
	\begin{itemize}
		\item ""B\"uchi shape"" if it has a single branch, "height" at most $2$, and if it has "height" $2$ its "root" is "round".
		\item ""coB\"uchi shape"" if it has a single branch, "height" at most $2$, and if it has "height" $2$ its "root" is "square".
		
		\item  \AP""Generalised B\"uchi shape"" if it has "height" at most $2$, and if it has "height" $2$ its "root" is "round".
		
		\item  \AP""Generalised coB\"uchi shape"" if it has "height" at most $2$, and if it has "height" $2$ its "root" is "square".
	\end{itemize}
\end{definition}

\begin{proposition}
	Let $\F\subseteq \powplus{\Gamma}$ be a family of non-empty subsets. Then $\MullerC{\F}{\GG}$ is a "B\"uchi@@language" "(resp. coB\"uchi) language" if and only if $\zielonkaTree{\F}$ has "B\"uchi@@shapeZT" "(resp. coB\"uchi) shape".
\end{proposition}
\begin{proof}
	This is just a special case of Proposition~\ref{prop-typ:parityZielonkaShape}.
\end{proof}

\begin{proposition}\label{prop-app:genBuchi-shape-ZT}
		Let $\F\subseteq \powplus{\Gamma}$ be a family of non-empty subsets. Then $\MullerC{\F}{\GG}$ is a "generalised B\"uchi@@language" "(resp. coB\"uchi) language" if and only if $\zielonkaTree{\F}$ has "generalised B\"uchi@@shapeZT" "(resp. generalised coB\"uchi) shape".
\end{proposition}
\begin{proof}
	We do the proof for the case "generalised B\"uchi@@language" (symmetric for generalised coB\"uchi).
	Assume that $\MullerC{\F}{\GG} =\genBuchiC{B}{\GG}$  for some family $B=\{B_1,\dots,B_k\}$. Then, $\GG\in \F$, as $\GG\cap B_i \neq \emptyset$, so the "root" of $\zielonkaTree{\F}$ is "round". If $C\subseteq \GG$ is rejecting, $C\cap B_i =\emptyset$ for all $i$, then it is the same for any subset $C'\subseteq C$, so "square" nodes of $\zielonkaTree{\F}$ are "leaves" and $\zielonkaTree{\F}$ has "height" at most $2$.
	
	Conversely, assume that $\zielonkaTree{\F}$ has "height" $2$ and that its "root" is "round" ($\GG\in \F$). Let $A_1,\dots, A_k$ be the labels of the $k$ "leaves" of $\zielonkaTree{\F}$ and define $B_i = A_i$. We claim that $\MullerC{\F}{\GG} =\genBuchiC{B}{\GG}$, for $B=\{B_1,\dots,B_k\}$. Indeed, if $C\in \F$ if and only if $C\nsubseteq A_i$  for any $i$ if and only if $C\cap B_i \neq \emptyset$ for all $i$.
\end{proof}

\begin{corollary}
	Let $\A$ 
	be a "deterministic" "generalised Büchi" (resp. "generalised coB\"uchi") "automaton"  "recognising@@automaton" a "Muller language" $L = \MullerC{\F}{\SS}$. There is a "deterministic""generalised Büchi" (resp. "generalised coB\"uchi") "automaton" recognising $L$ with just one state, that can be computed in polynomial time in the size of the representation of $\A$.
\end{corollary}
\begin{proof}
	We do the proof for the case "generalised Büchi".
	By Remark~\ref{rmk-app:genBuchi-expressive-power}, the "parity index" of~$L$ is at most $[0,1]$, so by Proposition~\ref{prop-typ:parity index-Muller-automaton}, the "Zielonka tree" of $\F$ has "generalised B\"uchi shape". Therefore, by Proposition~\ref{prop-app:genBuchi-shape-ZT}, $L$ is a "generalised B\"uchi@@language" that can be trivially "recognised" by a "generalised B\"uchi" "automaton" with just one state.
	
	The acceptance condition of such automaton can be deduced in linear time from the "Zielonka tree" $\zielonkaTree{\F}$, as indicated in the proof of Proposition~\ref{prop-app:genBuchi-shape-ZT}. The "Zielonka tree" $\zielonkaTree{\F}$ can be computed from the original automaton $\A$ using a similar argument than in the proof of Theorem~\ref{thm-min:minimisation_parity_automata}. Suppose that the "generalised Büchi" condition used by $\A$ is given by the sets $B_1,\dots, B_k\subseteq \GG$. Then, for each $i\in \{1,\dots,k\}$ we compute the restriction of $\A$ to the transitions using colours in $\GG\setminus B_i$, and perform a decomposition in "SCCs" of the obtained "graph". If $\SS_\S\subseteq \SS$ is the set of "input letters" appearing in one of those "SCC", then $\SS_\S\notin \F$. We put all the subsets of letters obtained in that way in a list $\mathit{altSets}$. The "leaves" of $\zielonkaTree{\F}$ correspond then to the maximal subsets of $\mathit{altSets}$.	
\end{proof}

\paragraph*{ACD and typeness for generalised acceptance conditions.}

\begin{definition}
	Let $\TS$ be a "Muller" "transition system" with a set of states $V$. We say that its "alternating cycle decomposition" $\acd{\TS}$ is a:
	\begin{itemize}
		\item \AP ""B\"uchi ACD"" if it is a "$[0,1]$-parity ACD".
		\item \AP ""coB\"uchi ACD"" if it is a "$[1,2]$-parity ACD".
		\item  \AP""Generalised B\"uchi ACD"" if for every state $v\in V$, the tree $\treeVertex{v}$ has "generalised B\"uchi shape".
		\item  \AP""Generalised coB\"uchi ACD"" if for every state $v\in V$, the tree $\treeVertex{v}$ has "generalised coB\"uchi shape".
		\item \AP ""$\Weak{d}$ ACD"" if it is a "parity ACD" and "trees" of $\acd{\TS}$ have "height" at most $d$.
	\end{itemize}
\end{definition}

\begin{remark}\label{rmk-typ:shape-weak-ACD}
	 $\acd{\TS}$ is a "$\Weak{d}$ ACD" if and only if it is a "$[0,d]$-parity ACD" and a "$[1,d+1]$-parity ACD".
\end{remark}

\begin{proposition}
	A "transition system" $\TS$ is "B\"uchi@@type" "(resp. coB\"uchi) type" if and only if $\acd{\TS}$ is a "B\"uchi ACD" (resp. "coB\"uchi ACD").
\end{proposition}
\begin{proof}
	This is a special case of Proposition~\ref{prop-typ:parity ACD type}.
\end{proof}

\begin{proposition}\label{prop-app:genBuchi-ACD type}
	A "transition system" $\TS$ is "generalised B\"uchi@@type" "(resp. generalised coB\"uchi) type" if and only if $\acd{\TS}$ is a "generalised B\"uchi ACD" (resp. "generalised coB\"uchi ACD").
\end{proposition}
\begin{proof}
	The result follows by applying the same argument and construction than in Proposition~\ref{prop-app:genBuchi-shape-ZT}, using as set of "output colours" the set of edges of $\TS$.
\end{proof}

\begin{proposition}\label{prop-app:weak-ACD-shape}
	A "transition system" $\TS$ is "$\Weak{d}$ type" if and only if $\acd{\TS}$ is a "$\Weak{d}$ ACD".
\end{proposition}
\begin{proof}
	Proposition~\ref{prop-typ:parity ACD type} already provides that $\TS$ is "parity type" if and only if $\acd{\TS}$ is a "parity ACD". As in the proof of the aforementioned proposition, we observe that $\TS$ and $\acdParityTransform{\TS}$ are "isomorphic@@TS". If $\acd{\TS}$ is a "$\Weak{d}$ ACD", then $\acdParityTransform{\TS}$ is $\Weak{d}$. Conversely, if $\acd{\TS}$ is not a "$\Weak{d}$ ACD", then $\TS$ contains a "$(d+1)$-flower", so the number of colours cannot be reduced (using the same argument as in the proof of Theorem~\ref{thm-acd:optimality-priorities_ACD-parity_transform}).
\end{proof}

\begin{corollary}\label{cor-app:weak-automata-weaktype}
	If $L\subseteq \SS^\oo$ is an "$\oo$-regular language" of "parity index" $\WeakIndex{d}$, then $L$ can be "recognised" by a "deterministic" $\Weak{d}$ automaton.
\end{corollary}
\begin{proof}
	Let $L$ be of "parity index" $\WeakIndex{d}$. By definition, $L$ is "recognised" by "parity" "automaton"~$\A$ using colours in $[0,d]$ (it is also recognised by an automaton using colours in $[1,d+1]$; we make an arbitrary choice). We will prove that $\A$ is in fact $\Weak{d}$. We will show that its "ACD" $\acd{\A}$ is "$\Weak{d}$ type", which allows to conclude by Proposition~\ref{prop-app:weak-ACD-shape}.
	Suppose that this was not the case, that is, that some "tree" of $\acd{\A}$ has "height" at least $d+1$. In this case, $\A$ would contain a "$(d+1)$-flower" (Lemma~\ref{lemma-acd:branch-ACD-induces-flower}), so by the Flower Lemma~\ref{lemma:flower-lemma}, $L$ has "parity index at least" $[0,d]$ or $[1,d+1]$, a contradiction.
\end{proof}

The following result generalises~\cite[Theorem~2]{BSW01Weakautomata}.
\begin{corollary}
	A "Muller" "transition system" is "$\Weak{d}$ type" if and only if it is both $[0,d]$ and "$[1,d+1]$-parity type".
\end{corollary}

\paragraph*{Deterministic automata using generalised acceptance conditions.}

\begin{corollary}
	Let $\underlyingGraph{\A}$ be the "underlying graph" of a "deterministic automaton". There are $[0,d-1]$ and "$[1,d]$-parity" conditions $\macc{p,0}$ and $\macc{p,1}$ such that $\Lang{\underlyingGraph{\A},\macc{p,0}} = \Lang{\underlyingGraph{\A},\macc{p,1}}$ if and only if there is a $\Weak{p}$ condition $\macc{W}$ such that $\Lang{\underlyingGraph{\A},\macc{W}} =\Lang{\underlyingGraph{\A},\macc{p,0}} = \Lang{\underlyingGraph{\A},\macc{p,1}}$.
\end{corollary}

\begin{proposition}
	Let $\A$ be a "deterministic" "Muller" "automaton", and assume that $\Lang{\A}$ can be "recognised" by a "deterministic" "B\"uchi" (resp. "coB\"uchi") "automaton"; that is, the "parity index" of $\Lang{\A}$ is at most $[0,1]$ (resp. at most $[1,2]$). Then, $\A$ is "generalised B\"uchi type" (resp. "generalised coB\"uchi type").
\end{proposition}
\begin{proof}
	We prove the result for the case generalised B\"uchi (analogous for coB\"uchi). We can assume that all the states of $\A$ are "accessible", as we can define a trivial acceptance condition in the part of $\A$ that is not accessible.  Since $\Lang{\A}$ has "parity index" at most $[0,1]$, the trees of the "ACD" of $\A$ have "height" at most $2$, and trees of "height" $2$ are "positive@@tree" (the "root" is a "round node@@acd"), by Proposition~\ref{prop-typ:parity index-Muller-automaton}, so it is a "generalised B\"uchi ACD", and by Proposition~\ref{prop-app:genBuchi-ACD type}, $\A$ is "generalised B\"uchi type".
\end{proof}

\paragraph*{Parity index from automata in normal form.}

\begin{corollary}\label{cor-app-weak:normal-form-index-par}
	Let $\A$ be a "deterministic" "parity" "automaton" in "normal form" using "colours" in $[0,d-1]$ (resp. $[1,d]$) such that all its states are "accessible". If $\A$ is $\Weak{d-1}$, then the "parity index" of $\Lang{\A}$ is $\WeakIndex{d-1}$. If not, the "parity index" of $\Lang{\A}$ is $[0,d-1]$ (resp. $[1,d]$).
\end{corollary}
\begin{proof}
	We assume that $\A$ uses colours in $[0,d-1]$ (in particular, it is "not negative@@TS").
	If $\A$ is not $\Weak{d-1}$, there is an "SCC" containing all the colours $[0,d-1]$.  By Proposition~\ref{prop-normF:flowers-in-normal_form}, such "SCC" contains a "positive@@flower" "$d$-flower", so by the Flower Lemma~\ref{lemma:flower-lemma}, the "parity index" of $\Lang{\A}$ is $[0,d-1]$.
	
	Suppose now that $\A$ is $\Weak{d-1}$. Let $\ell$ be a "cycle" of $\A$ in which the colour $d-1$ occurs. By Proposition~\ref{prop-normF:flowers-in-normal_form}, $\ell$ contains a "negative@@flower" "$(d-1)$-flower". As $\A$ is "not negative@@TS", it also contains a "positive@@flower" "$(d-1)$-flower". By the Flower Lemma~\ref{lemma:flower-lemma}, the "parity index" of $\Lang{\A}$ is $\WeakIndex{d-1}$.	
\end{proof}

\begin{corollary}\label{cor-app-weak:index-parity-type}
	Let $\A$ be a "deterministic" "parity" "automaton" such that all its states are "accessible" and the "parity index" of $\Lang{\A}$ is $[0,d-1]$ (resp. $[1,d]$ / $\WeakIndex{d}$). 
	Then, $\A$ is $[0,d-1]$ (resp. $[1,d]$ / $\WeakIndex{d}$)-"parity type". 
\end{corollary}
\begin{proof}
	We assume that $\Lang{\A}$ has "parity index" $[0,d-1]$ (the other cases are similar). By the Flower Lemma~\ref{lemma:flower-lemma}, $\A$ does not contain any "negative@@flower" "$d$-flower". 
	By Proposition~\ref{prop-typ:parity index-Muller-automaton}, the trees of the "ACD" of $\A$ have "height" at most $d$, and trees of "height" $d$ are "positive@@tree". That is, $\acd{\A}$ is a "$[0,d-1]$-parity ACD", and we conclude by applying Proposition~\ref{prop-typ:parity ACD type}.
\end{proof}

	\section{Transformations for games}
	\label{sec:appendix-games-transformations}
	\paragraph*{Games suitable for transformations.}
As we have indicated throughout the paper, defining transformations not preserving determinism in the case of "games" poses certain formal challenges. This difficulties appear both when  such transformations arise as the "product@@aut" $\G \compositionAut \A$ of a "game" $\G$ by an "non-deterministic" "automaton" $\A$, or when they are witnessed by an "HD mapping" $\pp\colon \G \to \G'$. The problem comes from the fact that the semantics of "non-determinism" in "automata" (or "history-determinism@@morph" of "morphisms") are inherently asymmetric,  and this asymmetry needs to be made compatible with the semantics of "games". 
The choices we have made to overcome this technical difficulty are:
\begin{itemize}
	\item Restrict transformations of "games" to "games" in a standard form, which we have called \emph{games "suitable for transformations"}.
	\item Add a restriction to "HD mappings" in the case of "games", introducing the notion of "HD-for-games" mapping. 
\end{itemize}

The main motivation for the standard form of "games" that we propose comes from viewing games as originating from logical formulas. Indeed, an equivalent model for "games" can be given as follows: vertices in the game graph are not partitioned into "Eve's" and "Adam's" nodes, instead, we assign a boolean formula to each transition that determines an interaction between the two players. The outcome of this interaction is (1) the next vertex, and (2) the "output colour" of the "acceptance condition". We can obtain a "game" of the kind we have defined in this paper by unfolding the boolean formulas of the transitions. There is a natural way to standardize such games: putting the boolean formulas in disjunctive normal form (DNF). Then, the unfolding of a game with formulas in DNF yields a "game" in which the partition into "Eve"-"Adam" nodes induces a bipartite graph with a particular structure: first, "Adam" chooses an "uncoloured" transition leading to a vertex controlled by "Eve" (with only one ingoing transition), and then "Eve" picks a transition producing some "output colour".

We recall that a "game" is "suitable for transformations" if it verifies that for every edge $e = v\re{}v'$, if $v$ is controlled by "Adam", then $e$ is "uncoloured" ($\gg(e)=\ee$), $v'\in \VEve$, and $e$ is the only incoming edge to $v'$ ($\mIn(v') = \{e\}$). 

Games in this form have an asymmetric structure that makes them suitable for any type of transformation. As any pair of consecutive transitions are of the form $v\re{\ee}\tilde{v}\re{c}v'$, with $\tilde{v}\in \VEve$, we can force it so that if a decision needs to be made in a product, Eve is the one who makes it.

\begin{lemma}
	For every game $\G$ with vertices $V$ and edges $E$, there exists a "game" $\widetilde{\G}$ that is "suitable for transformations", of size $|\widetilde{\G}| = \O(|E|)$, and equivalent to $\G$ in the following sense: there is an injective function $f\colon V\to \widetilde{V}$ such that "Eve" "wins" $\G$ from  $v$ if and only if she "wins" $\widetilde{\G}$ from  $f(v)$.
\end{lemma}
 \begin{proof}
 	We define $\widetilde{\G}$ as follows. We let its set of vertices be $\widetilde{V} = V \cup E$. Vertices of the form $v\in V$ will correspond to vertices coming from $\G$, and vertices $e\in E$ will be intermediate vertices added to force the "suitability for transformations" property. We let $"\widetilde {V}_\Adam" = \VAdam $ and $"\widetilde {V}_\Eve" = \VEve \cup E$. If $e=v\re{c}v'$ is an edge in $\G$, we add the edges  $v\re{\ee} e$ and $e\re{c} v'$ to $\G'$. 
 	It is clear that $\widetilde{\G}$ is "suitable for transformations" and that "Eve" "wins" $\G$ from  $v$ if and only if she "wins" $\widetilde{\G}$ from~$v$.
 \end{proof}
 
\paragraph*{ACD-HD-Rabin-transform-for-games.}
As discussed in Section~\ref{subsec-acd: HD-Rabin-transformation}, the "ACD-HD-Rabin-transform" of a "game" $\G$ does not always induce an "HD-for-games mapping" $\pp\colon \acdRabinTransform{\G} \to \G$, and $\G$ and $\acdRabinTransform{\G}$ do not necessarily have the same "winner". This is to be expected, as the "ACD-HD-Rabin-transform" does not take into account the partition into "Eve" and "Adam" nodes. In this paragraph we propose a small modification on the transformation to obtain a correct transformation for games.

\AP Let $\G$ be a "game". If there is an edge $v\re{}v'$ with $v\in\VAdam$, we say that $v'$ is an ""A-successor"". We remark that if $\G$ is "suitable for transformations", an "A-successor" is controlled by "Eve" and has a unique predecessor. 
We let $\intro*\VAsucc$ be the set of "A-successor" of $\G$ and $\intro*\Vnormal = V\setminus \VAsucc$. If $\G$ is "suitable for transformations", for each $v\in \VAsucc$ we let $\intro*\Apred(v)$ be its unique predecessor.

The idea to define the "ACD-HD-Rabin-transform-for-games" $\acdRabinTransformGFG{\G}$ is the following: starting from the regular "ACD-HD-Rabin-transform" $\acdRabinTransform{\G}$, we make some local changes to vertices that are "A-successors". First, if $v\in \VAdam$, we replace edges of the form $(v,x)\re{n} (v',x')$ in $\acdRabinTransform{\G}$ by $(v,x)\re{\ee} (v',x)$ (we forbid Adam to choose how to update the "ACD"-component). If such an edge is followed by $(v',x')\re{n'} (v'',x'')$ in $\acdRabinTransform{\G}$, then we add $(v',x)\re{n} (v'',x'')$ to $\acdRabinTransformGFG{\G}$ (we note that $v'\in \VEve$). That is, Eve chooses retroactively how to update the "ACD"-components performing two consecutive updates. We note that the node $n'$ is not output in the new "game"; this is not a problem, since $n$ must be an "ancestor" of $n'$ (we could say that $n$ contains more information regarding the "acceptance condition").

\begin{remark}\label{rmk-app:bounded-edges}
	Let $\G$ be a "game" "suitable for transformations", let $v\in \VAdam$ and $v\re{e_1}v'\re{e_2}v''$ be a path of size $2$ in $\G$ from $v$. It holds:
	\begin{itemize}
		\item If some "cycle" $\ell$ contains $e_2$, it also contains $e_1$.
		\item $\treeVertex{v'}$ is a "subtree" of $\treeVertex{v}$.
		\item Let $n_1\in \leaves(\treeVertex{v})$ and $n_2 = \jump_{\treeVertex{v'}}(n_1, \suppAcd(n_1,e_1))$. Then, $\suppAcd(n_2,e_2))$ is a "descendant" of $\suppAcd(n_1,e_1))$ in $\treeVertex{v'}$.
	\end{itemize}  
\end{remark}
 
\begin{definition}[ACD-HD-Rabin-transform-for-games]\label{def-ACD:RabinGFGTransformationACD}
	Let $\G$ be a "Muller" "game" "suitable for transformations". 
 	For each vertex $v\in V$ we let $\eta_v \colon \leaves(\treeVertex{v}) \to \{1,\dots, \memTree{\treeVertex{v}}\}$ be a mapping satisfying Property~\eqref{eq:property-star} from Lemma~\ref{lemma:property_star}.	
 	\AP We define the ""ACD-HD-Rabin-transform-for-games"" of~$\G$ to be the "Rabin" "transition system" $\intro*\acdRabinTransformGFG{\G}$ 
 	defined as follows.
 	\begin{description}
 		\setlength\itemsep{2mm}
 		\item[Vertices.] The set of vertices is \[\widetilde{V} = \bigcup_{v\in \Vnormal}  \{v\} \times [1,\memTree{\treeVertex{v}}]  \;\cup\, \bigcup_{v\in \VAsucc}  \{v\} \times [1,\memTree{\treeVertex{\Apred(v)}}] .\]
 		
 		\item[Players partition.] A vertex $(v,x)$ belongs to "Eve" if and only if $v$ belongs to "Eve" in $\G$.
 		
 		\item[Initial vertices.] $\widetilde{I} = \{ (v_0, x) \mid v_0\in I \tand x \in \{1,\dots, \memTree{\treeVertex{v_0}}\} \}$.
 		
 		\item[Edges and output colours.] 
 		Let $e = v\re{} v'$ in $\G$.
 		\begin{itemize}
 			\item 	If $v\in \VEve\cap \Vnormal$, we add $(v, x) \re{n} (v',x')$ to $\acdRabinTransformGFG{\G}$ exactly in the same cases as in the regular "ACD-HD-Rabin-transform".
 			\item If $v\in \VAdam$, we let $(v,x)\re{\ee}(v',x)$ in $\widetilde{E}$ for each $x\in \{1,\dots, \memTree{\treeVertex{v}}\}$. 
 			\item If $v\in \VAsucc$, we add $(v, x) \re{n} (v',x')$ to $\acdRabinTransformGFG{\G}$ if in the regular "ACD-HD-Rabin-transform" there is a path of size $2$ of the form 
 			\[ (\Apred(v),x)\re{n} (v, \tilde{x}) \re{\tilde{n}} (v',x').  \]
 		\end{itemize}
 			Formally, \[\widetilde{E} = \bigcup_{\substack{e=v\re{}v'\in E\\ v\in \Vnormal}}  \{e\} \times \leaves(\treeVertex{v}) \;\cup\, \bigcup_{\substack{e=v\re{}v'\in E\\v\in \VAsucc}}  \{e\} \times \leaves(\treeVertex{\Apred(v)}).\]
 		
 		\item[Rabin condition.] $R = \{(G_n, R_n)\}_{n\in \nodesAcdRound{\TS}}$, where $G_n$ and $R_n$ are defined as follows: Let $n$ be a "round@@acd" node, and let $n'$ be any node in $\nodesAcd{\TS}$,
 		\begin{equation*}
 			\begin{cases}
 				n' \in G_n & \text{ if } n'=n,\\
 				n' \in R_n & \text{ if }  n'\neq n \text{ and } n \text{ is not an "ancestor" of } n'. 
 			\end{cases}
 		\end{equation*} 	
 	
 	\end{description}
 \end{definition}

\subparagraph*{Correctness of the ACD-HD-Rabin-transform-for-games}
\correctnessAcdRabinTransformGames*
 
\begin{proof}
	The proof is analogous to that of the correctness of the usual "ACD-HD-Rabin-transform" (Proposition~\ref{prop-ACD:correctness_ACD-HD-Rabin-transform}). We define the mapping $\pp\colon \acdRabinTransformGFG{\G} \to \G$ as $\pp_V(v,x) = v$ and $\pp_E(e,l) = e$. It is clear that it is a "weak morphism", and it "preserve accepting runs" by Lemma~\ref{lemma-acd:Rabin-acc-sequences-nodes-ACD} and Remark~\ref{rmk-app:bounded-edges}.
	
	We define a "resolver@@morph" $(r_0,r)$ simulating $\pp$ similarly to the proof of Proposition~\ref{prop-ACD:correctness_ACD-HD-Rabin-transform}: We use $\acdParityTransform{\G}$ to guide the resolver. Let $\rr = v_0\re{}v_1re{}\dots$ be a run in $\G$, and let $(v_0,l_0)\re{}(v_1,l_1)\re{}\dots$ be the preimage of this run in $\acdParityTransform{\G}$. We simulate $\rr$ in $\acdRabinTransformGFG{\G}$ as follows: We ensure that at every moment $i$, if $v_i\notin \VAsucc$, the current vertex $(v_i,x_i)$ is such that $x_i = \eta_{v_i}(l_i)$. There distinguish two cases to simulate the edge $e_i = v_i\re{}v_{i+1}$:
	\begin{itemize}
		\item If $v_i\in \VAdam$, there is a single outgoing edge from $(v_i,x_i)$ mapped to the edge $v_i\re{}v_{i+1}$ in $\rr$: $(v_i,x_i)\re{\ee}(v_{i+1},x_i)$. This must be the edge picked by the "resolver@@mapping"
		\item If $v_i\in \VEve$, we pick  the edge $(v_i,x_i)\re{n_i}(v_{i+1},x_{x+1})$ such that $x_{i+1} = \eta_{l_{i+1}}$ and $n_i = \suppAcd(l_i, e_i)$.
	\end{itemize} 
	If $v_i\in \VAsucc$, the vertex $(v_i,x_i)$ will verify $x_i = x_{i-1} = \eta_{v_i}(l_i)$. 
	In this case, we pick  the edge $(v_i,x_i)\re{n_i}(v_{i+1},x_{x+1})$ such that $x_{i+1} = \eta_{l_{i+1}}$  and $n_i = \suppAcd(l_i, e_i)$. This is indeed an edge appearing in $\acdRabinTransformGFG{\G}$, as the path $(v_{i-1},x_{i-1})\re{n_{i-1}'}(v_i,x_i')\re{n_i}(v_{i+1},x_{x+1})$ exists in the regular $\acdRabinTransform{\G}$, with $x_i' = \eta_{v_i}(l_i)$.
	
	The "resolver@@morph" obtained in this way is "sound for $\acdRabinTransformGFG{\G}$", as there is a unique way to simulate edges issued from "Adam" vertices, and the rest of the edges are simulated in the same way as the resolver defined for the regular "ACD-HD-Rabin-transform", which we proved to be "sound@@morph".
\end{proof}
 
\subparagraph*{Optimality of the ACD-HD-Rabin-transform-for-games}
\optimalityAcdRabinTransformGames*
\begin{proof}
	The vertices of $\acdRabinTransformGFG{\G}$ corresponding to vertices in $\Vnormal$ are exactly the same that those in $\acdRabinTransform{\G}$: 
	\[ \{(v,x)\in \acdRabinTransformGFG{\G} \mid v\in \Vnormal\} = \{(v,x)\in \acdRabinTransform{\G} \mid v\in \Vnormal\}. \]
	Moreover, for $v\in\VAsucc$, there is one vertex of the form $(v,x)$ for each vertex $(\Apred(v),x)$, and each $v\in\VAsucc$ has exactly one predecessor in $\Vnormal$, so we conclude that:
	\[ |\acdRabinTransformGFG{\G}| \leq 2\cdot |\{(v,x)\in \acdRabinTransformGFG{\G} \mid v\in \Vnormal\}| \leq 2\cdot|\acdRabinTransform{\G}| \leq 2\cdot\G', \]
	where the last inequality follows from Theorem~\ref{thm-acd:optimality_ACD-HD-Rabin-transform}.
\end{proof} 	
	\section{Simplifications for prefix-independent conditions}
	\label{sec:appendix-prefix-independent}

We prove in this appendix results applying to "automata" "recognising" "prefix-independent" languages and "games" using "prefix-independent" winning conditions. We recall that a language $L\subseteq \SS^\oo$ is \emph{prefix-independent} if for all $w\in \SS^\oo$ and $u\in \SS^*$,  $uw\in L$ if and only if $w\in L$.

\prefixIndepAutomaton*
\begin{proof}
	Let $(r_0, r)$ be a "sound@@aut" "resolver@@aut" for $\A$ such that $q$ is "reachable using@@resolver" $(r_0, r)$ (there is a word $w_0\in \SS^*$ and a "run" $\rr_0 = r_0 \lrp{w_0} q$ "induced by $r$@@aut" over $w$). 
	We first show that  $\Lang{\initialTS{\A}{q}} \subseteq \Lang{\A}$. Let $w\in \SS^\oo$ be a word accepted from $q$. Then, $w_0w$ admits an "accepting run" from the original "initial state" (by "prefix-independence" of the "acceptance set"), so $w_0w\in \Lang{\A}$, and by the "prefix-independence" of $\Lang{\A}$, $w\in \Lang{\A}$ too.
		
	For the converse direction, we define a "sound@@aut" "resolver@@aut" $(r_0', r')$ for $\initialTS{\A}{q}$. We let $r_0' = q$, and $r'(\rr,a) = r(\rr_0\rr,a)$ be the strategy that acts as the resolver $r$ assuming that $\rr_0$ has happened in the past. It is clear that for every word $w\in \SS^\oo$, the "run induced by@@morph" $(r_0', r')$ over $w$ has a common suffix with the "run induced by@@morph" $(r_0, r)$ over $w_0w$. Therefore, by the "prefix-independence" assumptions:
	\[ w\in \Lang{\A} \iff w_0w \text{ is accepted using } (r_0, r) \iff w_0w \text{ is accepted using }(r_0', r').\qedhere \]
\end{proof}

\HDmappingsPrefIndep*
\begin{proof}
	First, $\pp\colon \initialTS{\TS}{V} \to \initialTS{\TS'}{V'}$ is trivially a "weak morphism". We claim that it "preserves accepting runs". Let $v\in V$ be a state in $\TS$ and let $\rr = v \lrp{w}$ be an "accepting run" from $v$. Since all the states are "reachable", there is some $v_0\in I$ and "finite run" $\rr_v = v_0 \rp{u} v$. Since $\pp$ is a "weak morphism" we have that $\ppRuns(\rr_v\rr) = \pp(v_0) \rp{u'} \pp(v) \lrp{w'}$. It holds that:
	\[ w \in \WW \overset{\begin{subarray}{c}
			\WW \\
			\text{ pref-indep. }
		\end{subarray}
	}\implies uw\in \WW \implies u'w'\in \WW' \overset{\begin{subarray}{c}
			\WW' \\
			\text{ pref-indep. }
		\end{subarray}
	}\implies w'\in \WW', \]
	where the central implication follows from the fact that $\pp$ "preserves accepting runs" between $\TS$ and $\TS'$. Therefore $\ppRuns(\rr)$ is also an "accepting run".
	
	In the rest of the proof we assume that $\pp$ is an "HD-for-games mapping", (which covers the "HD@@mapping" case).
	Let $(\rInit, r)$ be a "resolver@@mapping" "sound for $\TS$" "simulating"  $\pp\colon \TS \to \TS'$. We define a "resolver@@HDmapping" $(\trInit, \tilde{r})$ for the new mapping. 
	For every state $v$ of $\TS$, we fix a "finite run" $\rr_{v}\in \PathSetFin{\TS}{I}$ ending in $v$ that is "consistent with@@resolver" $(\rInit,r)$ over some $\rr'$, if such a run exists.  We let $V_{\mathrm{Reach}}\subseteq V$ be the set of vertices for which $\rr_{v}$ is well-defined. We note that for each $v'\in V'$ there exists at least one $v\in \inv{\pp}(v')$ such that $v\in V_{\mathrm{Reach}}$; indeed, if $\rr'_{v'}$ is a finite run reaching $v'$ in $\TS'$, one such~$v$ is $\mtargetPath(\rRuns(\rr_{v'}'))$ (that is, the vertex to which we arrive in $\TS$ when simulating $\rr_{v'}'$ via the original "resolver@@HDmapping"). We let $\trInit(v')$ be this vertex. 
	If $e'\in \mout(v')$ is an edge in $\TS'$, we let $\tilde{r}(\ee, e')= r(\rr_v,e')$, for $v=\trInit(v')$.
	For $\rr$ a non-empty "finite run" starting in $v\in V_{\mathrm{Reach}}$ and $e'\in E'$, we define $\tilde{r}(\rr, e')= r(\rr_{v}\rr, e')$. If $\rr$ starts in $v\notin V_{\mathrm{Reach}}$ we let $\tilde{r}(\rr, e')$ be any edge in $\inv{\pp}(e')$ (if $e'\in \mout(\pp(\mtarget(\rr)))$ we pick it in $\mout(\mtarget(\rr))$). We check that $(\trInit, \tilde{r})$ satisfies the four requirements to be a "resolver":
	\begin{enumerate}
		\item $\trInit(v')$ has been chosen in $\inv{\pp}(v')$.
		\item $\tilde{r}(e')$ is chosen in $\inv{\pp}(e')$.
		
		\item Let $e'\in \mout(v')$. We have defined $\tilde{r}(\ee, e') = r(\rr_v, e')$, where $\rr_v$ is a finite run "consistent with@@mapping" $r$ ending in $v=\rInit(v')$. By Property~\ref{item-HD-map:building-run} of a "resolver@@mapping", $r(\rr_v, e') \in \mout(v)$.
		
		\item Let $\rr = v_0 \lrp{\phantom{ww}} v \in \RunsFin{\initialTS{\TS}{V}}$ and $e'\in \mout(\pp(v))$. 
		If $v_0\notin V_{\mathrm{Reach}}$, then we have picked $\tilde{r}(\rr, e')$ in $\mout(v)$.
		If $v_0\in V_{\mathrm{Reach}}$, then $\tilde{r}(\rr, e') = r(\rr_{v_0}\rr, e')$; as $\rr_{v_0}\rr$ is a "run" ending in $v$ and $r$ verifies Property~\ref{item-HD-map:building-run} of a "resolver@@mapping" $r(\rr_{v_0}\rr, e')\in \mout(v)$.
	\end{enumerate} 
	
	Finally, we show that $(\trInit, \tilde{r})$ is "sound for $\TS$". Let  $\rr' = v'\lrp{w'} \,\in \Runs{\initialTS{\TS'}{V'}}$ be an "accepting run",  and let $\rr = v\lrp{w'}$ be a "run" "consistent with@@resolver" $(\trInit, \tilde{r})$ over $\rr'$. In particular, $v = \trInit(v')$. Let $\rr_{v} = v_0\lrp{u}v$ be the chosen run reaching $v$ and let $\rr_{v}' = v_0'\lrp{u'}v'$ be a finite run in $\RunsFin{\TS'}$ such that $\rr_v$ is "consistent with@@resolver" $(\rInit,r)$ over $\rr_{v}'$. It is immediate to check that $\rr_v\rr$ is  "consistent with@@resolver" $(\rInit,r)$ over $\rr_{v}'\rr'$. Since $\rr'$ is "accepting", we have that $w'\in \WW'$, and by "prefix-independence" of the "acceptance sets" and the fact that $\rr_v\rr$ is "accepting" if $\rr_{v}'\rr'$ is, we have:
	\[ w' \in \WW' \implies u'w'\in \WW' \implies uw\in \WW \implies w\in \WW,
	\] 
	so we conclude that $\rr$ is "accepting@@run" in $\TS$, as we wanted to show. 
\end{proof} 	
	\section{Simplifying automata with duplicated edges}
	\label{sec:appendix-simplifications-automata}
	\AP Given an automaton $\A=(Q, \Sigma, I, \GG, \transAut{}, \WW)$ we say that it has ""duplicated edges""  if there is some pair of states $q,q'\in Q$ and two different transitions between them labelled with the same input letter: $q\re{a:\aa}q'$, $q\re{a:\bb}q'$.

As commented in Remark~\ref{rmk-ZT:simplified-automata}, the construction of the "ZT-HD-Rabin-automaton" we have presented potentially introduces "duplicated edges", which can be seen as an undesirable property (even if some automata models such as the HOA format~\cite{HOAFormat2015} allow them). We show next that we can always derive an equivalent automaton without "duplicated edges". Intuitively, in the Rabin case, if we want to merge two transitions having as output letters $\aa$ and $\bb$, we add a fresh letter $(\aa\bb)$ to label the new transition. For each "Rabin pair", this new letter will simulate the best of either $\aa$ or $\bb$ depending upon the situation.

\begin{proposition}[Simplification of automata]\label{prop-app:simplification_Rabin}
	Let $\A$ be a "Muller" (resp. "Rabin") "automaton" presenting "duplicated edges". There exists a "Muller" (resp. "Rabin") automaton $\A'$ on the same set of states without "duplicated edges" such that $\Lang{\A} = \Lang{\A'}$. Moreover, if $\A$ is "history-deterministic",~$\A'$ can be chosen "history-deterministic". In the Rabin case, the number of "Rabin pairs" is also preserved. 
\end{proposition}
\begin{proof}
	For the Rabin case, let $\A'$ be an automaton that is otherwise as $\A$ except that instead of the transitions $\Delta$  of $\A$ it only has one $a$-transition $q \xrightarrow{a:x} q'\in \Delta'$ (with a fresh colour $x$ per transition) per state-pair $q,q'$ and letter $a\in \Sigma$. That is, $\DD'=\{(q,a,x_j,q') \: : \: (q,a,y,q')\in \DD \text{ for some } y\}$. The new "Rabin condition" $\{(G_1',R_1'),\ldots, (G_r',R_r')\}$ is defined as follows. For each transition $q\xrightarrow{a:x} q'$:
	\begin{itemize}
		\item $x\in G_i'$ if $q\xrightarrow{a:y} q'\in \Delta$ for some $y\in G_i$, 
		\item  $x\in R'_i$ if for all $q\xrightarrow{a:y} q'\in \Delta$,  $y\in R_i$. 
	\end{itemize}
	
	We claim that $\Lang{\A'}=\Lang{\A}$. Indeed, if $u\in \Lang{\A}$, as witnessed by some "run" $\rho$ and a "Rabin pair" $(G_i,R_i)$, then the corresponding run $\rho'$ in $\A'$ over $u$ is also "accepted by@@RabinPair" the "Rabin pair" $(G_i',R_i')$: the transitions of $\minf(\rho)\cap G_i$ induce transitions of $\minf(\rho')\cap G_i'$ and the fact that $\minf(\rho)\cap R_i=\emptyset$ guarantees that $\minf(\rho')\cap R_i'=\emptyset$.
	
	Conversely, if $u\in \Lang{\A'}$ as witnessed by a run $\rho'$ and "Rabin pair" $(G_i', R_i')$,
	then there is an accepting run $\rho$ over $u$ in $\A$: such a run can be obtained by choosing for each transition $q\xrightarrow{a:x}q'$ of $\rho'$ where $x\in G_i'$ a transition $q\xrightarrow{a:y}q'\in \Delta$ such that $y\in G_i$, which exists by definition of $\A'$, for each transition $q\xrightarrow{a:x}q'$ where $x\notin G_i\cup R_i$ a transition $q\xrightarrow{q,y}q'\in\Delta$ such that $y\notin R_i$, which also exists by definition of $\A'$, and for other transitions $q\xrightarrow{a:x}q'$ (that is, those for which $x\in R_i'$) an arbitrary transition $q\xrightarrow{a:y}q'\in \Delta$. Since $\rho'$ is accepting, we have $\minf(\rho')\cap G_i\neq \emptyset$ and $\minf(\rho)\cap R_i=\emptyset$, that is, $\rho$ is also accepting. 
	
	For the "Muller" case, the argument is even simpler. As above, we consider $\A'$ that is otherwise like $\A$ except that instead of the transitions $\Delta$ of $\A$, it only has one $a$-transition $q\xrightarrow{a:x}q'\in \Delta'$ (with a fresh colour per transition) per state-pair $q,q'$ and the accepting condition is defined as follows. A set of  transitions $T$ is accepting if and only if for each $t=q\xrightarrow{a:x}q'\in T$ there is a non-empty set $S_t \subseteq \{ q\xrightarrow{a:y}q'\in \Delta\}$ such that $\bigcup_{t\in T} S_t$ is accepting in $\A$. In other words, a set of transitions in $\A'$ is accepting if for each transition we can choose a non-empty subset of the original transitions in $\A$ that form an accepting run in $\A$.
	
	We claim that $\Lang{\A'}=\Lang{\A}$. Indeed if $u\in \Lang{\A}$, as witnessed by some run $\rho$, the run~$\rho'$ that visits the same sequence of states in $\A'$ is accepting as witnessed by the transitions that occur infinitely often in $\rho$.
	
	Conversely, assume $u\in \Lang{\A'}$, as witnessed by a run $\rho'$ and a non-empty subset $S_t$ for each transition $t$ that occurs infinitely often in $\rho'$ such that $\bigcup_{t\in \minf(\rho)} S_t$ is accepting in $\A$. Then there is an accepting run $\rho$ over $u$ in $\A$ that visits the same sequence of states as $\rho'$ and chooses instead of a transition $t\in \minf(\rho)$ each transition in $S_t$ infinitely often, and otherwise takes an arbitrary transition. The set of transitions $\rho$ visits infinitely often is exactly $\bigcup_{t\in\minf(\rho)} S_t$, and is therefore accepting.
	
	Finally, observe that in both cases, if $\A$ is "HD", then the automaton $\A'$ without duplicate edges is also "HD" since $\A'$ is obtained from $\A$ by merging transitions. Indeed, the "resolver" $r$ of~$\A$ induces a "resolver" $r'$ for $\A'$ by outputting the unique transition with the same letter and state-pair as $r$. By the same argument as above, the run induced by $r'$ is accepting if and only if the run induced by $r$ is.
\end{proof}

\begin{example}
	The "ZT-HD-Rabin-automaton" from Figure~\ref{fig-ZT:zielonka-Rabin automaton} has "duplicated transitions". In Figure~\ref{fig-app:simplified-Rabin} we present an equivalent "HD" "Rabin" "automaton" without duplicates. For this, we have merged the self-loops in state $1$ labelled with $a$ and $b$ respectively. We have added the "output colours" $(\aa\bb)$ and $(\theta\xi)$. 
	The new "Rabin pairs" are given by:
	
	\centering
	\begin{tabular}{l l}
		$G_\bb' = \{\bb, (\aa\bb)\}$, & $R_\bb'= \{\aa, \lambda, \xi, \zeta\}$,\\ 
		$G_\lambda' = \{\lambda\}$, & $R_\lambda'= \{\aa, \bb, (\aa\bb), \theta\}$. 
	\end{tabular}
	\begin{figure}[ht]
		\centering		
		\begin{tikzpicture}[square/.style={regular polygon,regular polygon sides=4}, align=center,node distance=2cm,inner sep=2pt]
			
			\node at (0,2) [state] (1) {$1$};
			\node at (3,2) [state] (2) {$2$};
			
			\path[->] 
			(1)  edge [out=250,in=290,loop] node[left] {$a:(\theta\xi)$ }   (1)
			(1)  edge [in=160,out=200,loop] node[left] {$b:(\aa\bb)$ }   (1)
			(1)  edge [in=70,out=110,loop] node[above] {$c:\aa$ }   (1)
			
			(1)  edge [in=210,out=-30] node[below] {$c:\lambda$ }   (2)
			
			(2)  edge [color=black] node[above] {$a:\lambda$ }   (1)
			(2)  edge [in=30,out=150, color=black] node[above] {$b:\aa$ }   (1)
			(2)  edge [in=-20,out=20,loop, color=black] node[right] {$c:\zeta$ }   (1);
			
		\end{tikzpicture}
		\caption{ The simplified "ZT-HD-Rabin-automaton".}
		\label{fig-app:simplified-Rabin}
		
	\end{figure}
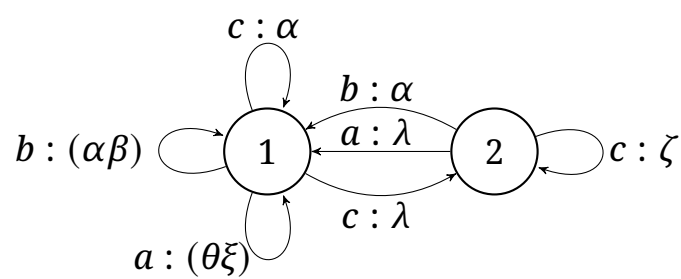
	
\end{example} 	
	
	\section{Proofs for Section~\ref{subsubsec-typeness:Typeness_Conditions}}
	\label{sec:appendix-typeness-proofs-Muller-conditions}
	\begin{proof}[Proof of Propositions~\ref{prop-typ:RabinZielonkaShape} and~\ref{prop-typ:StreettZielonkaShape}]
		We prove it for the "Rabin" case, "Streett" conditions being the dual notion.
			
		If all "round nodes" of $\zielonkaTree{\F} = (N =\roundnodes\sqcup \squarenodes, \ancestor, \nu :N \to \powplus{\GG})$ have at most one "child", we define a family of "Rabin pairs" $R = \{(G_n,R_n)\mid n\in \roundnodes\}$ such that $\Rabin{R}= \Muller{\F}$ as follows: for each "round node" $n\in \roundnodes$, we add a "Rabin pair" $(G_n,R_n)$. We let $G_n= \Gamma \setminus \nu(n)$. In order to define $R_n$, we observe that $n$ has at most one child $n'$, and we define $R_n=\nu(n)\setminus \nu(n')$, for $n'$ the only "child" of $n$, if it exists, or $R_n=\nu(n)$ if $n$ has no "children" at all. This is, the pair $(G_n,R_n)$ "accepts@@RabinPair" the sets of colours $A\subseteq \Gamma$ that contain some of the colours that disappear in the step $n\rightarrow n'$ and none of the colours appearing up in the "tree". We show that $\Rabin{R}= \Muller{\F}$. Let $A$ be a set of colours. If $A\in \F$, let $n$ be a maximal node (for $\ancestor$) containing $A$. It is a "round node" and there is some colour $c\in A$ not appearing in the only child of $n$. Therefore, $c\in G_n$ and $A \cap R_n=\emptyset$. Conversely, if $A\notin \F$, then for every "round node" $n$ with a child $n'$, either $A\subseteq \nu(n')$ (and therefore $A\cap G_n = \emptyset$) or $A \nsubseteq \nu(n)$ (and in that case $A\cap R_n \neq \emptyset$).
			
		We remark that this construction uses more "Rabin pairs" than necessary, since we could reuse Rabin pairs for nodes that are in the same level and that are not siblings.
			
		Conversely, suppose that $\Muller{\F} = \Rabin{R}$ for the "Rabin language associated to" $R=\{(G_1,R_1),\dots,(G_r,R_r)\}$. If $n \in \zielonkaTree{\F}$ is a "round node" ($A=\nu(n)\in \F$), then its label $A$ contains some colours that belongs to $G_{i_1},\dots,G_{i_k}$ and none belonging to $R_{i_1},\dots,R_{i_k}$ for some $i_1,\dots, i_k$, $k\geq 1$. A child of $n$ must not have these colours, so the only maximal subset of $A$ that is not in $\F$ is $A\setminus (G_{i_1}\cup \dots \cup G_{i_k})$.				
\end{proof}

\begin{proof}[Proof of Proposition~\ref{prop-typ:parityZielonkaShape}]
	We assume $\GG\in \F$ ($\minparityZ{\F}=0$), the other case is symmetric.
			
	Assume that $\zielonkaTree{\F}$ has a single branch of length $\maxparityZ{\F}+1$. We define a mapping $\phi\colon \GG \to [0, \maxparityZ{\F}]$ as  follows: for each colour $c\in \GG$ we let $n_c$ be the "deepest" node in $\zielonkaTree{\F}$ containing $c$, and we define $\phi(c) = \nu(n_c)$. It is easy to check that for all $w\in\GG^\oo$, $w\in \Muller{\F}$ if and only if $\phi(w)\in \parity$.

	Conversely, assume that we can assign colours to the elements of $\Gamma$ by $\phi:\Gamma \rightarrow [0, d]$, whose corresponding parity language is $\Muller{\F}$. We show that any node of the "Zielonka tree" $\zielonkaTree{\F}$ has at most one child. Indeed, let $n \in N$ and let $c\in \nu(n)$ such that $\phi(c) = \min\{\phi(c) \mid c\in \nu(n)\}$. We suppose that $\phi(c)$ is odd (the proof is symmetric for $\phi(c)$ even). Let $p = \min\{\phi(c') \mid c'\in \nu(n) \tand \phi(c') \text{ even}\}$. In every child of $n$ the elements with a smaller colour than $p$ must disappear, so the set of elements $\nu(n)\cap \{c\in \GG \mid \phi(c) \geq p \}$ is the only maximal subset of $\nu(n)$ belonging to $\F$. Moreover, in the label of the child of $n$ there is at least one colour less, so the "height" of $\zielonkaTree{\F}$ will be at most $d+1$.
\end{proof}

\end{document}